\documentclass[12pt,english,final]{article}
\usepackage{iftex}
\usepackage{mathtools}

\ifXeTeX
  \usepackage{fourier-otf}
\else

  \usepackage{fourier}
  \usepackage[T1]{fontenc}
  \usepackage[latin9]{inputenc}
\fi

\usepackage{cprotect}

\usepackage{url}
\usepackage{amsmath}
\usepackage{amsthm}
\ifXeTeX\else
  \usepackage{amssymb}
\fi
\usepackage{graphicx}
\usepackage{xcolor}
\usepackage[letterpaper]{geometry}
\usepackage{setspace}
\usepackage[authoryear,longnamesfirst]{natbib}
\usepackage{bibunits}
\makeatletter

\theoremstyle{definition}

\theoremstyle{plain}

\theoremstyle{plain}

\usepackage{subcaption}
\usepackage{diagbox}
\usepackage{appendix}
\usepackage{fancyvrb}
\usepackage{booktabs}

\ifXeTeX\else
  \usepackage{epstopdf}
\fi
\usepackage[font=normalsize]{caption}
\usepackage{enumerate}
\usepackage{mathrsfs}
\usepackage{bm}
\usepackage{pdfpages}
\usepackage{float}
\usepackage{afterpage}
\usepackage{pdflscape}
\usepackage{array}
\usepackage{multirow}
\usepackage{longtable}
\usepackage{makecell}
\usepackage{algorithm}
\usepackage{algpseudocode}
\usepackage{tikz}
\usetikzlibrary{decorations.pathreplacing}
\usepackage{titletoc}

\definecolor{darkblue}{rgb}{0,0,.6}
\definecolor{darkred}{rgb}{.6,0,0}
\definecolor{darkorange}{RGB}{204,85,0}

\usepackage[
 unicode=true,
 bookmarks=false,
 breaklinks=false,
 pdfborder={0 0 1},
 colorlinks=true,
 linkcolor=darkred,
 citecolor=darkblue,
 urlcolor=darkblue
]{hyperref}

\def\munderbar#1{\underline{\sbox\tw@{$#1$}\dp\tw@\z@\box\tw@}}

\newcommand{\commentout}[1]{}

\renewcommand{\paragraph}{\@startsection{paragraph}{4}{\z@}
  {1.5ex \@plus1ex \@minus.2ex}
  {-1em}
  {\normalfont\normalsize\bfseries}}

\theoremstyle{plain}
\providecommand{\assumptionname}{Assumption}
\providecommand{\observationname}{Observation}
\newtheorem*{observation*}{\protect\observationname}\providecommand{\observationname}{Observation}
\newtheorem{appassumption}{Assumption}[section]
\newtheorem{theorem}{Theorem}

\newcommand{\sym}[1]{\ifmmode^{#1}\else$^{#1}$\fi}

\makeatother

\usepackage{babel}
\providecommand{\definitionname}{Definition}
\providecommand{\lemmaname}{Lemma}
\providecommand{\propositionname}{Proposition}
\providecommand{\corollaryname}{Corollary}

\defaultbibliographystyle{aea}
\defaultbibliography{Inflation_Forecasting_refs}

\begin{document}
\begin{bibunit}
\title{\textbf{Forecasting Inflation with Microdata:\\An Adaptive Machine Learning Approach}\thanks{Chen: Stanford University; \url{cyc2152@stanford.edu}. Gao: Peking University; \url{chengao0716@gmail.com}. Hazell: London School of Economics; \url{j.hazell@lse.ac.uk}. Lei: Stanford Graduate School of Business; \url{lihualei@stanford.edu}. Lian: UC Berkeley and NBER; \url{chen_lian@berkeley.edu}. We thank Nicholas Tokay and Sven van Holten Charria for outstanding research assistance.}}
\author{\textbf{Catherine Chen} \qquad{}\textbf{Chen Gao} \qquad{}\textbf{Jonathon Hazell} \qquad{}\textbf{Lihua Lei} \qquad{}\textbf{Chen Lian}}
\date{\today}
\pagenumbering{gobble}
\maketitle
\begin{center}
Preliminary.
\end{center}

\begin{abstract}
  \noindent
Does microeconomic heterogeneity help to forecast aggregate inflation in a non-stationary environment? We develop a scan test for whether one forecast outperforms another, over an interval with unknown starting point and duration. To exploit any occasional forecasting power that the scan test detects, we design an adaptive machine learning pipeline. We encode the distribution of price changes into a high-dimensional vector, which we combine with a gradient boosted trees algorithm. We then combine this micro forecast with other benchmark forecasts, using an adaptive algorithm that makes use of the micro forecast only when it performs well. We apply the pipeline to UK microdata, with four main results. First, the micro forecast outperforms a univariate benchmark, but only in the volatile period after 2020. Second, the scan test detects periods of micro outperformance, so the micro forecast enters the combined forecast. Third, the combined forecast performs comparably to the univariate benchmark before 2020 and better at every horizon after 2020. Fourth, the value of microdata for the combined forecast materializes after 2020. We conclude that microdata are valuable for forecasting aggregate inflation, but only after large shocks.\end{abstract}
\thispagestyle{empty} \setlength{\abovedisplayskip}{.5em} \setlength{\belowdisplayskip}{.5em}

\clearpage{}

\pagenumbering{arabic}

\section{Introduction}\label{sec:introduction}

How can we forecast the economy? Traditional forecasting approaches rely on macroeconomic data. Yet aggregate outcomes are generated by millions of microeconomic decisions about price setting, employment, consumption, and production. In recent decades, high-dimensional micro datasets tracking this behavior have become available. These data create an opportunity to use micro-level information to improve forecasts of aggregate variables. In taking advantage of this opportunity, we can answer a fundamental question in macroeconomics: does microeconomic heterogeneity matter for aggregate dynamics---and if so, when?

This paper focuses on inflation, one of the most difficult macroeconomic variables to forecast. Even sophisticated models often struggle to outperform simple univariate benchmarks \citep{AtkesonOhanian2001}. A central challenge is non-stationarity: the process governing inflation dynamics changes over time. Figure~\ref{fig:us_uk_cpi} illustrates this clearly for the United Kingdom and the United States over 2015--2024, with markedly different inflation behavior before and after 2020.

\afterpage{
\begin{figure}[t!]
  \centering
  \caption{Headline CPI Inflation in the United States and United Kingdom, 2015--2024}
  \label{fig:us_uk_cpi}
  \includegraphics[width=\textwidth]{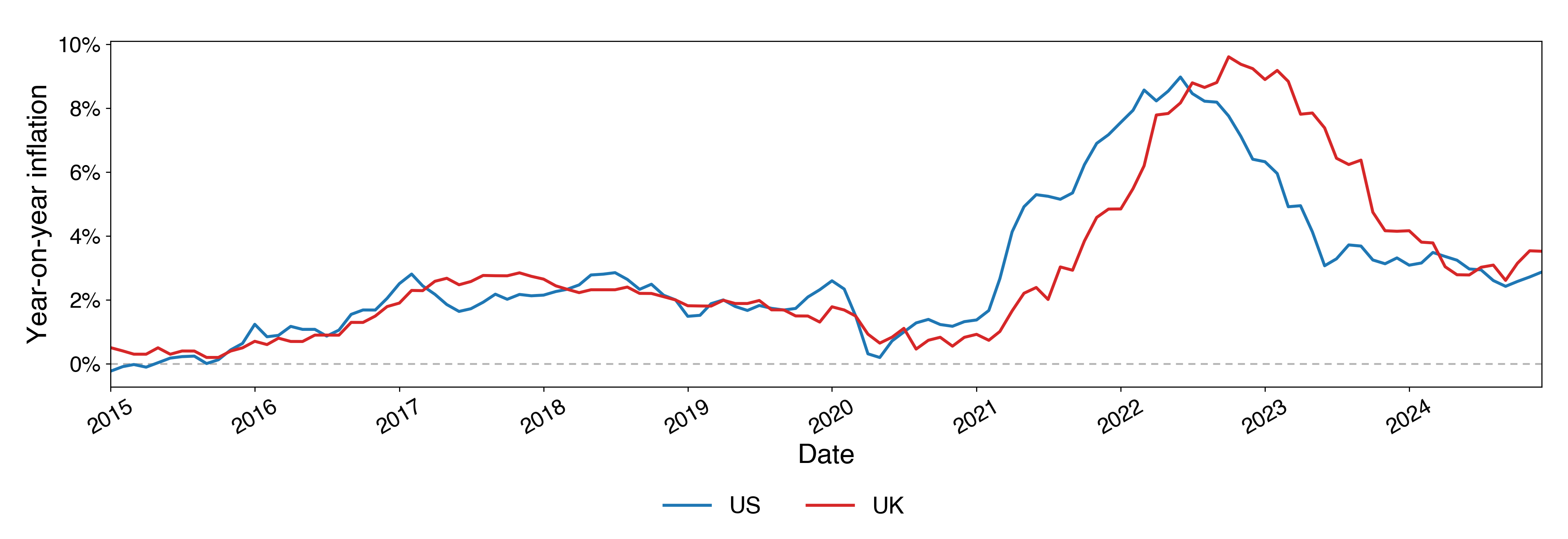}
  \smallskip
  \begin{minipage}{\textwidth}\setstretch{1.0}
    \footnotesize\textit{Notes:} Year-on-year headline CPI inflation for the United States (blue) and United Kingdom (red), January 2015 to December 2024. Data source: Federal Reserve Economic Data (FRED).
  \end{minipage}
\end{figure}
}

There is good reason to think that microdata---in our case, the millions of price changes observed each month---may improve inflation forecasts. Each price change reflects a forward-looking decision by an individual firm, and the cross section of such decisions may contain information about future aggregate inflation. In simple macroeconomic models, such as the linearized Phillips Curve with Calvo price setting, microdata does not matter for aggregate dynamics (e.g. \citealp{Gali2015}). Richer macroeconomic models contain more nuance: microdata should help forecast inflation in some environments but not others. When aggregate shocks are small, microdata do not add forecasting power conditional on macroeconomic aggregates in state-dependent models of price setting \citep{AuclertRogatoRognlieStraub2024}. When aggregate shocks are large, however, micro information can help predict aggregate inflation \citep{Blanco2021}. That is, microdata matters \textit{only occasionally} for aggregate dynamics.

In this work, we introduce a method to test for and exploit the power of high dimensional microdata for forecasting inflation, in a setting in which microdata is only occasionally powerful. We develop a scan test designed to test whether one forecast occasionally outperforms a benchmark, which we apply to the microdata-based forecast. We then create an adaptive machine learning pipeline in order to exploit the occasional forecasting power of microdata, motivated by existing tools from the adaptive machine learning literature \citep{HerbsterWarmuth1998,korotin2020adaptive}. These algorithms are well suited for high dimensional prediction in a non-stationary environment, and produce easy-to-interpret statistics that connect to macroeconomic models.

Our primary dataset consists of official UK CPI micro-price quotes. The data contain tens of thousands of monthly observations on individual prices across the economy and underlie the construction of aggregate inflation. The sample runs from 1996 to 2024, covering both the low and stable inflation period before 2020 and the more volatile period thereafter (see Figure~\ref{fig:us_uk_cpi}). Each observation is a price quote identified by an item code, shop, region, and shop type. We supplement these microdata with a large panel of macroeconomic predictors \citep{GouletCoulombeMarcellinoStevanovic2021}.

First, we develop a procedure for assessing whether microdata improve aggregate inflation forecasts, at least occasionally, allowing for non-stationary inflation dynamics. We develop a scan test with a composite one-sided null hypothesis that one forecast is at least as accurate as another at every point in time and at every forecasting horizon. The alternative hypothesis is that the second forecast outperforms over some contiguous interval. For each candidate interval and horizon, we compute a partially-studentized statistic based on cumulative loss differentials, and define the scan statistic as the maximum of this object over all admissible intervals and horizons. Because this maximum has a non-standard distribution, we obtain critical values using a computationally efficient Gaussian multiplier bootstrap, after first testing the loss differentials for serial dependence to determine the appropriate bootstrap regime. The interval and horizon that maximize the scan statistic point to the period and forecasting distance at which the micro forecast outperforms by the most. The scan statistic is a maximum over a large collection of candidate intervals and horizons, meaning we require the recent high-dimensional Gaussian approximation results \citep{ChernozhukovChetverikovKato2013,ChernozhukovChetverikovKato2017, CCW24} to derive our test.

In this paper, the primary use of the scan test is forecast expert selection. The aim is to pick candidate forecasts that should be combined in order to forecast inflation, with a microdata forecast being one such candidate. However the scan test is of independent interest: There are many circumstances when one wishes to know whether one forecast outperforms another during some ex ante unknown length of time that is shorter than the full sample. By contrast, standard forecast comparison tests focus on average performance over the full sample (e.g. \citealp{DieboldMariano1995}).\footnote{The test of \cite{GiacominiRossi2010} is an influential exception that we will discuss shortly.}

We then develop an adaptive machine learning pipeline to exploit the occasional forecasting power of microdata.
Our pipeline involves three steps. The first step encodes the raw micro-price data into a stable, high-dimensional vector of distributional statistics. This step is necessary because individual products regularly enter and exit the sample, so the raw data do not admit a fixed predictor vector. We therefore summarize, for each month and each of the 11 retained broad
COICOP1 consumption categories, the cross-sectional distribution of price
changes using statistics that are available in every period. These include the
fraction of prices that do not change, the deciles of the distribution of
non-zero price changes, and the mean non-zero price change. Stacking these
statistics across categories, and collecting their lags, yields an
interpretable but high-dimensional representation of the micro-price
environment, which can then be used for forecasting.

The second step acknowledges that because vector of predictors is high dimensional, conventional forecasting methods are prone to overfitting. We therefore use a machine learning method designed for high-dimensional prediction: gradient-boosted trees, implemented with XGBoost, which performs strongly in tabular forecasting settings \citep{ChenGuestrin2016,Grinsztajn2022}. We estimate the model on rolling windows, and develop a real time procedure for tuning the hyperparameters. We generate direct forecasts for horizons from 1 to 24 months ahead.

We contrast the microdata forecast with two benchmark forecasts from the literature. The first is a rolling-window seasonal ARIMA (SARIMA) model whose order is selected in real time using the same tuning procedure as our other forecasts, and which nests several leading univariate inflation forecasts used in the literature, including those of \citet{AtkesonOhanian2001} and \citet{StockWatson2007}. The second benchmark applies the same rolling-window XGBoost procedure used for microdata, but with a large set of macroeconomic predictors instead, allowing us to isolate the marginal contribution of micro-price information. We are careful to train and tune both benchmarks in the same way as our microdata forecast, to avoid advantaging the latter.

The third step combines the micro, macro, and univariate forecasts, while acknowledging that each modality may have only occasional forecasting power. We use a variant of the Fixed Share algorithm of \citet{HerbsterWarmuth1998} by \citet{korotin2020adaptive} which can work with multi-step ahead forecasting. The algorithm forms a weighted average of the three forecasts and updates those weights as multi-step forecast errors are realized, placing more weight on experts with better recent performance and using only information available in real time. There is a regret guarantee, valid in non-stationary environments, which bounds the algorithm's loss relative to the best sequence of experts chosen in hindsight. That is, unlike standard forecast combination methods that effectively target a single best model, Fixed Share is designed to track the best sequence of forecasts when the identity of the best forecast can change. The regret bound grows if forecasts are added that do not provide power. As such, we use the scan test to assess whether all three forecasts, micro, macro and univariate, should be included.

The Fixed Share algorithm also yields an interpretable time series of expert weights. In particular, the weight on the micro forecast measures how useful micro-price information is for forecasting inflation at a given date and horizon. This is valuable because macroeconomic models make predictions about when microdata should matter for aggregate inflation. We can therefore use the evolution of the Fixed Share weights to assess whether those theoretical predictions line up with the data. This interpretability is an important advantage of the approach, especially relative to machine learning methods that are often criticized as opaque ``black boxes'' \citep{MullainathanLudwig2024}.

Overall, our empirical results show that microdata are useful for aggregate
forecasting, but only during periods of large shocks. We report four main
empirical results. First, the microdata forecast outperforms the univariate
benchmark only after 2020: before 2020 the univariate benchmark is better at
every horizon, while after 2020 the micro forecast wins at nearly every
horizon, and by the most at long horizons.

Second, the scan tests imply that the Fixed Share combined forecast should
include all three modalities: micro, macro and univariate forecasts. We first
test whether the macro forecast adds forecasting power relative to the
univariate benchmark. We then test the micro forecast  against the
univariate and macro forecasts. The scan tests support adding both micro and macro forecasts.

Third, the combined forecast, which uses micro, macro and univariate
information, performs comparably to the univariate benchmark before 2020 and
improves on it at all 24 horizons in mean absolute error in the
post-2020 window.

Fourth, the value of microdata for the combined forecast materializes after
2020. The combined forecast leans on the univariate forecast before 2020 and
shifts weight toward the micro forecast afterwards, with the weight on the
micro forecast peaking in 2023. A more demanding test asks whether the
three-expert combined forecast outperforms, on average, a two-expert benchmark
that uses only macro and univariate information. It finds no outperformance before 2020, but finds outperformance after
2020. Microdata therefore add significantly to the two-expert benchmark but only in the post-2020 high-inflation window.

Finally, we ``open the black box'' by investigating which features of the microdata
matter for aggregate dynamics. We compute grouped Shapley values for different sets
of features that map onto objects in price-setting theory.\footnote{For example, these sets include the
extensive-margin price-adjustment block, the center of the price-change distribution,
its upper and lower tails, and monthly dummies. Various models of price setting
emphasize one of these margins, such as the extensive margin of price adjustment
\citep{KlenowKryvtsov2008,NakamuraSteinsson2008,Blanco2021} or asymmetries and
tail movements in the distribution of price changes
\citep{BallMankiw1995,Midrigan2011,Vavra2014,Alvarez2016,AlvarezBlaser2025}.}
Grouped Shapley values assess which groups contributed most to the forecast gains
of the microdata. After 2020, the positive contributions come from the center
and tails of the price-change distribution together with monthly dummies,
while the extensive-margin block contributes essentially nothing.

Our results therefore have implications for models of inflation. If macroeconomic variables recover the history of aggregate shocks
and the distributional state evolves according to a stable, estimable
transition law, macro information alone should be sufficient for forecasting
inflation. Microdata can improve feasible forecasts when this benchmark fails
in either of two ways. First, standard macroeconomic variables may not recover
the shocks that move the distributional state. Granular disturbances
originating in particular firms, products, or sectors are a leading example:
they can affect aggregate inflation without being fully reflected in observed
aggregates \citep{AlvarezBlaser2025}. Second, the mapping from the aggregate shock history
to the distributional state may be unstable or difficult to estimate,
particularly in periods of structural change. These failures are likely to become more consequential
during periods of high and volatile inflation, when large shocks are more
prevalent and price-setting nonlinearities make inflation more sensitive to
the distributional state \citep{Blanco2021}. The post-2020 forecast gains are
consistent with microdata revealing information that macroeconomic history
alone does not reliably recover during this period.

\textbf{Related literature.} Our paper contributes to the literature on forecasting inflation.
A classic result is that inflation is hard to forecast, meaning simple
univariate models such as random walks are difficult to beat \citep{AtkesonOhanian2001,
  StockWatson2007}. Notwithstanding this result, models with various degrees
of sophistication, using macroeconomic information, judgemental forecasts,
and inflation expectations by households, forecasters and financial
markets, have also been used to forecast inflation (see \citet{FaustWright2013} for a review). There
is a newer literature that seeks to forecast inflation using increasingly
rich macroeconomic datasets combined with machine learning methods.
For instance \citet{Medeiros2021} find that rich macroeconomic data and random
forest algorithms lead to forecasting gains for US inflation.\footnote{Alternative ways to forecast with a large number of predictors, and prevent overfitting, include Bayesian shrinkage or factor analysis \citep[e.g.,][]{StockWatson2002,BanburaGiannoneReichlin2010}.} Likewise, recent work attempts to
forecast inflation using aggregated statistics drawn from rich sources
of microdata \citep[e.g.,][]{Beck2024, BrandaoMarques2024}. We develop a
method that encodes the rich distributional information in microdata,
uses machine learning to construct a forecast, and adaptively combines
it with existing benchmark forecasts. The resulting weights are
interpretable and reveal when microdata add forecasting power. In doing so,
we contribute
more broadly to the literature on machine learning in time series
forecasting \citep[e.g.,][]{GiannoneLenzaPrimiceri2021, ChiEtAl2025}.

There is previous work that forecasts macroeconomic variables using
the distribution of microdata. A seminal, recent paper is \citet{ChangChenSchorfheide2024},
which develops a ``functional'' vector autoregression (VAR) approach. The
method involves summarizing the log density of a distribution of microdata
by a low dimensional collection of objects: finite dimensional sieves
with fixed basis functions and time varying coefficients. One can
then estimate a vector autoregression using aggregate variables and
the time varying sieve coefficients, allowing a forecast of macro
variables with microdata. The main focus is on a stationary environment.\footnote{See \citet{LenzaSavoia2024} for an application of this approach to firm heterogeneity
  in the Euro Area.} \citet{MeeksMonti2023} develop a related approach that summarizes
cross-sectional distributions with functional principal components,
and links them to inflation in a linear, Phillips Curve style
model. We offer a different approach: summarizing microdata with a
high dimensional series of statistics, which we combine with black-box forecasting algorithms and expert aggregation via the Fixed Share algorithm. Our approach is designed for a non-stationary environment.

This non-stationarity motivates us to develop a scan test for localized outperformance
of the microdata forecast relative to a benchmark. Classic
time series tests of forecasting performance test whether one forecast
outperforms another on average, throughout the full time series (e.g.
\citealp{DieboldMariano1995}). Instead, we test for whether there are occasional
periods of time in which a chosen forecast out-performs a benchmark, without
specifying ex ante when and for how long the forecast outperforms.

As such, our approach builds on the influential test for forecast comparisons of \citet{GiacominiRossi2010}. We take inspiration from their paper---which points out that comparisons across the full time series can miss localized episodes in which one forecast outperforms another---but study a different inferential problem. A first key difference in our framework is the null hypothesis. \cite{GiacominiRossi2010} test whether two forecasts have equal predictive ability. We instead study a one-sided null, asserting that one forecast does not outperform the other on any relevant interval. Rejection of equal predictive ability is not sufficient for our purposes, because rejecting the null does not identify which forecast performs better.\footnote{While \citet{GiacominiRossi2010} entertain a one-sided alternative hypothesis, their critical values are calibrated around the case of equal predictive ability. It is not clear whether this case is the least favorable for Type-I error in a one-sided null parameter space. The main obstacle is that the relevant t-statistics are partially studentized and hence need not be monotone in the data.}  A second difference is the nature of the forecast out-performance. The tests of \cite{GiacominiRossi2010} either assess equal predictive ability over intervals of unknown location but fixed length, or search for a single date after which one forecast outperforms a second; in both cases at a fixed forecast horizon. Our scan test searches for the possibility that one forecast outperforms another at an unknown start date and end date, for an unknown duration of time and at an unknown horizon. Lastly, since a forecast can outperform for an unknown interval, the scan
test considers a large collection of candidate
intervals. We develop an asymptotics that suits this setting. We study a high-dimensional asymptotic regime in which the minimal length of time for which the chosen forecast outperforms a benchmark grows logarithmically in the sample size; while the number of intervals grows exponentially in the sample size. For inference, we develop a high dimensional Gaussian multiplier bootstrap, which accounts for the large number of intervals and horizons that the test studies. Another novelty is partial studentization: instead of fully self-normalizing each short interval, we introduce an exponent that interpolates between unstudentized cumulative loss differentials and full local studentization. This stabilizes the scan when many short windows are searched, and we choose the exponent with simulations calibrated to the empirical loss-differential panel to avoid small-sample over-rejection.

Finally, we use an adaptive machine learning approach, the Fixed Share algorithm, to exploit that microdata helps to forecast inflation,
but only occasionally (i.e.\ there is non-stationarity). A large literature
studies how to combine forecasts \citep{BatesGranger1969, Timmermann2006},
including with Bayesian model averaging approaches
\citep[e.g.,][]{KoopKorobilis2012}. Many methods require a full specification of the likelihood function,
making it difficult to combine them with high dimensional covariates
and black-box machine learning methods. The Fixed Share algorithm provides worst-case regret guarantees
without distributional assumptions, in the presence of non-stationarity, and when some experts use high dimensional data and machine learning. An important previous paper also uses adaptive
machine learning, namely \citet{FouliardHowellReyStavrakeva2023}.

In addition to applying online learning in a different setting and using
microdata, we combine the algorithm with a scan test and include only forecasts
that provably improve the regret bound and are likely to improve the final
forecast's overall performance.

\section{Data}\label{sec:data_description}

This section describes the construction of the datasets used in the empirical
analysis. We combine two main sources of information. First, we use official
UK CPI micro-price quotes collected by the Office for National Statistics
(ONS), which provide the micro-price panel underlying the official
inflation statistics. Second, we assemble a large set of
macroeconomic and financial time series which we use to forecast inflation with macro data.

\subsection{Microdata on Prices}
\label{subsec:raw_data}

The forecasting target is monthly, headline, non-seasonally adjusted consumer
price index (CPI) inflation for the United Kingdom.\footnote{We measure inflation $\pi_t=100(\mathrm{CPI}_t/\mathrm{CPI}_{t-1}-1)$, using the ONS D7BT headline CPI
index.} Our primary dataset
consists of official micro-level price quotes underlying the CPI, collected by
the Office for National Statistics (ONS) and spanning January 1996 to
December 2024. For 1996--2016, we use the
replication package of \citet{AdamWeber2023}; for 2017 onward, we extend the
dataset using price quotes and item index files from the ONS
website.

Each observation is a price quote identified by four fields: a six-digit item
code, an anonymised shop code, a region code for the ONS collection regions,
and a shop type distinguishing chain retailers (``multiples'') from
independents.\footnote{Price-quote files from 2018 onward add two further
shop-type codes, covering roughly 1.7\% of observations.} For instance, an
item might be a large loaf of white sliced bread, collected from a Sainsbury's outlet in
a region such as London. Each item is associated with an expenditure weight; consumer price index (CPI) inflation is calculated by aggregating individual price changes using expenditure weights. There is also a flag for whether the price has a sale, as measured by the ONS price quote collector. For each observation, we calculate the monthly unit price change as
\[
  \Delta p_{jt} = 100 \times (\log p_{jt} - \log p_{j,t-1}),
\]
the month-on-month log price change in percentage points. There are around 90,000 price quotes in the typical month.

The raw ONS panel is first cleaned using standard validity, comparability, and price filters, which we list in Appendix~\ref{app:ons_cleaning}. \footnote{If the same product identifier (item, shop, region, and shop type) appears more than once in a month, the public data do not say which repeated quote should be linked to observations in adjacent months. Rather than choose one quote arbitrarily, we drop the whole series; this removes roughly 6.6\% of observations per quarter.} We organise the data at the cell level, where a cell is defined by a monthly date and a product category. The product categories are broad groupings of items at the COICOP1 level (there are 12 COICOP1 categories, such as food and non-alcoholic beverages). We balance the data by retaining only COICOP1 codes present in more than 95\% of months, which leaves 11 of the 12 COICOP1 categories.\footnote{The Education category (COICOP1 code 10) is the one category excluded by this filter: it appears in only 120 of the 347 sample months (February 1996 to January 2006), far below the 95\% retention threshold of 330 months.}

Table~\ref{tab:cells_coicop1} reports the percentile distribution of observations per cell, separately for all price changes and for non-zero price changes
only. The difference between the two panels reflects the well-documented
fact that roughly 90\% of prices are unchanged at monthly frequency. We include cell counts with and without items on sale.

We validate the data in two ways. First, we reconstruct the
aggregate CPI from the underlying micro quotes and verify that the implied
inflation rate closely tracks the official ONS series (see
Appendix~\ref{app:cpi_aggregation}). Second, we benchmark the moments of the
price-change distribution---mean, standard deviation, and excess kurtosis of
$|\Delta p|$---against the existing literature and find that our estimates
replicate the well-known stylised facts closely (see
Appendix~\ref{app:moments}).

Our main dataset for the forecasting exercise uses all observations in the
cleaned posted-price panel. For robustness, we study panels in which we treat
sales prices identified by the Office of National Statistics as missing, or
adjust them using a filter developed by
\citet{NakamuraSteinsson2008} with two parameterisations. The latter filters
identify likely sale episodes algorithmically and replace sale prices with
the inferred regular price, rather than dropping them; Appendix~\ref{app:NS_filter} describes the
algorithm and its parameters in detail. The forecasting robustness exercises
use the ONS-flag panel and the filter's baseline parameterisation (Filter A);
Appendices~\ref{app:nosales_encoding} and~\ref{app:nsa_encoding} report the
results.

\subsection{Macro Data}
\label{sec:macro_data}

In addition to the micro-price data, we assemble a panel of 36 monthly
macroeconomic and financial predictors. The candidate list is the union of
the series in the UK Monthly Database (UKMD) of
\citet{GouletCoulombeMarcellinoStevanovic2021}, series identified as
inflation-forecasting-relevant in prior work, and additional candidates judged useful for
forecasting UK CPI inflation. From this list we retain series that are
monthly, cover the estimation window (January 1996, or shortly after for a
few market series), and are either unrevised in practice (market prices,
exchange rates, policy and market interest rates, CPI indices, claimant
counts) or admit a vintage-clean construction.\footnote{We exclude a series
when a cleaner substitute is already included or when the series terminates
early, as with most Index of Production sub-categories and discontinued
benchmark rates. For average weekly earnings, we join the modern series to its
predecessor to extend coverage back to 1996. We do not use GDP, firms' or
households' inflation expectations, or inflation swaps, which are either not
monthly or only available more recently.} The
retained panel covers actual rents, wages,
labour-market slack, exchange rates, energy and food commodity prices and
futures, house prices, interest rates and gilt yields, equity indices and
volatility, trade, car registrations, the OECD composite leading indicator,
and air-freight price indices as supply-chain pressure proxies.\footnote{The retained macro
series are drawn from the ONS, the Bank of England, FRED/ALFRED, the World
Bank, OECD, and market data providers. The air-freight indices are subject to
revision and enter as ALFRED initial-release vintages, which begin in March
2010; earlier months use the latest available vintage.} Among the CPI
components, only actual rents are retained,
because rents are not covered by the micro-price quotes. All other CPI
sub-indices are excluded, as are CPIH sub-indices, RPI components, and PPI
sub-categories.\footnote{Throughout, the month-$t$ values of these series are
treated as part of the origin-$t$ information set: forecasts dated origin $t$
are formed after the month-$t$ releases.} Wherever possible we use non-seasonally adjusted series, which
avoids the look-ahead bias inherent in two-sided seasonal adjustment filters;
the two aggregate trade series are available only in seasonally adjusted
form.\footnote{The price data are not seasonally adjusted by design; the
aggregate consumer price index and its microdata constituents have not been
revised since 1996.} The sample runs through November 2024. The resulting dataset,
detailed in Appendix Table~\ref{tab:macro_all}, contains the 36 monthly
predictor series plus the headline, non-seasonally adjusted consumer price
index (CPI), and spans the period from the start of 1996 through November
2024.

\section{A Scan Test for Localized Forecast Outperformance}\label{sec:scantest}

This section introduces a scan test. The goal is to detect
localized periods in which one forecast model outperforms another, in a
non-stationary environment. Traditional forecast comparisons instead
focus on average performance over the full sample
\citep[e.g.][]{DieboldMariano1995}. Full sample tests are less useful in our setting. In standard macroeconomic models,
microdata matters for the dynamics of aggregate inflation after large shocks to
inflation, but not otherwise. As such, a forecasting model
that uses microdata might outperform benchmarks only at certain
times. The scan test is designed to test for this possibility.

The test is also more
general and applies to any comparison of two forecast models in a non-stationary
environment. The econometric theory for the test accounts for the composite
null, nonstationarity, temporal dependence, and high-dimensionality. In the
following section we will introduce a pipeline that uses the scan-test evidence to
guide an adaptive machine learning approach, which makes use of microdata only when it at least adds localized forecasting power.

\subsection{\texorpdfstring{Statistical Hypotheses and Methodological Details}{Statistical Hypotheses and Methodological Details}}\label{sec:scan_method}

\paragraph{Roadmap. }We now formalize our scan test. The aim is to decide whether a candidate forecast outperforms a second, benchmark forecast, over particular sub-periods and horizons; with rigorous statistical guarantees. As a roadmap, we first state the null and alternative hypotheses formally, then introduce a test statistic that measures statistical evidence against the null, and finally describe a Gaussian bootstrap procedure that approximates the null distribution of the statistic and converts it into a p-value. We present econometric theory that justifies the validity of our scan test. To maximize readibility, we only discuss technical assumptions informally and leave all details and proofs to Appendix \ref{sec:theory_scan_tests}.

\paragraph{Part I: Null and Alternative Hypotheses.}

Let $n$ denote the number of target months in the evaluation sample and $H$ the maximal forecasting horizon.
Consider a benchmark method $B$ and a competing method $M$. In our application,
$B$ could be a benchmark inflation forecast, such as the SARIMA model introduced
below, and $M$ could be the microdata-based inflation forecast. Writing
$\hat{\pi}_{t\mid t-h}$ for a forecast of inflation in target month $t$
made with information through origin month $t-h$, let
$e^{B}_{t\mid t-h} = \pi_{t} - \hat{\pi}^{B}_{t\mid t-h}$ and
$e^{M}_{t\mid t-h} = \pi_{t} - \hat{\pi}^{M}_{t\mid t-h}$ denote the two
methods' $h$-step-ahead forecast errors for target month $t$, and define the loss
differential, indexed by target month,
\[
  D_t^{(h)} = \bigl|e^{B}_{t\mid t-h}\bigr| - \bigl|e^{M}_{t\mid t-h}\bigr|.
\]
With this sign convention, positive values of $D_t^{(h)}$ indicate that method
$M$ outperforms method $B$ in target month $t$ at horizon $h$. 

The null hypothesis is
\[
  H_0 : \mathbb{E}[D_t^{(h)}] \le 0
  \quad \text{for all } t \in [n] := \{1,\dots,n\},\; h \in [H] := \{1,\dots,H\},
\]
so that under $H_0$, method $B$ is at least as good as method $M$ across all
target months and horizons. Note that it is possible for the null to
be rejected in both directions, that is, for both $M$ versus $B$ and $B$ versus
$M$. If so, there are some months or forecasting horizons in which method $B$
is better than method $M$, and other months or horizons where the opposite is
true.

To make rejection substantively
meaningful, we consider the restricted alternative
\[
  H_1 : \sum_{t\in [s,e]}\mathbb{E}[D_t^{(h)}] > 0
  \quad  \text{ for some interval } [s,e]
  \text{ with } L_{\min}\le e-s+1 \le L_{\max}
  \text{ and some } h \in [H].
\]
That is, method $M$ outperforms method $B$ on average over some contiguous
interval of target months of length between $L_{\min}$ and
$L_{\max}$, at some horizon
under study. The bounds $L_{\min}$ and $L_{\max}$ are
user-specified parameters. We will discuss their choices at the end of this
section. In the empirical implementation below, we use monthly data and take
$H=24$.

\paragraph{Part II: A Partially-Studentized Scan Statistic}

To test the null hypothesis against the localized alternative, we construct a
scan statistic. This statistic searches over all admissible target-month
intervals and forecasting horizons to determine when method $M$ has an
unusually good forecast loss relative to method $B$.

Let $\mathcal{I}$ be the collection of all target-month intervals
$I = [s,e] \subset [n]$ with $L_{\min}\le |I|\le L_{\max}$.

For each $I = [s,e]$, $h \in [H]$, and
$D_t = (D_t^{(1)},\dots,D_t^{(H)})^\top$, define
\[
  S_I^{(h)} = \sum_{t=s}^e D_t^{(h)}, \qquad
  \bar{D}_I^{(h)} = \frac{S_I^{(h)}}{|I|}, \qquad
  Q_I^{(h)} = \sum_{t=s}^e \bigl(D_t^{(h)} - \bar{D}_I^{(h)}\bigr)^2.
\]
In words, $S_I^{(h)}$ is the cumulative loss differential over the target-month
interval $I$,
$\bar{D}_I^{(h)}$ is the corresponding within-interval mean, and
$Q_I^{(h)}$ is the within-interval sum of squared deviations. For a tuning
exponent $\gamma\in[0,1/2]$, define the partially-studentized statistic for
interval $I$ and horizon $h$ as\footnote{When the denominator is numerically
smaller than $10^{-12}$, we set the corresponding observed or bootstrap
interval statistic to zero.}
\[
  T_I^{(h)}
  \;=\;
  \frac{S_I^{(h)}}{|I|^{1/2-\gamma}\bigl(Q_I^{(h)}\bigr)^\gamma}.
\]
The scan statistic is defined as the double maximum
\[
  T^{\max}
  \;=\;
  \max_{I \in \mathcal{I}}\,\max_{1 \le h \le H}\,T_I^{(h)}.
\]
This statistic searches jointly across target-month
intervals and forecasting horizons for the most pronounced localized
outperformance of method $M$ relative to method $B$.

Two features of $T^{\max}$ deserve emphasis. First, the joint
maximization over both intervals $I$ and horizons $h$ searches for localized
outperformance without fixing either the interval or the forecast horizon.
By searching over intervals, we do not have to specify the start month and
duration of any outperformance. Extending the search to $H$ horizons allows
detection of outperformance concentrated at particular forecast horizons rather
than uniformly across all of them. The cost of testing across multiple horizons
will be absorbed automatically by the bootstrap. Second, partial
studentization lets the statistic interpolate between no local variance
studentization $(\gamma=0)$ and full self-normalization $(\gamma=1/2)$.
We will show in Appendix~\ref{sec:theory_scan_tests} that the test is asymptotically
valid for any $\gamma \in [0,1/2]$. Nonetheless, because our empirical design
involves small samples, namely short intervals of length no more than 12,
simulations reveal substantial over-rejection for larger values of $\gamma$.
We discuss a principled approach to choosing $\gamma$ below. To the best of our
knowledge, partial studentization is a novel technique to address the
small-sample issue for Gaussian bootstrap.

\paragraph{Part III: Gaussian Multiplier Bootstrap Distribution.}\label{sec:bootstrap}

Owing to the joint maximization over many intervals and horizons, the null
distribution of $T^{\max}$ is nonstandard. We approximate it using a
Gaussian multiplier bootstrap, which delivers critical values and $p$-values
for testing $H_0$ against the localized alternative.

Define a draw of the bootstrap multiplier
\[
  \xi = (\xi_1,\dots,\xi_n)^\top \sim \mathcal{N}(0,\Theta_n)
\]
which is independent of the data. $\Theta_n$ is a positive semidefinite
$n \times n$ covariance matrix specified below. The role of $\Theta_n$ is to
encode the temporal dependence structure imposed on the bootstrap multipliers,
so its choice should mirror the dependence properties of the loss differential
process under the null.

For each interval $I=[s,e]$ and horizon $h$, define the multiplier-weighted
local scores
\[
  e_{t,I}^{(h)}
  =
  \xi_t\bigl(D_t^{(h)}-\bar{D}_I^{(h)}\bigr),
  \qquad
  S_I^{*(h)} = \sum_{t=s}^{e} e_{t,I}^{(h)}.
\]
The bootstrap denominator is recomputed from these local scores:
\[
  Q_I^{*(h)}
  =
  \sum_{t=s}^{e}\bigl(e_{t,I}^{(h)}-\bar e_I^{(h)}\bigr)^2
  =
  \sum_{t=s}^{e}\bigl(e_{t,I}^{(h)}\bigr)^2
  -
  \frac{\bigl(S_I^{*(h)}\bigr)^2}{|I|},
\]
where $\bar e_I^{(h)}=S_I^{*(h)}/|I|$. The bootstrap interval statistic is
\[
  T_I^{*(h)}
  =
  \frac{S_I^{*(h)}}{|I|^{1/2-\gamma}\bigl(Q_I^{*(h)}\bigr)^\gamma}.
\]
The bootstrap scan statistic is then defined as
\[
  T^{*,\max}
  =
  \max_{I \in \mathcal{I}} \max_{1 \le h \le H} T_I^{*(h)}.
\]
We next take $B$ bootstrap draws
$T^{*,(1),\max}, \ldots, T^{*,(B),\max}$, where each bootstrap
draw involves a new, independent draw of the vector $\xi$, followed by
recalculating $T^{*,\max}$ using the new draw. The $p$-value is
\[
  \hat{p}
  \;=\;
  \frac{1 + \sum_{b=1}^{B}\mathbf{1}\{T^{*,(b),\max} > T^{\max}\}}{B+1}.
\]
The $p$-value is the share of bootstrap draws in which the null produces a
statistic at least as extreme as the one observed in the data (after adding 1
to both the denominator and numerator as a finite sample correction).

\paragraph{Choice of $\Theta_n$}

For the Gaussian bootstrap to give an accurate approximation to the null distribution of the scan statistic, the joint distribution of $\{S_I^{*(h)}: L_{\min}\le |I|\le L_{\max}, h\in [H]\}$, conditional on the observed loss differentials, should approximate the sampling joint distribution of $\{S_I^{(h)}: L_{\min}\le |I|\le L_{\max}, h\in [H]\}$. Focusing on a single interval $I$, the covariance matrix of $S_I = (S_I^{(1)}, \ldots, S_I^{(H)})^\top$ is given by
\begin{equation}\label{eq:Cov_SI}
  \sum_{t,s\in I}\mathrm{Cov}(D_t, D_s),
  \end{equation}
while the covariance matrix of $S_I^* = (S_I^{*(1)}, \ldots, S_I^{*(H)})^\top$ conditional on $(D_1, \ldots, D_n)$ is given by
\begin{equation}\label{eq:Cov_SI*}
  \sum_{t,s\in I}(\Theta_n)_{ts} (D_t - \bar{D}_I)(D_s - \bar{D}_I)^\top,
  \end{equation}
where $\bar{D}_I = (\bar{D}_I^{(1)}, \ldots, \bar{D}_I^{(H)})^\top$.
To build more intuition around the choice of $\Theta_n$, we first consider the \textbf{Regime I} where $D_t$ are independent over time. In this case, we choose
\begin{equation}\label{eq:regimeI_Theta}
\Theta_n = I_n.
\end{equation}
In fact, under independence, the true covariance matrix of $S_I$ and the conditional covariance matrix of $S_I^*$ reduce to
\[\sum_{t\in I}\mathrm{Var}(D_t), \quad \sum_{t\in I}(D_t - \bar{D}_I)(D_t - \bar{D}_I)^\top.\]
When $|I|$ is sufficiently large, the standard concentration argument implies the closeness between the two matrices \citep[e.g.][Chapter 6]{wainwright2019high}.

Next, we consider the general case where $D_t$ are weakly dependent, or, more precisely, strongly mixing (see Assumption \ref{ass:D_mixing} in Appendix \ref{subsubapp:regimeD_main} for the formal definition). We refer to this as \textbf{Regime D}. Approximating equation \eqref{eq:Cov_SI} by equation \eqref{eq:Cov_SI*} is analogous to the problem of constructing heteroscedasticity and autocorrelation consistent standard errors \citep{newey1987simple, Andrews1991}. Following this intuition, we can choose $\Theta_n$ as
\begin{equation}\label{eq:regimeD_Theta}
  \Theta_n(t,u) = K\!\left(\frac{t-u}{b_n}\right),
\end{equation}
where $K$ is a kernel function and $b_n$ is a bandwidth.

Our implementation uses the quadratic spectral (QS) kernel, introduced by \cite{Andrews1991}, because it satisfies our regularity condition; see Appendix \ref{subsec:QS_kernel}. Unlike the standard Bartlett kernel in the Newey-West estimator \citep{newey1987simple}, the QS kernel does
not truncate at a finite lag and instead decays smoothly, placing nonzero
weight on all lags. It is known to be optimal in a mean squared error sense for
long-run variance estimation under certain conditions \citep{Andrews1991}.

Having chosen the QS kernel function, the remaining tuning parameter in Regime D is the bandwidth $b_n$, which controls how quickly dependence decays in the bootstrap covariance matrix. We apply the optimal bandwidth derived in \citet{Andrews1991} based on a data-driven plug-in rule based on an AR(1) approximation to
the loss differentials; see also \citet{CCW24} in the context that is closer to our setting. Specifically, for each horizon $h$, we fit the autoregression
\[
  D_t^{(h)} = c_h + \rho_h D_{t-1}^{(h)} + \varepsilon_{h,t},
\]
and obtain estimates $\hat{\rho}_h$ and $\hat{\sigma}_h^2$, where
$\hat{\sigma}_h^2$ is the innovation variance. We then aggregate these horizon-specific estimates into
\[
  \hat{a}
  =
  \frac{\sum_{h=1}^H 4 \hat{\rho}_h^2 \hat{\sigma}_h^4 (1-\hat{\rho}_h)^{-8}}
       {\sum_{h=1}^H \hat{\sigma}_h^4 (1-\hat{\rho}_h)^{-4}},
\]
and set the QS bandwidth to
\begin{equation}\label{eq:regimeD_bandwidth}
  b_n = 1.3221\,(\hat{a}n)^{1/5},
\end{equation}
where $n$ is the time-series sample size.\footnote{Two numerical safeguards
are applied in the implementation: the estimated AR(1) coefficients
$\hat{\rho}_h$ are clipped to $[-0.99,0.99]$ before forming $\hat{a}$, and the
eigenvalues of the kernel covariance matrix $\Theta_n$ are floored at a small
positive value before the multipliers are drawn.} The complete procedures for
Regime~I and Regime~D are given in Algorithms~\ref{alg:scan_regime1}
and~\ref{alg:scan_regimeD} in Appendix~\ref{app:pseudocode}.

\paragraph{Choice of $\gamma$}

We choose the partial-studentization exponent $\gamma$ by simulation under the null, calibrated to the empirical loss-differential sequence used in the forecast comparison. For each horizon, we fix the observed sequence $\{D_t^{(h)}\}$ and generate null draws
\[
  D_{t,k}^{(h)}=\xi_{t,k}D_t^{(h)}, \qquad
  \xi_{t,k}=\operatorname{sign}(Z_{t,k}),
\]
where $Z_{t,k}$ is a stationary Gaussian AR(1) process with autocorrelation $\rho$. This construction is useful because the random signs center the simulated loss differentials at the null of no systematic outperformance, while the multiplication by the observed $D_t^{(h)}$ preserves the empirical scale, heterogeneity across dates and horizons, and the cross-horizon co-movement in the loss-differential panel. We choose the grid for the simulation autocorrelation $\rho$ to cover the horizon-specific AR(1) estimates $\{\hat{\rho}_h\}_{h=1}^H$ defined above. For each value of $\gamma$ on a grid, each $\rho$, and each bootstrap regime, we simulate rejection probabilities at the nominal five-percent level. We then repeat the exercise for the final procedure that selects between Regime~I and Regime~D using the Ljung--Box pre-test. Across these designs, $\gamma=0.1$ avoids over-rejection most consistently, so we use $\gamma=0.1$ in the empirical scan tests. We include all details, including the DGP, simulation results, and robustness to the choice of $\gamma$ in Appendix~\ref{app:simulation_results}.

\paragraph{Part IV: Pre-Test for Serial Dependence and Regime Selection.}

Because the appropriate choice of $\Theta_n$ depends on whether the loss differentials are serially dependent, we conduct a pre-test before applying the
regime-specific bootstrap algorithms. This step determines whether inference
should proceed under the independent noise regime (Regime~I) or the dependent
noise regime (Regime~D).

To formally test for serial correlation, we apply the Ljung--Box test to
$\{D^{(h)}_t\}$ at multiple lags. The null hypothesis is that the
autocorrelations up to the tested lag are zero. Rejection of the null
indicates serial correlation.

For each horizon $h\in\{1,2,\dots,H\}$, let $r_{t,h} = D^{(h)}_t$ denote the
loss differential series. Fix a maximum lag $m$ and compute the Ljung--Box
statistic
\[
  Q_h \;=\; n(n+2)\sum_{k=1}^{m}\frac{\hat\rho_{h}(k)^2}{n-k},
\]
where $\hat\rho_{h}(k)$ is the sample autocorrelation of
$\{r_{t,h}\}_{t=1}^n$ at lag $k$. We aggregate over horizons via
$Q=\sum_{h=1}^{H} Q_h$.

To approximate the null distribution, we use a Gaussian multiplier bootstrap
that preserves the cross-horizon dependence at each target month $t$ while breaking
temporal dependence. Specifically, for each bootstrap draw $b=1,\dots,B$,
sample i.i.d.\ multipliers $\{\xi_t^{(b)}\}_{t=1}^n$ with
$\xi_t^{(b)}\sim\mathcal{N}(0,1)$, and form the bootstrapped panel
$r^{*(b)}_{t,h}=\xi_t^{(b)}\,r_{t,h}$ for all $h$. Compute $Q_h^{*(b)}$ and
$Q^{*(b)}=\sum_{h=1}^{H}Q_h^{*(b)}$ from
$\{r^{*(b)}_{t,h}\}_{t=1}^n$, and estimate the (one-sided) $p$-value as
\[
  \hat p \;=\; \frac{1+\sum_{b=1}^{B}\mathbf{1}\!\left\{Q^{*(b)} > Q\right\}}{B+1}.
\]

The null hypothesis of the Ljung--Box test corresponds to the absence of serial correlation, up to lag $m$, in the
loss differential sequences. Under this null, any observed autocorrelation
arises purely from random fluctuations. A small $p$-value indicates that the
observed level of autocorrelation is unlikely under temporal independence,
leading us to adopt the dependent noise regime (Regime~D). A large $p$-value
suggests insufficient evidence of dependence, and we proceed under the
independent noise regime (Regime~I).

In our application, we use $m=12$ lags and $B=4{,}999$ bootstrap draws at significance level $\delta=0.05$. The full procedure is given in Algorithm~\ref{alg:mb_ljung_box} in Appendix~\ref{app:pseudocode}.

\subsection{\texorpdfstring{Econometric Theory of Scan Tests in Regimes I and D}{Econometric Theory of Scan Tests in Regimes I and D}}

We face four major challenges. First, the null hypothesis is one-sided, and Type-I error must be controlled even when some or all of the means \(\mathbb{E}[D_t^{(h)}]\) are strictly negative. This is difficult because, due to partial studentization, the relevant \(t\)-statistics are not monotone in \(D_t^{(h)}\). As a result, the boundary of the null, namely \(\mathbb{E}[D_t^{(h)}]=0\) for all \(t\) and \(h\), is no longer necessarily the least favorable case. Second, the time series \(D_t^{(h)}\) is non-stationary, so standard functional central limit theorems do not directly apply. Third, the statistics \(D_t^{(h)}\) are highly correlated across forecasting horizons \(h\). Fourth, the number of intervals and horizons considered by the scan test can be much larger than the sample size. In our application, for example, the occasional-outperformance scan test considers \(1{,}449\) intervals and \(34{,}776\) interval-horizon pairs with only \(215\) observations. This disparity motivates a high-dimensional asymptotic regime in which the number of interval-horizon pairs grows much faster than the sample size.

All four challenges call for new development in our econometric theory. We address the first challenge by introducing a new condition that controls the volatility of $\mathbb{E}[D_t^{(h)}]$. Specifically, we define an interval-wise mean-flatness index:
\begin{equation*}
r_n:=\max_{L_{\min}\le |I|\le L_{\max}, h\in [H]}
\frac{\sum_{t\in I}\bigl(\mathbb{E}[D_t^{(h)}]-\mathbb{E}[\bar{D}_I^{(h)}]\bigr)^2}{\sum_{t\in I}\mathbb{E}\bigl[(D_t^{(h)} - \mathbb{E}[D_t^{(h)}])^{2}\bigr]} = \max_{L_{\min}\le |I|\le L_{\max}, h\in [H]}
\frac{\frac{1}{|I|}\sum_{t\in I}\bigl(\mathbb{E}[D_t^{(h)}]-\mathbb{E}[\bar{D}_I^{(h)}]\bigr)^2}{\frac{1}{|I|}\sum_{t\in I}\mathrm{Var}(D_t^{(h)})}.
\end{equation*}
The numerator measures the sample variance of $\mathbb{E}[D_t^{(h)}]$ over an interval and the denominator measures the average variance of $D_t^{(h)}$. Our theory requires $r_n$ to be small at certain slow rates; see Assumptions \ref{ass:intervalwise_mean} and \ref{ass:D_mean_flatness} for the rate requirements for Regimes I and D, respectively. In particular, $r_n=0$ whenever the mean is constant on each scanned interval.

More generally, $r_n$ is small insofar as the movement of $\mathbb{E}[D_t^{(h)}]$ is much slower than the intrinsic variability of $D_t^{(h)}$, which is plausible for challenging forecasting problems where $\mathrm{Var}(D_t^{(h)})$ is large.

We address the remaining three challenges by adapting recent developments in high-dimensional distributional approximations for independent and dependent data \citep{ChernozhukovChetverikovKato2013,ChernozhukovChetverikovKato2017, CCW24}. Specifically, the cited papers derive central limit theorems for the empirical average of independent but non-identically distributed observations or a weakly dependent, nonstationary multivariate time series. The theory allows the dimension to grow exponentially in a polynomial of the sample size. While these tools are helpful in understanding the distribution of $\{S_I^{(h)}: L_{\min}\le |I|\le L_{\max}, h\in [H]\}$, the numerators of the t-statistics, it remains challenging to handle the partially-studentized denominators $|I|^{1/2-\gamma}(Q_I^{(h)})^\gamma$, which vary with both $I$ and $h$.

We summarize our results in the following theorem.
\begin{theorem}[Informal; see Appendix Theorems \ref{I:thm:main} and \ref{D:thm:main} for formal statements]
  Let $\mathcal{I}=\{I\subset[n]:L_{\min}\le |I|\le L_{\max}\}$, let $p = H\,|\mathcal{I}|$ denote the number of admissible (interval, horizon) pairs, and let $c_{n,1-\alpha}^{*}$ be the conditional $(1-\alpha)$-quantile of $T^{*,\max}$. Assume $\gamma\in[0,1/2]$. If $\gamma<1/2$, assume also that $L_{\max}/L_{\min}$ is bounded; no such bounded-ratio condition is needed when $\gamma=1/2$. Then, under $H_0$, the scan test is asymptotically valid in the sense that
\[
\limsup_{n\to\infty}
\mathbb{P}\bigl(T^{\max}>c_{n,1-\alpha}^{*}\bigr)
\le \alpha
\qquad\forall \alpha\in(0,1),
\]
under either of the following assumptions, together with regularity conditions stated in Appendix~\ref{sec:theory_scan_tests}.
  \begin{itemize}
  \item (Regime I) $(D_t)_{t=1}^n$ are independent over time.
  \item (Regime D) $(D_t)_{t=1}^n$ are $\alpha$-mixing (see Assumption \ref{ass:D_mixing}).
  \end{itemize}
\end{theorem}

\subsection{\texorpdfstring{Three Practical Uses of the Scan Tests}{Three Practical Uses of the Scan Tests}}
\label{sec:scan_test_uses}

The scan-test framework has three practical uses in our empirical analysis. The
first two uses are for expert selection, and the third use is for the
confirmatory analysis.

The first is \emph{occasional outperformance}. In this use, the test asks whether a
candidate forecast outperforms a benchmark over some ex ante unknown interval
$I\in\mathcal{I}$ and some horizon $h$, with
$L_{\min}\le |I|\le L_{\max}$. In our application, we choose
$L_{\min}=6$ and $L_{\max}=12$.

The second use is \emph{overall outperformance}. In this use, we set
$\mathcal{I}=\{[1,n]\}$, or $L_{\min}=L_{\max}=n$. The statistic tests whether the candidate forecast improves
on the benchmark over the sample as a whole at some horizon:
\[
  H_0:\sum_{t=1}^{n}\mathbb{E}[D_t^{(h)}]\le 0
  \quad \text{for all } h=1,\ldots,H.
\]
This use separates sustained full-sample gains from the localized gains targeted
by the occasional-outperformance scan.

We use both occasional and overall outperformance for the expert selection.
This selection determines which candidate forecasts enter our online-learning
combined forecasts, which updates their weights adaptively.
Starting from a benchmark forecast, the researcher screens candidate experts
sequentially and includes a candidate only if it outperforms every expert
selected so far. By outperformance, we mean that, for each selected expert used
as the comparison forecast, either the occasional-outperformance test or the
overall-outperformance test rejects the corresponding no-outperformance null at
the 5\% nominal level. We use these two tests because they are complementary and
tend to be most powerful when the sample size is small, as in many macro
forecasting problems: occasional outperformance has stronger signal strength
because it targets localized gains, but it also has more noise from searching
over intervals; overall outperformance has weaker signal because it smooths over
time, but it also has less noise. The next section develops a pipeline that
combines the experts selected by this sequential screen.

The third use is a \emph{horizon-pooled test}. Given a user-chosen window of
target months $I^\star=[t_1,t_2]$ of length $T_w=t_2-t_1+1$, we average the
loss differential across horizons,
\[
  \bar D_t=\frac{1}{H}\sum_{h=1}^{H}D_t^{(h)},
\]
and restrict the interval collection to $\mathcal{I}=\{I^\star\}$, or
$L_{\min}=L_{\max}=T_w$. The resulting statistic is the partially-studentized
mean over that window,
\[
  T^{\max}
  =
  \frac{S_{I^\star}}{T_w^{1/2-\gamma}(Q_{I^\star})^\gamma},
  \qquad
  S_{I^\star}=\sum_{t=t_1}^{t_2}\bar D_t,
  \qquad
  Q_{I^\star}=\sum_{t=t_1}^{t_2}\bigl(\bar D_t-\bar D_{I^\star}\bigr)^2,
\]
where $\bar D_{I^\star}=S_{I^\star}/T_w$. The null becomes
\[
  H_0:\frac{1}{T_w}\sum_{t=t_1}^{t_2}\mathbb{E}[\bar D_t]\le 0.
\]
This horizon-pooled test keeps the one-sided directional null from the scan test,
but targets average performance across horizons in a pre-specified period. In
our empirical analysis, we use this confirmatory test to assess whether the
combined-expert forecast outperforms the other benchmarks. We reuse the Gaussian
multiplier bootstrap of Part~III and the Regime~I and Regime~D specifications
without modification.

\subsection{\texorpdfstring{Discussion: Existing Forecast Comparison Tests}{Discussion: Existing Forecast Comparison Tests}}

Our null hypothesis differs substantially from the
literature on forecast comparisons. Developing a valid test for this distinct
null is one of our main methodological contributions. For one thing, we consider
multiple horizons $h\in [H]$ all at once, whereas most forecast comparisons study a single horizon only. Moreover, standard forecast comparison tests
assume stationarity and test whether, on average across a whole sample, one
method outperforms a second (e.g. \citealp{DieboldMariano1995}). In effect, the
null hypothesis for a given $h$ is $\mathbb{E}[D_1^{(h)}] = \ldots = \mathbb{E}[D_n^{(h)}] \le
0$, under the assumption that $D_1^{(h)}, \ldots, D_n^{(h)}$ are stationary.

An influential series of tests by \citet{GiacominiRossi2010} relaxes stationarity, in
order to test whether one method outperforms a second only at certain times. Our scan test shares the same motivation: forecast
comparisons can miss localized episodes of relative performance. However, we target a different inferential problem. \citet{GiacominiRossi2010} test local equal predictive
ability, $\mathbb{E}[D_t^{(h)}]=0$ at each month, for a fixed horizon
$h$. They study two tests: a ``Fluctuation test'' that fixes the length of the
interval for which the candidate forecast outperforms and searches for the best
interval; and a ``One-Time Reversal test'', which looks for a single point in
the sample after which the candidate forecast outperforms. Their statistics are
standardized by a global heteroskedasticity and autocorrelation consistent
long-run variance, and
their critical values come from Brownian-motion/Brownian-bridge limits under a
high-level functional Central Limit Theorem with a positive scalar long-run variance.

Our scan test instead
considers a composite one-sided null, $\mathbb{E}[D_t^{(h)}]\le 0$ for all
months and horizons. Therefore rejection supports the directional claim
that one forecast outperforms another, somewhere. Our scan test searches
jointly over an unknown starting month, an unknown ending month, an unknown
duration, and an unknown horizon, without imposing any more restrictions on when
outperformance can happen. We use a high-dimensional Gaussian multiplier bootstrap to account
for the search across many intervals and horizons; the bootstrap also accounts for dependence across forecast horizons. Under
dependence, we do not require a single asymptotic long-run variance to exist: serial correlation is handled through the bootstrap covariance matrix $\Theta_n$,
with $\Theta_n=I_n$ under independence and a kernel covariance under dependence.
Finally, the critical values in the one-sided Fluctuation test of \cite{GiacominiRossi2010} are calibrated around both forecasts having equal ability. Our forecast comparison is
designed to control size uniformly over the full one-sided null by using a
one-sided max bootstrap critical value.

Our test is calibrated to handle the very small sample sizes that arise in the
occasional outperformance application, including intervals of length 6 to 12. In
Appendix~\ref{app:sim_gr}, we consider a suite of simulations based on our $D_t^{(h)}$ and
show that existing tests substantially
over-reject. In contrast, our scan test is well calibrated even for such small
sample sizes.

\section{Methods: Macro Forecasting with Micro Data}\label{sec:methods}

Section~\ref{sec:scantest} developed a scan test for localized forecast outperformance. We can use the test to ask whether microdata-based forecasts occasionally improve on benchmark forecasts. This section develops a forecasting pipeline that is designed to exploit the occasional forecasting power of microdata. There are three steps. First, we encode the distribution of price changes into a stable and high dimensional vector of statistics. Second, we combine the high dimensional vector with a machine learning algorithm in order to forecast inflation. We use a gradient boosted tree algorithm, which we tune in real time. Third, if the micro forecast is indeed useful according to the scan test, we combine it with the benchmark using an adaptive machine learning method. The algorithm gives more weight to the microdata forecast when it outperforms the benchmark and downweights it otherwise. Later, in Section~\ref{sec:group_shapley}, we discuss a method to ``open the black box'', letting us dissect which features of the microdata lead to its forecasting performance.

\subsection{Step 1: Encoding Micro Price Data}\label{sec:encoding_training_tuning}

\paragraph{Challenges.} The aim is to forecast inflation with microdata. The raw microdata from the consumer price index, however, are not in a form that can be used directly by the prediction methods below.

The first issue is mechanical. Machine learning methods such as XGBoost require a fixed-length input vector, so each training and test observation must contain the same predictors in the same order. The raw CPI micro panel does not satisfy this requirement because products, defined as item-store-region price quotes, regularly enter and exit the data. The typical product remains in the sample for less than a year. This turnover reflects shop behavior, for example, when a 6-pack of beer in a given grocery store is replaced by a close substitute. Because the set of observed products changes over time, using product-level price quotes directly would produce a predictor vector whose coordinates appear and disappear across months.

The second issue is statistical. Even if the raw product-level panel could be forced into a common design matrix, it would contain many noisy disaggregated predictors relative to the length of the monthly time series. A machine learning model trained on those inputs, even with regularization, risks fitting idiosyncratic product noise rather than systematic information about future inflation. More formally, statistical learning theory \citep[e.g.][]{wainwright2019high} suggests that, absent nonstationarity, forecasting accuracy can be decomposed into
\begin{equation}
  \label{eq:risk_decomp}
  \mathbb{E}[|Y - \hat{f}(X)|] \asymp \mathbb{E}\left|\mathbb{E}[Y\mid X] - \hat{f}(X)\right| + \sqrt{\frac{\mathrm{dim}(X)}{n}},
\end{equation}
where $n$ is the sample size and $\mathrm{dim}(X)$ is the number of predictors. The first term is the approximation error, which always goes down when more predictors are included. The second term is the estimation error, which grows in the number of predictors due to overfitting. With the small sample sizes typical of inflation forecasting, adding many noisy disaggregated predictors tends to increase estimation error more than it reduces approximation error. Overfitting would be even more salient under non-stationarity, as $X$ tends to be less predictive. While these two components cannot be easily assessed empirically in a changing environment, the guiding principle is to avoid overfitting by restricting the number of predictors.

\paragraph{Encoding method.} The first step is therefore to encode the raw
microdata into a stable vector of statistics. This encoding addresses the two
issues above. It replaces unstable product coordinates with fixed monthly
summaries, and it reduces the dimension of the predictor set by replacing
individual product-level price changes with distributional statistics computed
within broad cells. Cells are defined as the set of products in a COICOP1
category (e.g. food or clothing) in a given month. The baseline analysis uses a parsimonious summary
of the cross-sectional distribution of price changes. Specifically, for each
cell we calculate:

\begin{enumerate}
  \item The fraction of observations with $\Delta p_{jt} = 0$.
  \item The nine deciles (at probabilities $0.1,0.2,\ldots,0.9$) of the
  \emph{non-zero} price-change distribution within the cell.
  \item The mean of the non-zero price-change distribution.\footnote{Exact
  formulas for all statistics are given in
  Appendix~\ref{app:encoding_formulas}.}
\end{enumerate}
The fraction of zero price changes captures the extensive margin of price
adjustment.
The deciles and mean
summarize the distribution of realized price adjustments without using each
individual quote as a separate predictor. This is the dimension-reduction step
motivated by the statistical issue above: adding many disaggregated predictors
can increase estimation error in a short and non-stationary monthly sample. All cell-level statistics are weighted by the expenditure weights associated with the Consumer Price Index. Subsection~\ref{sec:micro_robustness} considers an alternative microdata encoding that constructs these statistics without CPI expenditure weights.

These statistics give 11 micro features for each retained COICOP1 category:
the fraction of observations with zero price changes, nine deciles of the
non-zero price-change distribution, and the mean non-zero price change. Since
the balanced sample retains 11 COICOP1 categories, stacking these features
across categories, and using 12 lags, yields
$11 \times 11 \times 12 = 1452$ microdata statistics at each point in time.
All series are constructed at the monthly frequency and are not seasonally
adjusted. For consistency with the benchmark forecasts, which also allow for
seasonality, we include target-month indicators as inputs for micro forecasts
so that expert comparisons and scan tests place the forecasts on comparable
footing.

\paragraph{Augmented encoding.} In
Appendix~\ref{app:augmented_encoding}, we also study an augmented encoding
that adds the remaining statistics from a richer price-setting summary. The
additional information consists of:
\begin{enumerate}
   \item Higher moments of the non-zero price-change distribution: variance
   with Bessel's correction ($m_2$), and the standardised third, fourth and
   fifth central moments ($m_3$, $m_4$, $m_5$), where $m_4$ is expressed as
   excess kurtosis.
  \item The nine deciles of time since the last reset, measured among
  observations whose current valid price change is zero and whose
  time-since-reset is defined. The clock resets on each valid non-zero
  price change, on the first valid price observation of a product spell,
  and whenever the gap between two consecutive observations of the same
  product exceeds 12 months.
  \item The mean and the same four higher moments of this restricted
  time-since-change distribution.
\end{enumerate}
Taken together with the baseline statistics, the augmented encoding gives 29
micro features for each retained COICOP1 category and $11 \times 29 = 319$
microdata statistics in each month, the same set of micro variables as the
six-group Shapley decomposition of Appendix~\ref{app:6_group_shapley}.
Consistent with the overfitting concern in
the second issue above, this augmented encoding gives similar but slightly
weaker forecasting performance than the baseline encoding in the final
  three-expert online-learning combination, with a larger gap at the standalone
  micro stage; full per-horizon comparisons are reported in
Appendix~\ref{app:augmented_encoding}.

\subsection{Step 2: Training and Tuning Individual Forecasts}\label{sec:training_tuning}

\paragraph{The Training Framework.}

After this encoding, the input to the forecasting model has fixed length but
remains high dimensional. Let $Z_t$ denote the vector of
microdata statistics observed in month $t$. For horizon
$h \in \{1,\ldots,24\}$ and origin month $t$,
$\hat{\pi}_{t+h\mid t}$ denotes the forecast of inflation in target
month $t+h$ formed with information through $t$. The forecasting
objective is to produce this forecast using the 12 most recent monthly
values of $Z_t$ available at $t$. We use a (non-linear)
local projection approach: for each horizon $h$, we train a separate model,
rather than relying on a single iterative model for all horizons,
as in multi-step-ahead forecasting by VAR.

Formally, define the predictor vector available at origin month $t$,
\[
  X_t^h
  =
  \bigl(Z_{t-11},\, Z_{t-10},\, \ldots,\, Z_{t},\, m_{t+h}\bigr),
\]
which stacks the 12 monthly microdata-statistics vectors observed through
month $t$ together with $m_{t+h} \in \{0,1\}^{11}$, the vector of eleven
target-month indicators (January omitted).

\paragraph{Training.}
To limit the influence of distant history, at origin $t$ we train the forecasting
model on a rolling window containing the most recent (up to) 120 realized target
months, using the pairs $(X_{t-j-h}^h,\, \pi_{t-j})$ for $j = 0,\ldots,119$
whenever available. We then apply the fitted model to $X_t^h$, which yields the
forecast $\hat{\pi}_{t+h\mid t}$. Because training uses only data through $t$,
this is a valid real-time $h$-step-ahead forecast.

\paragraph{The Real-time Tuning Framework.} Training ML algorithms involve multiple tuning hyperparameters (e.g., learning rates) that are crucial for forecasting performance. In cross-sectional prediction problems, it is routine to carve out a random validation set outside the test data to evaluate the machine learning algorithm under different hyperparameters. For time-series forecasting, it is important to respect the temporal ordering. Thus, a random validation set is not ideal. A common exercise in the literature is to (a) fix a set of hyperparameters from a grid, (b) train the model with a rolling window, and (c) evaluate the tuning parameter based on the overall forecasting accuracy, such as the mean absolute error (MAE) or root mean-squared error (RMSE), for the entire sample period. While steps (a) and (b) are standard, step (c) incurs look-ahead bias: the hyperparameter is chosen using forecasting errors from periods that would not yet have been observed when earlier forecasts were made. Since overall accuracy depends on all time steps, this procedure may lead to overfitting and unfair comparisons.

To avoid look-ahead bias, we modify step (c) by introducing a real-time
tuning framework. The idea is to score each candidate hyperparameter, at every
forecast origin $t$, on the 48 most recent target months that have already
realized, rather than on the full sample. This produces a time-varying
performance metric for each candidate. Formally, we start with a grid
$\mathcal{G}$ of values for the hyperparameter. For each origin $t$ and each
candidate value $\Gamma\in \mathcal{G}$, we replicate the $h$-step-ahead
forecasting exercise for the validation targets
$\pi_{t-47}, \pi_{t-46}, \ldots, \pi_t$. The forecast of a validation target
$\pi_{t-j}$, $j = 0,\ldots,47$, is formed at its own origin $t-j-h$: we use
the same rolling-window training procedure at that origin and apply the fitted model to
$X_{t-j-h}^h$. This restriction ensures that each validation forecast uses
only information available at its own origin, so it is a genuine
$h$-period-ahead forecast. Averaging the 48 absolute errors gives a real-time
mean absolute error (MAE) for each candidate $\Gamma$; we choose the
$\Gamma^*$ that minimizes this criterion over the grid $\mathcal{G}$.

Given the selected hyperparameter $\Gamma^*$, the deployed forecast at origin $t$
is the $h$-step-ahead forecast $\hat{\pi}_{t+h\mid t}$ obtained by applying the
rolling-window training procedure above with hyperparameter $\Gamma^*$. We
repeat the procedure at every origin and horizon
$h \in \{1, \ldots, 24\}$, so that the evaluated target months span 2007/01
through 2024/11.

Our framework works with any machine learning algorithm. Since this
predictor set is large relative to the length of the monthly time series, we
use gradient-boosted trees, a machine learning method suited to
high-dimensional covariates and tabular forecasting applications, which outperforms other methods such as random forest or LASSO in settings similar to ours
\citep{Grinsztajn2022}. Gradient boosted trees build a forecast as an additive ensemble of regression trees of moderate depth. Each new tree is fitted to the residual from the previous ensemble; the extent to which each new tree modifies the ensemble is the shrinkage or ``learning rate'', written $\nu$. Our implementation uses XGBoost
\citep{ChenGuestrin2016}.
The parameter grid $\mathcal{G}$ for XGBoost is a $7 \times 7$ grid of 49
candidates: shrinkage values
$\nu \in \{0.0003, 0.001, 0.003, 0.01, 0.03, 0.1, 0.3\}$ combined with tree
counts $n_{\text{est}} \in \{50, 100, 200, 400, 600, 800, 1000\}$, with the
tree depth fixed at the XGBoost default of six. Forecast accuracy is not
sensitive to this choice: coarser and shifted grids deliver mean absolute
errors within about one percent of the baseline grid.

\paragraph{A Computationally More Efficient Tuning Procedure.} The tuning
procedure described above nominally requires refitting the model 48 times for
each candidate hyperparameter vector at every origin. In practice no
refitting is needed, because the validation forecasts at successive
origins overlap: the same forecast can later serve either as a validation
forecast or as the deployed forecast. We therefore precompute, for each
candidate hyperparameter $\Gamma \in \mathcal{G}$ and each horizon, a single rolling sequence
of real-time forecasts: at every origin $t$, a model trained using the pairs
$(X_{t-j-h}^h,\, \pi_{t-j})$ for $j = 0,\ldots,119$, whenever available, produces
the forecast $\hat{\pi}^{\Gamma}_{t+h\mid t}$. Once
$\pi_{t+h}$ realizes, the absolute forecast error is cached. At each
origin $t$, averaging the trailing 48 cached forecast errors for a
candidate hyperparameter gives that candidate's real-time MAE.
Hyperparameter selection therefore reduces to a lookup of these trailing
MAEs, and the deployed forecast is the cached
$\hat{\pi}^{\Gamma^*}_{t+h\mid t}$. Thus, for each horizon, one rolling pass
per candidate replaces the 48 refits that would otherwise be required at every
origin, while preserving the same tuning criterion.
Algorithms~\ref{alg:mlcgi} and~\ref{alg:ml_deploy} in
Appendix~\ref{app:pseudocode} give the pseudocode; the accompanying notation
is collected in Table~\ref{tab:notation}.

\subsection{Step 3: Combining Forecasts with Adaptive Machine Learning}\label{sec:fixed_share}

Steps 1 and 2 produce a microdata-based forecast for each horizon. The remaining task is to decide whether and how this forecast should be combined with existing benchmarks, for which we develop an adaptive machine learning approach.

We study two benchmark forecasts for inflation. Our first benchmark is a univariate SARIMA
model, which historically tends to forecast inflation well (\citealp{AtkesonOhanian2001}; \citealp{StockWatson2007}). However both classic and more recent work also emphasize inflation forecasts based on large numbers of macroeconomic variables (e.g. \citealp{StockWatson2002Diffusion}; \citealp{Medeiros2021}). Therefore our second benchmark uses a large number of macroeconomic variables, combined with the XGBoost algorithm.

We combine forecasts using an adaptive machine learning algorithm: the Fixed Share algorithm of
\citet{HerbsterWarmuth1998}. The algorithm forms a weighted average of forecasts, or experts,
and updates the weights as forecast errors are realized.
This subsection first defines the SARIMA benchmark and the macro forecast,
then describes the Fixed Share algorithm and how to select experts for inclusion.

\paragraph{Benchmark I: SARIMA}\label{subsubsec:SARIMA}

Our first benchmark is a Seasonal Autoregressive Integrated Moving Average
(SARIMA) model estimated on a rolling window, with the order selected in real
time. The SARIMA order is denoted $\Upsilon = (p,d,q)(P,D,Q,M)$, with
fitted parameters $\Theta^{\Upsilon}_{c}$. We search over $p,d,q,P,D,Q \in \{0,1\}$ and fix
$M=12$ to account for monthly seasonality. We also include a candidate that forecasts $h$-step ahead inflation as the trailing
12-month average, $\hat{\pi}_{t+h\mid t} = \frac{1}{12}\sum_{j=0}^{11}\pi_{t-j}$; which we refer to as AO-12MA, following \citet{AtkesonOhanian2001}.

We choose this SARIMA forecast because it nests a wide range of popular, univariate specifications for inflation forecasting in the literature. Many of these specifications are known to forecast inflation as well as, or better than, richer alternatives. Not only does our specification nest the forecast of \cite{AtkesonOhanian2001}, but also the IMA(1,1) forecast of \cite{StockWatson2007}, the ``airline model'' of inflation forecasting used by statistics agencies (e.g. \citealp{Eurostat2015}), and the ARMA(1,1) model of \cite{AngBekaertWei2007}. As we shall see, this benchmark often performs excellently.

Each candidate order is fitted by maximum likelihood on the most recent 120
months of inflation data at every forecast origin. The
ten-year rolling window keeps the model responsive to shifts in inflation
dynamics without allowing distant observations to dominate the fit. Order selection follows the cached real-time tuning of Subsection~\ref{sec:training_tuning}, with the candidate orders $\Upsilon$ taking the role of the XGBoost hyperparameter vector $\Gamma$.
The full procedure is given in Algorithms~\ref{alg:sr1} and~\ref{alg:sr_deploy} in
Appendix~\ref{app:pseudocode}.

A SARIMA model expresses inflation as a linear function of its own lags,
seasonal lags, and moving-average terms, after differencing to remove trends
and seasonal unit roots.\footnote{Formally, a
$\mathrm{SARIMA}(p,d,q)(P,D,Q,s)$ model takes the form
$\Phi(B)\,\Phi_s(B^s)\,(1-B)^d\,(1-B^s)^D\,\pi_t =
\Omega(B)\,\Omega_s(B^s)\,\varepsilon_t$,
where $\Phi,\Phi_s,\Omega,\Omega_s$ are lag polynomials, $B$ is the backshift
operator, and $\varepsilon_t \sim \mathrm{i.i.d.}(0,\sigma^2)$.}
To produce an $h$-step ahead forecast, we iterate the
fitted model forward from the last observed value, replacing each future
innovation with zero (its conditional expectation given the history). Because
the model is univariate and forecasts are iterated rather than direct, a single
fit at each origin yields the forecast at every horizon, in contrast to the
machine-learning experts, which train a separate model for each horizon.

\paragraph{Benchmark II: Macro Forecast}\label{subsubsec:macro_xgboost}

Our second benchmark applies the rolling-window XGBoost procedure of
Algorithms~\ref{alg:mlcgi} and~\ref{alg:ml_deploy} to \emph{macro} rather than micro features.
Concretely, for a forecast formed at origin $t$, the predictors consist of the
twelve most recent monthly observations of each macro
time series described in Section~\ref{sec:data_description}, that is, lags
$h, h+1, \ldots, h+11$ of the target. All other
aspects of the algorithm, including the rolling window length (120 months),
hyperparameter grid, cached real-time tuning, target-month indicators, and direct
$h$-step forecasting, are held fixed. This
approach follows the high-dimensional macro forecasting literature, including \citet{Medeiros2021}, who use machine learning with a large macro panel to
forecast US inflation.\footnote{Our implementation differs from
  \citet{Medeiros2021} in two respects. First, we use gradient boosted trees
  rather than random forests. Second, we use non-seasonally adjusted data throughout to avoid look-ahead bias, whereas
  \citet{Medeiros2021} use seasonally adjusted series.}
\paragraph{Adaptive Combination of Forecasts with Fixed Share.}
Suppose that the scan test detects localized outperformance of a candidate expert compared to the benchmark. Then in principle, the candidate can improve forecasting performance. What remains is to find an adaptive algorithm, which makes use of the candidate when it outperforms the benchmark, but not otherwise. To this end, we combine forecasts using the Delayed Fixed Share algorithm of \citet[Algorithm~4]{korotin2020adaptive}, a variant of the Fixed Share algorithm of \citet{HerbsterWarmuth1998}.\footnote{Most standard online-learning algorithms, such as multiplicative weights (MW; \citealp{CesaBianchiLugosi2006}) and the original Fixed Share algorithm of \citet{HerbsterWarmuth1998}, assume that the outcome used to update weights is observed immediately after the forecast is made. They therefore do not directly apply to $h$-step-ahead forecasting with $h>1$, where the outcome is observed only after an $(h-1)$-step delay. A smaller branch of the literature studies online learning under delayed feedback; see, for example, \citet{joulani2013online} and \citet{korotin2020adaptive}.} The delayed variant is necessary because the loss from an $h$-step-ahead forecast is observed only when its target is realized; updating weights before then would introduce look-ahead bias. For convenience, we refer to this delayed variant simply as Fixed Share throughout.
The Fixed Share algorithm works as follows. Let $K$ denote the number of experts in the admitted pool, indexed by
$k = 1,\dots,K$. For
each horizon $h \in \{1,\dots,24\}$ and forecast origin $t$, let
$\hat{\pi}^{k}_{t+h\mid t}$ denote the $h$-step-ahead forecast of expert $k$ and
$\pi_t$ the realized inflation in month $t$. The algorithm maintains a probability vector of
weights $\{w^{h,k}_t\}_{k=1}^K$ over experts and uses these weights to form the
aggregate forecast
\[
  \hat{\pi}^{\mathrm{FS}}_{t+h\mid t}
  \;=\;
  \sum_{k=1}^{K} w^{h,k}_t\,\hat{\pi}^{k}_{t+h\mid t}.
\]

The loss on an $h$-step forecast is not realized until $h$ months after the
forecast is formed, so the weights must be updated using only losses that have
been realized by the origin. The algorithm therefore maintains, for each
horizon $h$, an internal weight vector $\{v^{h,k}_t\}_{k=1}^K$ over experts,
indexed by the \emph{target} month, the month whose inflation is being
forecast.
When the inflation rate of month $t$ becomes available, the $h$-step forecasts
of target $t$, formed at origin $t-h$, are evaluated using their absolute
forecast errors $\ell^{k}_{t\mid t-h} = |\pi_t - \hat{\pi}^{k}_{t\mid t-h}|$, and the
internal weights are updated in two steps. First, a Fixed Share mixing step
redistributes a fraction $\alpha \in [0,1]$ of the weight uniformly across all
experts,
preventing any expert from being permanently abandoned. Second, a
multiplicative weights (MW) step down-weights experts in proportion to their
realized absolute forecast error:
\[
  v^{h,k}_{t}
  \;\propto\;
  \Bigl[(1-\alpha)\,v^{h,k}_{t-1} + \frac{\alpha}{K}\Bigr]
  \exp\!\bigl(-\eta\,\ell^{k}_{t\mid t-h}\bigr),
  \qquad
  \sum_{k=1}^{K} v^{h,k}_{t} = 1,
\]
where $\eta > 0$ is the learning rate. The forecast weights used at origin
$t$ are obtained by propagating these internal weights through the $h$ months
separating the origin from its target, during which no further loss of horizon
$h$ is realized, so only the mixing step operates:
\[
  w^{h,k}_t
  \;=\;
  (1-\alpha)^h\, v^{h,k}_{t}
  \;+\;
  \frac{1-(1-\alpha)^h}{K}.
\]
The two parameters control distinct aspects of the algorithm's responsiveness.
A larger $\eta$ makes the multiplicative step more reactive: weight moves more
quickly toward experts with lower recent losses, at the cost of responding more
to short-run noise. A smaller $\eta$ smooths the update and retains more of the
prior weight. A larger $\alpha$ keeps more weight spread across experts through
the mixing step, which lets previously dominated experts regain influence once
their relative performance improves. The internal weights are initialized
uniformly, $v^{h,k}_{t_0} = 1/K$, and remain uniform until the first $h$-step
loss is realized, so early forecasts start from equal weights. The full path of
forecasts and weights for a
given $\eta$ is produced by Algorithm~\ref{alg:fixedshare_path} in
Appendix~\ref{app:pseudocode}.

A chief attraction of the Fixed Share algorithm is its interpretability. If one expert is particularly relevant for forecasting, its weight is high. Therefore one can inspect the sequence of weights to understand when certain data modalities, such as microdata, are relevant.\footnote{The weights also have a natural Bayesian interpretation. For a fixed horizon $h$, let the identity of the best expert be a latent state $n_s$ at time $s$. The Fixed Share algorithm considers a Markov prior on the sequence of $n_s$ with
$p(n_s=n\mid n_{s-1}=m)=(1-\alpha)\mathbf{1}\{n=m\}+\alpha/K$,
which says that this latent state usually persists but can switch uniformly across experts. Let $L_{t\mid s}=(\ell^1_{t\mid s},\ldots,\ell^K_{t\mid s})$ denote the vector of realized expert losses for target month $t$ from origin $s$, where $\ell^j_{t\mid s}=\ell(\pi_t,\hat{\pi}^j_{t\mid s})$. When the $h$-step loss for target month $s$ is realized, exponential weighting is equivalent to the likelihood
$p(L_{s\mid s-h}=\boldsymbol{\ell}\mid n_{s-h}=n)\propto \exp\{-\eta \ell\}$,
where $\ell$ denotes the $n$th component of $\boldsymbol{\ell}$.
In this notation, $v_t^{h,n}=p(n_{t-h}=n\mid L_{h+1\mid1},\ldots,L_{t\mid t-h})$, and the forecast-origin weight is $w_t^{h,n}=p(n_t=n\mid L_{h+1\mid1},\ldots,L_{t\mid t-h})$. We refer to \citet{korotin2020adaptive} for details.}

\paragraph{Regret Guarantee.}
The Fixed Share algorithm provides a worst-case performance guarantee,
clarifying precisely what the algorithm optimizes in a non-stationary
environment. The relevant object is \emph{dynamic regret}: the gap between the
cumulative loss of the online learner and the cumulative loss of an ideal
benchmark sequence of experts chosen in hindsight after seeing the full sample.
Formally, let
\[
  L_n^{\mathrm{FS},h}
  =
  \sum_{t=1}^{n} |\hat{\pi}_{t\mid t-h}^{\mathrm{FS}} - \pi_t|
\]
denote the cumulative loss of the Fixed Share forecast over $n$ periods. Let
$L_n^{(k_0,k_1,\ldots,k_S),h}$ denote the cumulative loss of the best sequence of
experts $k_0, k_1, \ldots, k_S$ in hindsight, when we allow the identity of the expert to switch
$S$ times over the sample. The dynamic regret of the algorithm is then defined as
\[
  \mathrm{Reg}_{S,n}^h = L_n^{\mathrm{FS},h} - \min_{k_0, k_1, \ldots, k_S}L_n^{(k_0, k_1, \ldots, k_S),h}.
\]
The standard notion of regret is a special case with $S=0$, which measures the comparison with the best ex-post single expert forecast \citep{orabona2019modern}.\par Let $\pi_t, \hat{\pi}_{t\mid t-h}^{\mathrm{FS}}$ be bounded. Tailoring Corollary~1 of \citet{korotin2020adaptive} to the Fixed Share update and our setting, we can easily obtain the following result:
\begin{equation}
  \mathrm{Reg}_{S,n}^h = O\left(\eta nh + \frac{1}{\eta}\left\{(S + 1)\log K + S \log \frac{1}{\alpha}
  + (n - S - 1)\log \frac{1}{1-\alpha}\right\}\right).
\end{equation}
where $K$ is the number of experts, $\eta$ is the learning rate, $\alpha$ is
the Fixed Share mixing parameter. Crucially, the guarantee holds for \textit{any}
sequence of outcomes. It does not require stationarity, ergodicity, or any
specific model of the data generating process, which is exactly why the method
is attractive for inflation forecasting in a non-stationary environment, with machine learning-based experts.

Each term on the right-hand side has a simple interpretation. The first term, $\eta nh$, is the cost of using a positive learning rate in the exponential-weighting update; it accumulates over the $n$ target months and grows with the forecast horizon $h$ because feedback arrives with delay. The bracketed terms constitute the complexity cost of tracking a switching comparator: $(S+1)\log K$ reflects the number of experts available on each segment, while the two $\alpha$ terms price the chosen pattern of switches and non-switches. Thus, adding experts raises the regret bound through the $\log K$ penalty. As a consequence, one should include only experts whose marginal contribution is non-negligible, in the sense that the decrease in the comparator loss outweighs the increase in the $\log K$ penalty term. The last two terms imply $\alpha$ should be small when $S <\!\!< n$. In particular, it is straightforward to show that the optimal $\alpha$ is given by $S/(n-1)$. In our empirical analysis, we set $\alpha = 0.02$, which is approximately $4/n$ at our sample size.

To gain more insight into the impact of the learning rate $\eta$, we set $\alpha$ to be the optimum. Then the regret can be written as
\begin{equation}
  \mathrm{Reg}_{S,n}^h = O\left(\eta nh + \frac{\log(Kn)}{\eta}\right).
\end{equation}
The tension
between the two pins down the role of the learning rate: a small $\eta$
yields stable weights but adapts too slowly, while a large $\eta$ reallocates across experts
quickly but overreacts to short-run noise.

Fixed Share can be understood as a generalization of the multiplicative weights
(MW) algorithm from the online learning literature \citep{CesaBianchiLugosi2006}.
The MW algorithm provides a similar regret guarantee, but only relative to the
single best expert in hindsight. That is too restrictive for our application:
we do not expect the same forecast to be best both before and after large
changes in inflation dynamics. The reason MW may fail in this environment is that
it updates weights multiplicatively based on losses, so an expert that performs
poorly early on sees its weight decay exponentially toward zero. Even if that
expert later becomes the best, it cannot recover meaningful weight. The key
innovation of Fixed Share is the mixing step governed by $\alpha$: after each
multiplicative update, a fraction of weight is redistributed across all
experts. This keeps every expert alive, so the algorithm can track switches in
the identity of the best forecast source over time.

It might be tempting to consider a more straightforward forecast which applies XGBoost to data from all three experts,  and lets a machine learning (ML) algorithm perform variable selection. We will perform this exercise as robustness, but it is not our baseline. While this practice is more common in the ML literature, it is prone to overfitting---the sample size for inflation forecasting is extremely low and most predictors are not highly predictive. More formally, equation \eqref{eq:risk_decomp} suggests that the overfitting cost of blindly combining modalities could easily outweigh the benefit for challenging forecasting problems. In contrast, the Fixed Share algorithm serves as a ``router'' to average multiple experts which rely on a subset of predictors. It is analogous to the sparse mixture-of-experts architecture that proves to be much more effective than dense feed-forward layers in training foundation AI models \citep{shazeer2017outrageously,lepikhingshard}. This approach lets the forecasting system retain a simple benchmark while allowing additional experts to receive weight only when their recent performance justifies it.

Furthermore, it is unrealistic to expect practitioners to replace a long-standing forecasting model, such as SARIMA, with a new and more black-box algorithm that has not yet been evaluated in real time, solely on the basis of backtesting performance. Our framework allows the forecaster to retain the status quo model, introduce new data modalities or machine-learning models as additional experts with zero or small initial weights, and then update these weights adaptively using the Fixed Share algorithm. The resulting weight shifts also provide a succinct and interpretable time-varying summary of the relative performance of each expert, whether the incumbent forecast or a newly introduced one, without imposing stationarity.

\paragraph{Fixed Learning Rate.}
Our baseline fixes $\eta = 0.5$ and $\alpha = 0.02$ throughout, with the
internal weights initialized uniformly at $1/K$; the combined forecasts cover
target months from January 2007 onward, with each forecast formed $h$ months
earlier at its forecast origin. This
specification serves as a transparent benchmark: the fixed learning rate is
consistent with standard practice in the online learning literature
\citep{HerbsterWarmuth1998}, and holding both $\eta$ and $\alpha$ fixed
reduces the implementation to a single pass of
Algorithm~\ref{alg:fixedshare_path} (Appendix~\ref{app:pseudocode})
with $\eta = 0.5$.

\paragraph{Varying Learning Rate.}

While a fixed $\eta = 0.5$ is our baseline, the right degree of
responsiveness may itself change over time: periods of rapid structural change
call for a larger $\eta$ so that weights shift quickly toward accurate experts,
while stable periods favor a smaller $\eta$ to avoid overreacting to noise.
As a robustness check, we therefore also consider a specification with a
varying learning rate, which
re-selects $\eta$ at each forecast origin based on realized forecast performance.

The procedure works as follows. Fix a grid
$\mathcal{H} = \{\eta^{(1)} < \eta^{(2)} < \cdots < \eta^{(G)}\}$
of $G$ candidate learning rates. For each $\eta^{(g)} \in \mathcal{H}$, the
Fixed Share algorithm is run in parallel, producing a path of combined
forecasts
$\hat{\pi}^{\mathrm{FS},(\eta^{(g)})}_{t+h\mid t}$. At each forecast origin $t$, we select the learning rate
$\eta$ whose path has delivered the lowest average realized error over the most
recent $W$ target months. Rather than re-evaluating all $G$ candidates at every step, we
restrict attention to the previously chosen value $\eta_{t-1}^\star$ and its
immediate grid neighbors, reflecting the assumption that the optimal learning
rate evolves gradually. Formally, let $g_{t-1}$ denote the index of
$\eta_{t-1}^\star$ in $\mathcal{H}$. The candidate index set at time $t$ is
\[
  \mathcal{I}_t
  \;=\;
  \bigl\{\max(1,\, g_{t-1}-1),\; g_{t-1},\; \min(G,\, g_{t-1}+1)\bigr\},
\]
where the $\max$ and $\min$ prevent the index from falling outside the grid
boundaries. Among these (at most three) candidates, we choose the one with
the smallest rolling-window MAE of realized $h$-step errors over the most
recent $W$ target months:
\[
  \eta_t^\star
  \;=\;
  \eta^{(g_t)},
  \qquad
  g_t = \arg\min_{g \in \mathcal{I}_t}
  \frac{1}{W}\sum_{s=t-W+1}^{t}
  \bigl|\pi_s - \hat{\pi}^{\mathrm{FS},(\eta^{(g)})}_{s\mid s-h}\bigr|,
\]
with ties kept at the incumbent value. Every term in the sum is available at
the origin: the forecast of target $s$ was formed at origin $s-h$, and $\pi_s$
is realized by $t \ge s$.
In our implementation, $\mathcal{H}$ consists of $G=21$ values
$\{0.25, 0.30, \ldots, 1.25\}$ in steps of $0.05$, with rolling window $W=60$
months and initial value
$\eta_0 = 0.5$, the fixed-rate baseline. The complete procedure is described in
Algorithms~\ref{alg:fixedshare_path}
and~\ref{alg:fs_precompute_neighbor_select_eta} in Appendix~\ref{app:pseudocode}.\footnote{We have considered other algorithms for updating $\eta$, such as AdaHedge (\citealp{deRooijVanErvenGrunwaldKoolen2014}); however they performed poorly in our setting.}

\paragraph{Scan Test for Expert Selection.}

As discussed above, to maximize the performance gain
from the Fixed Share algorithm, we need to choose experts judiciously such that
the improvement in the comparator loss $L_n^{(k_0, k_1, \ldots, k_S),h}$ exceeds
the $\log K$ penalty increase. We therefore use a sequential expert inclusion
procedure based on the scan test introduced in Section~\ref{sec:scantest}.
The idea is to start with a benchmark and add a new expert
only if it outperforms every expert selected so far, in the sense that either
the occasional-outperformance test or the overall-outperformance test
(Subsection~\ref{sec:scan_test_uses}) rejects
the corresponding no-outperformance null at the 5\% nominal level.

In practice, the most common forecast is given by the univariate SARIMA model,
which will always be considered as the first expert. We consider the macro
forecast as the second expert if it outperforms the univariate forecast. The
micro model described in
Sections~\ref{sec:encoding_training_tuning} and \ref{sec:training_tuning} joins the
pool of experts only if it outperforms both the univariate and the macro
forecasts, each used as the comparison forecast. Effectively, we set a rather
high bar for the usage of microdata. As a consequence, the inclusion of the
microdata-based expert suggests strong statistical evidence that the microdata
bring new information that is helpful for inflation forecasting.
 
In our expert-selection scan tests, we use $B=4{,}999$ bootstrap draws applied
to a panel of $n=215$ monthly observations spanning January 2007 through
November 2024, across $H=24$ forecast horizons. The occasional-outperformance
scan uses $L_{\min}=6$ and $L_{\max}=12$, giving
$\sum_{L=6}^{12}(n-L+1)=1{,}449$ candidate intervals and
$H \times 1{,}449 = 34{,}776$ interval-horizon pairs; the
overall-outperformance test uses the single full-sample interval,
$L_{\min}=L_{\max}=n$.

\section{Results: Forecasting Inflation with Microdata}\label{sec:results}

We now present the paper's four main empirical results. First, the micro
forecast outperforms the univariate benchmark only in the post-2020 window.
Before then, the univariate benchmark performs better at every horizon;
afterward, the micro forecast wins at nearly every horizon, with the largest
gains at long horizons.
Second, applying scan tests, we find that the Fixed Share combined forecast
should include all three modalities: micro, macro and univariate forecasts.
Third, this combined forecast performs comparably to the univariate
benchmark before 2020 and better than it at every horizon after 2020. Fourth,
we confirm the time-varying value of microdata. The combined forecast leans on
the univariate forecast before 2020 and shifts weight toward the micro forecast
afterward, especially at short and medium horizons. The full three-expert
combined forecast also outperforms the two-expert benchmark that excludes
microdata, on average, only in the post-2020 window.

\subsection{Micro Forecasts versus the Univariate Benchmark}
\label{sec:micro_vs_sarima}

We now present the first result: the micro forecast outperforms the univariate benchmark, but only after 2020. Figure~\ref{fig:xgb_micro_vs_sarima_h12} visualizes it. We plot the 12 month ahead forecast of monthly inflation, for both the univariate forecast and the microdata forecast, as well as actual inflation in the corresponding period. For interpretability, we take a three month moving average of all three series. Clearly, before 2020 the univariate forecast performs excellently. The microdata forecast has notably larger losses, especially in 2018-2019. However, after 2020 the micro forecast performs better. Neither forecast captures the rise in inflation in 2021-22, which according to both models was due to unforecastable shocks. However the microdata-based model does a better job of forecasting the subsequent propagation of these shocks, during the decline in inflation of 2023-24. The micro model predicts a gradual decline, whereas the univariate model predicts that inflation will be persistently high. At the end of the sample, as inflation becomes less volatile again, the univariate model again starts to perform better.

\begin{figure}[t!]
  \centering
  \caption{Micro vs.\ Univariate Forecast, 12 Months Ahead}
  \label{fig:xgb_micro_vs_sarima_h12}
  \includegraphics[width=\textwidth]{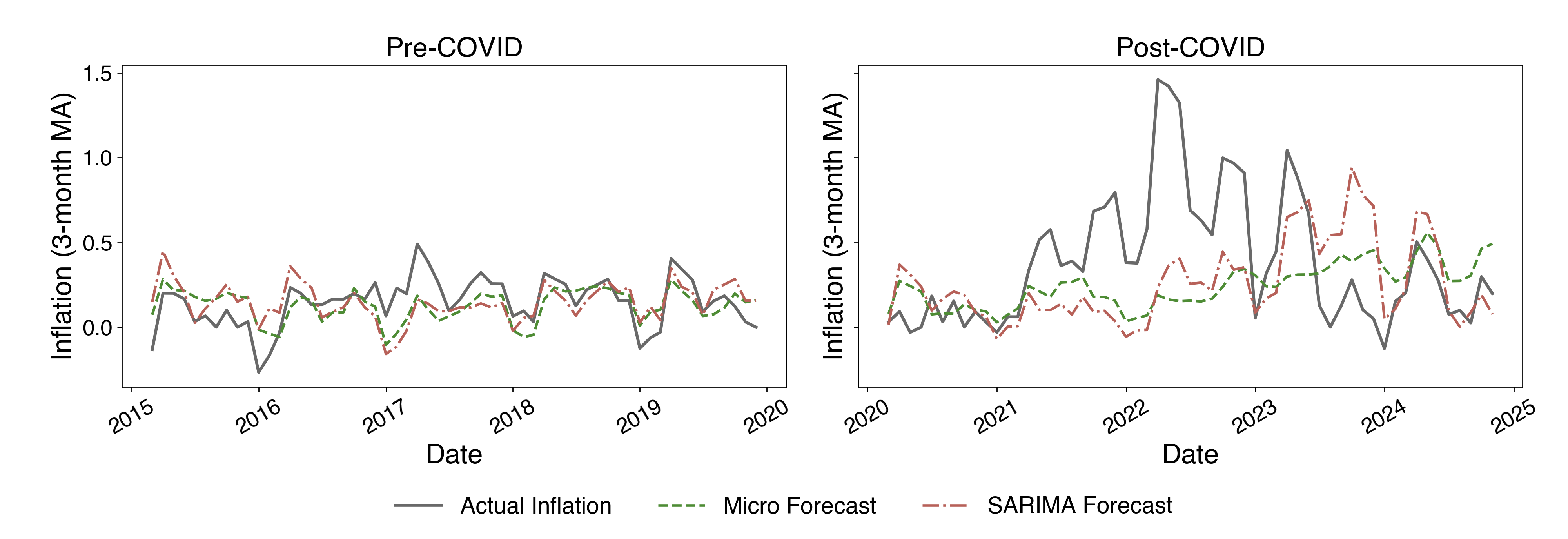}
  \smallskip
  \begin{minipage}{\textwidth}\setstretch{1.0}
    \footnotesize\textit{Notes:} Each panel plots 3-month trailing moving
    averages of realized monthly inflation, the micro
    12-month-ahead forecast, and the univariate (SARIMA) 12-month-ahead
    forecast. Dates on the x-axis are target months: a forecast dated
    $t$ was formed at origin $t-12$. The panels
    correspond to the 2015--2019 and 2020--2024 evaluation windows.
  \end{minipage}
\end{figure}

We then identify the forecast horizons that account for the micro model's post-2020 outperformance. Figure~\ref{fig:micro_vs_sarima_improvement_heatmap} presents a heatmap. The x axis is the horizon of the forecast; the two rows are the pre-2020 and post-2020 evaluation windows, and each cell reports the percent improvement in the mean absolute error (MAE) of the micro forecast relative to the univariate forecast. Evidently, the micro forecast underperforms at all horizons before 2020; whereas after 2020 the microdata outperforms at 23 of the 24 horizons, and by the most at long horizons, where the improvement reaches about 40 percent. Tables~\ref{tab:sarima_mae_two_windows} and~\ref{tab:fine_micro_vs_sarima_improvement} confirm these results by reporting the univariate-forecast MAE and the relative MAE improvement of the micro forecast; Appendix Figures~\ref{fig:xgb_micro_vs_sarima_h1}--\ref{fig:xgb_micro_vs_sarima_h24} plot the time series of forecasts at $h=1, 6,$ and $24$.
 
The occasional-outperformance scan test confirms the first result. It tests the null hypothesis that
microdata do not outperform the univariate forecast at any horizon or during
any time period. The Ljung--Box pre-test rejects serial
independence for this comparison ($p=0.034$), so the critical values use the
Regime~D bootstrap. The test rejects, with a bootstrap $p$-value of 0.041.
The scan test therefore provides statistical evidence that
microdata occasionally outperform the univariate forecast.

\begin{figure}[t!]
\centering
\caption{Micro Forecast Improvement Relative to Univariate}
\label{fig:micro_vs_sarima_improvement_heatmap}
\includegraphics[width=\textwidth]{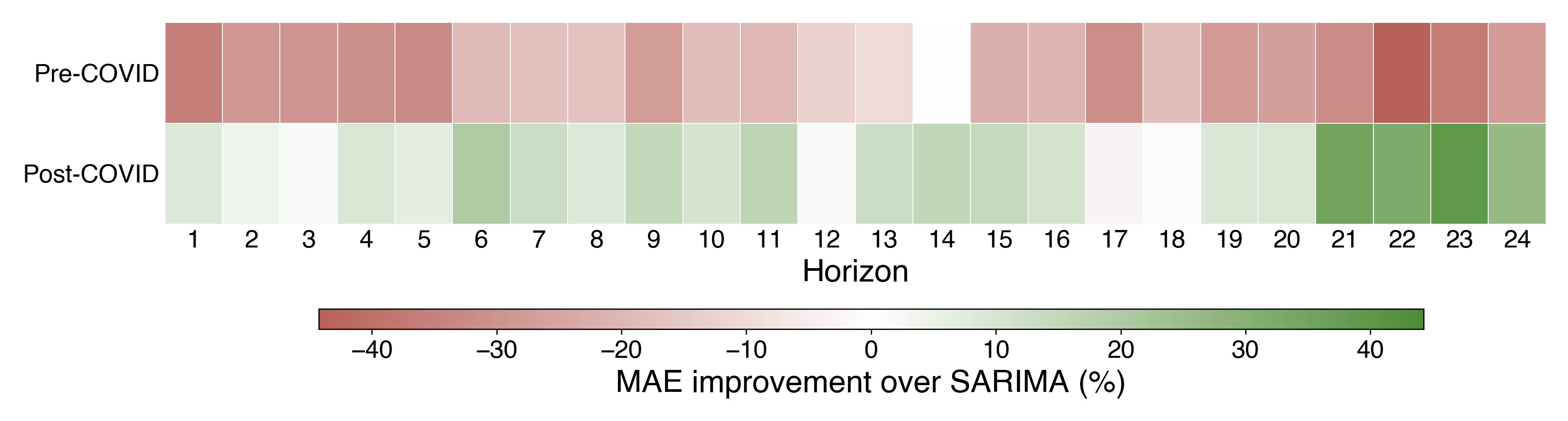}
  \smallskip
  \begin{minipage}{\textwidth}\setstretch{1.0}
    \footnotesize\textit{Notes:} Cells report the percentage improvement in
    mean absolute error (MAE) of the micro forecast relative to the univariate forecast by
    horizon $h=1,\ldots,24$ and evaluation window. Positive values indicate
    lower MAE for the micro forecast than for the univariate forecast; negative values indicate the
    reverse. Green shading denotes positive improvements and red shading
    denotes negative improvements. Numerical values are reported in
    Table~\ref{tab:fine_micro_vs_sarima_improvement}.
  \end{minipage}
\end{figure}

\subsection{Selecting Experts with the Scan Test}

We now present the second result: scan tests imply that the Fixed Share combined forecast should include all three experts: the micro, macro and univariate forecasts. As discussed in Section~\ref{sec:scantest}, we select experts sequentially. Starting from the univariate SARIMA benchmark, a candidate forecast is admitted only if it outperforms every expert selected so far, where outperformance means that either the occasional-outperformance test or the overall-outperformance test rejects the corresponding no-outperformance null at the 5\% level. We first test whether the macro forecast outperforms the univariate benchmark. We then test the micro forecast separately against the univariate benchmark and the macro forecast.

Appendix~\ref{app:macro_vs_sarima} contains details on the comparison between the macro forecast and the univariate benchmark. Before 2020, the macro forecast has a worse MAE than the univariate benchmark at all 24
horizons. After 2020, the macro forecast improves on the univariate forecast at 19 of the 24 horizons, by 10.5\% on average, with the largest percent gain in the MAE reaching 41.9\% at $h=23$. The scan test confirms that the macro expert adds forecasting power: the bootstrap $p$-values are 0.037 for occasional outperformance and 0.024 for overall outperformance.\footnote{For this comparison, the Ljung--Box pre-test does not reject serial independence ($p=0.223$), so both tests use the Regime~I bootstrap.} We therefore add the macro forecast to the selected expert set.

We next apply the same sequential rule to the micro forecast. As shown in the
preceding subsection, the occasional-outperformance test rejects the null of
no occasional outperformance of the micro forecast relative to the univariate
benchmark
($p=0.041$). Against the macro
forecast,\footnote{For the comparison with the macro forecast,
the Ljung--Box pre-test does not reject serial independence ($p=0.226$), so
the corresponding tests use the Regime~I bootstrap.} the overall-outperformance test
rejects the null of no overall outperformance of the micro forecast
($p=0.026$).\footnote{Against the univariate forecast, the
overall-outperformance test does not reject the null that the micro forecast
fails to improve on the univariate forecast over the full sample at any
horizon ($p=0.152$). Against the macro forecast, the
occasional-outperformance test does not reject the null that the micro forecast
fails to improve on the macro forecast over any admissible interval and
horizon ($p=0.258$). These results illustrate the tradeoff discussed in
Section~\ref{sec:scan_test_uses}. The occasional test targets a stronger
localized signal but incurs more noise from searching over intervals. The
overall test smooths over time, reducing noise but also weakening the signal.}
Thus, the micro
forecast passes the sequential screen against both previously selected
experts, through occasional outperformance relative to the univariate forecast
and overall outperformance relative to the macro forecast. The sequential rule
therefore admits the micro forecast, and all three experts enter the combined
forecast.

\subsection{Combining Experts with Fixed Share}

We combine the three forecasts, micro, macro and univariate, using Fixed Share. This exercise leads to our third result: the combined forecast performs comparably to the univariate forecast before 2020, and better than it at every horizon after 2020.

\begin{figure}[t!]
  \centering
  \caption{Combined Forecast and Combination Weights, 12 Months Ahead}
  \label{fig:online_h12}
  \includegraphics[width=\textwidth]{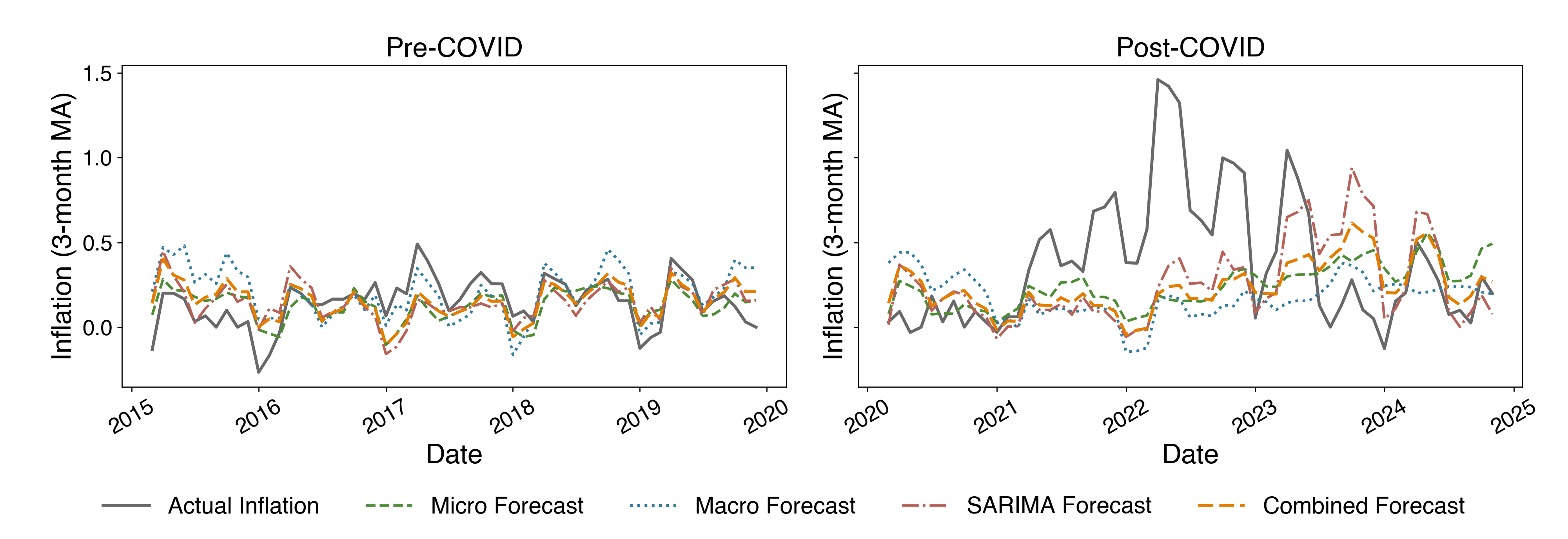}

  \medskip

  \includegraphics[width=\textwidth]{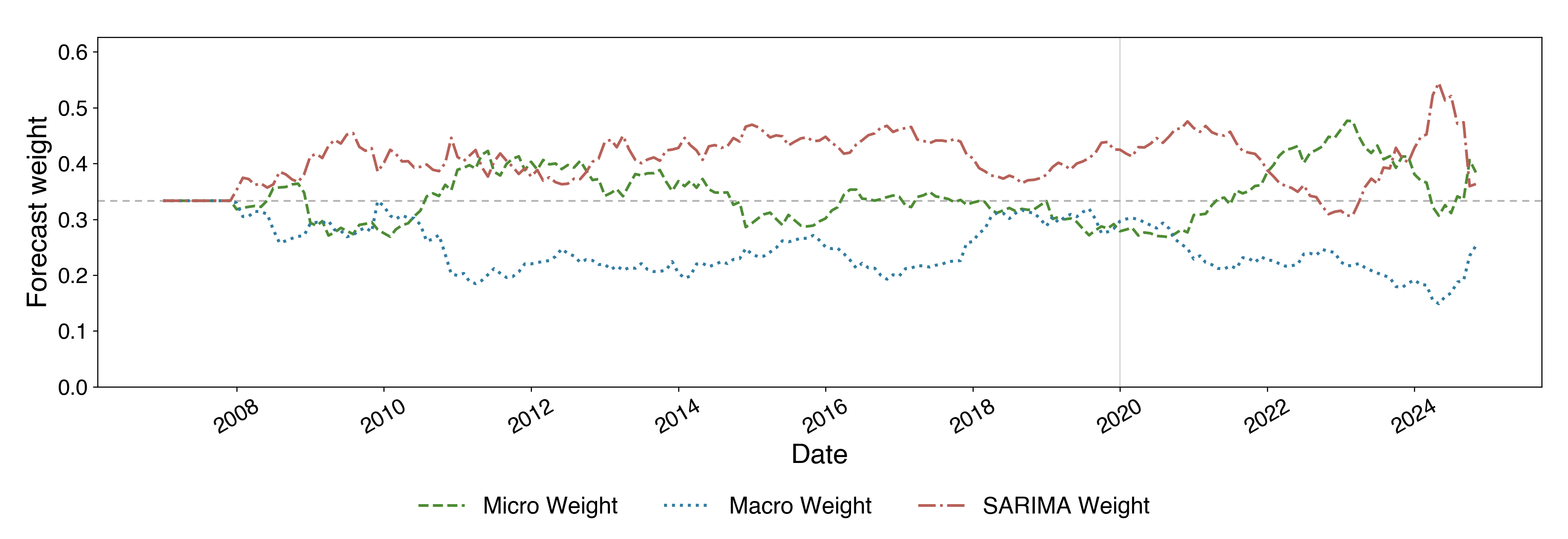}
  \smallskip
  \begin{minipage}{\textwidth}\setstretch{1.0}
    \footnotesize\textit{Notes:} The top panel plots 3-month trailing moving
    averages of realized inflation, the three expert forecasts (micro, macro,
    and univariate), and the Fixed Share combined forecast
    for the pre-2020 (2015--2019) and post-2020 (2020--2024) windows; dates
    on the x-axis are target months, so a forecast dated $t$ was formed
    at origin $t-12$. The bottom panel plots the
    combination weights on the three experts over the full sample, with the
    dashed horizontal line at the equal weight $1/3$ and the vertical line
    marking 2020.
  \end{minipage}
\end{figure}
  
Our baseline implementation holds the learning rate fixed at $\eta=0.5$, with
$K=3$ experts and $\alpha=0.02$, and combined forecasts covering target months
from January 2007 onward. As a robustness
check, we also run the varying-learning-rate version of
Algorithm~\ref{alg:fs_precompute_neighbor_select_eta}, described in
Subsection~\ref{sec:fixed_share}.\footnote{This version re-selects the
learning rate at each origin using the neighbor-only rolling search with
$W=60$ months. Its accuracy is nearly indistinguishable from the fixed-rate
baseline: it has lower MAE at roughly half of the horizons and higher MAE at
the rest, and the average gap is within 0.3 percent in both windows
(Appendix~\ref{app:three_expert_adaptive},
Table~\ref{tab:adaptive_mae_two_windows}). We therefore report the fixed-rate
baseline throughout.} Results are reported for the pre-2020 (2015--2019) and
post-2020 (2020--2024) evaluation windows.
 
Figure~\ref{fig:online_h12}, top panel, plots the 12-month-ahead forecast from
the Fixed Share combination, together with the three expert forecasts. The
adaptive combination performs well because it leans on the univariate forecast
when that benchmark is useful and shifts toward the micro forecast when
microdata add predictive content. Before 2020,
it stays close to the univariate forecast: in 2017--18, for instance, the
combined and univariate forecasts nearly coincide. During the 2023 disinflation, when the micro forecast is most useful, the
combined forecast adaptively breaks away from the univariate forecast, which
overpredicts inflation, and moves toward the micro forecast.

We confirm the advantages of the combined forecast in the heatmap of
Figure~\ref{fig:online_three_expert_improvement_heatmap}. The heatmap shows
that prior to 2020, the combined forecast has an MAE comparable to the
univariate model: the average gap across horizons is 1.2 percent, and no
shortfall exceeds 9 percent. After 2020, the combined forecast has a lower MAE
at every horizon, by 16 percent on average and by around 30 percent at the
longest horizons. Table~\ref{tab:fs_mae_two_windows} reports the underlying
MAEs of the combination and of each expert, by horizon and window.

\begin{figure}[t!]
\centering
\caption{Combined Forecast Improvement Relative to Univariate}
\label{fig:online_three_expert_improvement_heatmap}
\includegraphics[width=\textwidth]{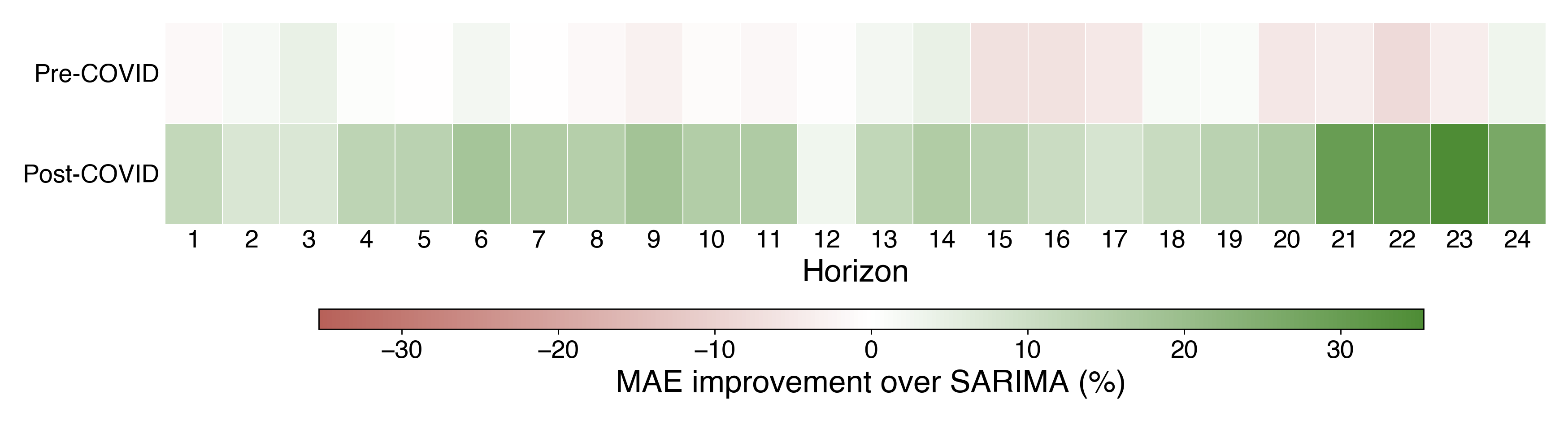}
  \smallskip
  \begin{minipage}{\textwidth}\setstretch{1.0}
    \footnotesize\textit{Notes:} Cells report the percentage MAE improvement
    of the three-expert Fixed Share combination relative to the univariate forecast
    by forecast horizon $h=1,\ldots,24$ and evaluation window. The three
    experts are the micro forecast, the macro forecast, and the univariate forecast. Positive values
    indicate lower MAE for the combination than for the univariate forecast;
    negative values indicate the reverse. Numerical values are reported in
    Table~\ref{tab:fs_mae_two_windows}.
  \end{minipage}
\end{figure}

\subsection{The Time Varying Value of Microdata}

We now present our fourth result, which confirms the time varying value of
microdata. One way to see the value of microdata is to inspect the weights of
the Fixed Share combination, which are plotted for the 12-month-ahead forecast in
the bottom panel of Figure~\ref{fig:online_h12}. The weight on
microdata is below its equal-weight benchmark of $1/3$ in the years
before 2020. After 2020 it rises, and it peaks in 2023, the period in which
the micro forecast is most informative about the disinflation. The pattern is
horizon specific: the average micro weight roughly doubles after 2020 at
$h=6$, from 0.21 to 0.41, rises from 0.32 to 0.36 at $h=12$, and drifts
slightly down at the longest horizons.\footnote{The
corresponding weight paths at other horizons are reported in
Appendix~\ref{app:figures}.}

We next show the time varying value of microdata by comparing the MAE of the
three-expert combined forecast to the two-expert benchmark.
Figure~\ref{fig:three_vs_two_expert_improvement_heatmap} presents a heatmap.
Before 2020, adding the micro forecast leaves the MAE essentially unchanged
on average: it helps modestly at some short and medium horizons and hurts at
the longest ones. After 2020, adding the micro forecast lowers MAE at 19 of
the 24 horizons, by 2.9 percent on average and by up to 9.5 percent.

We use the alternative, more demanding application of the scan test described
in Subsection~\ref{sec:scan_test_uses}. Specifically, we test the
three-expert combined forecast against the two-expert benchmark that combines
the univariate and macro forecasts, without microdata. The question is whether
adding the microdata expert improves forecast accuracy \textit{on average},
rather than only at some point in time. We apply the horizon-pooled test
described in Subsection~\ref{sec:scan_test_uses} separately before and after
2020 (for 2015--19 and 2020--24). Specifically, we study the
horizon-averaged loss differential
\[
  \bar D_t
  =
  \frac{1}{H}\sum_{h=1}^{H}D_t^{(h)}
  =
  \frac{1}{H}\sum_{h=1}^{H}
  \Bigl(\bigl|\varepsilon_{t,\text{two-expert}}^{(h)}\bigr|
  - \bigl|\varepsilon_{t,\text{three-expert}}^{(h)}\bigr|\Bigr),
\]
with $H=24$. The baseline regime for each window is chosen by the
multiplier-bootstrap Ljung--Box diagnostic computed on that window.

For the pre-2020 window, the Ljung--Box test does not
reject ($p=0.301$), so the baseline uses Regime~I. The test fails
to reject the null of no average improvement over the two-expert benchmark,
with a Regime~I bootstrap $p$-value of 0.548.
Adding the micro expert therefore
does not improve on the two-expert forecast on average before 2020.

\begin{figure}[t!]
  \centering
  \caption{Improvement from Adding the Micro Forecast, Relative to Two Experts}
  \label{fig:three_vs_two_expert_improvement_heatmap}
  \includegraphics[width=\textwidth]{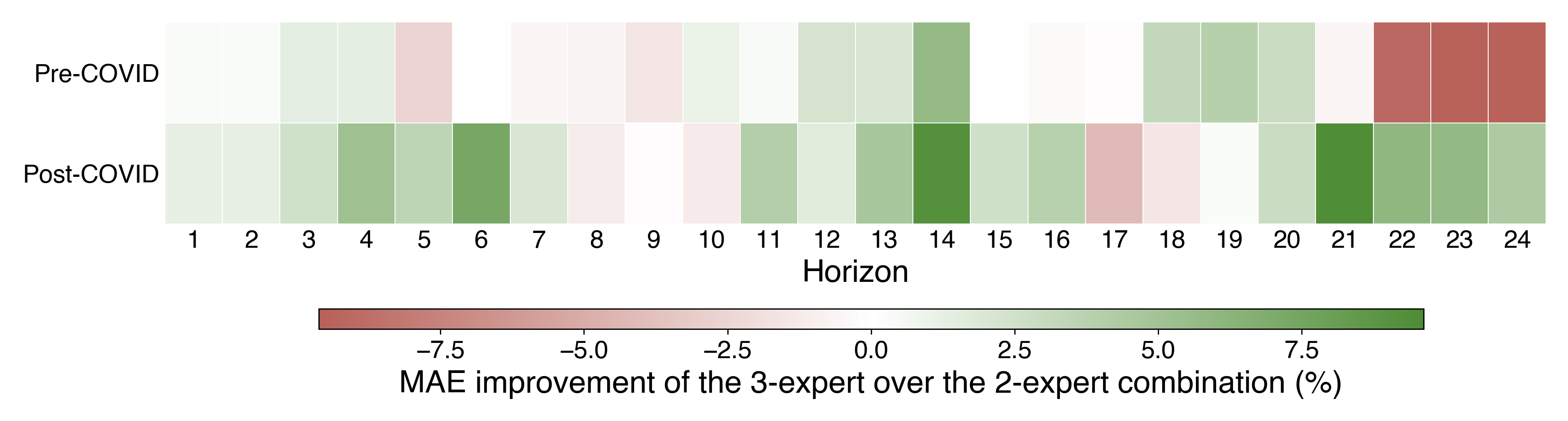}
  \smallskip
  \begin{minipage}{\textwidth}\setstretch{1.0}
    \footnotesize\textit{Notes:} Cells report the percentage MAE improvement
    of the three-expert Fixed Share combination relative to the two-expert
    Fixed Share benchmark by forecast horizon $h=1,\ldots,24$ and evaluation
    window. The three-expert combination uses the micro forecast, the macro
    forecast, and the univariate forecast; the two-expert benchmark uses only
    the macro and univariate forecasts. Positive values indicate lower MAE after
    adding the micro forecast; negative values indicate higher MAE. Numerical
    values are reported in
    Table~\ref{tab:three_vs_two_mae_improvement}.
  \end{minipage}
\end{figure}

After 2020, microdata is valuable for forecasting. For the post-COVID window,
the Ljung--Box test does not reject serial independence
($p=0.709$), so the baseline uses Regime~I. The test rejects decisively, with
a bootstrap $p$-value of 0.001. The full
three-expert combination therefore adds significantly to the two-expert
benchmark on average after 2020.

This exercise moves beyond the previous uses of the scan test, for expert selection. Those tests
showed that microdata improved upon the macro and univariate forecasts for at
least some horizon and interval. Now, we have shown a stronger claim:
the three-expert combined forecast, including microdata, improves upon the
corresponding combination without microdata, on average after 2020.

\subsection{Combined Three-Expert Forecast vs. a Single XGBoost Forecast}

A natural alternative to the combined three-expert forecast is a
single XGBoost forecast built from one pooled nonlinear regression. This pooled
forecast uses the full feature set across the three experts: the 121
micro distributional statistics, the macro covariates (the rents CPI
permutation), and lagged inflation, for 158 pre-lag predictors in total. The
comparison isolates the value of using Fixed Share to combine experts,
relative to a pooled XGBoost forecast based on the same information.

The single XGBoost forecast performs much worse before 2020, and moderately
worse afterwards. Before 2020, the percent MAE gaps between the combined
three-expert forecast and the single XGBoost forecast are $+25.7$ (on average
for horizons 1-6), $+8.2$ (horizons 7-12), $+20.1$ (horizons 13-18), and
$+25.1$ (horizons 19-24) percentage points, averaging a $19.8$
percentage-point gap.  In the post-2020 window, the gaps narrow to $+8.4$,
$+3.2$, $+5.2$, and $-0.6$ percentage points, averaging a $4.1$
percentage-point gap. The combined three-expert forecast retains the edge
through $h=20$; the single XGBoost forecast pulls ahead only at the longest
horizons ($h=21$ to $24$), by margins comparable to the three-expert
forecast's own lead at shorter horizons. The contrast highlights the advantages of the Fixed Share algorithm,
which exploits the univariate forecast when it is useful, and adaptively switches to the
micro forecast at the relevant time. The pre-2020 weakness of the single XGBoost
forecast is consistent with the overfitting concern formalized in
equation~\eqref{eq:risk_decomp}: when the microdata signal is limited, putting
all predictors into one pooled XGBoost regression can raise estimation error
rather than improve the forecast.

\subsection{Robustness to Alternative Microdata Encodings}
\label{sec:micro_robustness}

The baseline micro forecast incorporates three choices in encoding  microdata: the 11-feature encoding, expenditure-weighted within-category statistics,
and price quotes that include sales. We evaluate each choice against an alternative that appears to perform worse; the detailed results are reported in
Appendix~\ref{app:robustness_alt}.
 
\paragraph{Augmented encoding.}
We replace the 11-feature encoding with the full 29-feature augmented
encoding, which adds the variance (with
Bessel's correction), the standardised third and fifth central moments and the
excess kurtosis of the non-zero price-change distribution, the nine deciles of
the time-since-last-change distribution, and five moments of the same
distribution (Appendix~\ref{app:augmented_encoding}). The three-expert
combined forecast is nearly unchanged: relative to the baseline, the
augmented-encoding three-expert combined forecast's MAE performance deteriorates
by $2.8$ percentage points after 2020 and $0.3$ percentage points before 2020,
averaging across forecast horizons. The standalone micro forecast's MAE
performance deteriorates more relative to the baseline: by $2.5$ percentage
points before 2020 and $7.3$ percentage points after 2020. The deterioration in
MAE performance is likely due to overfitting from the higher-dimensional
encoding.

\paragraph{Unweighted statistics.}
Replacing expenditure-weighted statistics with equal-weighted statistics in the
encoding (Appendix~\ref{app:unweighted_encoding}) deteriorates the standalone
micro forecast's MAE performance relative to the baseline in both windows: by
$1.6$ percentage points in the pre-2020 window and by $3.3$ percentage points
in the post-2020 window, averaging across forecast horizons. The three-expert
combined forecast is more stable: its MAE performance is essentially unchanged
before 2020 and deteriorates by $0.9$ percentage points after 2020.
 
\paragraph{Prices excluding sales.}
Dropping ONS-sale-flagged quotes (Appendix~\ref{app:nosales_encoding})
improves the standalone micro forecast's MAE performance relative to the
baseline by $2.3$ percentage points in the pre-2020 window, but deteriorates
that performance by $8.9$ percentage points in the post-2020 window. The
three-expert combined forecast is more stable: its MAE performance is
essentially unchanged before 2020 and deteriorates by $2.6$ percentage points
after 2020, again leaving the baseline preferred. Sale-driven price changes
therefore carry useful predictive signal during the post-COVID inflation
surge, even though they add noise in stable times.

\paragraph{Sales-price imputation.}
Replacing temporary sale prices with the regular price inferred by the
\citet{NakamuraSteinsson2008} filter, parameterisation A
(Appendix~\ref{app:nsa_encoding}), deteriorates the standalone micro forecast's
MAE performance relative to the baseline by $8.9$ percentage points in the
pre-2020 window and $6.3$ percentage points in the post-2020 window. In the pre-2020 window the imputation step is more harmful than dropping the
flagged quotes outright, injecting noise of its own even when sales carry
little signal; in the post-2020 window it is somewhat less harmful than
dropping the quotes, though both alternatives still underperform the baseline.
The three-expert combined forecast is less affected:
its MAE performance is essentially unchanged before 2020 and deteriorates by
$2.1$ percentage points after 2020.

\section{Opening the Black Box: Grouped Shapley Decompositions}\label{sec:group_shapley}

Machine learning methods are sometimes criticized as a ``black box'' (e.g. \citealp{MullainathanLudwig2024}). We propose to ``open'' the black box by investigating which features of the microdata matter for its forecasting power. In particular, we use a grouped Shapley decomposition to ask which micro variables account for the incremental
value of the micro expert relative to the two-expert benchmark combining only the univariate and macro forecasts.

\paragraph{Construction.}
We partition the micro variables, computed within each COICOP1 category and
then stacked across categories (except the monthly dummies, which are common
across categories), into five groups:
\begin{enumerate}
  \item the fraction of zero price changes;
  \item the mean plus the middle three deciles of the price-change distribution;
  \item the top three deciles of the price-change distribution;
  \item the bottom three deciles of the price-change distribution;
  \item monthly dummies.
\end{enumerate}
Our choice of grouping is useful because it maps the forecasting exercise into
objects that price-setting theory treats as economically meaningful. The first
group captures the extensive margin of price adjustment,
a central empirical moment in studies of price rigidity
\citep{KlenowKryvtsov2008,NakamuraSteinsson2008,Blanco2021}. The second group
captures the center of the price-change distribution through the mean and
middle deciles, which summarize typical repricing magnitudes. The third and
fourth groups isolate the upper and lower tails of the price-change
distribution, allowing the decomposition to distinguish unusually large price
increases from unusually large price decreases or discounts. These tail groups
are natural objects for theories that emphasize asymmetry, state dependence,
and granular shocks \citep{BallMankiw1995,Midrigan2011,Vavra2014,Alvarez2016,AlvarezBlaser2025}. The fifth and last
group consists of monthly dummies, which absorb the seasonal structure in inflation.
The benchmarks already account for seasonality: the macro expert includes month
indicators, while the univariate SARIMA benchmark explicitly allows for monthly
seasonal dynamics. The monthly-dummy group therefore should not be read as
adding a seasonal adjustment that is missing from the competing forecasts. The
relevant comparison is whether seasonal timing improves forecasts when it is
combined with micro-price information, beyond what the macro and univariate
forecasts already extract from seasonality.

Grouped Shapley values reveal the marginal contribution of each group to the overall forecasting power of the microdata, in a way that accounts for the interactions and nonlinearity of the machine learning algorithm. We calculate grouped Shapley values as follows. For each horizon $h$ and nonempty coalition $S$ of grouped micro variables, we
estimate the micro expert using only the groups in $S$ and combine the
resulting forecast with the macro and univariate experts using the same Fixed
Share procedure as in Subsection~\ref{sec:fixed_share}. Let
$\mathrm{MAE}_{S,h}$ denote the mean absolute error of this three-expert
forecast, and let $\mathrm{MAE}_{\mathrm{baseline},h}$ denote
the mean absolute error of the two-expert alternative. We define
the coalition value as the signed percentage reduction in mean absolute error,
\[
  v_h(S)
  \;=\;
  100 \times
  \frac{\mathrm{MAE}_{\mathrm{baseline},h} - \mathrm{MAE}_{S,h}}
  {\mathrm{MAE}_{\mathrm{baseline},h}}.
\]
For the empty coalition, we set $v_h(\varnothing)=0$: it corresponds to the
benchmark without any micro inputs.

For each horizon, we evaluate all $2^5$ coalitions of these groups, compute
the corresponding values $v_h(S)$, and then assign each group its Shapley
value, that is, its average marginal contribution across all possible orders in
which the groups can enter the model. Formally, if $\mathcal{G}$ denotes the
set of the five groups and $g \in \mathcal{G}$ is one particular group, its
Shapley value at horizon $h$ is
\[
  \phi_{g,h}
  \;=\;
  \sum_{S \subseteq \mathcal{G}\setminus\{g\}}
  \frac{|S|!\,\bigl(|\mathcal{G}|-|S|-1\bigr)!}{|\mathcal{G}|!}
  \Bigl[v_h(S \cup \{g\}) - v_h(S)\Bigr].
\]
This expression averages the marginal contribution of group $g$ over all
possible subsets $S$ that exclude it, with the standard combinatorial weights.
Positive Shapley values indicate that a group lowers forecast errors relative
to the two-expert online-learning benchmark on average, while negative values
indicate that the group worsens forecast performance.
The unit of $\phi_{g,h}$ is percent decrease in MAE compared to the
two-expert benchmark MAE. For example, $\phi_{g,h}=1$ means that, at
horizon $h$, group $g$ accounts for an average MAE improvement equal to
one percent of the two-expert benchmark MAE, averaged over all coalitions
of the other groups. Because $v_h(\varnothing)=0$, the Shapley values add
up exactly: the group contributions sum to $v_h(\mathcal{G})$, the full
three-expert combination's percentage MAE improvement over the two-expert
benchmark. 

\paragraph{Computation.}
Re-tuning the micro expert for every coalition would multiply the full
hyperparameter search by $2^5$. We avoid this by freezing the
hyperparameters: each coalition's micro expert is re-trained at every origin
on its feature subset, but reuses the hyperparameter sequence that the full
micro expert already selected in real time
(Subsection~\ref{sec:training_tuning}). Freezing preserves the no-look-ahead
structure of the pipeline, because each reused hyperparameter was selected
from information available at its origin, and it reduces the cost of the full
decomposition to roughly $1.3$ times one training run of the micro expert.
Algorithms~\ref{alg:shapley_micro} and~\ref{alg:group_shapley} in
Appendix~\ref{app:pseudocode} give the complete procedure.

\paragraph{Results.}
We apply the grouped Shapley decomposition to measure
which groups of micro variables account for the MAE gain from adding the micro
expert to the two-expert benchmark that combines only the univariate and macro
forecasts.

\begin{table}[htbp!]
\centering
\caption{Grouped Shapley Values for the Micro Expert, Averaged over 24 Horizons}
\label{tab:micro_shapley_aggregate_24h}
\begin{tabular}{lrr}
\toprule
Feature group & Pre-2020 & Post-2020 \\
\midrule
Fraction of zero price changes & 0.0859 & -0.1215 \\
Mean + middle 3 deciles & -0.7840 & 0.6796 \\
Top 3 deciles & -0.1960 & 0.7720 \\
Bottom 3 deciles & 0.2602 & 0.5342 \\
Monthly dummies & 0.2360 & 1.0435 \\
\bottomrule
\end{tabular}
\begin{minipage}{0.92\textwidth}\vspace{0.4em}\footnotesize\textit{Notes:}
Entries are grouped Shapley contributions of each feature group to the MAE gain
from adding the micro expert to the two-expert online-learning benchmark, which
combines the univariate and macro forecasts. Values are in percent of benchmark
MAE and are averaged over horizons $h=1,\ldots,24$ within each evaluation
window. Positive entries reduce MAE; negative entries increase it.\end{minipage}
\end{table}

Table~\ref{tab:micro_shapley_aggregate_24h} reports Shapley values averaged over
horizons $h=1,\ldots,24$. In the pre-COVID window, every contribution is
small: the mean plus middle-three-deciles group is the most negative at
$-0.78$, the top-three-deciles group is mildly negative at $-0.20$, and the
bottom three deciles, monthly dummies, and the fraction of zero price changes
are mildly positive, at $0.26$, $0.24$, and $0.09$. Thus, before COVID, no group delivers a
large positive contribution. This mirrors the broader finding that microdata
add little forecasting power in a stable inflation environment.
   
In the post-COVID window, the positive contributions come from the mean plus
middle three deciles at $0.68$, the top three deciles at $0.77$, the bottom
three deciles at $0.53$, and monthly dummies at $1.04$.\footnote{The positive Shapley value
for monthly dummies does not mean that the benchmarks lack seasonal controls.
It means that, within the micro forecast, seasonal timing adds predictive
content when combined with micro-price statistics.} The contribution from the group for the
fraction of zero price changes is essentially zero on average, at $-0.12$.
The five contributions sum to $2.91$ percent, the three-expert combination's
average post-2020 MAE gain over the two-expert benchmark
(Table~\ref{tab:three_vs_two_mae_improvement}).  

Figure~\ref{fig:micro_shapley_six_horizon_ma} reports the same decomposition
at the horizon level: at each horizon $h$, the plotted value is the average of
the per-horizon grouped Shapley values over a six-horizon local window, using
the available adjacent horizons. Pre-COVID, the contributions are small at
short and medium horizons. The lower tail is positive only at short and medium
horizons and turns negative later, while the losses concentrate at the longest
horizons, where the mean plus middle-three-deciles group and the bottom three
deciles become sharply more negative, both falling below $-3$ by $h=24$.
Post-COVID, the bottom three deciles peak at medium horizons, around $h=11$ to
$14$, and remain positive through roughly  $h=16$ before turning negative at longer horizons; the top three deciles
contribute most at short and medium horizons before fading near the end of the
horizon range; the mean plus middle-three-deciles group is strongest at the
shortest horizons and again from about $h=20$ onward; and monthly dummies are
positive at every horizon and grow toward the long end. The fraction of zero
price changes is negative at short horizons and positive at the longest ones,
netting out to essentially zero on average. 
 
\begin{figure}[t!]
  \centering
  \caption{Grouped Shapley Values by Horizon}
  \label{fig:micro_shapley_six_horizon_ma}
  \begin{subfigure}[t]{0.49\textwidth}
    \centering
    \caption{Pre-COVID}
    \includegraphics[width=\textwidth]{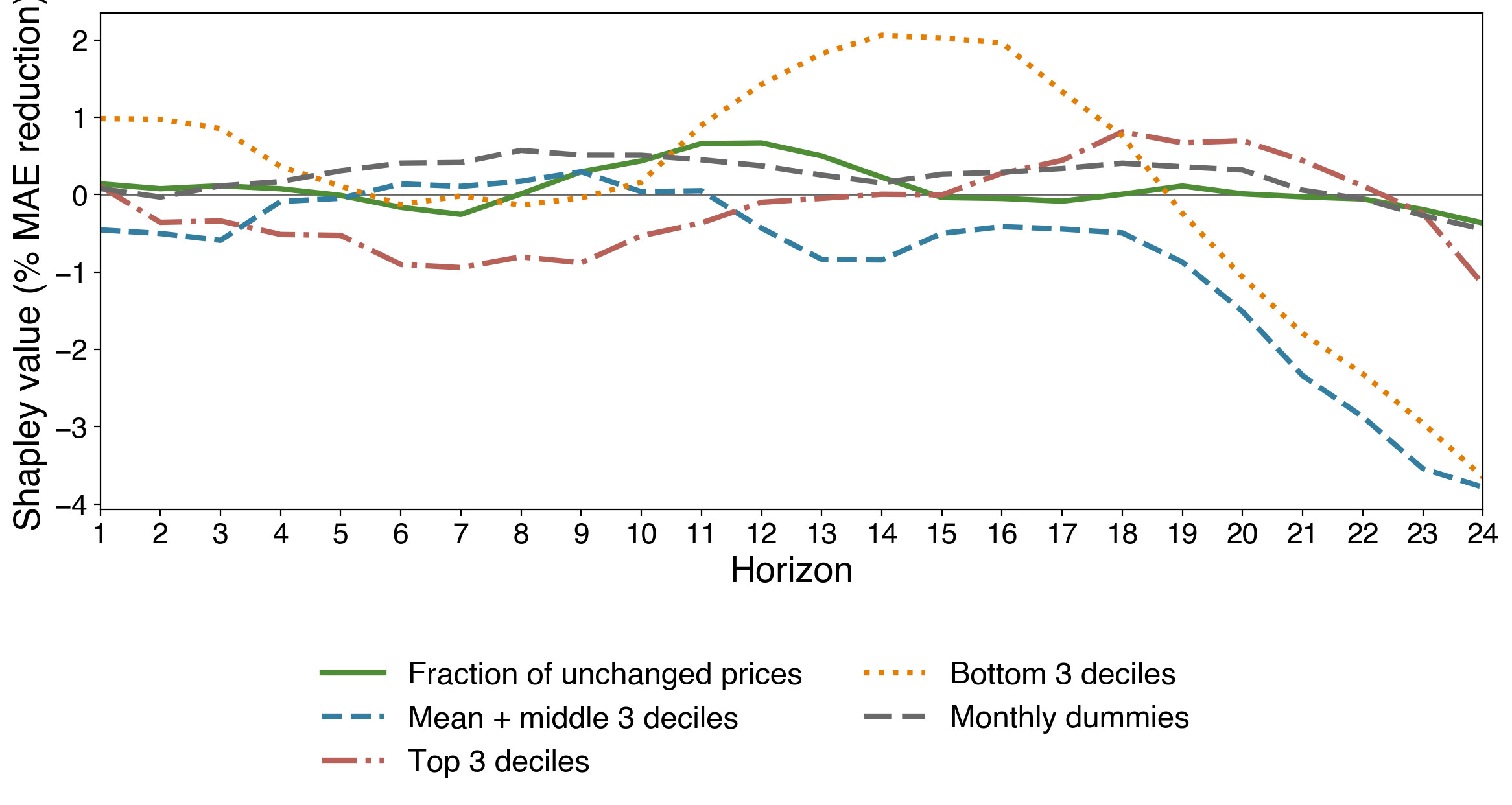}
  \end{subfigure}
  \hfill
  \begin{subfigure}[t]{0.49\textwidth}
    \centering
    \caption{Post-COVID}
    \includegraphics[width=\textwidth]{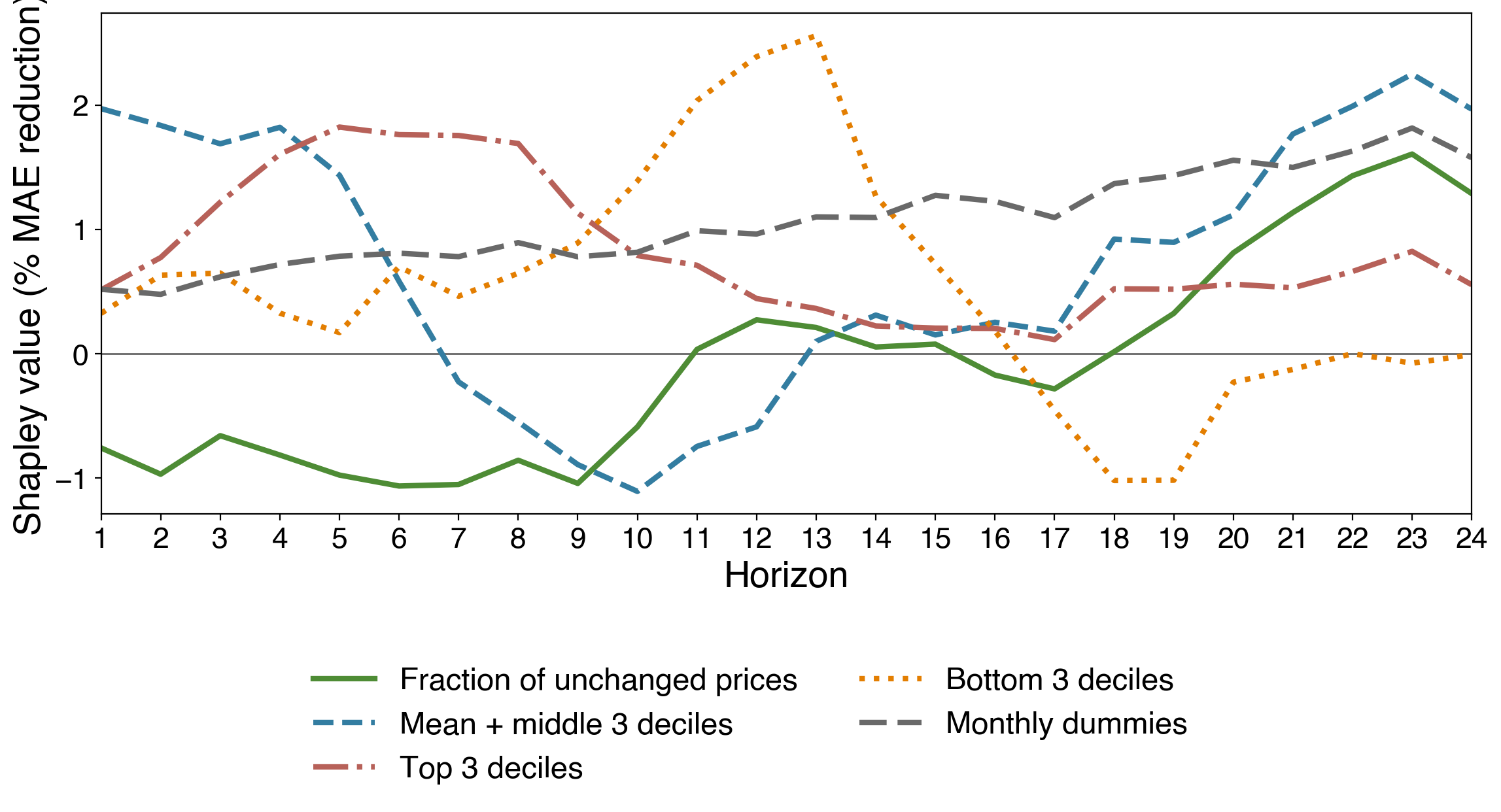}
  \end{subfigure}
  \smallskip
  \begin{minipage}{\textwidth}\setstretch{1.0}
    \footnotesize\textit{Notes:} The figure reports 5-group Shapley
    contributions for the incremental value of the micro forecast in the
    three-expert combined forecast relative to the two-expert
    benchmark that combines the univariate and macro forecasts. Values are smoothed
    across adjacent horizons using a 6-horizon moving average. Positive
    values indicate reductions in MAE from the corresponding feature group;
    negative values indicate increases in MAE. The horizontal line marks a
    zero contribution.
  \end{minipage}
\end{figure}

\section{Implications for Models of Inflation}\label{sec:models}

Overall, our results suggest that microdata matters for aggregate inflation dynamics, but only during periods of high and volatile inflation such as after 2020. We briefly discuss implications of our findings for models of price setting.

Consider first a benchmark. Let $\mu_t$ denote the cross-sectional
distribution of firms' pricing states relevant for future inflation (e.g.,
price gaps in standard menu-cost models), let $u_{-\infty:t}$ denote the
history of aggregate shocks through date $t$. Under a stable transition law, this distributional
state can be written as
\[
  \mu_t=\Gamma(u_{-\infty:t}).
\]
If the history of observed macro variables recovers $u_{-\infty:t}$ and
$\Gamma$ is consistently estimable from aggregate data, then $\mu_t$ can be
recovered from macro information alone. In this benchmark, observing microdata
separately cannot improve the optimal forecast of future inflation.

The familiar first-order result is a local version of this benchmark. In
standard menu-cost models, a first-order approximation is accurate when
aggregate shocks are small \citep{AuclertRogatoRognlieStraub2024}. Aggregate
inflation then admits a linear vector moving-average representation,
\[
  \pi_t \;=\; \mu_\pi + \sum_{l=0}^{\infty} \Psi_l\,
  u_{t-l},
\]
where $\mu_\pi$ denotes steady-state inflation, $u_t$ is the vector
of aggregate innovations, and $\{\Psi_l\}$ are impulse response coefficients.
If macro variables recover current and past aggregate shocks,
this representation implies that microdata do not add predictive content.

Microdata can improve feasible forecasts when this benchmark fails in one of
two ways. First, the observed macro variables may not recover the shocks
$u_{-\infty:t}$ that move the distributional state. A leading case is granular shocks. Disturbances to large
firms, products, or sectors need not wash out in a finite economy and can
therefore affect aggregate inflation. Micro-price distributions may contain
information about these disturbances beyond what appears in standard
aggregates \citep{AlvarezBlaser2025}. Dispersion, skewness, or sectoral relative-price
shocks can likewise change the distributional state and future inflation
without being summarized by the observed macro variables
\citep{BallMankiw1995,Vavra2014}.

Second, the mapping $\Gamma$ from the shock history to the distributional
state may not be consistently estimable. Changes in policy, regimes,
structural parameters, or the initial cross-sectional state can make this
mapping unstable, so that it cannot be recovered from past data.
Even when the mapping is stable, the combination of the inversion from macro
variables to shocks, the transition law for the distributional state, and the
law of motion for inflation may be difficult to estimate in a short sample,
especially in a large-shock region that is rarely observed before 2020.
Statistics of observed micro-level price changes can provide additional
information about the underlying distributional state and thereby improve a
feasible forecast.

Periods of high and volatile inflation such as after 2020 can lead to both
of these failures. The disturbances key to those episodes, for example
shocks to supply chains, are often granular in origin, so the observed macro
variables miss a larger share of the relevant shocks. Such episodes also
coincide with changes in policy and in the structure of the economy, which
can make the mapping from the shock history to the distributional state
unstable, while the data needed to re-estimate it remain scarce.
Price-setting nonlinearity amplifies both failures: in menu-cost models, the
frequency of price adjustment and inflation pass-through respond nonlinearly
to shock size, so inflation becomes more sensitive to the distributional
state after large shocks \citep{Blanco2021}. The observed distribution of
price changes can supply part of the information lost to unrecovered shocks
or to an unstable or poorly estimated mapping. Consistent with this
discussion, the incremental value of the micro expert appears after the
large inflationary shocks shown in Figure~\ref{fig:us_uk_cpi}, rather than
during the stable 2015--2019 period.

\section{Conclusion}\label{sec:conclusion}

This paper asks whether microdata matters for aggregate inflation dynamics. Standard models, as well as the non-stationary dynamics of inflation, suggest that microdata might be important only occasionally. Therefore we start by developing a new scan test, which assesses whether one forecast outperforms another at an unknown horizon and time period. This scan test can be used to detect the occasional forecasting power of microdata for inflation, but also used more broadly. We then develop an adaptive machine learning based pipeline to take advantage of this occasional forecasting power. There are three steps, which we apply to microdata underlying the UK consumer price index. First, we encode the distribution of price changes into a high dimensional vector of statistics. Second, we combine this vector with a machine learning algorithm, gradient boosted trees, tuned by a cached real-time procedure, in order to forecast inflation. Third, we combine the resulting forecast with the benchmarks through the Fixed Share algorithm, which raises the weight on microdata when it is useful for forecasting and lowers it otherwise. Microdata enter the combination only after the scan test detects localized outperformance beyond the existing experts.
  
There are four main empirical results. First, the microdata forecast
outperforms the univariate benchmark only after 2020; before 2020 the
univariate benchmark performs better at every horizon. Second, the scan tests imply that the Fixed Share combined forecast
should include all three modalities: micro, macro and univariate forecasts.
Third, the combined forecast performs comparably to the univariate benchmark before 2020 and
improves on it for all 24 horizons. Fourth, the value of microdata for the combined forecast materializes
after 2020. We also ``open the black box'' using grouped Shapley values
to study why microdata help forecast inflation. This decomposition suggests
that the positive contributions come from the center and tails of the
price-change distribution together with monthly dummies, while the extensive
margin of price adjustment contributes little.

Our findings connect to models of inflation dynamics. The finding that microdata improve forecasts only during high and volatile inflation is consistent with benchmark models of price setting, in which microdata carry additional predictive content when shocks are large. Several directions for future work emerge naturally. The same pipeline of distributional encoding, scan testing, and adaptive combination could be applied to forecast other macroeconomic aggregates, such as GDP growth or employment, using firm-level or household-level microdata. Extending the analysis to other countries would help assess whether the forecasting value of microdata during high and volatile inflation is a general feature of inflation dynamics.

\clearpage

\section*{Tables}
\begin{table}[!ht]
  \centering
  \caption{COICOP1 Cell Counts (Observations per Cell)}
  \label{tab:cells_coicop1}
  \begin{tabular}{@{}lrrrrrrr@{}}
    \toprule
    Treatment       & p5  & p10     & p25     & p50     & p75      & p90      & p95      \\
    \midrule
    \multicolumn{8}{@{}l}{\textit{Panel A: All price changes}}                           \\
    \addlinespace
    Drops ONS sales & 316 & 1{,}297 & 2{,}731 & 6{,}809 & 9{,}526  & 15{,}747 & 20{,}503 \\
    Includes sales  & 343 & 1{,}310 & 2{,}800 & 7{,}288 & 11{,}111 & 17{,}609 & 21{,}676 \\
    \addlinespace
    \midrule
    \multicolumn{8}{@{}l}{\textit{Panel B: Non-zero price changes only}}                 \\
    \addlinespace
    Drops ONS sales & 23  & 55      & 211     & 423     & 739      & 1{,}600  & 2{,}778  \\
    Includes sales  & 45  & 89      & 264     & 788     & 1{,}597  & 2{,}813  & 4{,}328  \\
    \bottomrule
  \end{tabular}
  \smallskip
  \begin{minipage}{\textwidth}\setstretch{1.0}
    \footnotesize \textit{Notes:} This table reports percentiles of the distribution
    of cell counts across all (date, COICOP1 code) cells, separately for all price
    observations (Panel~A) and for non-zero price changes only (Panel~B). The
    ``Drops ONS sales'' panel sets ONS-sale-flagged prices to missing, sets the
    following observation's price change to missing, reruns short-spell deletion,
    and recomputes time-since-change. NS-filtered datasets have cell counts equal
    to the ``includes sales'' figures, since observations are adjusted rather than
    dropped.
  \end{minipage}
\end{table}

\clearpage

\begin{table}[htbp!]
\centering
\caption{Univariate Forecast Mean Absolute Error Before and After 2020}
\label{tab:sarima_mae_two_windows}
\begin{tabular}{crr}
\toprule
$h$ & Pre 2020 MAE & Post 2020 MAE \\
\midrule
1 & 0.1364 & 0.3915 \\
2 & 0.1373 & 0.3842 \\
3 & 0.1424 & 0.3985 \\
4 & 0.1512 & 0.3957 \\
5 & 0.1476 & 0.3975 \\
6 & 0.1503 & 0.4268 \\
7 & 0.1479 & 0.4764 \\
8 & 0.1532 & 0.4843 \\
9 & 0.1474 & 0.5070 \\
10 & 0.1547 & 0.4504 \\
11 & 0.1479 & 0.4883 \\
12 & 0.1556 & 0.4235 \\
13 & 0.1825 & 0.4753 \\
14 & 0.1820 & 0.4878 \\
15 & 0.1651 & 0.4893 \\
16 & 0.1655 & 0.4768 \\
17 & 0.1605 & 0.4722 \\
18 & 0.1694 & 0.4736 \\
19 & 0.1792 & 0.4720 \\
20 & 0.1670 & 0.4712 \\
21 & 0.1656 & 0.6530 \\
22 & 0.1693 & 0.6208 \\
23 & 0.1686 & 0.7061 \\
24 & 0.1778 & 0.5808 \\
\bottomrule
\end{tabular}
\smallskip
\begin{minipage}{0.82\textwidth}\setstretch{1.0}
\footnotesize\textit{Notes:} The table reports mean absolute error (MAE) for
the rolling-window univariate forecast at forecast horizons $h=1,\ldots,24$.
The pre-2020 evaluation window is 2015--2019 and the post-2020 evaluation
window is 2020--2024. Errors are computed against the unsmoothed monthly
inflation target.
\end{minipage}
\end{table}

\clearpage

\begin{table}[htbp!]
\centering
\caption{Micro percent MAE improvement over SARIMA, pre- and post-COVID.}
\label{tab:fine_micro_vs_sarima_improvement}
\begin{tabular}{crr}
\toprule
$h$ & Pre-COVID \% impr. & Post-COVID \% impr. \\
\midrule
1 & -35.67 & 8.49 \\
2 & -28.92 & 4.46 \\
3 & -29.10 & 2.14 \\
4 & -30.97 & 9.67 \\
5 & -32.52 & 6.73 \\
6 & -19.51 & 20.31 \\
7 & -17.33 & 13.00 \\
8 & -17.05 & 8.63 \\
9 & -27.49 & 14.63 \\
10 & -18.19 & 10.43 \\
11 & -19.84 & 16.04 \\
12 & -12.81 & 1.79 \\
13 & -10.32 & 12.95 \\
14 & -0.13 & 15.28 \\
15 & -22.62 & 14.27 \\
16 & -21.06 & 11.03 \\
17 & -31.47 & -3.33 \\
18 & -18.48 & 0.82 \\
19 & -28.04 & 9.07 \\
20 & -26.65 & 9.34 \\
21 & -32.13 & 35.42 \\
22 & -44.26 & 32.27 \\
23 & -36.87 & 39.33 \\
24 & -27.88 & 26.46 \\
\bottomrule
\end{tabular}
\end{table}

\clearpage

\begin{table}[htbp!]
\centering
\caption{Fixed-Share Combination Mean Absolute Error Before and After 2020}
\label{tab:fs_mae_two_windows}
\begin{tabular}{crrrrrrrr}
\toprule
& \multicolumn{4}{c}{Pre 2020} & \multicolumn{4}{c}{Post 2020} \\
\cmidrule(lr){2-5} \cmidrule(lr){6-9}
$h$ & FS & Micro & Macro & SARIMA & FS & Micro & Macro & SARIMA \\
\midrule
1 & 0.1385 & 0.1850 & 0.1708 & 0.1364 & 0.3444 & 0.3583 & 0.3607 & 0.3915 \\
2 & 0.1348 & 0.1770 & 0.1678 & 0.1373 & 0.3548 & 0.3670 & 0.3962 & 0.3842 \\
3 & 0.1363 & 0.1838 & 0.1629 & 0.1424 & 0.3690 & 0.3899 & 0.4137 & 0.3985 \\
4 & 0.1502 & 0.1980 & 0.1679 & 0.1512 & 0.3436 & 0.3575 & 0.4080 & 0.3957 \\
5 & 0.1478 & 0.1956 & 0.1559 & 0.1476 & 0.3429 & 0.3707 & 0.4035 & 0.3975 \\
6 & 0.1467 & 0.1796 & 0.1681 & 0.1503 & 0.3505 & 0.3401 & 0.4221 & 0.4268 \\
7 & 0.1481 & 0.1736 & 0.1631 & 0.1479 & 0.4019 & 0.4145 & 0.4072 & 0.4764 \\
8 & 0.1557 & 0.1793 & 0.1835 & 0.1532 & 0.4130 & 0.4425 & 0.4040 & 0.4843 \\
9 & 0.1523 & 0.1879 & 0.1843 & 0.1474 & 0.4145 & 0.4328 & 0.4079 & 0.5070 \\
10 & 0.1562 & 0.1829 & 0.1876 & 0.1547 & 0.3819 & 0.4035 & 0.4050 & 0.4504 \\
11 & 0.1507 & 0.1773 & 0.1763 & 0.1479 & 0.4107 & 0.4100 & 0.4622 & 0.4883 \\
12 & 0.1561 & 0.1756 & 0.1878 & 0.1556 & 0.4116 & 0.4160 & 0.4656 & 0.4235 \\
13 & 0.1783 & 0.2014 & 0.1921 & 0.1825 & 0.4170 & 0.4137 & 0.4582 & 0.4753 \\
14 & 0.1741 & 0.1822 & 0.2004 & 0.1820 & 0.4122 & 0.4132 & 0.4755 & 0.4878 \\
15 & 0.1761 & 0.2024 & 0.2042 & 0.1651 & 0.4213 & 0.4195 & 0.4340 & 0.4893 \\
16 & 0.1763 & 0.2003 & 0.2079 & 0.1655 & 0.4266 & 0.4242 & 0.4569 & 0.4768 \\
17 & 0.1686 & 0.2110 & 0.1928 & 0.1605 & 0.4332 & 0.4879 & 0.4360 & 0.4722 \\
18 & 0.1668 & 0.2007 & 0.1868 & 0.1694 & 0.4221 & 0.4697 & 0.4303 & 0.4736 \\
19 & 0.1771 & 0.2294 & 0.2034 & 0.1792 & 0.4077 & 0.4292 & 0.4244 & 0.4720 \\
20 & 0.1762 & 0.2115 & 0.2062 & 0.1670 & 0.3964 & 0.4272 & 0.4172 & 0.4712 \\
21 & 0.1725 & 0.2188 & 0.1988 & 0.1656 & 0.4569 & 0.4217 & 0.4339 & 0.6530 \\
22 & 0.1831 & 0.2442 & 0.1877 & 0.1693 & 0.4332 & 0.4205 & 0.4090 & 0.6208 \\
23 & 0.1754 & 0.2307 & 0.1782 & 0.1686 & 0.4569 & 0.4284 & 0.4105 & 0.7061 \\
24 & 0.1719 & 0.2273 & 0.1907 & 0.1778 & 0.4267 & 0.4271 & 0.4113 & 0.5808 \\
\bottomrule
\end{tabular}
\smallskip
\begin{minipage}{0.95\textwidth}\setstretch{1.0}
\footnotesize\textit{Notes:} Mean absolute error of the patient fixed-share
combination (FS, $\alpha=0.02$, $\eta=0.5$, $\rho=0$, grid=fine) against
its experts (Micro, Macro, SARIMA), by forecast horizon.
The pre-2020 window is 2015--2019 and the post-2020 window is 2020--2024.
\end{minipage}
\end{table}

\clearpage

\begin{table}[htbp!]
\centering
\caption{Adding the Micro Expert: Three- vs Two-Expert Fixed Share}
\label{tab:three_vs_two_mae_improvement}
\begin{tabular}{crrrrrr}
\toprule
& \multicolumn{3}{c}{Pre 2020} & \multicolumn{3}{c}{Post 2020} \\
\cmidrule(lr){2-4} \cmidrule(lr){5-7}
$h$ & 3-expert & 2-expert & $\Delta$ (\%) & 3-expert & 2-expert & $\Delta$ (\%) \\
\midrule
1 & 0.1385 & 0.1390 & 0.32 & 0.3444 & 0.3487 & 1.24 \\
2 & 0.1348 & 0.1352 & 0.33 & 0.3548 & 0.3596 & 1.32 \\
3 & 0.1363 & 0.1382 & 1.39 & 0.3690 & 0.3786 & 2.54 \\
4 & 0.1502 & 0.1523 & 1.37 & 0.3436 & 0.3624 & 5.19 \\
5 & 0.1478 & 0.1440 & -2.63 & 0.3429 & 0.3553 & 3.48 \\
6 & 0.1467 & 0.1469 & 0.12 & 0.3505 & 0.3782 & 7.34 \\
7 & 0.1481 & 0.1472 & -0.56 & 0.4019 & 0.4104 & 2.08 \\
8 & 0.1557 & 0.1546 & -0.70 & 0.4130 & 0.4082 & -1.18 \\
9 & 0.1523 & 0.1500 & -1.51 & 0.4145 & 0.4139 & -0.15 \\
10 & 0.1562 & 0.1579 & 1.12 & 0.3819 & 0.3773 & -1.22 \\
11 & 0.1507 & 0.1513 & 0.40 & 0.4107 & 0.4283 & 4.11 \\
12 & 0.1561 & 0.1598 & 2.30 & 0.4116 & 0.4183 & 1.59 \\
13 & 0.1783 & 0.1820 & 2.04 & 0.4170 & 0.4374 & 4.67 \\
14 & 0.1741 & 0.1848 & 5.78 & 0.4122 & 0.4539 & 9.18 \\
15 & 0.1761 & 0.1760 & -0.04 & 0.4213 & 0.4325 & 2.58 \\
16 & 0.1763 & 0.1757 & -0.35 & 0.4266 & 0.4439 & 3.89 \\
17 & 0.1686 & 0.1683 & -0.16 & 0.4332 & 0.4158 & -4.18 \\
18 & 0.1668 & 0.1723 & 3.20 & 0.4221 & 0.4156 & -1.55 \\
19 & 0.1771 & 0.1844 & 3.95 & 0.4077 & 0.4087 & 0.26 \\
20 & 0.1762 & 0.1814 & 2.88 & 0.3964 & 0.4083 & 2.91 \\
21 & 0.1725 & 0.1715 & -0.59 & 0.4569 & 0.5050 & 9.53 \\
22 & 0.1831 & 0.1678 & -9.13 & 0.4332 & 0.4607 & 5.97 \\
23 & 0.1754 & 0.1600 & -9.62 & 0.4569 & 0.4850 & 5.79 \\
24 & 0.1719 & 0.1570 & -9.48 & 0.4267 & 0.4463 & 4.39 \\
\bottomrule
\end{tabular}
\smallskip
\begin{minipage}{0.95\textwidth}\setstretch{1.0}
\footnotesize\textit{Notes:} Mean absolute error of the 3-expert patient
fixed-share combination (micro + macro + SARIMA) and the 2-expert macro + SARIMA
combination ($\alpha=0.02$, $\eta=0.5$, $\rho=0$, grid=fine), by
horizon. $\Delta$ is the percent MAE improvement from adding the micro expert
(positive = the 3-expert combination is better). Pre-2020 = 2015--2019,
post-2020 = 2020--2024.
\end{minipage}
\end{table}

\clearpage
\begin{singlespace}
  \putbib
\end{singlespace}
\end{bibunit}

\clearpage
\appendix
\begin{bibunit}
\numberwithin{equation}{section}
\numberwithin{table}{section}
\numberwithin{figure}{section}
\setcounter{table}{0}
\setcounter{figure}{0}
\renewcommand{\thetable}{\Alph{section}.\arabic{table}}
\renewcommand{\thefigure}{\Alph{section}.\arabic{figure}}

\phantomsection
\begin{center}
  {\LARGE\bfseries Appendix for ``Forecasting Inflation with Microdata: An Adaptive Machine Learning Approach''}
\end{center}
\bigskip

This appendix collects supporting material for the main
text. Appendix~\ref{app:data} details the construction and cleaning of the
micro-price and macro datasets. Appendix~\ref{sec:theory_scan_tests} develops
the econometric theory of the scan tests, and
Appendix~\ref{app:simulation_results} reports the simulation studies behind
their finite-sample calibration, including the choice of the
partial-studentization exponent and a comparison with the
Giacomini--Rossi fluctuation test. Appendix~\ref{app:figures} contains
additional figures, and Appendix~\ref{app:pseudocode} collects the notation
and pseudocode for all algorithms. Appendix~\ref{app:robustness_alt} examines
robustness to alternative microdata encodings, aggregation schemes, and data
filters, and Appendix~\ref{app:additional_results} presents additional
empirical results.
\medskip

\begingroup
\begin{singlespace}
\normalsize
\titlecontents{section}[1.25em]{\addvspace{0.25em}\bfseries}
  {\contentslabel{1.25em}}{\hspace*{-1.25em}}{\hfill\contentspage}
\titlecontents{subsection}[3.25em]{}
  {\contentslabel{2em}}{\hspace*{-2em}}{\titlerule*[0.5pc]{.}\contentspage}
\startcontents[appendices]
\printcontents[appendices]{l}{1}{\setcounter{tocdepth}{2}}
\end{singlespace}
\endgroup
\clearpage

\section{Data Appendix}
\label{app:data}

This appendix gathers implementation details on data construction that are
useful for replication but are not central to the main exposition. We organise
the material into six subsections: cleaning of the ONS micro-price panel
(Appendix~\ref{app:ons_cleaning}), aggregation and CPI replication
(Appendix~\ref{app:cpi_aggregation}), benchmarking of the price-change
distribution (Appendix~\ref{app:moments}), treatment of temporary sales
(Appendix~\ref{app:NS_filter}), macro predictors
(Appendix~\ref{app:macro_predictors}), and exact formulas for distributional
statistics (Appendix~\ref{app:encoding_formulas}).

\subsection{Price Quotes Data Cleaning}
\label{app:ons_cleaning}

\subsubsection{Cleaning}
\label{app:ons_cleaning_steps}

We follow \citet{AdamWeber2023} and apply the steps below in order.

\begin{enumerate}
  \item \textit{Sample and region filters.} We drop January 1996, since it is a
partial initial month and cannot be used to compute lagged price changes. Since
the monthly price-change variable requires a valid consecutive-month lagged
price, the encoded micro-feature panel begins in March 1996 rather than January
1996. Thus, while the raw quote panel spans January 1996 through December 2024,
the encoded monthly micro-feature panel spans March 1996 through December 2024.
We also remove observations with invalid region code 0 and exclude Region 1
(``Catalogue collections''), which covers items collected centrally, such as
shelter, university tuition fees, and rail fares, for which only item-level
indices are published; following \citet{Blanco2021}, we exclude these
observations. Region 13 (``NI'') is Northern Ireland; these observations
(1,357,662 rows, 4.31\% of the cleaned panel) are retained.

  \item \textit{Validity filter.} Retain observations with valid current prices.
  For observations that are not flagged as non-comparable substitutions, we also
  require the base price to be valid according to the Office for National
  Statistics codes. For non-comparable substitutions, identified by the
  \texttt{N} or \texttt{Z} flags, we retain the observation for price-change
  encoding but set its base price to missing before any base-price-dependent
  calculation.

  \item \textit{Duplicate removal.} If a product identifier appears more than once in any given period, the entire series is dropped. This deliberately conservative rule removes roughly 6.6\% of otherwise valid observations per quarter (ranging between 5.1\% and 8.3\%). Most colliding cells show distinct prices, consistent with the ONS caveat that distinct outlets can share a shop code, so a per-month resolution rule would have to guess which quote to keep; we drop the series instead.

  \item \textit{Price rounding.} All prices are rounded to two decimal places.

  \item \textit{Implausible prices.} Zero prices, zero base prices, and prices below \pounds 0.001 are removed.

  \item \textit{Substitution spells.} Comparable substitutions, identified by the
  \texttt{C} or \texttt{X} flags, non-comparable substitutions, identified
  by the \texttt{N} or \texttt{Z} flags, and permanent pack-size, weight, or
  volume changes, identified by the \texttt{W} flag, start new product spells. Therefore
  the substitution month is treated as the first observation of a new spell:
  its contemporaneous price change is missing, while the following month's
  price change is computed relative to the new product whenever the next
  observation is consecutive.
  \item \textit{Short-spell removal.} Rather than splitting product series at gaps, we store prices in a time $\times$ product matrix where missing months are represented as \texttt{NaN}; gaps therefore propagate naturally as missing observations. Single isolated observations (with \texttt{NaN} on both sides) are set to missing, and any product column with fewer than two non-missing observations is dropped. This rule conditions on the adjacent month $t{+}1$: whether an isolated month-$t$ quote is kept can depend on the presence of a quote in month $t{+}1$, so the cleaned month-$t$ panel embeds one month of look-ahead relative to a strictly real-time vintage.\footnote{\citet{AdamWeber2023} instead split each product series into new sub-series at any gap exceeding one month. Their stricter rule removes 3.3\% of price-quote observations and 23\% of product series.}
\end{enumerate}
At the initial cleaning checkpoint, the validity, duplicate, price, and
flag-based filters reduce the raw ONS file from 42,013,725 to 33,053,579
price-quote rows, a reduction of 21.33\%. After the remaining baseline
filters, including region exclusion, substitution spell splitting, short-spell
deletion, COICOP1 balancing, and final price-change recomputation, the final
baseline panel contains 31,001,584 non-missing price observations. Relative to
the raw ONS file, this corresponds to an effective removal share of 26.21\%.
The final number of valid non-missing price-change observations used to
construct the encoded micro features is 27,847,570, corresponding to 33.72\%
fewer observations than the raw price-quote count.

\subsection{Aggregation and CPI Replication}
\label{app:cpi_aggregation}

Aggregation proceeds in four steps.

\subsubsection{Step 1 -- Stratum index}

Within each item $i$ and stratum cell $k$, the price quotes are aggregated into
a stratum-level price index using a shop-weighted geometric mean of log price
relatives to the February base:
\begin{equation}
  \text{stratum index}_{t,k}
  = 100 \times \exp\!\left(
  \frac{\displaystyle\sum_{j \in k} w_j^{\text{shop}} \cdot \log(p_{jt}/p_{j0})}
  {\displaystyle\sum_{j \in k} w_j^{\text{shop}}}
  \right),
\end{equation}
where $j$ runs over all products in stratum $k$, $p_{j0}$ is the February base price carried in the ONS base-price field, and $w_j^{\text{shop}}$ is the shop weight. Observations with missing base prices, including retained \texttt{N}/\texttt{Z}
non-comparable substitution rows, are excluded from this Jevons calculation.
They are retained only for the construction of product-level price-change
features. Shop weights are set to \texttt{NaN} wherever the corresponding price is missing, so the denominator always sums only over observed products. If a stratum cell code is zero (no sub-item stratification), a uniform stratum weight of 100 is applied.

\subsubsection{Step 2 -- Item index}

Stratum indices are combined into the item index by a weighted arithmetic mean
using stratum weights $w_k^{\text{stratum}}$:
\begin{equation}
  \text{item index}_{t,i}
  = \frac{\displaystyle\sum_k w_{k,t}^{\text{stratum}} \cdot \text{stratum index}_{t,k}}
  {\displaystyle\sum_k w_{k,t}^{\text{stratum}}}.
\end{equation}
Each price quote $j$ in stratum $k$, item $i$, COICOP1 category $c$, and
month $t$ receives an implicit product weight
\begin{equation}
  \label{eq:omega_jt}
  \omega_{j,t}
  = \underbrace{\frac{w_j^{\text{shop}}}{\sum_{j'\in k} w_{j'}^{\text{shop}}}}_{\text{share within stratum}}
  \times
  \underbrace{\frac{w_{k,t}^{\text{stratum}}}{\sum_{k'\in i} w_{k',t}^{\text{stratum}}}}_{\text{share within item}}
  \times
  \underbrace{\frac{w_{i,t}^{\text{item}}}{\sum_{i'\in c} w_{i',t}^{\text{item}}}}_{\text{share within COICOP1}}.
\end{equation}
The first two factors give the price quote's proportional contribution to the
item index; multiplying by the third factor extends $\omega_{j,t}$ to the
quote's proportional contribution within its COICOP1 category in month $t$,
so that $\sum_{j \in c,t} \omega_{j,t} = 1$ for each category-month cell.
This three-factor weight is the basis for all expenditure-weighted
price-change features used in the forecasting exercise
(Appendix~\ref{app:encoding_formulas}).
Single-stratum items receive $\text{share within item}=1$; items whose
stratum weights are all zero in a month are likewise assigned a uniform
within-item weight $1/|\{k': k' \in i\}|$.

\subsubsection{Step 3 -- Quality screen}

Each reconstructed item index is compared to the corresponding official ONS
item series over their joint sample. An item is deemed \emph{suitable} if the
RMSE of the log-level difference between the two series is below 2\% and there
are no gaps in the reconstructed series. Of the 1{,}318 reconstructed items,
944 are suitable under this screen over the full joint sample.

\subsubsection{Step 4 -- CPI aggregation and chaining}

The reconstructed headline index is a weighted average of suitable item
indices, with item weights normalised to sum to one in each period:
\begin{equation}
  \text{CPI}_t^{\text{unchained}}
  = \sum_i \tilde{w}_{i,t}^{\text{item}} \cdot \text{item index}_{t,i},
  \qquad
  \tilde{w}_{i,t}^{\text{item}}
  = \frac{w_{i,t}^{\text{item}}}{\sum_{i'} w_{i',t}^{\text{item}}}.
\end{equation}
Because each annual index is expressed relative to the February base of that
year, the unchained series restarts at 100 each February. The resulting chained index tracks the official ONS CPI closely over
1996--2016. The post-2016 extension is somewhat noisier, likely reflecting
major ONS weight revisions introduced in 2017.

To assess the quality of the aggregation more transparently, we compare the
inflation rate implied by our reconstructed CPI with the corresponding official
ONS inflation series, for the baseline panel with all
posted prices.

Figure~\ref{fig:cpi_replication_sales} plots the
corresponding year-on-year inflation series. The replicated
series closely tracks the official ONS benchmark, indicating that the sequence
of cleaning, weighting, item-index construction, and chaining reproduces the
aggregate inflation dynamics well. This figure therefore serves as an
important validation exercise for the data construction pipeline: it shows
that the resulting micro-price panel remains consistent
with the official CPI at the aggregate level.

\clearpage
\begin{figure}[t!]
  \centering
  \caption{Official ONS Inflation and Replicated CPI Inflation: All Posted Prices}
  \label{fig:cpi_replication_sales}
  \includegraphics[width=0.82\textwidth]{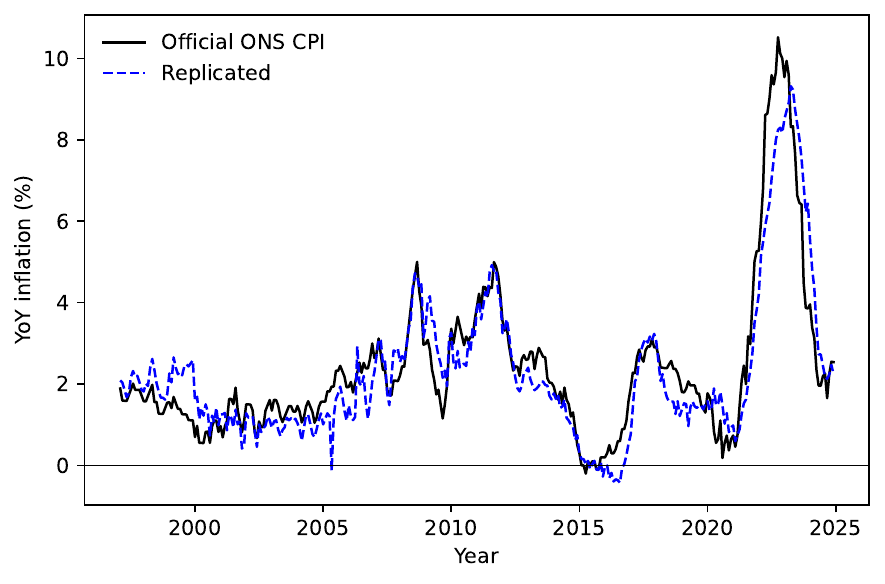}
  \smallskip
  \begin{minipage}{0.82\textwidth}\setstretch{1.0}
    \footnotesize\textit{Notes:} The figure compares official UK ONS
    non-seasonally-adjusted year-on-year CPI inflation with the rate implied
    by the reconstructed chained index built from the baseline micro-price
    panel using all posted prices. The sample runs from February 1997 to
    December 2024. The solid black line denotes official CPI inflation and
    the dashed blue line denotes the replicated series. Over the joint
    sample, the root mean squared difference is 0.76 percentage points and
    the correlation is 0.92.
  \end{minipage}
\end{figure}

\subsection{Benchmarking the Price-Change Distribution}
\label{app:moments}

To validate the pipeline we compute moments of $|\Delta p_t|$ following the
protocol of \citet{Blanco2021}. Each non-zero price change is standardised at
the item level by its weighted within-item mean and standard deviation, using
the CPI product weight $\omega_{j,t}$ defined in
Appendix~\ref{app:cpi_aggregation}. The top and bottom 2\% of the
standardised distribution are trimmed, zero price changes are excluded, and
the reported mean and standard deviation are weighted averages of $|\Delta p|$
in original (proportion) units, while the reported kurtosis is the
$\omega$-weighted raw Pearson kurtosis of the standardised price change.
These aggregate statistics are not used in estimation; they serve as a check
that the cleaning, weighting, and aggregation are correctly implemented.

Table~\ref{tab:moments} reports the resulting moments across five sales
treatments and lists published estimates from the literature.

\begin{table}[!ht]
  \centering
  \caption{Moments of $|\Delta p|$ across Sales Treatments and Literature (2\% Trim)}
  \label{tab:moments}
  \begin{tabular}{@{}lccc@{}}
    \toprule
                                        & Mean $|\Delta p|$ & Std dev $|\Delta p|$ & Kurtosis \\
    \midrule
    \multicolumn{4}{@{}l}{\textit{This paper}}                                                       \\
    \quad All prices, no trim           & 0.157             & 0.177                & 7.75            \\
    \quad All prices (2\% trim)         & 0.141             & 0.150                & 2.90            \\
    \quad ONS sales dropped             & 0.088             & 0.099                & 3.23            \\
    \quad NS Filter A                   & 0.115             & 0.138                & 3.24            \\
    \quad NS Filter B                   & 0.127             & 0.145                & 3.13            \\
    \addlinespace
    \multicolumn{4}{@{}l}{\textit{Literature}}                                                       \\
    \quad \citet{KlenowKryvtsov2008}    & 0.113             & --                   & --              \\
    \quad \citet{NakamuraSteinsson2008} & 0.085             & --                   & --              \\
    \quad \citet{Blanco2021}            & 0.110             & --                   & 4.02            \\
    \quad \citet{KaradiReiff2019}       & 0.099             & --                   & 3.98            \\
    \quad \citet{Blanco2021} (UK)       & 0.153             & 0.205                & 3.81            \\
    \bottomrule
  \end{tabular}
  \smallskip
  \begin{minipage}{\textwidth}\setstretch{1.0}
    \footnotesize \textit{Notes:} All ``This paper'' moments are computed on
    the baseline cleaned ONS micro-price panel using the CPI product weight
    $\omega_{j,t}$ defined in
    Appendix~\ref{app:cpi_aggregation}. Mean and standard deviation of
    $|\Delta p|$ are reported in original (proportion) units; kurtosis is the
    raw (Pearson) kurtosis of price changes after item-level standardisation,
    matching the protocol of \citet{Blanco2021}. The 2\% trim drops the top
    and bottom 2\% of the item-standardised distribution. For the trimmed
    rows, the baseline (with sales) uses all valid price quotes; the
    ``ONS sales dropped'' row sets ONS \texttt{S}-flagged observations to
    missing; ``NS Filter A'' and ``NS Filter B'' apply the
    \citet{NakamuraSteinsson2008} sale filter with the parameters of
    Table~\ref{tab:ns_params}. \citet{NakamuraSteinsson2008} report the median
    $|\Delta p|$ rather than the mean. The literature values use different
    datasets, different sale-handling, and (in some cases) different
    standardisation rules; see Appendix~\ref{app:moments} for a discussion of
    comparability.
  \end{minipage}
\end{table}

The most direct comparison is the UK row of \citet{Blanco2021}, who use the same
ONS micro-price data, weight observations by item-level CPI expenditure
weights, standardise at the item level, and trim the top and bottom 2\% of
the standardised distribution. Our NS Filter A and NS Filter B rows, which
remove temporary sales by the \citet{NakamuraSteinsson2008} filter, are the
closest like-for-like analogues: they imply a mean $|\Delta p|$ of
$0.115$--$0.127$ and kurtosis of $3.13$--$3.24$, bracketing the
\citet{Blanco2021} UK values of $0.153$ and $3.81$. The remaining gap
reflects differences in the sale-detection rule (the V-shape filter used in
\citet{Blanco2021} is less aggressive than NS Filter A), differences in the
item-coverage and base-price cleaning, and a slightly different sample
period.

The remaining literature rows are listed for context only.
\citet{KlenowKryvtsov2008}, \citet{NakamuraSteinsson2008}, the US row of
\citet{Blanco2021}, and \citet{KaradiReiff2019} use US scanner or CPI
microdata under distinct sample-construction, weighting, and standardisation
rules, and are not directly comparable to our ONS replication. The kurtosis
values in the literature panel are raw Pearson kurtosis as reported in those
papers, so a value near $3$ indicates a near-normal standardised distribution
and a value near $1$ a near two-point distribution in the sense of
\citet{GolosovLucas2007}.

\subsection{Treatment of Sales}
\label{app:NS_filter}
We consider three sale treatments. The first is an ONS-sale-dropped panel.
For this panel, observations with the ONS \texttt{S} flag are not used as
valid prices: their price, log price, and contemporaneous price change are set
to missing. Since the price change in the month after a sale depends on the
sale price as the lagged price, the following observation's price change is
also set to missing within each product series. After this operation, we rerun
the short-spell deletion rule, setting isolated observations to missing and
dropping product series with fewer than two valid prices. Finally, we recompute
the time-since-change variable from the resulting filtered price-change series.

The ONS \texttt{S} flag may incompletely capture temporary price
reductions. We therefore also construct two  sale-adjusted
panels, following \cite{NakamuraSteinsson2008}. The filter operates product
by product on the price level series. It is
governed by three integer parameters: a short rebound window $JJ$, a minimum
number of distinct price points $KK$, and a longer look-ahead window $LL$. At
each period $t$ with observed price $p_t$, a \emph{regular price} $r_t$ is
assigned via the following ordered rules:

\begin{enumerate}
  \item \textbf{Continuity.} $p_t = r_{t-1}$: set $r_t = r_{t-1}$.
  \item \textbf{Price increase.} $p_t > r_{t-1}$: set $r_t = p_t$.
  \item \textbf{Drop with rebound.} $p_t < r_{t-1}$ and $p_{t+s} = r_{t-1}$ for some $s \le JJ$: classify period $t$ as a sale, set $r_t = r_{t-1}$.
  \item \textbf{Drop, no rebound, sufficient variation.} If the price does not rebound within $JJ$ periods but the following $LL$ observations contain at least $KK$ distinct price values, treat the drop as permanent and set $r_t = p_t$.
  \item \textbf{Drop, no rebound, modal-price rule.} If the maximum price in the $LL$-period window reappears later in the series, adopt it as the regular price.
  \item \textbf{Default.} Set $r_t = p_t$.
\end{enumerate}

The output is the adjusted series $\{r_t\}$ and a binary indicator equal to one
when $r_t \neq p_t$. We identify sales using the following two parameterisations:

\begin{table}[h]
  \centering
  \caption{Nakamura--Steinsson Filter Parameters}
  \label{tab:ns_params}
  \begin{tabular}{@{}lccc@{}}
    \toprule
                            & Rebound window $JJ$ & Min.\ distinct prices $KK$ & Look-ahead $LL$ \\
    \midrule
    Filter A (baseline)     & 3                   & 3                          & 3               \\
    Filter B (conservative) & 1                   & 1                          & 3               \\
    \bottomrule
  \end{tabular}
  \smallskip
  \begin{minipage}{0.78\textwidth}\setstretch{1.0}
    \footnotesize\textit{Notes:} Filter A is the baseline sales-filter
    parameterisation used in the robustness exercises. Filter B is a more
    conservative parameterisation. The columns report the rebound window,
    minimum number of distinct prices in a spell, and look-ahead horizon used
    by the Nakamura--Steinsson sale-filter algorithm. For Filter B the
    look-ahead parameter is inert: with $KK=1$ the modal-price branch that
    uses $LL$ is never reached, so Filter B produces identical sale flags for
    any value of $LL$ (the canonical parameterisation sets $L=1$).
  \end{minipage}
\end{table}

\subsection{Macro Predictors}
\label{app:macro_predictors}

Table~\ref{tab:macro_all} lists the baseline macroeconomic and financial
predictors used in the macro forecasting exercise.

\begingroup
\footnotesize
\setlength{\tabcolsep}{4pt}
\begin{longtable}{p{4.6cm} p{6.7cm} p{1.9cm} c}
  \caption{Macro Panel Series by Category}\label{tab:macro_all} \\
  \toprule
  \textbf{Variable} & \textbf{Description} & \textbf{Source} & \textbf{Tr.} \\
  \midrule
  \endfirsthead
  \multicolumn{4}{l}{\small\textit{(continued)}} \\[2pt]
  \toprule
  \textbf{Variable} & \textbf{Description} & \textbf{Source} & \textbf{Tr.} \\
  \midrule
  \endhead
  \midrule
  \multicolumn{4}{r}{\small Continued on next page} \\
  \endfoot
  \bottomrule
  \multicolumn{4}{p{13.5cm}}{\footnotesize \textit{Notes:} This table lists the 37 series of the macro panel. ``Tr.''\ is the stationarity transformation applied before estimation: $\Delta\ln$~=~$100\times$ log difference; level~=~untransformed; $\Delta$~=~first difference of the level (VIX). All series are non-seasonally adjusted except the two trade series ($^{\dagger}$), which are only available seasonally adjusted and enter as the first-estimate row of the ONS trade revisions triangles. $^{*}$\texttt{CPI\_\allowbreak{}ALL\_\allowbreak{}D7BT} is the history of the target index itself and is never used as a predictor.\par} \\
  \endlastfoot

  \multicolumn{4}{l}{\textit{Consumer prices}} \\[2pt]
  \texttt{CPI\_\allowbreak{}ALL\_\allowbreak{}D7BT}$^{*}$ & CPI index 00 all items & ONS & $\Delta\ln$ \\
  \texttt{CPI\_\allowbreak{}RENTS\_\allowbreak{}D7CE} & CPI index 04.1 actual rents for housing & ONS & $\Delta\ln$ \\
  \addlinespace[4pt]
  \multicolumn{4}{l}{\textit{Wages and labour-market slack}} \\[2pt]
  \texttt{AWE\_\allowbreak{}NSA\_\allowbreak{}TOTAL\_\allowbreak{}PAY\_\allowbreak{}COMPOSITE} & ONS AWE whole-economy NSA total-pay composite: MD9M historic level through 1999 and KA5Q current index from 2000 & ONS & $\Delta\ln$ \\
  \texttt{CLAIMS} & Claimant Count, K02000001 UK, people, NSA, thousands & ONS & $\Delta\ln$ \\
  \texttt{CLAIMS\_\allowbreak{}RATE} & Claimant Count rate, K02000001 UK, people, NSA, percentage & ONS & level \\
  \addlinespace[4pt]
  \multicolumn{4}{l}{\textit{Exchange rates}} \\[2pt]
  \texttt{GBP\_\allowbreak{}BROAD} & Broad effective exchange rate index, sterling, Jan 2005=100 & BoE & $\Delta\ln$ \\
  \texttt{GBP\_\allowbreak{}CAN} & GBP/CAD spot rate, monthly average & BoE & $\Delta\ln$ \\
  \texttt{GBP\_\allowbreak{}EUR} & GBP/EUR spot rate, monthly average & BoE & $\Delta\ln$ \\
  \texttt{GBP\_\allowbreak{}JAP} & GBP/JPY spot rate, monthly average & BoE & $\Delta\ln$ \\
  \texttt{GBP\_\allowbreak{}US} & GBP/USD spot rate, monthly average & BoE & $\Delta\ln$ \\
  \addlinespace[4pt]
  \multicolumn{4}{l}{\textit{Energy, food and commodity prices and futures}} \\[2pt]
  \texttt{COMP\_\allowbreak{}P} & World Bank primary commodity price index & World Bank & $\Delta\ln$ \\
  \texttt{CRUDE\_\allowbreak{}GAS\_\allowbreak{}FUT} & Natural gas futures front-month (FN1 or UK/European gas future depending vendor) & ICE & $\Delta\ln$ \\
  \texttt{CRUDE\_\allowbreak{}OIL\_\allowbreak{}FUT} & Brent crude oil front-month futures (CO1) & ICE & $\Delta\ln$ \\
  \texttt{FOODFUT} & Wheat futures index (INFAWHP) & Bloomberg & $\Delta\ln$ \\
  \texttt{OIL\_\allowbreak{}PRICE} & Crude oil prices: Brent Europe, dollars per barrel, monthly NSA & FRED & $\Delta\ln$ \\
  \addlinespace[4pt]
  \multicolumn{4}{l}{\textit{House prices}} \\[2pt]
  \texttt{HALIFAX\_\allowbreak{}HPI\_\allowbreak{}NSA} & Halifax UK house price index, non-seasonally adjusted & Halifax & $\Delta\ln$ \\
  \texttt{NATIONWIDE\_\allowbreak{}HPI\_\allowbreak{}NSA} & Nationwide UK house price index, non-seasonally adjusted & Nationwide & $\Delta\ln$ \\
  \addlinespace[4pt]
  \multicolumn{4}{l}{\textit{Interest and mortgage rates}} \\[2pt]
  \texttt{BANK\_\allowbreak{}RATE} & Official Bank Rate, monthly average & BoE & level \\
  \texttt{BGS\_\allowbreak{}10yrs\_\allowbreak{}yld} & 10-year gilt nominal par yield & BoE & level \\
  \texttt{BGS\_\allowbreak{}20yrs\_\allowbreak{}yld} & 20-year gilt nominal par yield & BoE & level \\
  \texttt{BGS\_\allowbreak{}5yrs\_\allowbreak{}yld} & 5-year gilt nominal par yield & BoE & level \\
  \texttt{MORT\_\allowbreak{}FRATE\_\allowbreak{}2YRS} & 2-year fixed mortgage rate, 75\% LTV & BoE & level \\
  \texttt{MORT\_\allowbreak{}FRATE\_\allowbreak{}5YRS} & 5-year fixed mortgage rate, 75\% LTV & BoE & level \\
  \texttt{SONIA} & Sterling overnight index average rate & FRED/BoE & level \\
  \addlinespace[4pt]
  \multicolumn{4}{l}{\textit{Financial markets and uncertainty}} \\[2pt]
  \texttt{FTSE250} & FTSE 250 index & FTSE & $\Delta\ln$ \\
  \texttt{FTSE\_\allowbreak{}ALL} & FTSE All Share index & FTSE & $\Delta\ln$ \\
  \texttt{SP500} & S\&P 500 index & S\&P & $\Delta\ln$ \\
  \texttt{UK\_\allowbreak{}focused\_\allowbreak{}equity} & iShares MSCI United Kingdom ETF (EWU) & iShares & $\Delta\ln$ \\
  \texttt{VIX} & CBOE Volatility Index & CBOE & $\Delta$ \\
  \addlinespace[4pt]
  \multicolumn{4}{l}{\textit{Trade}} \\[2pt]
  \texttt{ONS\_\allowbreak{}TOTAL\_\allowbreak{}TRADE\_\allowbreak{}EXPORTS\_\allowbreak{}REV\_\allowbreak{}TRI}$^{\dagger}$ & UK trade revisions triangle: total trade exports, current prices, seasonally adjusted & ONS & $\Delta\ln$ \\
  \texttt{ONS\_\allowbreak{}TOTAL\_\allowbreak{}TRADE\_\allowbreak{}IMPORTS\_\allowbreak{}REV\_\allowbreak{}TRI}$^{\dagger}$ & UK trade revisions triangle: total trade imports, current prices, seasonally adjusted & ONS & $\Delta\ln$ \\
  \addlinespace[4pt]
  \multicolumn{4}{l}{\textit{Supply-chain pressure}} \\[2pt]
  \texttt{AIR\_\allowbreak{}FREIGHT\_\allowbreak{}IN\_\allowbreak{}ASIA\_\allowbreak{}IC1312} & Inbound Price Index (International Services): Air Freight for Asia & FRED/ALFRED & $\Delta\ln$ \\
  \texttt{AIR\_\allowbreak{}FREIGHT\_\allowbreak{}IN\_\allowbreak{}EUROPE\_\allowbreak{}IC1311} & Inbound Price Index (International Services): Air Freight for Europe & FRED/ALFRED & $\Delta\ln$ \\
  \texttt{AIR\_\allowbreak{}FREIGHT\_\allowbreak{}OUT\_\allowbreak{}ASIA\_\allowbreak{}IS2312} & Outbound Price Index (International Services): Air Freight for Asia & FRED/ALFRED & $\Delta\ln$ \\
  \texttt{AIR\_\allowbreak{}FREIGHT\_\allowbreak{}OUT\_\allowbreak{}EUROPE\_\allowbreak{}IS2311} & Outbound Price Index (International Services): Air Freight for Europe & FRED/ALFRED & $\Delta\ln$ \\
  \addlinespace[4pt]
  \multicolumn{4}{l}{\textit{Activity and sentiment}} \\[2pt]
  \texttt{ACEA\_\allowbreak{}UK\_\allowbreak{}CAR\_\allowbreak{}REG\_\allowbreak{}NSA} & New passenger car registrations, United Kingdom & OECD & $\Delta\ln$ \\
  \texttt{CLI} & OECD composite leading indicator for UK & OECD & level \\
  \addlinespace[4pt]
\end{longtable}
\endgroup

\clearpage
\subsection{Exact Formulas for Distributional Statistics}
\label{app:encoding_formulas}

This subsection gives the exact definitions of the statistics used to encode
the microdata, described in Section~\ref{sec:methods}. Let $c$ index a
COICOP1 category and $t$ a calendar month. A cell $(c,t)$ contains all
retained price observations $j$ in category $c$ during month $t$; let
$\mathcal{J}_{ct}$ denote this set, with $N_{ct}=|\mathcal{J}_{ct}|$.
Price changes $\Delta p_{jt}$ are defined as in
Section~\ref{sec:data_description}. Some observations in $\mathcal{J}_{ct}$
may have missing $\Delta p_{jt}$, for example when a valid consecutive-month
price change cannot be computed.

\paragraph{Weighted price-change statistics.}

Each price observation carries the CPI product weight $\omega_{j,t}$ defined
in Equation~\eqref{eq:omega_jt}. For each cell, denote the total weight and
the total weight on zero and non-zero price-change observations by
\[
  W_{ct} = \sum_{j \in \mathcal{J}_{ct}} \omega_{j,t},
  \qquad
  \mathcal{J}_{ct}^{0}
  =
  \{j \in \mathcal{J}_{ct}: \Delta p_{jt}=0\},
  \qquad
  W_{ct}^{0} = \sum_{j \in \mathcal{J}_{ct}^{0}} \omega_{j,t},
\]
\[
  \mathcal{J}_{ct}^{+}
  =
  \{j \in \mathcal{J}_{ct}: \Delta p_{jt}\ne 0
  \text{ and } \Delta p_{jt}\text{ is observed}\},
  \qquad
  W_{ct}^{+} = \sum_{j \in \mathcal{J}_{ct}^{+}} \omega_{j,t}.
\]
The weighted fraction of zero price changes is
\[
  f_{ct}^{0} = \frac{W_{ct}^{0}}{W_{ct}}.
\]
The nine deciles $q_{k,ct}$ for $k \in \{0.1, 0.2, \ldots, 0.9\}$ are the
$\omega$-weighted empirical quantiles of
$\{\Delta p_{jt}\}_{j \in \mathcal{J}_{ct}^{+}}$, defined as the smallest
observed value $q$ such that
\[
  \frac{1}{W_{ct}^{+}}
  \sum_{j \in \mathcal{J}_{ct}^{+}\,:\,\Delta p_{jt} \le q} \omega_{j,t}
  \;\ge\; k.
\]
The five moments use the Kish effective sample size
\[
  n_{\text{eff},ct}^{+}
  = \frac{\bigl(W_{ct}^{+}\bigr)^{2}}
          {\sum_{j \in \mathcal{J}_{ct}^{+}} \omega_{j,t}^{2}},
\]
which equals $n_{ct}^{+}$ when all weights within the cell are equal and
shrinks toward $1$ as the weights concentrate on a few observations. Let
\[
  \bar{m}_{ct} = \frac{1}{W_{ct}^{+}}
                 \sum_{j \in \mathcal{J}_{ct}^{+}} \omega_{j,t}\,\Delta p_{jt},
  \qquad
  s_{ct}
  = \sqrt{\,
    \frac{n_{\text{eff},ct}^{+}}{n_{\text{eff},ct}^{+}-1}
    \cdot
    \frac{1}{W_{ct}^{+}}
    \sum_{j \in \mathcal{J}_{ct}^{+}} \omega_{j,t}
    \bigl(\Delta p_{jt} - \bar{m}_{ct}\bigr)^{2}
    \,}.
\]
Then
\begin{align*}
  m_{1,ct} &= \bar{m}_{ct}, \\
  m_{2,ct} &= s_{ct}^{2}, \\
  m_{k,ct} &= \frac{1}{W_{ct}^{+}\,s_{ct}^{k}}
              \sum_{j \in \mathcal{J}_{ct}^{+}}
              \omega_{j,t}\bigl(\Delta p_{jt} - \bar{m}_{ct}\bigr)^{k},
              \quad k = 3, 5, \\
  m_{4,ct} &= \frac{1}{W_{ct}^{+}\,s_{ct}^{4}}
              \sum_{j \in \mathcal{J}_{ct}^{+}}
              \omega_{j,t}\bigl(\Delta p_{jt} - \bar{m}_{ct}\bigr)^{4} - 3.
\end{align*}
Here $m_1$ is the $\omega$-weighted mean, $m_2$ is the Kish/Bessel-corrected
weighted sample variance, $m_3$ is weighted skewness, $m_4$ is weighted
excess kurtosis (so that a normal distribution has $m_4 = 0$), and $m_5$ is
the weighted standardised fifth central moment. The higher-order moments
$m_3$--$m_5$ use the Bessel-corrected standard deviation $s_{ct}$ in the
denominator but divide the weighted raw central moments by $W_{ct}^{+}$ in
the numerator. The unweighted robustness check in
Appendix~\ref{app:unweighted_encoding} replaces $\omega_{j,t}$ throughout
with uniform weights, in which case $W_{ct}^{+}$ collapses to $n_{ct}^{+}$
and $n_{\text{eff},ct}^{+}$ to $n_{ct}^{+}$, and the formulas above reduce
to the standard unweighted analogues.

\paragraph{Comparison with \citet{Alvarez2016} kurtosis.}
\citet{Alvarez2016} report raw (non-excess) kurtosis, denoted
$\mathrm{Kur}(\Delta p)$, which equals $m_4 + 3$ in our notation.
Under normality $\mathrm{Kur} = 3$; under a \citet{GolosovLucas2007} two-point
distribution $\mathrm{Kur} = 1$; under Calvo $\mathrm{Kur} = 6$.
Their sufficient-statistic result for cumulative output is
$M = \tfrac{\delta}{6\varepsilon}\,\mathrm{Kur}(\Delta p)/N(\Delta p)$,
where $N(\Delta p)$ is the frequency of price changes.
After correcting for cross-sectional heterogeneity and measurement error,
\citet{Alvarez2016} estimate $\mathrm{Kur}(\Delta p) \approx 4$ in U.S.\ data,
corresponding to $m_4 \approx 1$ in our convention.

\citet{Alvarez2016} compute kurtosis on price changes that have been
standardised at the UPC-outlet (cell) level, with very small changes
discarded before pooling. Our COICOP1-month forecasting features defined
above use $\omega_{j,t}$-weighted within-cell moments without prior
item-level standardisation, while the validation exercise in
Appendix~\ref{app:moments} follows the protocol of \citet{Blanco2021}, which
applies $\omega_{j,t}$ weighting together with item-level standardisation
and a 2\% trim of the standardised distribution. Both of our pipelines
differ from \citet{Alvarez2016} in either the standardisation step or the
trimming rule, so the resulting kurtosis values are not directly comparable
to $\mathrm{Kur}(\Delta p)\approx 4$.

\paragraph{Time-since-change statistics.}

For each product $j$ in month $t$, after comparable and non-comparable
substitution events have been split into separate product spells, define the
time since the last reset as
\[
  \delta_{jt} = t - t^{*}_{jt},
\]
where $t^{*}_{jt}$ is the most recent reset month $t' \le t$. Reset months are
defined as follows: (i)~the first valid price observation of a product spell;
(ii)~any valid consecutive-month price-change observation with
$\Delta p_{jt'} \ne 0$; and (iii)~any observation following a gap of more than
12 months within the same product spell. Thus the entry month of a product
spell has $\delta_{jt}=0$ even though its price change is missing. Because
time-since-change features are computed only from observations with zero valid
price changes, entry months do not enter these moments.

For the time-since-change features, we use observations with zero valid
price changes and a defined time-since-reset. Let
\[
  \mathcal{J}_{ct}^{0,\delta}
  = \bigl\{j \in \mathcal{J}_{ct} : \Delta p_{jt} = 0
             \text{ and } \delta_{jt} \text{ is non-missing}\bigr\},
  \qquad
  W_{ct}^{0,\delta}
  = \sum_{j \in \mathcal{J}_{ct}^{0,\delta}} \omega_{j,t}.
\]
The nine deciles and five moments of
$\{\delta_{jt}\}_{j \in \mathcal{J}_{ct}^{0,\delta}}$ are computed by the
same $\omega_{j,t}$-weighted formulas as above, substituting $\delta_{jt}$
for $\Delta p_{jt}$, $\mathcal{J}_{ct}^{0,\delta}$ for
$\mathcal{J}_{ct}^{+}$, and $W_{ct}^{0,\delta}$ for $W_{ct}^{+}$.

\clearpage

\section{Econometric Theory of Scan Tests}\label{sec:theory_scan_tests}

\numberwithin{equation}{section}
\makeatletter
\@ifundefined{proposition}{\newtheorem{proposition}[theorem]{Proposition}}{}
\@ifundefined{lemma}{\newtheorem{lemma}[theorem]{Lemma}}{}
\@ifundefined{corollary}{\newtheorem{corollary}[theorem]{Corollary}}{}
\@ifundefined{remark}{\theoremstyle{remark}\newtheorem{remark}{Remark}}{}
\makeatother
\providecommand{\1}{\mathbf{1}}
\providecommand{\R}{\mathbb R}
\providecommand{\E}{\mathbb E}
\providecommand{\Pbb}{\mathbb P}
\providecommand{\Var}{\mathrm{Var}}
\providecommand{\Cov}{\mathrm{Cov}}
\providecommand{\op}{o_p}
\providecommand{\Op}{O_p}
\providecommand{\ops}{o_{P^*}}
\providecommand{\Ops}{O_{P^*}}
\providecommand{\psito}{\xrightarrow{p}}
\providecommand{\dto}{\Rightarrow}
\providecommand{\st}{\,\colon\,}
\providecommand{\norm}[1]{\left\lVert #1\right\rVert}
\providecommand{\cI}{\mathcal{I}}
\providecommand{\cA}{\mathcal{A}}
\providecommand{\cR}{\mathcal{R}}
\providecommand{\psitwo}[1]{\lVert #1 \rVert_{\psi_2}}
\providecommand{\D}{D}
\providecommand{\abs}[1]{\left\lvert #1\right\rvert}
\providecommand{\one}{\mathbf{1}}
\providecommand{\red}[1]{{\color{red} #1}}
\renewcommand{\d}{d}

\subsection{Setup}
\subsubsection{Review of the scan test under simplified notation}

We use triangular-array notation and suppress the row index.  Fix a
deterministic exponent $\gamma\in[0,1/2]$.
The exponent $\gamma$ is part of the test specification and is not selected
from the data.  We use $\alpha\in(0,1)$ below for the nominal test level.
Let
$[q]:=\{1,\ldots,q\}$, let $\mathcal I$ be a collection of intervals in
$[n]$, and define
\begin{equation}\label{eq:Lmin_p}
  L_{\min}:=\min_{I\in\mathcal I}|I|,
  \qquad
  L_{\max}:=\max_{I\in\mathcal I}|I|,
  \qquad
  A:=\mathcal I\times[H],
  \qquad
  p:=|A|=H|\mathcal I|.
\end{equation}
Because $p\le n^2$,
\begin{equation}\label{eq:logpn_logn}
  \log(pn)\asymp\log n.
\end{equation}
For $a\in A$, write $a=(I,h)$ for its associated tuple, and set
\begin{equation}
  S_a:=\sum_{t\in I}D_t^{(h)},
  \qquad
  \bar D_a:=\frac{S_a}{|I|},
  \qquad
  Q_a
  :=\sum_{t\in I}\left(D_t^{(h)}-\bar D_a\right)^2.
\end{equation}
\begin{equation}\label{eq:observed_Q}
  T_a:=\frac{S_a}{|I|^{1/2-\gamma}Q_a^\gamma},
  \quad
  T_{\max}:=\max_{a\in A}T_a.
\end{equation}
For $\gamma=0$, we define $q^0=1$ for every $q\ge0$, so
$T_a=S_a$.  For $\gamma>0$, a ratio with a zero denominator is defined
to be zero.  The one-sided null
is
\begin{equation}\label{eq:null_standalone}
  H_0:\quad
  \E[D_t^{(h)}]\le0
  \quad\text{for every }t\in[n],\ h\in[H].
\end{equation}

Let $\xi=(\xi_1,\ldots,\xi_n)^\top\sim N(0,\Theta_n)$ independently of
the data, where $\Theta_n$ is positive semidefinite with unit diagonal.
Regime~I uses $\Theta_n=I_n$; Regime~D uses the kernel covariance matrix
defined in Assumption~\ref{ass:D_kernel}.  For $a=(I,h)\in A$, write
\[
  S_a^*:=
  \sum_{t\in I}\xi_t\left(D_t^{(h)}-\bar D_a\right),
  \qquad
  \bar D_a^*:=\frac{S_a^*}{|I|},
  \qquad
  Q_a^*
  :=\sum_{t\in I}\left\{
    \xi_t\left(D_t^{(h)}-\bar D_a\right)
    -\bar D_a^*
  \right\}^2.
\]
The bootstrap statistic is
\begin{equation}\label{eq:implemented_bootstrap}
  T_a^*:=
  \frac{S_a^*}{|I|^{1/2-\gamma}(Q_a^*)^\gamma},
  \quad
  T_{\max}^*:=\max_{a\in A}T_a^*.
\end{equation}
The conditional critical value for the bootstrap statistic above is
\begin{equation}
  c_{n,1-\alpha}^{*}
  :=
  \inf\{x:\Pbb^*(T_{\max}^*\le x)\ge1-\alpha\}.
\end{equation}
Here $\Pbb^*$ denotes probability conditional on the data; conditional
stochastic orders are understood in outer probability.
The test rejects when
$T_{\max}>c_{n,1-\alpha}^{*}$.

\subsubsection{Other notation used for proofs}

Write
\begin{equation}
  \mu_t^{(h)}:=\E[D_t^{(h)}],
  \qquad
  \d_t^{(h)}:=D_t^{(h)}-\mu_t^{(h)}.
\end{equation}
For $a=(I,h)\in A$, write
\begin{equation}
  S_a^\circ:=\sum_{t\in I}\d_t^{(h)}.
\end{equation}
Define
\begin{equation}\label{eq:dta_def}
  \bar\d_a:=\frac1{|I|}\sum_{t\in I}\d_t^{(h)},
  \qquad
  d_{t,a}:=\d_t^{(h)}-\bar\d_a.
\end{equation}
\begin{equation}\label{eq:muta_def}
  \bar\mu_a:=\frac1{|I|}\sum_{t\in I}\mu_t^{(h)},
  \qquad
  \mu_{t,a}:=\mu_t^{(h)}-\bar\mu_a.
\end{equation}
\begin{equation}
  q_a:=\sum_{t\in I}d_{t,a}^2,
  \qquad
  \nu_a:=\sum_{t\in I}\E[\d_t^{(h)2}].
\end{equation}
\begin{equation}\label{eq:denom_Adef}
  q_a^\circ:=q_a-\nu_a=\sum_{t\in I}d_{t,a}^2-\nu_a,
  \qquad
  C_a:=\sum_{t\in I}d_{t,a}\mu_{t,a}.
\end{equation}
The centered comparison statistics are
\begin{equation}
  T_{\max}^\circ
  :=\max_{a\in A}
  \frac{S_a^\circ}{|I|^{1/2-\gamma}Q_a^\gamma}.
\end{equation}
The interval-wise mean-flatness index is
\begin{equation}\label{eq:rn}
  \nu_{a,\mu}:=\sum_{t\in I}\mu_{t,a}^2,
  \qquad
  r_n
  :=
  \max_{a=(I,h)\in A}
  \frac{\nu_{a,\mu}}{\nu_a}.
\end{equation}
It controls the shape of the mean within an interval, not its level.  In
particular, arbitrary constant negative shifts leave $r_n$ equal to zero.

For positive deterministic sequences $a_n$ and $b_n$, $a_n=\Omega(b_n)$
means $b_n=O(a_n)$, and $a_n=\omega(b_n)$ means $b_n=o(a_n)$.

\subsubsection{Common assumptions}

\begin{appassumption}\label{ass:subgaussian}
The centered coordinates $\d_t^{(h)}=D_t^{(h)}-\mu_t^{(h)}$ are uniformly
sub-Gaussian: for a constant
$B_0<\infty$,
\[
  \sup_{t\in[n],\,h\in[H]}\|\d_t^{(h)}\|_{\psi_2}\le B_0.
\]
The Orlicz norms are defined in Section~\ref{app:miscellany}.
\end{appassumption}

\begin{appassumption}\label{ass:bounded_second_moment}
For constants $0<\sigma_-^2\le\sigma_+^2<\infty$,
\[
  \sigma_-^2
  \le\E[\d_t^{(h)2}]
  \le\sigma_+^2
  \qquad\text{for every }t\in[n],\ h\in[H].
\]
\end{appassumption}

\begin{appassumption}
\label{ass:length_ratio}
For a constant $C_L<\infty$,
\[
  \left(\frac{L_{\max}}{L_{\min}}\right)^{1-2\gamma}
  \le C_L.
\]
For fixed $\gamma<1/2$, this requires $L_{\max}/L_{\min}=O(1)$.
For $\gamma=1/2$, it imposes no restriction.
\end{appassumption}

\subsection{Preliminary results}

\subsubsection{Lemmas used for both regimes}

The following algebraic decomposition is used in both regimes to compare the
observed denominator with its oracle variance scale.
\begin{lemma}\label{lem:denominator_decomposition}
For every $a=(I,h)\in A$,
\[
  \frac{Q_a-\nu_a}{\nu_a}
  =
  \frac{q_a^\circ}{\nu_a}
  +
  2\frac{C_a}{\nu_a}
  +
  \frac{\nu_{a,\mu}}{\nu_a}.
\]
\end{lemma}

\begin{proof}
Because $\sum_{t\in I}\mu_{t,a}=0$,
\[
  D_t^{(h)}-\bar D_a=d_{t,a}+\mu_{t,a}.
\]
Expanding the square and using the definition of $C_a$ gives
\[
  Q_a=q_a+2C_a+\nu_{a,\mu}.
\]
Also,
\[
  q_a
  =
  \sum_{t\in I}\left(\d_t^{(h)}-\bar\d_a\right)^2
  =
  \sum_{t\in I}\d_t^{(h)2}-|I|\bar\d_a^2,
\]
and the definition \eqref{eq:denom_Adef} gives
$q_a-\nu_a=q_a^\circ$.  The displayed identity follows by subtracting
$\nu_a$ and dividing by $\nu_a$.
\end{proof}

The next lemma gives the deterministic perturbation bound for changing the
variance normalizer in a power-normalized statistic.
\begin{lemma}\label{lem:power_denominator_perturbation}
Let $(s_a,V_a,\widetilde V_a)_{a\in A}$ be real numbers with
$V_a,\widetilde V_a>0$, and set
\[
  \Delta_V:=\max_{a\in A}\left|\frac{V_a}{\widetilde V_a}-1\right|.
\]
On $\{\Delta_V\le1/2\}$,
\[
  \max_{a\in A}
  \left|
  \frac{s_a}{|I|^{1/2-\gamma}V_a^\gamma}
  -
  \frac{s_a}{|I|^{1/2-\gamma}\widetilde V_a^\gamma}
  \right|
  \le
  C\gamma\Delta_V
  \max_{a\in A}
  \frac{|s_a|}{|I|^{1/2-\gamma}\widetilde V_a^\gamma},
\]
where $C$ is universal.
\end{lemma}

\begin{proof}
For $\gamma=0$ the two sides on the left are identical.  For
$\gamma>0$, write $u_a=V_a/\widetilde V_a$.  On
$\{\Delta_V\le1/2\}$, $u_a\in[1/2,3/2]$ for every $a$.  The function
$u\mapsto u^{-\gamma}$ has derivative bounded by $C\gamma$ on this
interval, and therefore $|u_a^{-\gamma}-1|\le C\gamma|u_a-1|$.  Multiplying
by $|s_a|/(|I|^{1/2-\gamma}\widetilde V_a^\gamma)$ and taking maxima proves
the claim.
\end{proof}

The next proposition provides the distribution-function perturbation bound
used when one statistic is uniformly close to another.
\begin{proposition}\label{prop:cdf_perturbation}
Let $X$ and $Y$ be real-valued random variables on a probability space with
probability measure $\mathsf P$.  For every $\eta>0$,
\[
  \sup_{x\in\R}
  \left|
  \mathsf P(X\le x)-\mathsf P(Y\le x)
  \right|
  \le
  \mathsf P(|X-Y|>\eta)
  +
  \sup_{x\in\R}\mathsf P(|Y-x|\le\eta).
\]
\end{proposition}

\begin{proof}
For every $x\in\R$,
\[
  \{X\le x\}
  \subset
  \{Y\le x+\eta\}\cup\{|X-Y|>\eta\},
\]
and
\[
  \{Y\le x-\eta\}
  \subset
  \{X\le x\}\cup\{|X-Y|>\eta\}.
\]
The first inclusion gives
\[
  \mathsf P(X\le x)-\mathsf P(Y\le x)
  \le
  \mathsf P(|X-Y|>\eta)
  +
  \mathsf P(x<Y\le x+\eta),
\]
and the second inclusion gives
\[
  \mathsf P(Y\le x)-\mathsf P(X\le x)
  \le
  \mathsf P(|X-Y|>\eta)
  +
  \mathsf P(x-\eta<Y\le x).
\]
Taking the supremum over $x$ proves the result.
\end{proof}

The next lemma converts convergence in probability into a deterministic error
threshold.
\begin{lemma}\label{lem:deterministic_threshold}
If $Z_n\ge0$ and $Z_n=o_p(1)$, then there exists a deterministic sequence
$\rho_n\downarrow0$ such that
\[
  \Pbb(Z_n>\rho_n)\to0.
\]
\end{lemma}

\begin{proof}
For each integer $j\ge1$, choose $N_j$ large enough that
$\Pbb(Z_n>1/j)\le1/j$ whenever $n\ge N_j$, and take the sequence
$(N_j)$ increasing.  Set $\rho_n=1/j$ for $N_j\le n<N_{j+1}$.  Then
$\rho_n\downarrow0$ and $\Pbb(Z_n>\rho_n)\le\rho_n$ for all sufficiently
large $n$.
\end{proof}

\subsubsection{Bounding the tail probability of $T_{\max}$ by $T_{\max}^\circ$}

The following lemma gives the pathwise comparison between the feasible
statistic and its centered counterpart under the null.
\begin{lemma}\label{lem:null_pathwise_dominance}
Under \eqref{eq:null_standalone},
\[
  T_{\max}\le T_{\max}^{\circ}.
\]
\end{lemma}

\begin{proof}
Under \eqref{eq:null_standalone}, for every $a=(I,h)\in A$,
\[
  \sum_{t\in I}\mu_t^{(h)}\le 0.
\]
Therefore
\[
  S_a=S_a^{\circ}+\sum_{t\in I}\mu_t^{(h)}\le S_a^{\circ}
  \qquad\text{for every }a\in A.
\]
The two coordinates $T_a$ and the corresponding coordinate in
$T_{\max}^{\circ}$ have the same realized denominator
$|I|^{1/2-\gamma}Q_a^\gamma$, so the coordinatewise inequality implies
$T_{\max}\le T_{\max}^{\circ}$.
\end{proof}

The next lemma gives the common tail-probability comparison between
$T_{\max}^{\circ}$ and $\mathbb{T}_{\max}^{\circ}$.
\begin{lemma}\label{lem:centered_oracle_perturbation}
Set
\[
  \Delta_Q:=\max_{a\in A}\left|\frac{Q_a}{\nu_a}-1\right|.
\]
If
\[
  \Pbb(\Delta_Q>1/2)
  +
  \Pbb\!\left(
    C\gamma\Delta_Q\bar{\mathbb{T}}_{\max}^{\circ}>\eta_n
  \right)
  \to0
\]
for a universal constant $C<\infty$, and
\[
  \sup_{x\in\R}
  \Pbb\!\left(|\mathbb{T}_{\max}^{\circ}-x|\le\eta_n\right)
  \to0,
\]
then
\[
  \sup_{x\in\R}
  \left|
  \Pbb(T_{\max}^{\circ}>x)
  -
  \Pbb(\mathbb{T}_{\max}^{\circ}>x)
  \right|
  =
  \sup_{x\in\R}
  \left|
  \Pbb(T_{\max}^{\circ}\le x)
  -
  \Pbb(\mathbb{T}_{\max}^{\circ}\le x)
  \right|
  =
  o(1).
\]
\end{lemma}

\begin{proof}
On $\{\Delta_Q\le1/2\}$, $Q_a,\nu_a>0$ for every $a\in A$.
Lemma~\ref{lem:power_denominator_perturbation}, applied with
$s_a=S_a^\circ$, $V_a=Q_a$, and $\widetilde V_a=\nu_a$, gives
\[
  \max_{a\in A}
  \left|
  \frac{S_a^\circ}{|I|^{1/2-\gamma}Q_a^\gamma}
  -
  \mathbb{T}_a^\circ
  \right|
  \le
  C\gamma\Delta_Q\bar{\mathbb{T}}_{\max}^{\circ}.
\]
Using $|\max_a x_a-\max_a y_a|\le\max_a|x_a-y_a|$ gives
\[
  |T_{\max}^{\circ}-\mathbb{T}_{\max}^{\circ}|
  \le
  C\gamma\Delta_Q\bar{\mathbb{T}}_{\max}^{\circ}.
\]
Thus
$\Pbb(|T_{\max}^{\circ}-\mathbb{T}_{\max}^{\circ}|>\eta_n)\to0$ under the
displayed assumptions.  Proposition~\ref{prop:cdf_perturbation}, applied
with $X=T_{\max}^{\circ}$ and $Y=\mathbb{T}_{\max}^{\circ}$, proves the
distribution-function claim.  The tail-probability claim is identical
because tail probabilities are complements of distribution functions.
\end{proof}

\subsubsection{Approximating the bootstrap distribution by a fixed-denominator version}

The proofs use the following auxiliary fixed-denominator bootstrap statistics
\begin{equation}\label{eq:fixed_bootstrap}
  T_{0,a}^*:=
  \frac{S_a^*}{|I|^{1/2-\gamma}Q_a^\gamma},
  \qquad
  T_{0,\max}^*:=\max_{a\in A}T_{0,a}^*.
\end{equation}
Conditionally on the data, $(T_{0,a}^*)_{a\in A}$ is Gaussian.  The actual
bootstrap statistic implemented in \eqref{eq:implemented_bootstrap} is not Gaussian,
because the same multipliers enter its numerator and denominator.  But we
will show that the two versions have the same asymptotic bootstrap
distribution.

The following lemma is used in both regimes.  It permits nonconstant means
through the index $r_n$ and therefore does not impose a bound on the level of
the mean.
It gives a uniform concentration bound for the random bootstrap denominator
and compares the random- and fixed-denominator bootstrap distribution
functions.

\begin{lemma}
\label{lem:random_denominator}
Let $k_n=\Omega(1)$ be a deterministic sequence, and suppose Assumptions~\ref{ass:bounded_second_moment} and
\ref{ass:length_ratio} hold.  Suppose that, for some universal constants
$0<c<C<\infty$, with probability $1-o(1)$ over data, uniformly over
$a=(I,h)\in A$,
\begin{align}
  &Q_a\ge c|I|,
  \label{eq:RD_Qlower}\\
  &\max_{t\in I}(\d_t^{(h)})^2\le C\log n,
  \qquad
  \sum_{t\in I}(\d_t^{(h)})^4\le C|I|,
  \label{eq:RD_fourth}\\
  &c\le\Var^*\!\left(\frac{S_a^*}{\sqrt{Q_a}}\right)\le C.
  \label{eq:RD_var}
\end{align}
Suppose also that $\Theta_n$ has unit diagonal and
\begin{equation}\label{eq:RD_rows}
  \max_t\sum_{u=1}^n|\Theta_n(t,u)|\le Ck_n,
  \qquad
  \max_t\sum_{u=1}^n\Theta_n(t,u)^2\le Ck_n.
\end{equation}
Then, in probability,
\begin{equation}\label{eq:RD_rate}
  \Delta_Q^*
  :={}
  \max_{a\in A}
  \left|\frac{Q_a^*}{Q_a}-1\right|
  =
  O_{P^*}\!\left(
    \sqrt{\frac{k_n\log n}{L_{\min}}}
    +\frac{k_n(\log n)^2}{L_{\min}}
    +k_nr_n\log n
  \right).
\end{equation}
If
\begin{equation}\label{eq:RD_growth}
  \frac{k_n(\log n)^3}{L_{\min}}\to0,
  \qquad
  k_nr_n(\log n)^2\to0,
\end{equation}
then
\begin{equation}\label{eq:RD_cdf}
  \sup_{x\in\R}
  \left|
    \Pbb^*(T_{\max}^*\le x)
    -\Pbb^*(T_{0,\max}^*\le x)
  \right|
  =
  o_p(1).
\end{equation}
\end{lemma}

\begin{proof}
Fix $a=(I,h)$.
Since the bootstrap sample mean over $I$ is $\bar D_a^*$,
\[
  Q_a^*
  =
  \left\{
    \sum_{t\in I}\left(D_t^{(h)}-\bar D_a\right)^2\xi_t^2
    -|I|(\bar D_a^*)^2
  \right\}.
\]
Define
\[
  w_{t,a}:=\frac{(D_t^{(h)}-\bar D_a)^2}{Q_a}.
\]
Because $\sum_{t\in I}w_{t,a}=1$ and $\Theta_n(t,t)=1$, the exact identity
\begin{equation}\label{eq:RD_identity}
  \frac{Q_a^*}{Q_a}
  =
  1+\sum_{t\in I}w_{t,a}(\xi_t^2-1)-\frac{(S_a^*)^2}{|I|Q_a}
\end{equation}
holds.  Consequently,
\begin{equation}\label{eq:RD_decomposition}
  \Delta_Q^*
  \le
  \max_{a=(I,h)\in A}\left|\sum_{t\in I}w_{t,a}(\xi_t^2-1)\right|
  +\frac1{L_{\min}}\max_a\left(\frac{S_a^*}{\sqrt{Q_a}}\right)^2.
\end{equation}

The bounds in \eqref{eq:RD_fourth} imply
\[
  \max_{t\in I}d_{t,a}^2\le C\log n,
  \qquad
  \sum_{t\in I}d_{t,a}^4\le C|I|
\]
by $|\bar\d_a|\le\max_{t\in I}|\d_t^{(h)}|$ and Jensen's inequality.
Since $D_t^{(h)}-\bar D_a=d_{t,a}+\mu_{t,a}$, Assumption
\ref{ass:bounded_second_moment}, \eqref{eq:rn}, and
\eqref{eq:RD_Qlower}--\eqref{eq:RD_fourth} imply
\begin{align}
  \max_{t\in I}w_{t,a}
  &\le C\left(\frac{\log n}{|I|}+r_n\right),
  \label{eq:RD_maxweight}\\
  \sum_{t\in I}w_{t,a}^2
  &\le C\left(\frac1{|I|}+r_n^2\right).
  \label{eq:RD_l2weight}
\end{align}
Indeed, $\sum_{t\in I}\mu_{t,a}^2=\nu_{a,\mu}\le r_n\nu_a\le C r_n|I|$ by
\eqref{eq:rn} and Assumption~\ref{ass:bounded_second_moment}.  Hence
$\max_{t\in I}\mu_{t,a}^2\le C r_n|I|$, and the inequality
$(u+v)^2\le2(u^2+v^2)$ gives
\[
  \max_{t\in I}\left(D_t^{(h)}-\bar D_a\right)^2
  \le C(\log n+r_n|I|).
\]
Dividing by $Q_a\ge c|I|$ yields \eqref{eq:RD_maxweight}.  Similarly,
\[
  \sum_{t\in I}\left(D_t^{(h)}-\bar D_a\right)^4
  \le
  8\sum_{t\in I}d_{t,a}^4
  +8\sum_{t\in I}\mu_{t,a}^4,
  \qquad
  \sum_{t\in I}\mu_{t,a}^4
  \le\left(\sum_{t\in I}\mu_{t,a}^2\right)^2
  \le C r_n^2|I|^2.
\]
Together with $\sum_{t\in I}d_{t,a}^4\le C|I|$, this gives
\[
  \sum_{t\in I}\left(D_t^{(h)}-\bar D_a\right)^4
  \le C|I|+Cr_n^2|I|^2.
\]
Dividing by $Q_a^2\ge c|I|^2$ gives \eqref{eq:RD_l2weight}.

Write $\xi_I=\Theta_I^{1/2}Z_I$, where $Z_I\sim N(0,I_{|I|})$.  Then
\[
  \sum_{t\in I}w_{t,a}(\xi_t^2-1)
  =Z_I^\top B_aZ_I-\operatorname{tr}(B_a),
  \qquad
  B_a:=\Theta_I^{1/2}\operatorname{diag}(w_{t,a}:t\in I)\Theta_I^{1/2}.
\]
The trace equals one because
\[
  \operatorname{tr}(B_a)
  =
  \sum_{t\in I}w_{t,a}\Theta_n(t,t)
  =
  \sum_{t\in I}w_{t,a}
  =1.
\]
Using $2w_{t,a}w_{u,a}\le w_{t,a}^2+w_{u,a}^2$,
\eqref{eq:RD_rows}, and \eqref{eq:RD_l2weight},
\[
\begin{aligned}
  \|B_a\|_F^2
  &=
  \sum_{t,u\in I}w_{t,a}w_{u,a}\Theta_n(t,u)^2\\
  &\le
  \frac12\sum_{t,u\in I}
  \{w_{t,a}^2+w_{u,a}^2\}\Theta_n(t,u)^2\\
  &=
  \sum_{t\in I}w_{t,a}^2\sum_{u\in I}\Theta_n(t,u)^2\\
  &\le
  Ck_n\sum_{t\in I}w_{t,a}^2
  \le Ck_n\left(\frac1{|I|}+r_n^2\right).
\end{aligned}
\]
Furthermore, \eqref{eq:RD_rows} and \eqref{eq:RD_maxweight} give
\[
  \|B_a\|_{\mathrm{op}}
  \le
  Ck_n\left(\frac{\log n}{|I|}+r_n\right).
\]
The Gaussian quadratic-form inequality in
Proposition~\ref{prop:gaussian_quadratic} therefore yields, for $u\ge1$,
\[
  \Pbb^*\!\left(\left|\sum_{t\in I}w_{t,a}(\xi_t^2-1)\right|>C\left[\sqrt{k_nu\left(\frac1{|I|}+r_n^2\right)}+k_nu\left(\frac{\log n}{|I|}+r_n\right)\right]\right)\le2e^{-u}.
\]
By \eqref{eq:logpn_logn}, taking $u$ proportional to $\log n$ and applying
a union bound gives, after
absorbing $r_n\sqrt{k_n\log n}$ into $k_nr_n\log n$ because
$k_n=\Omega(1)$ and $\log n\to\infty$,
\begin{equation}\label{eq:RD_Rrate}
  \max_{a=(I,h)\in A}\left|\sum_{t\in I}w_{t,a}(\xi_t^2-1)\right|
  =O_{P^*}\!\left(
    \sqrt{\frac{k_n\log n}{L_{\min}}}
    +\frac{k_n(\log n)^2}{L_{\min}}
    +k_nr_n\log n
  \right)
\end{equation}
in probability over the data.  By \eqref{eq:logpn_logn}, \eqref{eq:RD_var},
and Proposition~\ref{prop:subgaussian_max} applied conditionally on the data,
\begin{equation}\label{eq:RD_Gmax}
  \max_a\left|\frac{S_a^*}{\sqrt{Q_a}}\right|=O_{P^*}(\sqrt{\log n})
\end{equation}
in probability over the data.  Equations \eqref{eq:RD_decomposition}--\eqref{eq:RD_Gmax}
and the implication
\[
  \frac{\log n}{L_{\min}}
  =
  o\!\left(\frac{k_n(\log n)^2}{L_{\min}}\right),
\]
which follows from $k_n=\Omega(1)$ and $\log n\to\infty$, prove
\eqref{eq:RD_rate}.

If $\gamma=0$, then
$T_{\max}^*=T_{0,\max}^*$ identically.  Suppose
$\gamma>0$.  Lemma~\ref{lem:power_denominator_perturbation}, applied with
$s_a=S_a^*$, $V_a=Q_a^*$, and $\widetilde V_a=Q_a$, gives
\begin{equation}\label{eq:RD_Tstar_perturb}
  |T_{\max}^*-T_{0,\max}^*|
  \le C\gamma\Delta_Q^*
  \max_a|T_{0,a}^*|.
\end{equation}
Moreover,
\[
  T_{0,a}^*
  =
  \left(\frac{Q_a}{|I|}\right)^{1/2-\gamma}\frac{S_a^*}{\sqrt{Q_a}}.
\]
Equations \eqref{eq:RD_Qlower}--\eqref{eq:RD_fourth}, the definition of
$r_n$, Assumption~\ref{ass:bounded_second_moment}, and
\eqref{eq:RD_growth} imply
with probability $1-o(1)$,
\[
  c
  \le
  \min_{a=(I,h)\in A}
  \left(\frac{Q_a}{|I|}\right)^{1/2-\gamma}
  \le
  \max_{a=(I,h)\in A}
  \left(\frac{Q_a}{|I|}\right)^{1/2-\gamma}
  \le
  C.
\]
Indeed,
\eqref{eq:RD_Qlower} gives $Q_a/|I|\ge c$.  For the upper bound,
\[
  Q_a
  =
  \sum_{t\in I}(d_{t,a}+\mu_{t,a})^2
  \le
  2\sum_{t\in I}d_{t,a}^2+2\nu_{a,\mu}.
\]
By Cauchy--Schwarz and \eqref{eq:RD_fourth},
$\sum_{t\in I}d_{t,a}^2\le |I|^{1/2}
(\sum_{t\in I}d_{t,a}^4)^{1/2}\le C|I|$.  Also
$\nu_{a,\mu}\le r_n\nu_a\le Cr_n|I|$, and \eqref{eq:RD_growth} with
$k_n=\Omega(1)$ and $\log n\to\infty$ gives $r_n=o(1)$.  Hence
$Q_a/|I|\le C$ uniformly.
Since $1/2-\gamma\in[0,1/2]$,
\[
  c\le
  \left(\frac{Q_a}{|I|}\right)^{1/2-\gamma}
  \le C
\]
uniformly.  Therefore, with probability $1-o(1)$ over data,
\begin{equation}\label{eq:RD_T0_var_bounds}
  c
  \le
  \min_{a\in A}\Var^*(T_{0,a}^*)
  \le
  \max_{a\in A}\Var^*(T_{0,a}^*)
  \le C,
\end{equation}
because
\[
  \Var^*(T_{0,a}^*)
  =
  \left(\frac{Q_a}{|I|}\right)^{1-2\gamma}
  \Var^*\!\left(\frac{S_a^*}{\sqrt{Q_a}}\right)
\]
and \eqref{eq:RD_var} holds.  On this event,
$(T_{0,a}^*)_{a\in A}$ is conditionally Gaussian with uniformly bounded
conditional $\psi_2$ norms.  Proposition~\ref{prop:subgaussian_max},
applied conditionally on the data with $m=p$, gives
\begin{equation}\label{eq:RD_T0max_rate}
  \max_a|T_{0,a}^*|
  =O_{P^*}(\sqrt{\log p})
  =O_{P^*}(\sqrt{\log n})
\end{equation}
in probability over the data, where the last bound uses
\eqref{eq:logpn_logn}.
The two conditions in \eqref{eq:RD_growth} imply
$\Delta_Q^*\log n=o_{P^*}(1)$ in probability over the data.  Combining
\eqref{eq:RD_growth}, \eqref{eq:RD_T0max_rate}, and
\eqref{eq:RD_Tstar_perturb} gives
\begin{equation}\label{eq:RD_Tstar_T0max_close}
  |T_{\max}^*-T_{0,\max}^*|
  =o_{P^*}((\log n)^{-1/2})
\end{equation}
in probability over the data.  Choose a deterministic sequence
$\eta_n=o((\log n)^{-1/2})$ such that
\begin{equation}\label{eq:RD_Tstar_close_prob}
  \Pbb^*(|T_{\max}^*-T_{0,\max}^*|>\eta_n)
  =
  o_p(1),
\end{equation}
which is possible by \eqref{eq:RD_Tstar_T0max_close}.
Proposition~\ref{prop:gaussian_anti_concentration},
applied conditionally on the data to $T_{0,\max}^*$ using
\eqref{eq:RD_T0_var_bounds}, gives
\begin{equation}\label{eq:RD_T0_anti_conc}
  \sup_{x\in\R}
  \Pbb^*(|T_{0,\max}^*-x|\le\eta_n)
  =
  o_p(1).
\end{equation}
Lemma~\ref{lem:centered_oracle_perturbation}, applied conditionally on the
data with $\mathsf P_n=\Pbb^*$, $X_n=T_{\max}^*$, and
$Y_n=T_{0,\max}^*$, together with \eqref{eq:RD_Tstar_close_prob} and
\eqref{eq:RD_T0_anti_conc}, proves \eqref{eq:RD_cdf}.
\end{proof}

\subsubsection{Gaussian approximation of the fixed-denominator bootstrap distribution}

Let
\[
  \omega_a
  :=\left(\frac{\nu_a}{|I|}\right)^{1/2-\gamma},
  \qquad
  \mathbb{T}_a^{\circ}
  :=\frac{S_a^\circ}{|I|^{1/2-\gamma}\nu_a^\gamma}
  =\omega_a\frac{S_a^\circ}{\sqrt{\nu_a}}.
\]
Define
\begin{equation}\label{eq:mathbbT_def}
  \begin{aligned}
  \mathbb{T}^{\circ}:=(\mathbb{T}_a^{\circ})_{a\in A},
  \qquad
  \mathbb{T}_{\max}^{\circ}:=\max_{a\in A}\mathbb{T}_a^{\circ},
  \qquad
  \bar{\mathbb{T}}_{\max}^{\circ}:=\max_{a\in A}|\mathbb{T}_a^{\circ}|.
  \end{aligned}
\end{equation}
Let
\[
  \Sigma_{\circ}:=\Cov(\mathbb{T}^{\circ}),
  \qquad
  G_{\circ}\sim N(0,\Sigma_{\circ}).
\]
This definition includes all serial cross-covariances in Regime~D.  Under
Regime~I independence it reduces to
\[
  (\Sigma_{\circ})_{ab}
  =
  \frac{\sum_{t\in I_a\cap I_b}
    \E[\d_t^{(h_a)}\d_t^{(h_b)}]}
  {|I_a|^{1/2-\gamma}|I_b|^{1/2-\gamma}
   \nu_a^\gamma \nu_b^\gamma}.
\]
Write
\[
  G_{\circ,\max}:=\max_{a\in A}G_{\circ,a},
  \qquad
  q_{\circ,n,u}
  :=\inf\{x\in\R:\Pbb(G_{\circ,\max}\le x)\ge u\}.
\]

The following lemma shows that the fixed-denominator bootstrap approximation
follows from consistency of its conditional covariance matrix.
For matrices indexed by $A$, write
\[
  \|\Sigma\|_{\max}:=\max_{a,b\in A}|\Sigma_{ab}|.
\]
\begin{lemma}\label{lem:fixed_bootstrap_covariance}
Let
\[
  \widehat\Sigma_{ab}
  :=
  \Cov^*(T_{0,a}^*,T_{0,b}^*),
  \qquad a,b\in A.
\]
Suppose that, for constants $0<c<C<\infty$,
\[
  c\le\min_{a\in A}(\Sigma_{\circ})_{aa}
  \le
  \max_{a\in A}(\Sigma_{\circ})_{aa}\le C,
\]
and, with probability $1-o(1)$ over data,
\[
  c\le\min_{a\in A}\widehat\Sigma_{aa},
  \qquad
  \max_{a\in A}\widehat\Sigma_{aa}\le C.
\]
If
\begin{equation}\label{eq:fixed_denominator_covariance_condition}
  \|\widehat\Sigma-\Sigma_{\circ}\|_{\max}
  =
  o_p\bigl((\log p)^{-2}\bigr),
\end{equation}
then
\[
  \sup_{x\in\R}
  \left|
  \Pbb^*(T_{0,\max}^{*}\le x)
  -
  \Pbb(G_{\circ,\max}\le x)
  \right|
  =
  o_p(1).
\]
\end{lemma}

\begin{proof}
Let
\[
  \Delta_\Sigma:=\|\widehat\Sigma-\Sigma_{\circ}\|_{\max}.
\]
The covariance consistency assumption gives a deterministic sequence
$\varepsilon_n\downarrow0$ such that
\[
  \varepsilon_n(\log p)^2\to0,
  \qquad
  \Pbb(\Delta_\Sigma>\varepsilon_n)\to0.
\]
On the event where $\Delta_\Sigma\le\varepsilon_n$ and the displayed
variance bounds for $\widehat\Sigma$ hold, the conditional law of
$(T_{0,a}^*)_{a\in A}$ is centered Gaussian with covariance
$\widehat\Sigma$, and $G_{\circ}$ is centered Gaussian with covariance
$\Sigma_{\circ}$.  Proposition~\ref{prop:gaussian_comparison} yields
\[
  \sup_{x\in\R}
  \left|
  \Pbb^*(T_{0,\max}^{*}\le x)
  -
  \Pbb(G_{\circ,\max}\le x)
  \right|
  \le
  C\varepsilon_n^{1/3}
  \{1\vee\log(p/\varepsilon_n)\}^{2/3}.
\]
The right-hand side is $o(1)$ because
$\varepsilon_n(\log p)^2\to0$ and
$\varepsilon_n\{\log(1/\varepsilon_n)\}^2\to0$.  Since the event has
probability tending to one, the desired conclusion follows.
\end{proof}

The next lemma transfers oracle Gaussian approximation and bootstrap
approximation into asymptotic size control.
\begin{lemma}\label{lem:size_transfer}
Fix $\alpha\in(0,1)$.  Under \eqref{eq:null_standalone}, suppose that the
conditions of Lemmas~\ref{lem:random_denominator} and
\ref{lem:fixed_bootstrap_covariance} hold.  Suppose also that
\[
  \sup_{x\in\R}
  \left|
  \Pbb(T_{\max}^{\circ}\le x)
  -
  \Pbb(G_{\circ,\max}\le x)
  \right|
  =
  o(1).
\]
Then
\[
  \limsup_{n\to\infty}
  \Pbb\bigl(T_{\max}>c_{n,1-\alpha}^{*}\bigr)
  \le \alpha.
\]
\end{lemma}

\begin{proof}
Lemmas~\ref{lem:random_denominator} and
\ref{lem:fixed_bootstrap_covariance} imply
\[
  Z_n:=
  \sup_{x\in\R}
  \left|
  \Pbb^*(T_{\max}^{*}\le x)
  -
  \Pbb(G_{\circ,\max}\le x)
  \right|
  =
  o_p(1).
\]
Lemma~\ref{lem:deterministic_threshold} gives a deterministic sequence
$\rho_n\downarrow0$ such that $\Pbb(Z_n>\rho_n)\to0$.  Let
\[
  F_n^*(x):=\Pbb^*\bigl(T_{\max}^{*}\le x\bigr),
  \qquad
  E_n:=\left\{
  \sup_{x\in\R}\bigl|F_n^*(x)-\Pbb(G_{\circ,\max}\le x)\bigr|\le\rho_n
  \right\}.
\]
Then $\Pbb(E_n^c)\to0$.  On $E_n$,
$F_n^*(c_{n,1-\alpha}^{*})\ge1-\alpha$, and hence
\[
  \Pbb(G_{\circ,\max}\le c_{n,1-\alpha}^{*})
  \ge1-\alpha-\rho_n.
\]
For all sufficiently large $n$, $\rho_n<1-\alpha$, so
\[
  c_{n,1-\alpha}^{*}
  \ge q_{\circ,n,1-\alpha-\rho_n}
  \qquad\text{on }E_n.
\]
Lemma~\ref{lem:null_pathwise_dominance} gives
$T_{\max}\le T_{\max}^{\circ}$ under \eqref{eq:null_standalone}.  Therefore
\[
  \Pbb\bigl(T_{\max}>c_{n,1-\alpha}^{*}\bigr)
  \le
  \Pbb\bigl(T_{\max}^{\circ}>q_{\circ,n,1-\alpha-\rho_n}\bigr)
  +\Pbb(E_n^c).
\]
Let
\[
  \Delta_{\mathrm{cdf},n}:=
  \sup_{x\in\R}
  \left|
  \Pbb(T_{\max}^{\circ}\le x)
  -
  \Pbb(G_{\circ,\max}\le x)
  \right|.
\]
By assumption, $\Delta_{\mathrm{cdf},n}\to0$, so
\[
  \Pbb\bigl(T_{\max}^{\circ}>q_{\circ,n,1-\alpha-\rho_n}\bigr)
  \le
  \Pbb\bigl(G_{\circ,\max}>q_{\circ,n,1-\alpha-\rho_n}\bigr)
  +\Delta_{\mathrm{cdf},n}
  \le
  \alpha+\rho_n+\Delta_{\mathrm{cdf},n}.
\]
Taking the limit superior proves the claim.
\end{proof}

\subsubsection{Roadmap of the proofs}\label{subsub:roadmap_proofs}

The final goal in each regime is
\[
  \limsup_{n\to\infty}
  \Pbb\bigl(T_{\max}>c_{n,1-\alpha}^{*}\bigr)
  \le\alpha
  \qquad\forall\alpha\in(0,1).
\]
By Lemma~\ref{lem:size_transfer}, it is enough to show
\[
  \sup_{x\in\R}
  \left|
  \Pbb(T_{\max}^{\circ}\le x)
  -
  \Pbb(G_{\circ,\max}\le x)
  \right|
  \to0
\]
and to verify the conditions of Lemmas~\ref{lem:random_denominator} and
\ref{lem:fixed_bootstrap_covariance}.  The display is reduced by
Lemma~\ref{lem:centered_oracle_perturbation}
to proving
\begin{equation}\label{eq:roadmap_centered_gaussian}
  \sup_{x\in\R}
  \left|
  \Pbb(\mathbb T_{\max}^{\circ}\le x)
  -
  \Pbb(G_{\circ,\max}\le x)
  \right|
  \to0.
\end{equation}
The bootstrap part is reduced by Lemmas~\ref{lem:random_denominator} and
\ref{lem:fixed_bootstrap_covariance} to proving the covariance consistency
condition
\begin{equation}\label{eq:roadmap_bootstrap_covariance}
  \|\widehat\Sigma-\Sigma_{\circ}\|_{\max}
  =
  o_p\bigl((\log p)^{-2}\bigr).
\end{equation}
All remaining work in Regimes I and D is therefore to verify the centered
Gaussian approximation \eqref{eq:roadmap_centered_gaussian}, the covariance
consistency condition \eqref{eq:roadmap_bootstrap_covariance}, and the
denominator and anti-concentration conditions required by Lemmas
\ref{lem:centered_oracle_perturbation}, \ref{lem:random_denominator}, and
\ref{lem:fixed_bootstrap_covariance}.

\subsection{Regime I Validity}

\subsubsection{Additional assumptions for Regime I}

We consider the following additional assumptions for Regime I.

\begin{appassumption}\label{ass:independence}
The vectors $D_1,\dots,D_n$ are independent across $t$.
\end{appassumption}

\begin{appassumption}\label{ass:intervalwise_mean}
The interval-wise mean is asymptotically flat:
\[
r_n=o\bigl((\log n)^{-2}\bigr).
\]
\end{appassumption}

\begin{appassumption}\label{ass:Lmin}
The minimal interval length grows at least polynomially in $n$:
\[
  L_{\min}=\omega\bigl((\log n)^7\bigr).
\]
\end{appassumption}

\begin{remark}
Assumption \ref{ass:intervalwise_mean} constrains the \emph{shape} of the mean on each scanned interval, not its level. If $\mu_t^{(h)}\equiv -c_h$ is constant in $t$, then $r_n=0$ for every $n$ regardless of the magnitude of $c_h$.
\end{remark}

\subsubsection{Main theorem}

\begin{theorem}\label{I:thm:main}
Fix the exponent $\gamma\in[0,1/2]$.  Assume Assumptions
\ref{ass:subgaussian}, \ref{ass:bounded_second_moment},
\ref{ass:length_ratio}, and
\ref{ass:independence}--\ref{ass:Lmin}.
Let $c_{n,1-\alpha}^{*}$ be the conditional $(1-\alpha)$-quantile
in \eqref{eq:implemented_bootstrap}.  Then the Gaussian multiplier
bootstrap with $\Theta_n=I_n$ is asymptotically conservative: under
\eqref{eq:null_standalone},
\[
\limsup_{n\to\infty}
\Pbb\bigl(T_{\max}>c_{n,1-\alpha}^{*}\bigr)
\le \alpha
\qquad\forall \alpha\in(0,1).
\]
\end{theorem}

Following the roadmap in Section~\ref{subsub:roadmap_proofs}, it remains to prove
\eqref{eq:roadmap_centered_gaussian} and
\eqref{eq:roadmap_bootstrap_covariance} and to verify the conditions in
Lemmas~\ref{lem:centered_oracle_perturbation}--\ref{lem:fixed_bootstrap_covariance}.
In what follows, Section~\ref{subsub:I_preliminary_bounds} gives preliminary
results, Sections~\ref{subsub:I_centered_gaussian} and
\ref{subsub:I_bootstrap_covariance} prove
\eqref{eq:roadmap_centered_gaussian} and
\eqref{eq:roadmap_bootstrap_covariance}, respectively, and
Section~\ref{subsub:I_verify_reduction}
completes the proof of Theorem~\ref{I:thm:main} by verifying the conditions
in Lemmas~\ref{lem:centered_oracle_perturbation}--\ref{lem:fixed_bootstrap_covariance}.
For Sections~\ref{subsub:I_centered_gaussian}--\ref{subsub:I_verify_reduction},
we assume Assumptions
\ref{ass:subgaussian}, \ref{ass:bounded_second_moment},
\ref{ass:length_ratio}, and
\ref{ass:independence}--\ref{ass:Lmin} all hold.

\subsubsection{Preliminary bounds}\label{subsub:I_preliminary_bounds}

The three lemmas in this subsection supply the Regime I inputs needed for
the roadmap.  Lemma~\ref{I:lem:Uw} controls
$\bar{\mathbb{T}}_{\max}^{\circ}$ and the intersection sums
$S_b^\circ/\sqrt{\nu_b}$; these bounds control sample-centering terms in
Lemma~\ref{I:lem:Delta}, in the covariance expansion for
\eqref{eq:roadmap_bootstrap_covariance}, and in the verification of
Lemma~\ref{lem:centered_oracle_perturbation}.  Lemma~\ref{I:lem:Delta}
gives $\Delta_Q=o_p(1)$; this is used to compare $T_{\max}^{\circ}$ with
$\mathbb{T}_{\max}^{\circ}$, to lower bound $Q_a$ in
Lemma~\ref{lem:random_denominator}, and to control the denominator
perturbation in \eqref{eq:roadmap_bootstrap_covariance}.
Lemma~\ref{I:lem:fourth} verifies the maximum and fourth-moment condition
\eqref{eq:RD_fourth} required by Lemma~\ref{lem:random_denominator}.

\begin{lemma}\label{I:lem:Uw}
Under Assumptions \ref{ass:subgaussian}, \ref{ass:bounded_second_moment}, and \ref{ass:independence},
\[
  \bar{\mathbb{T}}_{\max}^{\circ}
  =\Op\bigl(\sqrt{\log n}\bigr).
\]
Moreover, if
\[
\mathcal K:=\{I\cap J\st I,J\in\mathcal I,\ I\cap J\neq \varnothing\},
\]
then
\[
\max_{\substack{b=(K,h):\\ K\in\mathcal K,\ h\in[H]}}
\frac{|S_b^{\circ}|}{\sqrt{\nu_b}}
=
\Op\bigl(\sqrt{\log n}\bigr).
\]
\end{lemma}

\begin{proof}
Fix $a=(I,h)$. By Assumptions \ref{ass:subgaussian}, \ref{ass:independence}, and Proposition \ref{prop:sum_subgaussian},
\[
\Bigl\|S_a^{\circ}\Bigr\|_{\psi_2}
\le C B_0\sqrt{|I|}
\]
for a constant $C<\infty$ depending only on universal constants. Since
\[
\nu_a\ge \sigma_-^2 |I|,
\]
we obtain the uniform bound
\[
\left\|\frac{S_a^{\circ}}{\sqrt{\nu_a}}\right\|_{\psi_2}
\le \frac{CB_0}{\sigma_-}.
\]
Hence, for some constants $c_1,c_2>0$,
\[
\Pbb\!\left(\left|\frac{S_a^{\circ}}{\sqrt{\nu_a}}\right|>u\right)
\le c_1 e^{-c_2u^2}
\qquad(u\ge 0)
\]
uniformly in $a\in A$.  Since
$\mathbb{T}_a^{\circ}=\omega_aS_a^\circ/\sqrt{\nu_a}$ and
Assumption~\ref{ass:bounded_second_moment} bounds $\omega_a$ uniformly
above, there exist constants $c_1',c_2'>0$ such that
\[
\Pbb\bigl(|\mathbb{T}_a^{\circ}|>u\bigr)
\le c_1' e^{-c_2'u^2}
\qquad(u\ge 0)
\]
uniformly in $a\in A$.  A union bound over the $p$ coordinates and the
choice $u=M\sqrt{\log n}$, with $M$ large enough, prove the first claim.

For the second claim, note that $|\mathcal K|\le |\mathcal I|^2$, hence
\[
H|\mathcal K|\le H|\mathcal I|^2=\frac{p^2}{H}\le p^2.
\]
For each $b=(K,h)$ with $K\in\mathcal K$ and $h\in[H]$, Assumptions
\ref{ass:subgaussian} and \ref{ass:independence}, Proposition
\ref{prop:sum_subgaussian}, and $\nu_b\ge\sigma_-^2|K|$ give
$\|S_b^{\circ}/\sqrt{\nu_b}\|_{\psi_2}\le C$.  Applying the union
bound over $H|\mathcal K|\le p^2$ variables gives the same
$\sqrt{\log n}$ order.
\end{proof}

The next lemma shows that the observed denominator uniformly approximates
the oracle variance scale in Regime~I.
\begin{lemma}\label{I:lem:Delta}
Under Assumptions \ref{ass:subgaussian}, \ref{ass:bounded_second_moment}, and \ref{ass:independence}--\ref{ass:Lmin},
\[
\Delta_Q:=\max_{a\in A}\left|\frac{Q_a}{\nu_a}-1\right|
=
\Op\!\left(\sqrt{\frac{\log n}{L_{\min}}}+r_n\right).
\]
In particular, $\Delta_Q=o_p(1)$ and $\Pbb(\Delta_Q>1/2)\to 0$.
\end{lemma}

\begin{proof}
By Lemma~\ref{lem:denominator_decomposition}, it suffices to control the
three terms in the decomposition of $(Q_a-\nu_a)/\nu_a$ uniformly
over $a\in A$.

\medskip
\noindent\emph{Step 1: control of $\frac{q_a^\circ}{\nu_a}$.}
By \eqref{eq:denom_Adef},
\[
  q_a^\circ
  =
  \sum_{t\in I}\Bigl(\d_t^{(h)2}-\E[\d_t^{(h)2}]\Bigr)
  -
  |I|\bar\d_a^2.
\]
By Assumption \ref{ass:subgaussian} and Proposition
\ref{prop:product_subgaussian},
$\d_t^{(h)2}-\E[\d_t^{(h)2}]$ is uniformly sub-exponential.
Hence Bernstein's inequality yields, for every fixed $a=(I,h)$ and every
$u\ge 1$,
\[
\Pbb\!\left(
  \left|\sum_{t\in I}
    \Bigl(\d_t^{(h)2}-\E[\d_t^{(h)2}]\Bigr)\right|
  >C\bigl(\sqrt{|I|u}+u\bigr)\right)
\le 2e^{-u}
\]
for some constant $C<\infty$.  Since $p=|A|\le n^2$, the choice
$u=x+2\log n$ gives $pe^{-u}\le e^{-x}$.  Taking this $u$ and applying a
union bound over $a\in A$ gives, for every $x\ge1$, with probability at
least $1-2e^{-x}$,
\[
  \max_{a=(I,h)\in A}
  \frac{
    \left|\sum_{t\in I}
      \bigl(\d_t^{(h)2}-\E[\d_t^{(h)2}]\bigr)\right|
  }{\nu_a}
  \le C\max_{a\in A}
  \left\{
    \sqrt{\frac{x+2\log n}{|I_a|}}
    +\frac{x+2\log n}{|I_a|}
  \right\}.
\]
Since $|I_a|\ge L_{\min}$, this proves
\[
  \max_{a=(I,h)\in A}
  \frac{
    \left|\sum_{t\in I}
      \bigl(\d_t^{(h)2}-\E[\d_t^{(h)2}]\bigr)\right|
  }{\nu_a}
  =\Op\!\left(
    \sqrt{\frac{\log n}{L_{\min}}}
    +\frac{\log n}{L_{\min}}
  \right).
\]
Assumption~\ref{ass:bounded_second_moment} gives
\[
  \frac{\nu_a}{|I|}\ge \sigma_-^2,
  \qquad
  \omega_a
  =
  \left(\frac{\nu_a}{|I|}\right)^{1/2-\gamma}
  \ge(\sigma_-^2)^{1/2-\gamma},
  \qquad
  \omega_a^{-1}\le C.
\]
Since $S_a^\circ/\sqrt{\nu_a}=\omega_a^{-1}\mathbb T_a^\circ$,
Lemma~\ref{I:lem:Uw} gives
\[
\max_{a\in A}\frac{(S_a^{\circ})^2}{|I_a|\nu_a}
\le
\frac{C(\bar{\mathbb{T}}_{\max}^{\circ})^2}{L_{\min}}
=
\Op\!\left(\frac{\log n}{L_{\min}}\right).
\]
Combining the last two displays with the decomposition of $q_a^\circ$ at
the start of this step gives
\[
  \max_{a\in A}\frac{|q_a^\circ|}{\nu_a}
  =
  \Op\!\left(
    \sqrt{\frac{\log n}{L_{\min}}}
    +\frac{\log n}{L_{\min}}
  \right).
\]

\medskip
\noindent\emph{Step 2: control of $\frac{C_a}{\nu_a}$.}
Since $\sum_{t\in I}\mu_{t,a}=0$, the random variable $C_a$ in
\eqref{eq:denom_Adef} satisfies
$C_a=\sum_{t\in I}\mu_{t,a}\d_t^{(h)}$.  Conditionally on the deterministic
coefficients $\mu_{t,a}$, it is a weighted sum of the independent centered
sub-Gaussian variables $\{\d_t^{(h)}:t\in I\}$. By Proposition
\ref{prop:sum_subgaussian},
\[
\|C_a\|_{\psi_2}
\le C B_0\left(\sum_{t\in I}\mu_{t,a}^2\right)^{1/2}
=
C B_0 \nu_{a,\mu}^{1/2}.
\]
Therefore, for every fixed $a$ and every $u\ge 0$,
\[
\Pbb\bigl(|C_a|>C\nu_{a,\mu}^{1/2}\sqrt{u}\bigr)\le 2e^{-u}.
\]
A union bound over $a\in A$ yields directly
\[
  \max_{a\in A}\frac{|C_a|}{\nu_a}
  =\Op\!\left(
    \max_{a\in A}\frac{\sqrt{\nu_{a,\mu}\log n}}{\nu_a}
  \right).
\]
Since $\nu_{a,\mu}\le r_n \nu_a$ and $\nu_a\ge \sigma_-^2L_{\min}$,
\[
\max_{a\in A}\frac{|C_a|}{\nu_a}
=
\Op\!\left(\sqrt{\frac{r_n\log n}{L_{\min}}}\right).
\]

\medskip
\noindent\emph{Step 3: control of $\frac{\nu_{a,\mu}}{\nu_a}$.}
By definition \eqref{eq:rn} of $r_n$,
\[
\max_{a\in A}\frac{\nu_{a,\mu}}{\nu_a}=r_n.
\]

Assumptions~\ref{ass:intervalwise_mean} and \ref{ass:Lmin} imply
\[
  \frac{\log n}{L_{\min}}
  =
  o\!\left(\sqrt{\frac{\log n}{L_{\min}}}\right),
  \qquad
  \sqrt{\frac{r_n\log n}{L_{\min}}}
  =
  O\!\left(\sqrt{\frac{\log n}{L_{\min}}}\right).
\]
Collecting the three bounds proves
\[
\Delta_Q
=
\Op\!\left(
\sqrt{\frac{\log n}{L_{\min}}}
+
r_n
\right),
\]
as claimed.
\end{proof}

The next lemma gives uniform maximum and fourth-moment bounds for the
centered variables $\d_t^{(h)}=D_t^{(h)}-\mu_t^{(h)}$ in Regime~I.
\begin{lemma}
\label{I:lem:fourth}
Under Assumptions~\ref{ass:subgaussian}, \ref{ass:independence}, and
\ref{ass:Lmin},
\[
  \max_{a=(I,h)\in A}\max_{t\in I}(\d_t^{(h)})^2
  =O_p(\log n),
  \qquad
  \max_{a=(I,h)\in A}\frac1{|I|}\sum_{t\in I}(\d_t^{(h)})^4
  =O_p(1).
\]
\end{lemma}

\begin{proof}
Assumption~\ref{ass:subgaussian}, \eqref{eq:logpn_logn}, the union bound over
$nH\le np$, and Proposition~\ref{prop:subgaussian_max} give
\[
  \max_{t\in[n],\,h\in[H]}|\d_t^{(h)}|
  =O_p(\sqrt{\log n}).
\]
This proves the first conclusion.

Uniform sub-Gaussianity implies
\[
  \sup_{t,h}\E[(\d_t^{(h)})^4]<\infty,
  \qquad
  \sup_{t,h}\Pbb\bigl((\d_t^{(h)})^4>u\bigr)
  \le2e^{-c\sqrt u}.
\]
Let $X_t^{(h)}=(\d_t^{(h)})^4-\E[(\d_t^{(h)})^4]$.  For the variance term
$V$ in \eqref{eq:mpr-V}, where
$\varphi_M(x)=(x\wedge M)\vee(-M)$,
\[
\begin{aligned}
  \Cov\!\left(\varphi_M(X_s^{(h)}),\varphi_M(X_t^{(h)})\right)
  &=0
  \quad(s\ne t),\\
  \sup_{M,t,h}\Var\!\left(\varphi_M(X_t^{(h)})\right)
  &\le \sup_{t,h}\E[(X_t^{(h)})^2]
  \le C\sup_{t,h}\E[(\d_t^{(h)})^8]<\infty.
\end{aligned}
\]
Hence $V=O(1)$ uniformly.  Proposition~\ref{prop:alpha_mixing_concentration}
may therefore be applied with $\gamma_1=1$ and $\gamma_2=1/2$.  Taking
$y=Cm$ in \eqref{eq:mpr-strong-mixing} gives, for an interval of length $m$,
\begin{equation}\label{eq:I_fourth_tail}
  \Pbb\left(\left|\sum_{t\in I}X_t^{(h)}\right|>Cm\right)
  \le C(m+1)e^{-cm^{1/3}}+Ce^{-cm}.
\end{equation}
Assumption~\ref{ass:Lmin} gives $(\log n)^7/L_{\min}\to0$, which implies
$\log n=o(L_{\min}^{1/3})$.  A union bound over $A$ in
\eqref{eq:I_fourth_tail} proves the result.
\end{proof}

\subsubsection{Proof of the centered Gaussian approximation}\label{subsub:I_centered_gaussian}

For $t\in[n]$ and $a=(I,h)\in A$, define
\begin{equation}\label{eq:Zta}
Z_{t,a}
:=\frac{\sqrt n\,\1\{t\in I\}\d_t^{(h)}}
{|I|^{1/2-\gamma}\nu_a^\gamma}.
\end{equation}
Let $Z_t:=(Z_{t,a})_{a\in A}$.  By the definition of
$\mathbb{T}^{\circ}$ in \eqref{eq:mathbbT_def},
\[
\mathbb{T}^{\circ}=\frac{1}{\sqrt n}\sum_{t=1}^n Z_t.
\]
Set
\[
  B_n^I
  :=1\vee B_0(1\vee\sigma_-^{-1})
  \sqrt{\frac n{L_{\min}}}.
\]
Assumption~\ref{ass:bounded_second_moment} gives
\[
  \sigma_-^2|I|
  \le
  \nu_a
  =
  \sum_{u\in I}\E[\d_u^{(h)2}]
  \le
  \sigma_+^2|I|.
\]
In particular,
\[
  |I|^{1/2-\gamma}\nu_a^\gamma
  \ge
  |I|^{1/2-\gamma}(\sigma_-^2|I|)^\gamma
  =
  \sigma_-^{2\gamma}|I|^{1/2},
  \qquad
  \sigma_-^{-2\gamma}\le 1\vee\sigma_-^{-1},
\]
because $0\le\gamma\le1/2$.  Therefore
\[
\|Z_{t,a}\|_{\psi_2}
\le
\frac{B_0\sqrt n}{|I|^{1/2-\gamma}\nu_a^\gamma}
\le
B_0(1\vee\sigma_-^{-1})\sqrt{\frac n{|I|}}
\le B_n^I
\qquad\forall t,a.
\]
The standard embedding of Orlicz norms gives
$\|Z_{t,a}\|_{\psi_1}\le
C\|Z_{t,a}\|_{\psi_2}\le CB_n^I$.
Thus the condition
$\max_{t,a}\|Z_{t,a}\|_{\psi_1}\le B$ in
Proposition~\ref{prop:max_gaussian_CLT} holds with $B=CB_n^I$, and
\[
  \frac{(CB_n^I)^2\log^7(pn)}{n}
  =
  C^2\frac{(B_n^I)^2\log^7(pn)}{n}.
\]
Finally,
\[
\Var\!\left(\frac{1}{\sqrt n}\sum_{t=1}^n
Z_{t,a}\right)
=\Var(\mathbb{T}_a^{\circ})
=\left(\frac{\nu_a}{|I|}\right)^{1-2\gamma}
\qquad\forall a\in A.
\]
The same bound from Assumption~\ref{ass:bounded_second_moment} gives
\[
  (\sigma_-^2)^{1-2\gamma}
  \le
  \left(\frac{\nu_a}{|I|}\right)^{1-2\gamma}
  \le
  (\sigma_+^2)^{1-2\gamma}
  \qquad\forall a\in A,
\]
so the coordinate variances are uniformly bounded above and away from zero.
By Assumption \ref{ass:Lmin},
\[
\frac{(B_n^I)^2(\log(pn))^7}{n}
\lesssim
\frac{(\log n)^7}{L_{\min}}
\to 0,
\]
where $\log(pn)\asymp\log n$ by \eqref{eq:logpn_logn} and
the final convergence is Assumption~\ref{ass:Lmin}.  The conditions of
Proposition~\ref{prop:max_gaussian_CLT} are therefore satisfied. Thus,
\begin{equation}\label{eq:I_oracle_U_cdf}
\sup_{x\in\R}
\Bigl|
\Pbb\bigl(\mathbb{T}_{\max}^{\circ}\le x\bigr)
-
\Pbb\bigl(G_{\circ,\max}\le x\bigr)
\Bigr|
\to 0.
\end{equation}
This is \eqref{eq:roadmap_centered_gaussian}.

\subsubsection{Proof of the fixed-denominator bootstrap covariance}\label{subsub:I_bootstrap_covariance}

For $a=(I,h)$ and $b=(J,g)$, set $K:=I\cap J$.
Recall from \eqref{eq:dta_def} and \eqref{eq:muta_def} that
$D_t^{(h)}-\bar D_a=d_{t,a}+\mu_{t,a}$ for $t\in I$.
By direct calculation,
\[
\widehat\Sigma_{ab}
=
\frac{\widehat\Omega_{ab}}
{|I_a|^{1/2-\gamma}|I_b|^{1/2-\gamma}Q_a^\gamma
Q_b^\gamma}.
\]
Here
\[
\begin{aligned}
\widehat\Omega_{ab}
&:=\sum_{t\in K}
\left(D_t^{(h)}-\bar D_a\right)
\left(D_t^{(g)}-\bar D_b\right).
\end{aligned}
\]
The target covariance entry is
\[
(\Sigma_{\circ})_{ab}
=
\frac{\Omega_{ab}}
{|I_a|^{1/2-\gamma}|I_b|^{1/2-\gamma}\nu_a^\gamma \nu_b^\gamma},
\qquad
\Omega_{ab}:=\sum_{t\in K}\E\bigl[\d_t^{(h)}\d_t^{(g)}\bigr].
\]
Hence
\begin{align}
\widehat\Sigma_{ab}-(\Sigma_{\circ})_{ab}
&=
\frac{\widehat\Omega_{ab}-\Omega_{ab}}
{|I_a|^{1/2-\gamma}|I_b|^{1/2-\gamma}Q_a^\gamma
Q_b^\gamma}
\label{I:eq:covdec1}
\\
&\quad+
(\Sigma_{\circ})_{ab}
\left\{
\left(\frac{Q_a}{\nu_a}\right)^{-\gamma}
\left(\frac{Q_b}{\nu_b}\right)^{-\gamma}-1
\right\}.
\label{I:eq:covdec2}
\end{align}
We bound the two terms uniformly.

\medskip
\noindent\emph{Step 1: bounding \eqref{I:eq:covdec2}.}
On $\{\Delta_Q\le1/2\}$, for every $a,b\in A$,
\[
  \frac{Q_a}{\nu_a},\frac{Q_b}{\nu_b}\in[1/2,3/2],
  \qquad
  \left|
  \left(\frac{Q_a}{\nu_a}\right)^{-\gamma}
  \left(\frac{Q_b}{\nu_b}\right)^{-\gamma}-1
  \right|
  \le C\left(
  \left|\frac{Q_a}{\nu_a}-1\right|
  +\left|\frac{Q_b}{\nu_b}-1\right|
  \right)
  \le C\Delta_Q .
\]
Cauchy--Schwarz and Assumption~\ref{ass:bounded_second_moment} give
\[
\begin{aligned}
|\Omega_{ab}|
&\le
\left(\sum_{t\in K}\E[\d_t^{(h)2}]\right)^{1/2}
\left(\sum_{t\in K}\E[\d_t^{(g)2}]\right)^{1/2}
\le \sqrt{\nu_a\nu_b},\\
|(\Sigma_{\circ})_{ab}|
&\le
\left(\frac{\nu_a}{|I|}\right)^{1/2-\gamma}
\left(\frac{\nu_b}{|J|}\right)^{1/2-\gamma}
\le C .
\end{aligned}
\]
Therefore
\begin{equation}\label{eq:I_cov_denominator_perturb}
\max_{a,b}
\left|
(\Sigma_{\circ})_{ab}
\left\{
\left(\frac{Q_a}{\nu_a}\right)^{-\gamma}
\left(\frac{Q_b}{\nu_b}\right)^{-\gamma}-1
\right\}
\right|
\le C\Delta_Q
=
\Op\!\left(\sqrt{\frac{\log n}{L_{\min}}}+r_n\right),
\end{equation}
by Lemma \ref{I:lem:Delta}.

\medskip
\noindent\emph{Step 2: bounding \eqref{I:eq:covdec1}.}
Expand
\begin{equation}\label{eq:Omegahat_ab}
\widehat\Omega_{ab}
=
\sum_{t\in K} d_{t,a}d_{t,b}
+
\sum_{t\in K} d_{t,a}\mu_{t,b}
+
\sum_{t\in K} d_{t,b}\mu_{t,a}
+
\sum_{t\in K} \mu_{t,a}\mu_{t,b}.
\end{equation}
We analyze the four pieces.

\medskip
\noindent\underline{Bounding the first term of \eqref{eq:Omegahat_ab}.}
A direct expansion gives
\begin{align*}
\sum_{t\in K} d_{t,a}d_{t,b}-\Omega_{ab}
&=
\hat{\omega}_{ab}^{\circ}
-
\bar\d_b S_{(K,h)}^{\circ}
-
\bar\d_a S_{(K,g)}^{\circ}
+
|K|\bar\d_a\bar\d_b,
\end{align*}
where
\[
\hat{\omega}_{ab}^{\circ}:=
\sum_{t\in K}\Bigl(\d_t^{(h)}\d_t^{(g)}-
\E[\d_t^{(h)}\d_t^{(g)}]\Bigr).
\]
By Assumption \ref{ass:subgaussian}, $\d_t^{(h)}$ and $\d_t^{(g)}$ are both
sub-Gaussian.  By Proposition \ref{prop:product_subgaussian},
$\d_t^{(h)}\d_t^{(g)}$ is sub-exponential; Bernstein's inequality and a
union bound over the $p^2$ pairs $(a,b)$ yield
\[
\max_{a,b}
\frac{|\hat{\omega}_{ab}^{\circ}|}{\sqrt{\nu_a \nu_b}}
=
\Op\!\left(
\sqrt{\frac{\log n}{L_{\min}}}
+
\frac{\log n}{L_{\min}}
\right).
\]
Also, Assumption~\ref{ass:bounded_second_moment} gives
$\omega_b^{-1}\le C$, and
$S_b^\circ/\sqrt{\nu_b}=\omega_b^{-1}\mathbb T_b^\circ$.  Hence
Lemma~\ref{I:lem:Uw} gives
\[
\max_b\frac{|\bar\d_b|}{\sqrt{\nu_b}}
=
\max_b\frac{|S_b^{\circ}|}{|I_b|\sqrt{\nu_b}}
\le
\frac{C\bar{\mathbb{T}}_{\max}^{\circ}}{L_{\min}}
=
\Op\!\left(\frac{\sqrt{\log n}}{L_{\min}}\right),
\]
and
\[
\max_{a,b}\frac{|S_{(K,h)}^{\circ}|}{\sqrt{\nu_a}}
=
\Op\bigl(\sqrt{\log n}\bigr),
\quad
\max_{a,b}\frac{|S_{(K,g)}^{\circ}|}{\sqrt{\nu_b}}
=
\Op\bigl(\sqrt{\log n}\bigr),
\]
so the two sample-mean terms are both
\[
\Op\!\left(\frac{\log n}{L_{\min}}\right).
\]
Finally,
\[
\max_{a,b}
\frac{|K|\,|\bar\d_a\bar\d_b|}{\sqrt{\nu_a \nu_b}}
\le
\max_{a,b}
\frac{|K|}{|I_a||I_b|}
\frac{|S_a^{\circ}|}{\sqrt{\nu_a}}
\frac{|S_b^{\circ}|}{\sqrt{\nu_b}}
\le
\frac{C(\bar{\mathbb{T}}_{\max}^{\circ})^2}{L_{\min}}
=
\Op\!\left(\frac{\log n}{L_{\min}}\right),
\]
where the last line follows from Lemma \ref{I:lem:Uw}.
Therefore
\begin{equation}
\max_{a,b}
\frac{\left|\sum_{t\in K} d_{t,a}d_{t,b}-\Omega_{ab}\right|}{\sqrt{\nu_a \nu_b}}
=
\Op\!\left(
\sqrt{\frac{\log n}{L_{\min}}}
+
\frac{\log n}{L_{\min}}
\right).
\label{I:eq:centerednum}
\end{equation}

\medskip
\noindent\underline{Bounding the second and third terms of \eqref{eq:Omegahat_ab}.}
Consider $\sum_{t\in K}d_{t,a}\mu_{t,b}$; the other term is symmetric.
Then
\begin{align*}
\sum_{t\in K}d_{t,a}\mu_{t,b}
&=
\sum_{t\in K}\d_t^{(h)}\mu_{t,b}
-
\bar\d_a\sum_{t\in K}\mu_{t,b}.
\end{align*}
The first term is a weighted sum of the independent centered sub-Gaussian
variables $\{\d_t^{(h)}:t\in K\}$ with coefficient vector
$(\mu_{t,b})_{t\in K}$, hence is
sub-Gaussian with scale at most
$CB_0(\sum_{t\in K}\mu_{t,b}^2)^{1/2}\le CB_0\nu_{b,\mu}^{1/2}$.
For fixed $(a,b)$ and $u\ge1$, \eqref{eq:rn} and
$\nu_a\ge\sigma_-^2L_{\min}$ give
\[
\Pbb\!\left(
\frac{\left|\sum_{t\in K}\d_t^{(h)}\mu_{t,b}\right|}
{\sqrt{\nu_a\nu_b}}
>
C\sqrt{\frac{r_nu}{L_{\min}}}
\right)
\le 2e^{-u}.
\]
Since there are at most $p^2\le n^4$ pairs $(a,b)$, taking
$u=x+4\log n$ gives
\[
\max_{a,b}
\frac{\left|\sum_{t\in K}\d_t^{(h)}\mu_{t,b}\right|}{\sqrt{\nu_a \nu_b}}
=
\Op\!\left(\sqrt{\frac{r_n\log n}{L_{\min}}}\right).
\]
For the second term, Cauchy--Schwarz yields
\[
\left|\sum_{t\in K}\mu_{t,b}\right|
\le \sqrt{|K|\nu_{b,\mu}}
\le \sqrt{|I_a|\nu_{b,\mu}}.
\]
Assumption~\ref{ass:bounded_second_moment} gives $\omega_a^{-1}\le C$ and
$S_a^\circ/\sqrt{\nu_a}=\omega_a^{-1}\mathbb T_a^\circ$.  Thus
Lemma~\ref{I:lem:Uw} gives
\[
\frac{|\bar\d_a|\,\left|\sum_{t\in K}\mu_{t,b}\right|}{\sqrt{\nu_a \nu_b}}
\le
\frac{|S_a^{\circ}|}{\sqrt{\nu_a}}
\frac{1}{\sqrt{|I_a|}}
\sqrt{\frac{\nu_{b,\mu}}{\nu_b}}
\le
C\bar{\mathbb{T}}_{\max}^{\circ}\sqrt{\frac{r_n}{L_{\min}}}
=
\Op\!\left(\sqrt{\frac{r_n\log n}{L_{\min}}}\right).
\]
Hence
\begin{equation}
\max_{a,b}
\frac{\left|\sum_{t\in K}d_{t,a}\mu_{t,b}\right|}{\sqrt{\nu_a \nu_b}}
=
\Op\!\left(\sqrt{\frac{r_n\log n}{L_{\min}}}\right),
\label{I:eq:mixed1}
\end{equation}
and the same bound holds for $\sum_{t\in K}d_{t,b}\mu_{t,a}$.

\medskip
\noindent\underline{Bounding the last term of \eqref{eq:Omegahat_ab}.}
By Cauchy--Schwarz and the definition of $r_n$,
\begin{equation}\label{eq:Omegahat_last_term}
\max_{a,b}
\frac{\left|\sum_{t\in K}\mu_{t,a}\mu_{t,b}\right|}{\sqrt{\nu_a \nu_b}}
\le
\max_{a,b}
\sqrt{\frac{\sum_{t\in K}\mu_{t,a}^2}{\nu_a}}
\sqrt{\frac{\sum_{t\in K}\mu_{t,b}^2}{\nu_b}}
\le r_n.
\end{equation}

Combining \eqref{I:eq:centerednum}, \eqref{I:eq:mixed1}, and \eqref{eq:Omegahat_last_term} yields
\[
\max_{a,b}
\frac{|\widehat\Omega_{ab}-\Omega_{ab}|}{\sqrt{\nu_a \nu_b}}
=
\Op\!\left(\sqrt{\frac{\log n}{L_{\min}}}+r_n\right).
\]
On $\{\Delta_Q\le1/2\}$, Assumption~\ref{ass:length_ratio} gives
\[
\max_{a,b}
\frac{|\widehat\Omega_{ab}-\Omega_{ab}|}
{|I_a|^{1/2-\gamma}|I_b|^{1/2-\gamma}Q_a^\gamma
Q_b^\gamma}
\le
C\max_{a,b}\frac{|\widehat\Omega_{ab}-\Omega_{ab}|}{\sqrt{\nu_a \nu_b}}
=
\Op\!\left(\sqrt{\frac{\log n}{L_{\min}}}+r_n\right).
\]
Together with \eqref{eq:I_cov_denominator_perturb} and the decomposition
\eqref{I:eq:covdec1}--\eqref{I:eq:covdec2}, this proves
\[
\|\widehat\Sigma-\Sigma_{\circ}\|_{\max}
=
\Op\!\left(\sqrt{\frac{\log n}{L_{\min}}}+r_n\right).
\]
Finally, Assumptions \ref{ass:intervalwise_mean} and \ref{ass:Lmin} imply
\[
  \left(\sqrt{\frac{\log n}{L_{\min}}}+r_n\right)(\log n)^2
  =
  \frac{(\log n)^{5/2}}{\sqrt{L_{\min}}}
  +
  r_n(\log n)^2
  \to0.
\]
Thus,
\[
\|\widehat\Sigma-\Sigma_{\circ}\|_{\max}
=o_p\bigl((\log n)^{-2}\bigr).
\]

\subsubsection{Verifying conditions of the reduction lemmas}\label{subsub:I_verify_reduction}

We verify the hypotheses of the three reduction lemmas used in the roadmap
for Regime I.
We first verify Lemma~\ref{lem:centered_oracle_perturbation}.  For
$\gamma=0$, the two denominators are identical.  Uniformly over
$\gamma\in[0,1/2]$, Lemmas~\ref{I:lem:Delta} and \ref{I:lem:Uw} give
\[
C\gamma\Delta_Q\bar{\mathbb{T}}_{\max}^{\circ}
=
\Op\!\left(
  \left(\sqrt{\frac{\log n}{L_{\min}}}+r_n\right)\sqrt{\log n}
\right).
\]
Assumptions~\ref{ass:intervalwise_mean} and \ref{ass:Lmin} imply
\[
\left(\sqrt{\frac{\log n}{L_{\min}}}+r_n\right)\log n
=
\frac{(\log n)^{3/2}}{\sqrt{L_{\min}}}
+r_n\log n
\to 0.
\]
Choose a deterministic sequence $M_n\uparrow\infty$ sufficiently slowly that, with
\[
  \eta_n:=M_n\left(\sqrt{\frac{\log n}{L_{\min}}}+r_n\right)\sqrt{\log n},
\]
\[
  \eta_n\to0,
  \qquad
  \eta_n\sqrt{1\vee\log(p/\eta_n)}\to0.
\]
Then
\begin{equation}\label{eq:I_oracle_perturb_prob}
  \Pbb(\Delta_Q>1/2)
  +
  \Pbb(C\gamma\Delta_Q\bar{\mathbb{T}}_{\max}^{\circ}>\eta_n)\to0.
\end{equation}
Since $\Sigma_{\circ}=\Cov(\mathbb{T}^{\circ})$ and
$\mathbb{T}^{\circ}$ is defined in \eqref{eq:mathbbT_def}, Regime I
independence gives, for $a=(I,h)$,
\[
  (\Sigma_{\circ})_{aa}
  =
  \Var(\mathbb{T}_a^{\circ})
  =
  \left(\frac{\nu_a}{|I|}\right)^{1-2\gamma}.
\]
Assumption~\ref{ass:bounded_second_moment} therefore yields
\[
  (\sigma_-^2)^{1-2\gamma}
  \le
  (\Sigma_{\circ})_{aa}
  \le
  (\sigma_+^2)^{1-2\gamma}
  \qquad\forall a\in A.
\]
Proposition~\ref{prop:gaussian_anti_concentration} therefore implies
\begin{equation}\label{eq:I_oracle_anti_conc}
\sup_{x\in\R}\Pbb\bigl(|G_{\circ,\max}-x|\le \eta_n\bigr)\to 0.
\end{equation}
For every $x\in\R$, \eqref{eq:roadmap_centered_gaussian} gives
\[
  \Pbb\bigl(|\mathbb{T}_{\max}^{\circ}-x|\le\eta_n\bigr)
  \le
  \Pbb\bigl(|G_{\circ,\max}-x|\le\eta_n\bigr)
  +
  2\sup_{z\in\R}
  \left|
  \Pbb\bigl(\mathbb{T}_{\max}^{\circ}\le z\bigr)
  -
  \Pbb\bigl(G_{\circ,\max}\le z\bigr)
  \right|.
\]
Thus \eqref{eq:roadmap_centered_gaussian} and
\eqref{eq:I_oracle_anti_conc} imply
\[
  \sup_{x\in\R}
  \Pbb\bigl(|\mathbb{T}_{\max}^{\circ}-x|\le\eta_n\bigr)
  \to0.
\]
Together with \eqref{eq:I_oracle_perturb_prob}, this verifies the
conditions of Lemma~\ref{lem:centered_oracle_perturbation}.

We next verify Lemma~\ref{lem:random_denominator} with $k_n=1$.  Lemma
\ref{I:lem:Delta} and $\nu_a\ge\sigma_-^2|I|$ give
$Q_a\ge c|I|$ uniformly with probability $1-o(1)$ over data.
Lemma~\ref{I:lem:fourth} verifies \eqref{eq:RD_fourth}.  Since
$\Theta_n=I_n$,
\[
  \Var^*\!\left(\frac{S_a^*}{\sqrt{Q_a}}\right)
  =
  \frac{\sum_{t\in I}\left(D_t^{(h)}-\bar D_a\right)^2}
  {Q_a}=1,
\]
which verifies \eqref{eq:RD_var}, and \eqref{eq:RD_rows} holds with
$k_n=1$.  Finally, Assumptions~\ref{ass:intervalwise_mean} and
\ref{ass:Lmin} give
\[
  \frac{(\log n)^3}{L_{\min}}\to0,
  \qquad
  r_n(\log n)^2\to0.
\]
Thus \eqref{eq:RD_growth} holds.

It remains to verify Lemma~\ref{lem:fixed_bootstrap_covariance}.  Under
Regime I independence,
$(\Sigma_{\circ})_{aa}=(\nu_a/|I|)^{1-2\gamma}$, so Assumption
\ref{ass:bounded_second_moment} gives diagonal entries bounded above and
away from zero.  Equation~\eqref{eq:roadmap_bootstrap_covariance} gives
$\|\widehat\Sigma-\Sigma_{\circ}\|_{\max}=o_p((\log n)^{-2})$, and hence the
same diagonal bounds hold for $\widehat\Sigma$ with probability $1-o(1)$
over data.  This verifies all conditions of
Lemma~\ref{lem:fixed_bootstrap_covariance}.

We now complete the proof of Theorem~\ref{I:thm:main}.
Sections~\ref{subsub:I_centered_gaussian} and
\ref{subsub:I_bootstrap_covariance} prove
\eqref{eq:roadmap_centered_gaussian} and
\eqref{eq:roadmap_bootstrap_covariance}, respectively, and the preceding
paragraphs verify the conditions of
Lemmas~\ref{lem:centered_oracle_perturbation}--\ref{lem:fixed_bootstrap_covariance}.
Therefore the conditions of Lemma~\ref{lem:size_transfer} hold.  Applying
Lemma~\ref{lem:size_transfer} proves the theorem.

\subsection{Regime D Validity}

\subsubsection{Additional assumptions for Regime D}\label{subsubapp:regimeD_main}

We impose the following assumptions.

\begin{appassumption}\label{ass:D_mixing}
The centered vector sequence $\d_t=(\d_t^{(1)},\ldots,\d_t^{(H)})$,
$t=1,\ldots,n$, is geometrically $\alpha$-mixing: there exist constants
$K_1>1$, $K_2>0$, and $\gamma_{\mathrm{mix}}>0$ such that
\[
\alpha(\ell)
:=
\sup_{1\le k\le n-\ell}
\sup_{A\in\sigma(\d_1,\dots,\d_k),\ B\in\sigma(\d_{k+\ell},\dots,\d_n)}
\bigl|\Pbb(A\cap B)-\Pbb(A)\Pbb(B)\bigr|
\le K_1e^{-K_2\ell^{\gamma_{\mathrm{mix}}}}
\qquad\forall \ell\ge 1.
\]
\end{appassumption}

\begin{appassumption}\label{ass:D_nondegenerate}
The centered oracle coordinates are nondegenerate: there exist constants $0<c_-\le c_+<\infty$ such that
\[
c_-\le \Var\!\left(\frac{S_a^{\circ}}{\sqrt{\nu_a}}\right)\le c_+
\qquad\forall a\in A.
\]
\end{appassumption}

\begin{appassumption}\label{ass:D_kernel}
The kernel is symmetric, continuously differentiable, satisfies $K(0)=1$,
\[
  \sup_{x\in\R}|K'(x)|<\infty,
\]
and obeys
\[
|K(x)|\le C_K(1+|x|)^{-\vartheta}
\quad\forall x\in\R
\]
for some constants $C_K<\infty$ and $\vartheta>1$.
The bandwidth satisfies
\[
b_n\asymp n^{\rho},
\quad
0<\rho<\frac{\vartheta-1}{3\vartheta-2}.
\]
The matrices
\begin{equation}\label{eq:D_Theta}
  \Theta_n(t,u):=K\!\left(\frac{t-u}{b_n}\right),
  \qquad 1\le t,u\le n,
\end{equation}
are positive semidefinite.
\end{appassumption}

\begin{remark}
  We prove in Section~\ref{subsec:QS_kernel} that the quadratic spectral
  kernel introduced by \citet{Andrews1991} satisfies
  Assumption~\ref{ass:D_kernel}.
  Section~\ref{subsec:D_data_dependent_bandwidth} further shows that the test
  remains valid in Regime~D with the data-driven bandwidth used in the main
  text.
\end{remark}

\begin{appassumption}\label{ass:D_growth}
Let
\[
B_n:=1\vee\frac{B_0\sqrt n}{\sigma_-\sqrt{L_{\min}}}.
\]
Assume
\[
(\log n)^{3-\gamma_{\mathrm{mix}}}=o\bigl(n^{\gamma_{\mathrm{mix}}/3}\bigr),
\]
and
\[
(B_n)^{2/3}\frac{(\log n)^{(1+2\gamma_{\mathrm{mix}})/(3\gamma_{\mathrm{mix}})}}{n^{1/9}}
+
B_n\frac{(\log n)^{7/6}}{n^{1/9}}
\to 0.
\]
Let $c_1=c_1(\rho,\vartheta)>0$ and
$c_2=c_2(2,\gamma_{\mathrm{mix}},\vartheta)>0$ denote the constants
delivered by Proposition~\ref{prop:consistency_covariance_matrix} with its
kernel parameters specialized to $\gamma_1=2$ and
$\gamma_2=\gamma_{\mathrm{mix}}$.  Assume
\[
(B_n)^2n^{-c_1}(\log n)^{c_2+2}\to 0,
\qquad
(B_n)^2n^{-\rho}(\log n)^2\to 0.
\]
In addition, $b_n(\log n)^2/n\to0$.
\end{appassumption}

\begin{appassumption}\label{ass:D_mean_flatness}
The interval-wise mean is asymptotically flat:
\[b_nr_n(\log n)^2\to 0.\]
\end{appassumption}

\begin{appassumption}\label{ass:D_Lmin}
The bandwidth, dimension, mixing exponent, and minimum interval length
satisfy
\[
  \frac{(\log n)^5}{L_{\min}}\to0,
  \qquad
  \frac{b_n(\log n)^3}{L_{\min}}\to0,
\]
\[
  \frac{(\log n)^{1+2/\gamma_{\mathrm{mix}}}}{L_{\min}}\to0,
  \qquad
  \frac{(\log n)^{2+1/\gamma_{\mathrm{mix}}}}{L_{\min}}\to0.
\]
When $\gamma_{\mathrm{mix}}=1$, they are implied by the first condition.
\end{appassumption}

\begin{remark}
When $\gamma_{\mathrm{mix}}=1$ and the quadratic spectral kernel is used,
Section~\ref{subsec:QS_kernel} gives $\vartheta=2$ and
$b_n\asymp n^\rho$ with $0<\rho<1/4$.  Take
$\rho=1/4-\epsilon$ for a fixed $0<\epsilon<1/4$.
Proposition~\ref{prop:chang_chen_wu_explicit_exponents} gives
\[
  c_1(1/4-\epsilon,2)=\frac{4\epsilon}{3},
  \qquad
  c_2(2,1,2)=\frac{10}{3},
  \qquad
  c_2^\sharp(2,1,2)=\frac23.
\]
Since $p\le n^2$, we use the improved exponent $c_2^\sharp=2/3$ in
Assumption~\ref{ass:D_growth}.  The definition in
Assumption~\ref{ass:D_growth} satisfies
\[
  1\vee\frac{B_0\sqrt n}{\sigma_-\sqrt{L_{\min}}}
  \asymp \sqrt{\frac{n}{L_{\min}}},
\]
and Assumption~\ref{ass:D_growth} reduces to
\begin{equation}\label{eq:QS_D_growth_simplified}
  \begin{gathered}
  \frac{n^{7/18}(\log n)^{7/6}}{\sqrt{L_{\min}}}\to0,
  \qquad
  \frac{n^{1-4\epsilon/3}(\log n)^{8/3}}{L_{\min}}\to0,
  \qquad
  \frac{n^{3/4+\epsilon}(\log n)^2}{L_{\min}}\to0.
  \end{gathered}
\end{equation}
Indeed, $(\log n)^2=o(n^{1/3})$ and
$n^{-3/4-\epsilon}(\log n)^2\to0$ hold automatically.  Also, $L_{\min}\le n$
gives the omitted high-dimensional CLT term
\[
  \frac{n^{2/9}(\log n)/L_{\min}^{1/3}}
  {n^{7/18}(\log n)^{7/6}/L_{\min}^{1/2}}
  =
  \left(\frac{L_{\min}}{n\log n}\right)^{1/6}
  \to0.
\]
The first condition in \eqref{eq:QS_D_growth_simplified} is implied by the
other two.  Indeed, those two conditions require $L_{\min}$ to dominate
polynomial powers with exponents $1-4\epsilon/3$ and $3/4+\epsilon$,
respectively.  Since
\[
  \max\{1-4\epsilon/3,\,3/4+\epsilon\}\ge 6/7,
\]
the resulting polynomial slack is stronger than the requirement
$L_{\min}\gg n^{7/9}(\log n)^{7/3}$ imposed by the first condition.
Assumption~\ref{ass:D_mean_flatness} becomes
\begin{equation}\label{eq:QS_D_mean_flatness_simplified}
  n^{1/4-\epsilon}r_n(\log n)^2\to0,
\end{equation}
and Assumption~\ref{ass:D_Lmin} reduces to
\[
  \frac{n^{1/4-\epsilon}(\log n)^3}{L_{\min}}\to0,
\]
which is implied by the first condition in
\eqref{eq:QS_D_growth_simplified} and also implies the other three length
conditions in Assumption~\ref{ass:D_Lmin}.
Combining \eqref{eq:QS_D_growth_simplified} and
\eqref{eq:QS_D_mean_flatness_simplified}, Assumptions
\ref{ass:D_growth}--\ref{ass:D_Lmin} hold if
\[
  L_{\min}
  =
  \omega\!\left(
    \max\left\{
      n^{1-4\epsilon/3}(\log n)^{8/3},
      n^{3/4+\epsilon}(\log n)^2
    \right\}
  \right),
  \qquad
  r_n=o\!\left(\frac{1}{n^{1/4-\epsilon}(\log n)^2}\right).
\]
In Section~\ref{sec:scan_method}, we choose $b_n\asymp n^{1/5}$, giving $\epsilon=1/20$.
In this case, the condition reduces to
\[
  L_{\min}=\omega\!\left(n^{14/15}(\log n)^{8/3}\right),
  \qquad
  r_n=o\!\left(\frac{1}{n^{1/5}(\log n)^2}\right).
\]
\end{remark}

\subsubsection{Main theorem}

\begin{theorem}\label{D:thm:main}
Fix the exponent $\gamma\in[0,1/2]$.  Assume Assumptions
\ref{ass:subgaussian}, \ref{ass:bounded_second_moment},
\ref{ass:length_ratio}, \ref{ass:D_mixing}--\ref{ass:D_kernel},
and \ref{ass:D_growth}--\ref{ass:D_Lmin}.
Let $c_{n,1-\alpha}^{*}$ be the conditional $(1-\alpha)$-quantile
in \eqref{eq:implemented_bootstrap}.  Then the Gaussian multiplier
bootstrap with covariance
$\Theta_n(t,u)=K((t-u)/b_n)$ is asymptotically conservative: under
\eqref{eq:null_standalone},
\[
\limsup_{n\to\infty}
\Pbb\bigl(T_{\max}>c_{n,1-\alpha}^{*}\bigr)
\le \alpha
\qquad\forall \alpha\in(0,1).
\]
\end{theorem}

Following the roadmap in Section~\ref{subsub:roadmap_proofs}, we are left to prove
\eqref{eq:roadmap_centered_gaussian} and
\eqref{eq:roadmap_bootstrap_covariance} and verify the conditions in
Lemmas~\ref{lem:centered_oracle_perturbation}--\ref{lem:fixed_bootstrap_covariance}.
In the following, Section~\ref{subsub:D_preliminary_bounds} provides
preliminary results, Section~\ref{subsub:D_centered_gaussian} proves
\eqref{eq:roadmap_centered_gaussian},
Section~\ref{subsub:D_bootstrap_covariance} proves
\eqref{eq:roadmap_bootstrap_covariance}, and
Section~\ref{subsub:D_verify_reduction} completes the proof of
Theorem~\ref{D:thm:main} by
verifying the conditions in Lemmas~\ref{lem:centered_oracle_perturbation}--\ref{lem:fixed_bootstrap_covariance}.
For Sections~\ref{subsub:D_centered_gaussian}--\ref{subsub:D_verify_reduction},
we assume Assumptions
\ref{ass:subgaussian}, \ref{ass:bounded_second_moment},
\ref{ass:length_ratio}, \ref{ass:D_mixing}--\ref{ass:D_kernel},
and \ref{ass:D_growth}--\ref{ass:D_Lmin} all hold.

\subsubsection{Preliminary bounds}\label{subsub:D_preliminary_bounds}

\begin{lemma}\label{D:lem:rowsum}
Under Assumption \ref{ass:D_kernel}, there exists $C_\Theta<\infty$ such that
\[
\max_{1\le t\le n}\sum_{u=1}^n |\Theta_n(t,u)|\le C_\Theta b_n,
\qquad
\max_{1\le t\le n}\sum_{u=1}^n \Theta_n(t,u)^2\le C_\Theta b_n
\qquad\text{for all }n.
\]
Consequently,
\[
|x^\top\Theta_n y|\le C_\Theta b_n\norm{x}_2\norm{y}_2
\qquad\forall x,y\in\R^n.
\]
\end{lemma}

\begin{proof}
Fix $t$. Since $|K(x)|\le C_K(1+|x|)^{-\vartheta}$ with $\vartheta>1$,
\[
\sum_{u=1}^n |\Theta_n(t,u)|
\le
\sum_{u=1}^n C_K\left(1+\frac{|t-u|}{b_n}\right)^{-\vartheta}
\le C_K\sum_{m\in\mathbb Z}\left(1+\frac{|m|}{b_n}\right)^{-\vartheta}.
\]
Split the last sum into $|m|\le b_n$ and $|m|>b_n$.
The first part contributes $O(b_n)$.
For the second part,
\[
\sum_{|m|>b_n}\left(1+\frac{|m|}{b_n}\right)^{-\vartheta}
\le Cb_n\sum_{k\ge 1}(1+k)^{-\vartheta}=O(b_n)
\]
because $\vartheta>1$.
This proves the row-sum bound.
For the squared row-sum bound,
\[
  \sum_{u=1}^n \Theta_n(t,u)^2
  \le
  C\sum_{m\in\mathbb Z}\left(1+\frac{|m|}{b_n}\right)^{-2\vartheta}.
\]
The terms with $|m|\le b_n$ contribute $O(b_n)$, while the terms with
$|m|>b_n$ are bounded by
$Cb_n\sum_{k\ge1}(1+k)^{-2\vartheta}=O(b_n)$ because
$2\vartheta>1$.
The quadratic-form bound follows from
\[
|x^\top\Theta_n y|
\le
\norm{x}_2\,\norm{\Theta_n y}_2
\le
\norm{x}_2\,\left(\max_t\sum_u |\Theta_n(t,u)|\right)\norm{y}_2.
\]
\end{proof}

\begin{lemma}\label{D:lem:Uw}
Under Assumptions \ref{ass:subgaussian},
\ref{ass:bounded_second_moment}, \ref{ass:length_ratio},
\ref{ass:D_mixing}--\ref{ass:D_kernel}, and \ref{ass:D_growth},
\[
  \bar{\mathbb{T}}_{\max}^{\circ}=O_p(\sqrt{\log n}).
\]
\end{lemma}

\begin{proof}
Define
\[
Z_{t,a}^{\circ}
:=\frac{\sqrt n\,\1\{t\in I\}\d_t^{(h)}}
{|I|^{1/2-\gamma}\nu_a^\gamma},
\qquad a=(I,h)\in A.
\]
Then
\[
\mathbb{T}^{\circ}=\frac{1}{\sqrt n}\sum_{t=1}^n
Z_t^{\circ}.
\]
As in Section~\ref{subsub:I_centered_gaussian},
\[
\norm{Z_{t,a}^{\circ}}_{\psi_2}
\le B_0(1\vee\sigma_-^{-1})\sqrt{\frac n{L_{\min}}}
\lesssim B_n
\qquad\forall t,a.
\]
The sequence inherits the same geometric mixing rate.  Furthermore,
\[
\Var(\mathbb{T}_a^{\circ})
=\left(\frac{\nu_a}{|I|}\right)^{1-2\gamma}
\Var\!\left(\frac{S_a^\circ}{\sqrt{\nu_a}}\right),
\]
so Assumptions~\ref{ass:D_nondegenerate} and
\ref{ass:bounded_second_moment} bound every coordinate variance above and
away from zero.
Hence Proposition \ref{prop:dependent_CLT} applies under Assumption \ref{ass:D_growth} and gives
\[
\sup_{A\in\mathcal R^p}
\Bigl|\Pbb(\mathbb{T}^{\circ}\in A)
-\Pbb(G_{\circ}\in A)\Bigr|\to 0.
\]
Applying this with $A=[-x,x]^p$ yields
\[
\sup_{x\in\R}
\Bigl|\Pbb(\bar{\mathbb{T}}_{\max}^{\circ}\le x)
-\Pbb\left(\max_{a\in A}|G_{\circ,a}|\le x\right)\Bigr|
\to0.
\]
Proposition~\ref{prop:subgaussian_max} applied to
$(G_{\circ,a})_{a\in A}$ gives
\[
\max_{a\in A}|G_{\circ,a}|=\Op\bigl(\sqrt{\log n}\bigr),
\]
and hence $\bar{\mathbb{T}}_{\max}^{\circ}=\Op(\sqrt{\log n})$.
\end{proof}

\begin{lemma}\label{lem:centered-square}
Assume Assumptions \ref{ass:subgaussian}, \ref{ass:bounded_second_moment}, and \ref{ass:D_mixing}--\ref{ass:D_Lmin}. Then there exist constants $C_A,c_A,\eta_0\in(0,\infty)$, depending only on $(B_0,K_1,K_2,\gamma_{\mathrm{mix}})$, such that for every interval $I\subset[n]$ and every $x\ge 1$ satisfying
\[
 x \le c_A |I|^{\gamma_{\mathrm{mix}}/(2+\gamma_{\mathrm{mix}})},
\]
one has
\[
\Pbb\!\left(
  \left|\sum_{t\in I}
  \left\{(\d_t^{(h)})^2-\E[(\d_t^{(h)})^2]\right\}\right|
  > C_A\sqrt{|I|x}
\right)
\le 2e^{-x}.
\]
\end{lemma}

\begin{proof}
Fix an interval $I$ and let $s$ be its left endpoint.  Define
\[
  Z_r:=(\d_{s+r-1}^{(h)})^2-\E[(\d_{s+r-1}^{(h)})^2],
  \qquad r=1,\dots,|I|.
\]
Since each $Z_r$ is a measurable function of $D_{s+r-1}$, the sequence
$(Z_r)_{r=1}^{|I|}$ inherits geometric strong mixing condition
(Assumption \ref{ass:D_mixing}). After possibly decreasing the exponent
constant, the prefactor $K_1$ can be absorbed into the exponential rate,
so there exists $\widetilde K_2>0$, depending only on
$(K_1,K_2,\gamma_{\mathrm{mix}})$, such that
\[
\alpha_Z(\ell) \le \exp\!\bigl(-\widetilde K_2\ell^{\gamma_{\mathrm{mix}}}\bigr)
\qquad\text{for all }\ell\ge 1.
\]
By Assumption \ref{ass:subgaussian} and Proposition
\ref{prop:product_subgaussian},
\[
  \| (\d_t^{(h)})^2 \|_{\psi_1}\le B_0^2
  \quad\Longrightarrow\quad
  \|(\d_t^{(h)})^2-\E[(\d_t^{(h)})^2]\|_{\psi_1}\le 2B_0^2.
\]
Therefore there exists $b_0<\infty$, depending only on $B_0$, such that
\[
\sup_{1\le r\le |I|} \Pbb(|Z_r|>u) \le \exp\!\bigl(1-u/b_0\bigr)
\qquad\text{for all }u>0.
\]
Thus Proposition~\ref{prop:alpha_mixing_concentration} applies with
\[
\gamma_1=\gamma_{\mathrm{mix}},
\qquad
\gamma_2=1,
\qquad
\gamma=
\left(\frac{1}{\gamma_{\mathrm{mix}}}+1\right)^{-1}
=\frac{\gamma_{\mathrm{mix}}}{1+\gamma_{\mathrm{mix}}}.
\]
Consequently,
\[
\frac{\gamma}{2-\gamma}
=
\frac{\gamma_{\mathrm{mix}}}{2+\gamma_{\mathrm{mix}}}.
\]
Hence there exist constants $C_0,C_1,C_2,\eta_0\in(0,\infty)$, depending only on $(B_0,K_1,K_2,\gamma_{\mathrm{mix}})$, such that for every $y\ge C_0(\log |I|)^{\eta_0}$,
\[
\Pbb\!\left(\max_{1\le j\le |I|}\left|\sum_{r=1}^j Z_r\right|\ge y\right)
\le
(|I|+1)\exp\!\left(-\frac{y^{\gamma}}{C_1}\right)
+
\exp\!\left(-\frac{y^2}{C_2+C_2 |I|V}\right).
\]
Let $q(u)$ denote the tail quantile in
Proposition~\ref{prop:alpha_mixing_concentration},
\[
  q(u):=\inf\{x\ge0:\sup_r\Pbb(|Z_r|>x)\le u\}.
\]
Since $\sup_r\Pbb(|Z_r|>x)\le \exp(1-x/b_0)$, for $0<u\le1$,
\[
  x\ge b_0\log(e/u)
  \quad\Longrightarrow\quad
  \sup_r\Pbb(|Z_r|>x)\le u.
\]
Thus $q(u)\le b_0\log(e/u)$.  Proposition~\ref{prop:V_bound} therefore gives
\[
  \int_0^{2\alpha_Z(k)}q(u)^2\,du
  \lesssim
  \alpha_Z(k)\log^2\!\bigl(e/\alpha_Z(k)\bigr).
\]
Because $\alpha_Z(k)\le\exp(-\widetilde K_2k^{\gamma_{\mathrm{mix}}})$,
\[
  \sum_{k=0}^\infty
  \int_0^{2\alpha_Z(k)}q(u)^2\,du
  \lesssim
  \sum_{k=0}^\infty
  \alpha_Z(k)\log^2\!\bigl(e/\alpha_Z(k)\bigr)
  \lesssim
  \sum_{k=0}^\infty
  e^{-\widetilde K_2k^{\gamma_{\mathrm{mix}}}}
  (1+k^{\gamma_{\mathrm{mix}}})^2
  <\infty.
\]
Hence the variance term $V$ in
Proposition~\ref{prop:alpha_mixing_concentration} is bounded uniformly in
$|I|$, so $V=O(1)$.  Since the full interval sum is one of the partial sums,
there exists $C<\infty$ such that, for every
$y\ge C_0(\log |I|)^{\eta_0}$,
\begin{equation}\label{eq:D_centered_square_full_tail}
\begin{split}
\Pbb\!\left(
  \left|\sum_{t\in I}
  \left\{(\d_t^{(h)})^2-\E[(\d_t^{(h)})^2]\right\}\right|
  \ge y
\right)
&\le
(|I|+1)\exp\!\left(-\frac{y^{\gamma}}{C_1}\right)
+
\exp\!\left(-\frac{y^2}{C |I|}\right).
\end{split}
\end{equation}

Now set $y=A\sqrt{|I|x}$ with $A>0$ to be chosen. Because $x\ge 1$, one has $y\ge A\sqrt{|I|}$, which dominates $C_0(\log |I|)^{\eta_0}$ for all sufficiently large $|I|$; the finitely many small values of $|I|$ are absorbed into the final constant $C_A$. Also, from
\[
 x \le c_A |I|^{\gamma/(2-\gamma)}
 = c_A |I|^{\gamma_{\mathrm{mix}}/(2+\gamma_{\mathrm{mix}})},
\]
it follows that $x\lesssim |I|^{\gamma/2}x^{\gamma/2}$.  Thus the first
term in \eqref{eq:D_centered_square_full_tail} satisfies, by choosing $A$
large enough and then $c_A$ small enough,
\[
(|I|+1)\exp\!\left(-\frac{A^{\gamma}|I|^{\gamma/2}x^{\gamma/2}}{C_1}\right)
\le e^{-x}.
\]
For the second term in \eqref{eq:D_centered_square_full_tail},
\[
\exp\!\left(-\frac{A^2x}{C}\right)
\le e^{-x}
\]
for $A$ chosen large enough. Combining the two bounds gives the claimed inequality after renaming $A$ as $C_A$.
\end{proof}

\begin{lemma}\label{D:lem:Delta}
Under Assumptions \ref{ass:subgaussian},
\ref{ass:bounded_second_moment}, and
\ref{ass:D_mixing}--\ref{ass:D_Lmin},
\[
  \Delta_Q
  =O_p\!\left(
    \sqrt{\frac{\log n}{L_{\min}}}
    +\sqrt{r_n}
  \right).
\]
In particular, $\Delta_Q=o_p(1)$ and
$\Pbb(\Delta_Q>1/2)\to0$.
\end{lemma}

\begin{proof}
Lemma~\ref{lem:denominator_decomposition} expresses both
$q_a-\nu_a$ and $Q_a-q_a$ in terms of
$q_a^\circ$, $C_a$, and $\nu_{a,\mu}$.  We control these three terms
uniformly.

\medskip
\noindent\emph{Step 1: control of $\frac{q_a^\circ}{\nu_a}$.}
By \eqref{eq:denom_Adef},
\[
  q_a^\circ
  =
  \sum_{t\in I}\Bigl(\d_t^{(h)2}-\E[\d_t^{(h)2}]\Bigr)
  -
  |I|\bar\d_a^2.
\]
By Lemma \ref{lem:centered-square}, there exist constants $C_A,c_A>0$ such that, for every fixed $a=(I,h)$ and every $x\ge 1$ with
\[
x+2\log n\le c_A |I|^{\gamma_{\mathrm{mix}}/(2+\gamma_{\mathrm{mix}})},
\]
one has
\[
\Pbb\!\left(
  \left|\sum_{t\in I}
    \Bigl(\d_t^{(h)2}-\E[\d_t^{(h)2}]\Bigr)\right|
  >C_A\sqrt{|I|(x+2\log n)}\right)
\le 2e^{-(x+2\log n)}.
\]
Assumption \ref{ass:D_Lmin} implies
\[
\log n=o\bigl(L_{\min}^{\gamma_{\mathrm{mix}}/(2+\gamma_{\mathrm{mix}})}\bigr),
\]
so the restriction is eventually satisfied uniformly over all intervals for
any constant $x$.
A union bound over the $p$ coordinates, retaining the coordinate-specific
normalization, gives
\[
  \max_{a=(I,h)\in A}
  \frac{
    \left|\sum_{t\in I}
      \bigl(\d_t^{(h)2}-\E[\d_t^{(h)2}]\bigr)\right|
  }{\nu_a}
  =\Op\!\left(
    \max_{a\in A}\sqrt{\frac{\log n}{|I_a|}}
  \right)
  =\Op\!\left(\sqrt{\frac{\log n}{L_{\min}}}\right).
\]

Assumption~\ref{ass:bounded_second_moment} gives
\[
  \frac{\nu_a}{|I|}\ge \sigma_-^2,
  \qquad
  \omega_a
  =
  \left(\frac{\nu_a}{|I|}\right)^{1/2-\gamma}
  \ge(\sigma_-^2)^{1/2-\gamma},
  \qquad
  \omega_a^{-1}\le C.
\]
Since $S_a^\circ/\sqrt{\nu_a}=\omega_a^{-1}\mathbb T_a^\circ$,
Lemma~\ref{D:lem:Uw} gives
\[
\max_{a\in A}\frac{(S_a^{\circ})^2}{|I|\nu_a}
\le
\frac{C(\bar{\mathbb{T}}_{\max}^{\circ})^2}{L_{\min}}
=
\Op\!\left(\frac{\log n}{L_{\min}}\right).
\]
Combining the last two displays with the decomposition of $q_a^\circ$ at
the start of this step gives
\begin{equation}\label{eq:D_Delta_q_bound}
  \max_{a\in A}\frac{|q_a^\circ|}{\nu_a}
  =
  \Op\!\left(
    \sqrt{\frac{\log n}{L_{\min}}}
    +\frac{\log n}{L_{\min}}
  \right).
\end{equation}

\medskip
\noindent\emph{Step 2: control of $\frac{C_a}{\nu_a}$.}
The equality $\sum_{t\in I}\mu_{t,a}=0$ gives
$C_a=\sum_{t\in I}d_{t,a}\mu_{t,a}$.  Hence Cauchy--Schwarz yields the
pathwise bound
\[
  \frac{|C_a|}{\nu_a}
  \le
  \left(\frac{q_a}{\nu_a}\right)^{1/2}
  \left(\frac{\nu_{a,\mu}}{\nu_a}\right)^{1/2}
  \le
  \sqrt{r_n}
  \left(\frac{q_a}{\nu_a}\right)^{1/2}.
\]
Equation \eqref{eq:D_Delta_q_bound} shows that
$\max_a q_a/\nu_a=1+o_p(1)$.  Therefore
\begin{equation}\label{eq:D_Delta_C_bound}
  \max_{a\in A}\frac{|C_a|}{\nu_a}=O_p(\sqrt{r_n}).
\end{equation}

\medskip
\noindent\emph{Step 3: control of $\frac{\nu_{a,\mu}}{\nu_a}$.}
By definition \eqref{eq:rn} of $r_n$,
\begin{equation}\label{eq:D_Delta_R_bound}
\max_{a\in A}\frac{\nu_{a,\mu}}{\nu_a}=r_n.
\end{equation}

Combining \eqref{eq:D_Delta_q_bound}, \eqref{eq:D_Delta_C_bound}, and
\eqref{eq:D_Delta_R_bound} gives
\[
  \max_{a\in A}
  \left|\frac{Q_a-\nu_a}{\nu_a}\right|
  =
  O_p\!\left(
    \sqrt{\frac{\log n}{L_{\min}}}
    +\frac{\log n}{L_{\min}}
    +\sqrt{r_n}+r_n
  \right),
\]
Assumption~\ref{ass:D_Lmin} implies
$\log n/L_{\min}=o(\sqrt{\log n/L_{\min}})$, and
Assumption~\ref{ass:D_mean_flatness}, together with $b_n\to\infty$, implies
$r_n=o(1)$ and hence $r_n=O(\sqrt{r_n})$.  Thus the displayed bound on
$\max_{a\in A}|(Q_a-\nu_a)/\nu_a|$ gives the asserted rate for $\Delta_Q$.
Assumptions
\ref{ass:D_mean_flatness} and \ref{ass:D_Lmin} make
\[
  \sqrt{\frac{\log n}{L_{\min}}}
  +\sqrt{r_n}
  \to0,
\]
proving the stated probability conclusions.
\end{proof}

The next lemma gives uniform maximum and fourth-moment bounds for the
centered variables $\d_t^{(h)}=D_t^{(h)}-\mu_t^{(h)}$ in Regime~D.
\begin{lemma}
\label{D:lem:fourth}
Under Assumptions~\ref{ass:subgaussian}, \ref{ass:D_mixing}, and
\ref{ass:D_Lmin},
\[
  \max_{a=(I,h)\in A}\max_{t\in I}(\d_t^{(h)})^2=O_p(\log n),
  \qquad
  \max_{a=(I,h)\in A}\frac1{|I|}\sum_{t\in I}(\d_t^{(h)})^4=O_p(1).
\]
\end{lemma}

\begin{proof}
The first claim follows from Assumption~\ref{ass:subgaussian},
\eqref{eq:logpn_logn}, $nH\le np$, and
Proposition~\ref{prop:subgaussian_max}, exactly as in
Lemma~\ref{I:lem:fourth}.  It remains to prove the fourth-moment bound.

Set $X_t^{(h)}=(\d_t^{(h)})^4-\E[(\d_t^{(h)})^4]$.  Uniform
sub-Gaussianity gives a semiexponential tail with exponent $1/2$:
\[
  \sup_{t,h}\Pbb(|X_t^{(h)}|>u)
  \le C\exp(-c\sqrt u),
  \qquad u>0.
\]
The sequence inherits the mixing coefficients in Assumption
\ref{ass:D_mixing}.  Apply Proposition
\ref{prop:alpha_mixing_concentration} with
\[
  \gamma_1=\gamma_{\mathrm{mix}},
  \qquad
  \gamma_2=\frac12,
  \qquad
  \gamma_4
  :=\left(\frac1{\gamma_{\mathrm{mix}}}+2\right)^{-1}.
\]
The fourth powers have uniformly bounded second moments and their quantile
envelope obeys $q(u)\lesssim\log^2(e/u)$.  Consequently,
Proposition~\ref{prop:V_bound} gives
\[
  \int_0^{2\alpha_X(k)}q(u)^2\,du
  \lesssim
  \alpha_X(k)\log^4\!\bigl(e/\alpha_X(k)\bigr).
\]
By Assumption~\ref{ass:D_mixing},
$\alpha_X(k)\le\exp(-\widetilde K_2 k^{\gamma_{\mathrm{mix}}})$ for a
constant $\widetilde K_2>0$.  Hence
\[
  \sum_{k>0}\int_0^{2\alpha_X(k)}q(u)^2\,du
  \lesssim
  \sum_{k>0}
  e^{-\widetilde K_2 k^{\gamma_{\mathrm{mix}}}}
  (1+k^{\gamma_{\mathrm{mix}}})^4
  <\infty.
\]
Together with the uniform boundedness of
$\E[(X_t^{(h)})^2]$, equation~\eqref{eq:remark3-strong-mixing} gives
$V=O(1)$ uniformly over $I$ and $h$.  Taking $y=Cm$ in
\eqref{eq:mpr-strong-mixing} gives
\begin{equation}\label{eq:D_fourth_tail}
  \Pbb\left(\left|\sum_{t\in I}X_t^{(h)}\right|>Cm\right)
  \le C(m+1)e^{-cm^{\gamma_4}}+Ce^{-cm}.
\end{equation}
Because $1/\gamma_4=2+1/\gamma_{\mathrm{mix}}$, Assumption
\ref{ass:D_Lmin} implies $\log n=o(L_{\min}^{\gamma_4})$.  A union bound
over the $p$ interval-coordinate pairs in \eqref{eq:D_fourth_tail} proves
the second claim.
\end{proof}

\subsubsection{Proof of the centered Gaussian approximation}\label{subsub:D_centered_gaussian}

As in Lemma~\ref{D:lem:Uw}, set
\begin{equation}\label{eq:D_Zcirc}
Z_{t,a}^{\circ}
:=\frac{\sqrt n\,\1\{t\in I\}\d_t^{(h)}}
{|I|^{1/2-\gamma}\nu_a^\gamma}.
\end{equation}
Then
\[
\mathbb{T}^{\circ}=\frac{1}{\sqrt n}\sum_{t=1}^n Z_t^\circ.
\]
The sequence inherits the same geometric mixing rate, and Assumptions
\ref{ass:D_nondegenerate} and \ref{ass:length_ratio} yield
\[
c\le \Var\!\left(\frac{1}{\sqrt n}\sum_{t=1}^n
Z_{t,a}^\circ\right)\le C
\qquad\forall a\in A.
\]
By Assumptions \ref{ass:subgaussian} and
\ref{ass:bounded_second_moment},
\begin{equation*}\tag{\theequation a}\label{eq:D_Zcirc_psi2_bound}
\|Z_{t,a}^\circ\|_{\psi_2}
\le B_0(1\vee\sigma_-^{-1})\sqrt{\frac n{L_{\min}}}
\le B_n,
\qquad\forall t,a.
\end{equation*}
First, we assume $p=\Omega(n^\kappa)$ for some $\kappa>0$.
Equation~\eqref{eq:D_Zcirc_psi2_bound} verifies the sub-Gaussian bound in
Proposition~\ref{prop:dependent_CLT}; the variance bounds above verify its
nondegeneracy condition, and Assumption~\ref{ass:D_mixing} gives the required
geometric mixing.  The first two rate conditions in
Assumption~\ref{ass:D_growth} are exactly the CLT rate conditions in
Proposition~\ref{prop:dependent_CLT}.  Hence, with
$S_n=\mathbb T^\circ$ and $G=G_\circ$, the one-sided rectangle
$(-\infty,x]^p$ gives the event
$\{\mathbb T_{\max}^\circ\le x\}$ for each $x$.  Then
Proposition~\ref{prop:dependent_CLT} implies
\begin{equation}\label{eq:D_oracle_U_cdf}
\sup_{x\in\R}
\Bigl|
\Pbb\bigl(\mathbb{T}_{\max}^{\circ}\le x\bigr)
-
\Pbb\bigl(G_{\circ,\max}\le x\bigr)
\Bigr|
\to 0.
\end{equation}
This is \eqref{eq:roadmap_centered_gaussian}.

It remains to cover the case in which $p=o(n^\kappa)$ for every
$\kappa>0$.  Proposition~\ref{prop:dependent_CLT} requires a polynomial
dimension.  Set
\[
  \widetilde A:=A\times[m_n],
  m_n:=\left\lceil\frac{n^{1/2}}{p}\right\rceil,
  \qquad
  \widetilde p:=|\widetilde A|=m_np,
\]
and let $R_n:\R^p\to\R^{\widetilde p}$ be the coordinate-replication map
\[
  (R_n z)_{(a,j)}:=z_a,\qquad (a,j)\in A\times[m_n].
\]
The duplicated sequence is
\[
  \widetilde Z_t^\circ:=R_nZ_t^\circ,\qquad t=1,\ldots,n,
\]
so only the coordinate index is enlarged; the time index is not replicated or
reordered.  Equivalently, for $\widetilde a=(a,j)\in\widetilde A$,
\[
  \widetilde Z_{t,\widetilde a}^{\circ}:=Z_{t,a}^{\circ},
  \qquad
  \widetilde{\mathbb T}_{\widetilde a}^{\circ}:=\mathbb T_a^\circ.
\]
Thus
\[
  \widetilde{\mathbb T}_{\max}^{\circ}
  :=\max_{\widetilde a\in\widetilde A}
  \widetilde{\mathbb T}_{\widetilde a}^{\circ}
  =
  \max_{a\in A}\mathbb T_a^\circ
  =
  \mathbb T_{\max}^{\circ}.
\]
Since $p=o(n^{1/2})$, we have
$n^{1/2}\le \widetilde p<n^{1/2}+p\asymp n^{1/2}$, so the duplicated array
satisfies $\widetilde p=\Omega(n^{1/3})$, the dimension requirement in
Proposition~\ref{prop:dependent_CLT}.  For $\widetilde a=(a,j)$,
\[
  \max_{t,\widetilde a}\|\widetilde Z_{t,\widetilde a}^{\circ}\|_{\psi_2}
  =
  \max_{t,a}\|Z_{t,a}^{\circ}\|_{\psi_2},
  \qquad
  \Var\!\left(\frac1{\sqrt n}\sum_{t=1}^n
  \widetilde Z_{t,\widetilde a}^{\circ}\right)
  =
  \Var\!\left(\frac1{\sqrt n}\sum_{t=1}^n Z_{t,a}^{\circ}\right).
\]
Moreover, for any $B\subset[n]$,
\[
  \sigma(\widetilde Z_t^\circ:t\in B)
  =
  \sigma(R_nZ_t^\circ:t\in B)
  =
  \sigma(Z_t^\circ:t\in B),
\]
because $Z_{t,a}^\circ=\widetilde Z_{t,(a,1)}^\circ$.  Hence, with the usual
definition
\[
  \alpha_Z(\ell)
  :=
  \sup_t
  \alpha\!\left(\sigma(Z_s^\circ:s\le t),
  \sigma(Z_s^\circ:s\ge t+\ell)\right),
\]
we have $\alpha_{\widetilde Z}(\ell)=\alpha_Z(\ell)$ for every $\ell$.
Moreover $\log \widetilde p\asymp\log n$.  Hence the first two displays in
Assumption~\ref{ass:D_growth} imply the two rate conditions in
Proposition~\ref{prop:dependent_CLT} with dimension $\widetilde p$:
\[
  (\log \widetilde p)^{3-\gamma_{\mathrm{mix}}}
  =
  o\bigl(n^{\gamma_{\mathrm{mix}}/3}\bigr),
  \qquad
  B_n^{2/3}
  \frac{(\log \widetilde p)^{(1+2\gamma_{\mathrm{mix}})/(3\gamma_{\mathrm{mix}})}}{n^{1/9}}
  +
  B_n
  \frac{(\log \widetilde p)^{7/6}}{n^{1/9}}
  \to0.
\]

Let $\widetilde G_\circ$ be Gaussian with the covariance matrix of
$\widetilde{\mathbb T}^{\circ}$.  Proposition~\ref{prop:dependent_CLT} applied
to the duplicated array gives
\[
\sup_{x\in\R}
\Bigl|
\Pbb\bigl(\widetilde{\mathbb T}_{\max}^{\circ}\le x\bigr)
-
\Pbb\bigl(\widetilde G_{\circ,\max}\le x\bigr)
\Bigr|
\to0.
\]
Construct $\widetilde G_\circ$ by duplicating the coordinates of $G_\circ$;
then $\widetilde G_{\circ,\max}=G_{\circ,\max}$ pathwise.  Therefore
\eqref{eq:D_oracle_U_cdf}, and hence
\eqref{eq:roadmap_centered_gaussian}, also holds in the sub-polynomial case.

\subsubsection{Proof of the fixed-denominator bootstrap covariance}\label{subsub:D_bootstrap_covariance}

For $a=(I,h)$ and $b=(J,g)$, the conditional covariance matrix of the
fixed-denominator bootstrap satisfies
\begin{equation}\label{eq:D_boot_cov_entry}
\widehat\Sigma_{ab}
=
\frac{\sum_{t\in I}\sum_{u\in J}\Theta_n(t,u)
\left(D_t^{(h)}-\bar D_a\right)
\left(D_u^{(g)}-\bar D_b\right)}
{|I|^{1/2-\gamma}|J|^{1/2-\gamma}Q_a^\gamma Q_b^\gamma}.
\end{equation}
Use $Z_t^\circ$ from \eqref{eq:D_Zcirc} and set
$\bar Z^\circ:=n^{-1}\sum_{t=1}^n Z_t^\circ$.  Let
\[
  \widehat{\Sigma}^\circ
  :=\frac1n\sum_{t=1}^n\sum_{u=1}^n
  \Theta_n(t,u)(Z_t^{\circ}-\bar Z^\circ)
  (Z_u^{\circ}-\bar Z^\circ)^\top .
\]
The array $(Z_t^\circ)_{t\le n}$ is centered, inherits the geometric mixing
rate from Assumption~\ref{ass:D_mixing}, and satisfies
$\max_{t,a}\|Z_{t,a}^\circ\|_{\psi_2}\le B_n$ by
\eqref{eq:D_Zcirc_psi2_bound}.  Recall that
\[
  \Sigma_{\circ}
  =
  \Cov(\mathbb T^\circ)
  =
  \Cov\!\left(\frac1{\sqrt n}\sum_{t=1}^n Z_t^\circ\right).
\]
Therefore
Proposition~\ref{prop:consistency_covariance_matrix}, applied to
$(Z_t^\circ)_{t\le n}$, yields
\[
  \|\widehat{\Sigma}^\circ-\Sigma_{\circ}\|_{\max}
  =
  O_p\!\left(B_n^2n^{-c_1}(\log p)^{c_2}\right)
  +O\!\left(B_n^2n^{-\rho}\right).
\]
Because $\log p\le\log n$, the condition
$B_n^2n^{-c_1}(\log n)^{c_2+2}\to0$ makes the first term
$o_p((\log n)^{-2})$, while
$B_n^2n^{-\rho}(\log n)^2\to0$ makes the second term
$o((\log n)^{-2})$.  Hence
\begin{equation}\label{eq:D_centered_kernel_cov}
  \|\widehat{\Sigma}^\circ-\Sigma_{\circ}\|_{\max}
  =o_p((\log n)^{-2}).
\end{equation}

It remains to compare the feasible covariance in
\eqref{eq:D_boot_cov_entry} with the globally centered oracle covariance
$\widehat{\Sigma}^\circ$.  For $a=(I,h)$, equations
\eqref{eq:dta_def}--\eqref{eq:muta_def} give
\begin{equation}\label{eq:D_feasible_decomposition}
  D_t^{(h)}-\bar D_a=d_{t,a}+\mu_{t,a},
  \qquad t\in I,
\end{equation}
Define
\[
  w_{t,I}:=\frac{\1\{t\in I\}}{|I|}-\frac1n,
  \qquad t\in[n].
\]
Then \eqref{eq:D_Zcirc} gives
\begin{equation}\label{eq:D_oracle_centering_decomposition}
  Z_{t,a}^\circ-\bar Z_a^\circ
  =
  \frac{\sqrt n}{|I|^{1/2-\gamma}\nu_a^\gamma}
  \left\{
    \1\{t\in I\}d_{t,a}
    +S_a^\circ w_{t,I}
  \right\}.
\end{equation}
Also,
\begin{equation}\label{eq:D_cov_transfer_norms}
  \sum_{t\in I}(d_{t,a}+\mu_{t,a})^2=Q_a,
  \qquad
  \sum_{t\in I}\mu_{t,a}^2=\nu_{a,\mu},
  \qquad
  \sum_{t=1}^n
  (S_a^\circ)^2 w_{t,I}^2
  \le \frac{(S_a^\circ)^2}{|I|}.
\end{equation}
To see the denominator bound, write $x_c:=Q_c/\nu_c$.  On
$\{\Delta_Q\le1/2\}$, the definition of $\Delta_Q$ gives
$x_c\in[1/2,3/2]$ for every $c\in A$.  Since
$f(x):=x^{-\gamma}$ has
\[
  \sup_{1/2\le x\le 3/2}|f'(x)|
  =
  \sup_{1/2\le x\le 3/2}\gamma x^{-\gamma-1}
  \le C,
\]
the mean value theorem yields, uniformly in $c$,
\begin{equation}\label{eq:D_cov_transfer_denominator}
  \max_{c\in A}
  \left|
  \left(\frac{\nu_c}{Q_c}\right)^\gamma-1
  \right|
  =
  \max_{c\in A}|x_c^{-\gamma}-1|
  \le
  C\max_{c\in A}|x_c-1|
  \le C\Delta_Q.
\end{equation}
For any $x,y\in\R^n$, Lemma~\ref{D:lem:rowsum} gives
\begin{equation}\label{eq:D_cov_transfer_rowsum}
  \left|\sum_{t=1}^n\sum_{u=1}^n\Theta_n(t,u)x_ty_u\right|
  \le Cb_n\norm{x}_2\norm{y}_2.
\end{equation}
To obtain an entrywise bound for
$\widehat\Sigma_{ab}-\widehat\Sigma_{ab}^{\circ}$,
\eqref{eq:D_boot_cov_entry} and
\eqref{eq:D_oracle_centering_decomposition} give
\[
\begin{aligned}
  \widehat\Sigma_{ab}
  &=
  \frac{1}{|I|^{1/2-\gamma}|J|^{1/2-\gamma}Q_a^\gamma Q_b^\gamma}
  \sum_{t=1}^n\sum_{s=1}^n
  \Theta_n(t,s)\1\{t\in I\}(d_{t,a}+\mu_{t,a})
  \\
  &\hspace{7em}\times
    \1\{s\in J\}(d_{s,b}+\mu_{s,b}),
    \\
  \widehat\Sigma_{ab}^\circ
  &=
  \frac{1}{|I|^{1/2-\gamma}|J|^{1/2-\gamma}\nu_a^\gamma\nu_b^\gamma}
  \sum_{t=1}^n\sum_{s=1}^n\Theta_n(t,s)
  \left\{\1\{t\in I\}d_{t,a}
    +S_a^\circ w_{t,I}\right\}
  \\
  &\hspace{7em}\times
    \left\{\1\{s\in J\}d_{s,b}
      +S_b^\circ w_{s,J}\right\}.
\end{aligned}
\]
Define the same oracle-centered covariance with the feasible denominator:
\[
\begin{aligned}
  \widetilde{\widehat{\Sigma}}_{ab}^\circ
  &:=
  \frac{1}{|I|^{1/2-\gamma}|J|^{1/2-\gamma}Q_a^\gamma Q_b^\gamma}
  \sum_{t=1}^n\sum_{s=1}^n\Theta_n(t,s)
  \left\{\1\{t\in I\}d_{t,a}
    +S_a^\circ w_{t,I}\right\}
  \\
  &\hspace{7em}\times
    \left\{\1\{s\in J\}d_{s,b}
      +S_b^\circ w_{s,J}\right\}.
\end{aligned}
\]
Then
\[
  \widetilde{\widehat{\Sigma}}_{ab}^\circ
  =
  (Q_a/\nu_a)^{-\gamma}(Q_b/\nu_b)^{-\gamma}
  \widehat\Sigma_{ab}^{\circ}.
\]
On $\{\Delta_Q\le1/2\}$, \eqref{eq:D_cov_transfer_denominator} gives
\[
  \left|
  (Q_a/\nu_a)^{-\gamma}(Q_b/\nu_b)^{-\gamma}-1
  \right|
  \le C\Delta_Q .
\]
Also, \eqref{eq:D_centered_kernel_cov} and $\|\Sigma_{\circ}\|_{\max}\le C$
give $\|\widehat{\Sigma}^{\circ}\|_{\max}=O_p(1)$.  Hence, with probability
$1-o(1)$,
\[
  \max_{a,b\in A}
  \left|\widetilde{\widehat{\Sigma}}_{ab}^\circ
  -\widehat\Sigma_{ab}^{\circ}\right|
  \le
  C\Delta_Q,
\]
which controls the denominator replacement.

For the numerator difference
$\widehat\Sigma_{ab}-\widetilde{\widehat{\Sigma}}_{ab}^{\circ}$, the
$\1\{t\in I\}d_{t,a}\1\{s\in J\}d_{s,b}$ term cancels:
\[
\begin{aligned}
&\1\{t\in I\}(d_{t,a}+\mu_{t,a})
  \1\{s\in J\}(d_{s,b}+\mu_{s,b})
\\
&\quad -
\left\{\1\{t\in I\}d_{t,a}
  +S_a^\circ w_{t,I}\right\}
\left\{\1\{s\in J\}d_{s,b}
  +S_b^\circ w_{s,J}\right\}
\\
&=
\1\{t\in I\}\1\{s\in J\}
\left(\mu_{t,a}d_{s,b}+d_{t,a}\mu_{s,b}
  +\mu_{t,a}\mu_{s,b}\right)
\\
&\quad
-\1\{t\in I\}d_{t,a}
  S_b^\circ w_{s,J}
-S_a^\circ w_{t,I}\1\{s\in J\}d_{s,b}
\\
&\quad
-S_a^\circ S_b^\circ w_{t,I}w_{s,J}.
\end{aligned}
\]
Equations \eqref{eq:D_cov_transfer_norms} and
\eqref{eq:D_cov_transfer_denominator}, and the definition \eqref{eq:rn},
give, on $\{\Delta_Q\le1/2\}$,
\[
  \frac{\left\|\bigl(\1\{t\in I\}d_{t,a}\bigr)_{t=1}^n\right\|_2}
       {\sqrt{\nu_a}}
  \le
  \left\{2\frac{Q_a}{\nu_a}
  +2\frac{\nu_{a,\mu}}{\nu_a}\right\}^{1/2}
  \le C,
  \qquad
  \frac{\left\|\bigl(\1\{t\in I\}\mu_{t,a}\bigr)_{t=1}^n\right\|_2}
       {\sqrt{\nu_a}}
  =
  \sqrt{\frac{\nu_{a,\mu}}{\nu_a}},
\]
\[
  \frac{
  \left\|\bigl(S_a^\circ w_{t,I}\bigr)_{t=1}^n\right\|_2}
  {\sqrt{\nu_a}}
  \le
  \frac{|S_a^\circ|}{\sqrt{|I|\nu_a}}
  \le
  \frac{|S_a^\circ|}{\sqrt{\nu_a}}.
\]
By \eqref{eq:D_cov_transfer_rowsum}, uniformly in $a=(I,h)$ and $b=(J,g)$,
\begin{equation}\label{eq:D_cov_transfer_bound}
\begin{split}
\left|\widehat\Sigma_{ab}
-\widetilde{\widehat{\Sigma}}^\circ_{ab}\right|
\le{}& C\frac{b_n}{\sqrt{|I||J|}}
\Biggl[
\frac{|S_a^\circ|}{\sqrt{\nu_a}}
+\frac{|S_b^\circ|}{\sqrt{\nu_b}}
+\sqrt{\frac{\nu_{a,\mu}}{\nu_a}}
+\sqrt{\frac{\nu_{b,\mu}}{\nu_b}}
\\
&\hspace{9em}
+\left(\frac{|S_a^\circ|}{\sqrt{\nu_a}}
  +\sqrt{\frac{\nu_{a,\mu}}{\nu_a}}\right)
 \left(\frac{|S_b^\circ|}{\sqrt{\nu_b}}
  +\sqrt{\frac{\nu_{b,\mu}}{\nu_b}}\right)
\Biggr]
\end{split}
\end{equation}
By Assumption~\ref{ass:bounded_second_moment}, \eqref{eq:mathbbT_def}, and
\eqref{eq:rn},
\[
  \max_{a\in A}\frac{|S_a^\circ|}{\sqrt{\nu_a}}
  =
  \max_{a=(I,h)\in A}
  \left(\frac{|I|}{\nu_a}\right)^{1/2-\gamma}
  |\mathbb T_a^\circ|
  \le C\bar{\mathbb T}_{\max}^\circ,
  \qquad
  \max_{a\in A}\sqrt{\frac{\nu_{a,\mu}}{\nu_a}}\le\sqrt{r_n}.
\]
Since $|I|,|J|\ge L_{\min}$, \eqref{eq:D_cov_transfer_bound} yields
\[
  \max_{a,b\in A}
  \left|\widehat\Sigma_{ab}
  -\widetilde{\widehat{\Sigma}}^\circ_{ab}\right|
  \le
  C\frac{b_n}{L_{\min}}
  \left\{
    \bar{\mathbb T}_{\max}^\circ+\sqrt{r_n}
    +(\bar{\mathbb T}_{\max}^\circ+\sqrt{r_n})^2
  \right\}.
\]
The triangle inequality and the denominator-replacement bound give
\[
\begin{aligned}
  \|\widehat\Sigma-\widehat{\Sigma}^{\circ}\|_{\max}
  &\le
  \max_{a,b\in A}
  \left|\widehat\Sigma_{ab}
    -\widetilde{\widehat{\Sigma}}_{ab}^{\circ}\right|
  +
  \max_{a,b\in A}
  \left|\widetilde{\widehat{\Sigma}}_{ab}^{\circ}
    -\widehat\Sigma_{ab}^{\circ}\right|
  \\
  &\le
  C\Delta_Q
  +
  C\frac{b_n}{L_{\min}}
  \left\{
    \bar{\mathbb T}_{\max}^\circ+\sqrt{r_n}
    +(\bar{\mathbb T}_{\max}^\circ+\sqrt{r_n})^2
  \right\}.
\end{aligned}
\]
Lemma~\ref{D:lem:Delta}, Lemma~\ref{D:lem:Uw}, and
Assumptions~\ref{ass:D_mean_flatness} and \ref{ass:D_Lmin} imply
\[
\Delta_Q=o_p((\log n)^{-2}),
\qquad
\frac{b_n}{L_{\min}}
\left\{
\bar{\mathbb T}_{\max}^\circ+\sqrt{r_n}
+(\bar{\mathbb T}_{\max}^\circ+\sqrt{r_n})^2
\right\}
=o_p((\log n)^{-2}).
\]
Therefore
\begin{equation}\label{eq:D_cov_transfer}
  \|\widehat\Sigma-\widehat{\Sigma}^\circ\|_{\max}
  =o_p((\log n)^{-2}).
\end{equation}
Combining \eqref{eq:D_centered_kernel_cov} and
\eqref{eq:D_cov_transfer} proves \eqref{eq:roadmap_bootstrap_covariance}.

\subsubsection{Verifying conditions of the reduction lemmas}\label{subsub:D_verify_reduction}

We first verify Lemma~\ref{lem:centered_oracle_perturbation}.  For
$\gamma=0$, $T_{\max}^\circ=\mathbb T_{\max}^\circ$.  For $\gamma>0$,
Lemmas~\ref{D:lem:Delta} and \ref{D:lem:Uw} imply
\[
  C\gamma\Delta_Q\bar{\mathbb{T}}_{\max}^{\circ}
  =
  O_p\!\left(
  \left\{
    \sqrt{\frac{\log n}{L_{\min}}}
    +\sqrt{r_n}
  \right\}\sqrt{\log n}
  \right).
\]
Moreover,
\[
  \left\{
    \sqrt{\frac{\log n}{L_{\min}}}
    +\sqrt{r_n}
  \right\}\log n
  \lesssim
  \frac{(\log n)^{3/2}}{\sqrt{L_{\min}}}
  +\sqrt{r_n}\log n
  \to0,
\]
where Assumption~\ref{ass:D_Lmin} controls the first term, while
$b_nr_n(\log n)^2\to0$ and $b_n\to\infty$ control the second term.  Choose
$M_n\uparrow\infty$ sufficiently slowly that
\[
  \eta_n:=
  M_n
  \left\{
    \sqrt{\frac{\log n}{L_{\min}}}
    +\sqrt{r_n}
  \right\}\sqrt{\log n}
  \to0,
  \qquad
  \eta_n\sqrt{1\vee\log(p/\eta_n)}\to0.
\]
Then
\begin{equation}\label{eq:D_oracle_perturb_prob}
  \Pbb(\Delta_Q>1/2)
  +
  \Pbb\bigl(C\gamma\Delta_Q\bar{\mathbb{T}}_{\max}^{\circ}>\eta_n\bigr)
  \to0.
\end{equation}
Since $\Sigma_{\circ}=\Cov(\mathbb{T}^{\circ})$ and
$\mathbb{T}^{\circ}$ is defined in \eqref{eq:mathbbT_def}, for $a=(I,h)$,
\[
  (\Sigma_{\circ})_{aa}
  =
  \Var(\mathbb{T}_a^{\circ})
  =
  \left(\frac{\nu_a}{|I|}\right)^{1-2\gamma}
  \Var\!\left(\frac{S_a^\circ}{\sqrt{\nu_a}}\right).
\]
Assumptions~\ref{ass:bounded_second_moment} and
\ref{ass:D_nondegenerate} therefore give constants $0<c<C<\infty$ such
that
\[
  c\le(\Sigma_{\circ})_{aa}\le C
  \qquad\forall a\in A.
\]
Proposition~\ref{prop:gaussian_anti_concentration} gives
\begin{equation}\label{eq:D_oracle_anti_conc}
  \sup_{x\in\R}\Pbb\bigl(|G_{\circ,\max}-x|\le\eta_n\bigr)
  \to0.
\end{equation}
For every $x\in\R$, \eqref{eq:D_oracle_U_cdf} gives
\[
  \Pbb\bigl(|\mathbb T_{\max}^{\circ}-x|\le\eta_n\bigr)
  \le
  \Pbb\bigl(|G_{\circ,\max}-x|\le\eta_n\bigr)
  +2\sup_{z\in\R}
  \left|
  \Pbb\bigl(\mathbb T_{\max}^{\circ}\le z\bigr)
  -\Pbb\bigl(G_{\circ,\max}\le z\bigr)
  \right|.
\]
Equations~\eqref{eq:D_oracle_U_cdf} and \eqref{eq:D_oracle_anti_conc}
therefore imply
\[
  \sup_{x\in\R}
  \Pbb\bigl(|\mathbb T_{\max}^{\circ}-x|\le\eta_n\bigr)
  \to0.
\]
Together with \eqref{eq:D_oracle_perturb_prob}, this verifies the
conditions of Lemma~\ref{lem:centered_oracle_perturbation}.

We next verify Lemma~\ref{lem:random_denominator} with $k_n=b_n$.
Lemma~\ref{D:lem:Delta} and
$\nu_a\ge\sigma_-^2|I|$ verify \eqref{eq:RD_Qlower}, and
Lemma~\ref{D:lem:fourth} verifies \eqref{eq:RD_fourth}.  The row conditions
\eqref{eq:RD_rows} follow from Lemma~\ref{D:lem:rowsum}.  The conditional
variance condition \eqref{eq:RD_var} follows from
\eqref{eq:roadmap_bootstrap_covariance}.  The diagonal bound established
above gives constants $0<c_0<C_0<\infty$ such that
\[
  c_0\le \min_{a\in A}(\Sigma_\circ)_{aa}
  \le \max_{a\in A}(\Sigma_\circ)_{aa}\le C_0.
\]
Also, by definition of $\widehat\Sigma$ in
Lemma~\ref{lem:fixed_bootstrap_covariance},
\[
  \widehat\Sigma_{aa}
  =
  \Cov^*(T_{0,a}^*,T_{0,a}^*)
  =
  \Var^*(T_{0,a}^*).
\]
Equation~\eqref{eq:roadmap_bootstrap_covariance} implies
\[
  \Pbb\!\left(
    \|\widehat\Sigma-\Sigma_\circ\|_{\max}\le c_0/2
  \right)\to1.
\]
On this event,
\[
  \min_{a\in A}\Var^*(T_{0,a}^*)
  =
  \min_{a\in A}\widehat\Sigma_{aa}
  \ge
  \min_{a\in A}(\Sigma_\circ)_{aa}
  -\|\widehat\Sigma-\Sigma_\circ\|_{\max}
  \ge c_0/2,
\]
and
\[
  \max_{a\in A}\Var^*(T_{0,a}^*)
  =
  \max_{a\in A}\widehat\Sigma_{aa}
  \le
  \max_{a\in A}(\Sigma_\circ)_{aa}
  +\|\widehat\Sigma-\Sigma_\circ\|_{\max}
  \le C_0+c_0/2.
\]
Renaming constants gives, with probability $1-o(1)$ over data,
\[
  c\le \min_{a\in A}\Var^*(T_{0,a}^*)
  \le \max_{a\in A}\Var^*(T_{0,a}^*)\le C.
\]
Since Lemma~\ref{D:lem:Delta} and Assumption~\ref{ass:bounded_second_moment}
make $Q_a/|I|$ uniformly bounded above and away from zero with probability
$1-o(1)$ over data,
\[
  \Var^*\!\left(\frac{S_a^*}{\sqrt{Q_a}}\right)
  =
  \left(\frac{|I|}{Q_a}\right)^{1-2\gamma}
  \Var^*(T_{0,a}^*)
\]
is also uniformly bounded above and away from zero.  Thus
\eqref{eq:RD_var} holds.  Assumptions~\ref{ass:D_mean_flatness} and
\ref{ass:D_Lmin} imply
\[
  \frac{b_n(\log n)^3}{L_{\min}}\to0,
  \qquad
  b_nr_n(\log n)^2\to0.
\]
Thus \eqref{eq:RD_growth} holds.

It remains to verify Lemma~\ref{lem:fixed_bootstrap_covariance}.  By
the diagonal-bound argument used to verify \eqref{eq:RD_var}, the diagonal
entries of $\Sigma_\circ$ and $\widehat\Sigma$ satisfy the two
nondegeneracy conditions in Lemma~\ref{lem:fixed_bootstrap_covariance}.
Together with \eqref{eq:roadmap_bootstrap_covariance}, this verifies all
conditions of
Lemma~\ref{lem:fixed_bootstrap_covariance}.

We now complete the proof of Theorem~\ref{D:thm:main}.
Section~\ref{subsub:D_centered_gaussian} proves
\eqref{eq:roadmap_centered_gaussian},
Section~\ref{subsub:D_bootstrap_covariance} proves
\eqref{eq:roadmap_bootstrap_covariance}, and the preceding paragraphs
verify the conditions of Lemmas~\ref{lem:centered_oracle_perturbation}--\ref{lem:fixed_bootstrap_covariance}.
Therefore the conditions of Lemma~\ref{lem:size_transfer} hold.  Applying
Lemma~\ref{lem:size_transfer} proves the theorem.

\subsubsection{Quadratic spectral kernel satisfies Assumption \ref{ass:D_kernel}}\label{subsec:QS_kernel}
\begin{lemma}\label{D:lem:QS_kernel}
Let
\[
K_{\mathrm{QS}}(x)
:=
\begin{cases}
\dfrac{25}{12\pi^2x^2}\left\{\dfrac{\sin(6\pi x/5)}{6\pi x/5}-\cos(6\pi x/5)\right\}, & x\neq 0,\\[1ex]
1, & x=0.
\end{cases}
\]
Then $K_{\mathrm{QS}}$ is symmetric and continuously differentiable on $\R$, satisfies $K_{\mathrm{QS}}(0)=1$, and obeys
\[
\sup_{x\in\R}|K_{\mathrm{QS}}'(x)|<\infty,
\qquad
|K_{\mathrm{QS}}(x)|\le C_{\mathrm{QS}}(1+|x|)^{-2}
\qquad\forall x\in\R
\]
for some constant $C_{\mathrm{QS}}<\infty$. Consequently, Assumption \ref{ass:D_kernel} holds for the quadratic spectral kernel with $\vartheta=2$ whenever
\[
b_n\asymp n^{\rho},
\qquad
0<\rho<\frac14.
\]
\end{lemma}

\begin{proof}
Symmetry is immediate because $K_{\mathrm{QS}}$ depends on $x$ only through $x^2$ and the even functions $\sin(z)/z$ and $\cos(z)$. For $x\neq 0$, the formula is smooth.
Set $c:=6\pi/5$. Using the Taylor expansions
\[
\frac{\sin(cx)}{cx}=1-\frac{c^2x^2}{6}+\frac{c^4x^4}{120}+O(x^6),
\qquad
\cos(cx)=1-\frac{c^2x^2}{2}+\frac{c^4x^4}{24}+O(x^6),
\]
we obtain
\[
\frac{\sin(cx)}{cx}-\cos(cx)
=
\frac{c^2x^2}{3}-\frac{c^4x^4}{30}+O(x^6).
\]
Therefore,
\[
K_{\mathrm{QS}}(x)
=
1-\frac{18\pi^2}{125}x^2+O(x^4)
\qquad (x\to 0).
\]
Hence $K_{\mathrm{QS}}$ extends continuously at $0$ with value $1$, and the expansion also gives $K_{\mathrm{QS}}'(0)=0$; thus $K_{\mathrm{QS}}\in C^1(\R)$.

The positive-semidefiniteness condition follows from an explicit Fourier
representation.  With $c=6\pi/5$, define
\[
  f_c(u):=\frac{3}{4c}\left(1-\frac{u^2}{c^2}\right)
  \1\{|u|\le c\}.
\]
This is a nonnegative density, and direct integration gives
\begin{equation}\label{eq:QS_fourier}
  K_{\mathrm{QS}}(x)
  =\int_{-c}^c e^{iux}f_c(u)\,du
  =\frac{3}{c^2x^2}
  \left\{\frac{\sin(cx)}{cx}-\cos(cx)\right\}.
\end{equation}
Because $f_c$ has compact support, differentiation under the integral is
valid and gives
\[
  K_{\mathrm{QS}}'(x)
  =\int_{-c}^c iu e^{iux}f_c(u)\,du,
  \qquad
  \sup_{x\in\R}|K_{\mathrm{QS}}'(x)|
  \le\int_{-c}^c|u|f_c(u)\,du<\infty.
\]
Therefore, for arbitrary real numbers $v_1,\ldots,v_n$,
\[
  \sum_{t,u=1}^n v_tv_u
  K_{\mathrm{QS}}\!\left(\frac{t-u}{b_n}\right)
  =\int_{-c}^c
  \left|\sum_{t=1}^n v_te^{iut/b_n}\right|^2f_c(u)\,du
  \ge0.
\]
Thus every matrix in \eqref{eq:D_Theta} is positive semidefinite.

For the tail bound, again with $c=6\pi/5$, for $x\neq 0$ we have
\[
|K_{\mathrm{QS}}(x)|
\le
\frac{25}{12\pi^2x^2}\left(\left|\frac{\sin(cx)}{cx}\right|+|\cos(cx)|\right)
\le
\frac{25}{12\pi^2x^2}\left(1+\frac{1}{c|x|}\right).
\]
In particular, for $|x|\ge 1$,
\[
|K_{\mathrm{QS}}(x)|
\le
\frac{25}{12\pi^2}\left(1+\frac{1}{c}\right)|x|^{-2}
\le
4\frac{25}{12\pi^2}\left(1+\frac{1}{c}\right)(1+|x|)^{-2}.
\]
On the compact set $[-1,1]$, continuity implies $\sup_{|x|\le 1}|K_{\mathrm{QS}}(x)|<\infty$. Enlarging the constant if necessary yields
\[
|K_{\mathrm{QS}}(x)|\le C_{\mathrm{QS}}(1+|x|)^{-2}
\qquad\forall x\in\R.
\]
Thus Assumption \ref{ass:D_kernel} is satisfied with $\vartheta=2$. The bandwidth restriction in that assumption then becomes
\[
0<\rho<\frac{\vartheta-1}{3\vartheta-2}=\frac{1}{4}.
\]
\end{proof}

\subsubsection{Extension of Theorem~\ref{D:thm:main} to data-dependent bandwidths}
\label{subsec:D_data_dependent_bandwidth}

This subsection shows that the test remains valid in Regime D with the
data-driven bandwidth used in the main text.
The implementation uses the quadratic spectral kernel of
\citet{Andrews1991}.  For each horizon $h$, it
fits
\[
  D_t^{(h)}=c_h+\rho_hD_{t-1}^{(h)}+\epsilon_{h,t},
\]
forms
\begin{equation}\label{eq:D_data_driven_stepsize}
  \widehat a
  :=
  \frac{
    \sum_{h=1}^H
    4\widehat\rho_h^2\widehat\sigma_h^4(1-\widehat\rho_h)^{-8}
  }{
    \sum_{h=1}^H
    \widehat\sigma_h^4(1-\widehat\rho_h)^{-4}
  },
  \qquad
  \widehat b_n:=1.3221(\widehat a n)^{1/5}.
\end{equation}
The resulting multiplier covariance is
\[
  \widehat\Theta_n(t,u)
  :=
  K_{\mathrm{QS}}\!\left(\frac{t-u}{\widehat b_n}\right).
\]

The next corollary shows that Theorem~\ref{D:thm:main} continues to work for
the data-driven stepsize in \eqref{eq:D_data_driven_stepsize} under mild
convergence of $\widehat a$.
\begin{corollary}\label{cor:D_data_dependent_bandwidth_hat_a}
Suppose that Assumptions~\ref{ass:subgaussian},
\ref{ass:bounded_second_moment}, \ref{ass:length_ratio},
\ref{ass:D_mixing}--\ref{ass:D_nondegenerate},
and \ref{ass:D_growth}--\ref{ass:D_Lmin} hold with $b_n$ replaced by
$b_n^\circ=1.3221(a_0n)^{1/5}$ for some constant $a_0>0$.  If, for some
constant $\epsilon>0$,
\[
  \widehat a-a_0=o_p\bigl(n^{-1/5-\epsilon}\bigr),
\]
then the conclusion of Theorem~\ref{D:thm:main} remains valid when the
bootstrap statistic is computed using the Andrews bandwidth
$\widehat b_n=1.3221(\widehat a n)^{1/5}$.
\end{corollary}

\begin{proof}
Let
\[
  \Theta_n^\circ(t,u)
  :=
  K_{\mathrm{QS}}\!\left(\frac{t-u}{b_n^\circ}\right).
\]
The condition on $\widehat a$ gives
\[
  \frac{\widehat b_n}{b_n^\circ}
  =
  \left(\frac{\widehat a}{a_0}\right)^{1/5}
  =
  1+o_p\bigl(n^{-1/5-\epsilon}\bigr),
  \qquad
  |\widehat b_n-b_n^\circ|=o_p(n^{-\epsilon}),
\]
and therefore, with probability $1-o(1)$ over data,
\[
  c b_n^\circ\le\widehat b_n\le C b_n^\circ.
\]
The explicit quadratic spectral formula and the expansion at zero in the
proof of Lemma~\ref{D:lem:QS_kernel} give
\[
  |K_{\mathrm{QS}}'(x)|\le C(1+|x|)^{-2}
  \qquad\forall x\in\R.
\]
On the event $c b_n^\circ\le\widehat b_n\le Cb_n^\circ$, the mean value
theorem gives
\[
\begin{aligned}
  &\max_{1\le t\le n}
  \sum_{u=1}^n
  \left|
    K_{\mathrm{QS}}\!\left(\frac{t-u}{\widehat b_n}\right)
    -
    K_{\mathrm{QS}}\!\left(\frac{t-u}{b_n^\circ}\right)
  \right|
  \\
  &\le
  C\left|\frac1{\widehat b_n}-\frac1{b_n^\circ}\right|
  \sum_{m=1}^{n-1}m\left(1+\frac{m}{b_n^\circ}\right)^{-2}
  \\
  &\le
  C|\widehat b_n-b_n^\circ|\log n
  =
  o_p\bigl((\log n)^{-2}\bigr),
\end{aligned}
\]
because $|\widehat b_n-b_n^\circ|\log n=o_p((\log n)^{-2})$.

For any positive bandwidth $b$, the Fourier representation
\eqref{eq:QS_fourier} gives positive semidefiniteness of the matrix
$K_{\mathrm{QS}}((t-u)/b)$, and $K_{\mathrm{QS}}(0)=1$ gives unit diagonal.
Thus $\widehat\Theta_n$ is a valid covariance matrix for the multipliers.

The row-sum argument in Lemma~\ref{D:lem:rowsum} gives, with probability
$1-o(1)$ over data,
\begin{equation}\label{eq:D_hatTheta_row_bounds}
  \max_t\sum_u|\widehat\Theta_n(t,u)|\le C\widehat b_n\le Cb_n^\circ,
  \qquad
  \max_t\sum_u\widehat\Theta_n(t,u)^2\le C\widehat b_n\le Cb_n^\circ.
\end{equation}

It remains to verify the fixed-denominator covariance condition
\eqref{eq:fixed_denominator_covariance_condition} in
Lemma~\ref{lem:fixed_bootstrap_covariance}; by \eqref{eq:logpn_logn}, in
Regime~D it is enough to show
$\|\widehat\Sigma(\widehat\Theta_n)-\Sigma_\circ\|_{\max}
=o_p((\log n)^{-2})$.
For any covariance matrix $\Theta$, define
\[
  \widehat\Sigma_{ab}(\Theta)
  :=
  \frac{\sum_{t\in I_a}\sum_{u\in I_b}
    \Theta(t,u)
    \left(D_t^{(h_a)}-\bar D_a\right)
    \left(D_u^{(h_b)}-\bar D_b\right)}
  {|I_a|^{1/2-\gamma}|I_b|^{1/2-\gamma}Q_a^\gamma Q_b^\gamma}.
\]
Let
\[
  v_{a,t}
  :=
  \frac{\1\{t\in I_a\}\left(D_t^{(h_a)}-\bar D_a\right)}
  {|I_a|^{1/2-\gamma}Q_a^\gamma}.
\]
Then
\[
  \widehat\Sigma_{ab}(\widehat\Theta_n)
  -
  \widehat\Sigma_{ab}(\Theta_n^\circ)
  =
  v_a^\top(\widehat\Theta_n-\Theta_n^\circ)v_b.
\]
Moreover,
\[
  \norm{v_a}_2
  =
  \left(\frac{Q_a}{|I_a|}\right)^{1/2-\gamma}.
\]
Lemma~\ref{D:lem:Delta} and Assumption~\ref{ass:bounded_second_moment}
imply $\max_a\norm{v_a}_2=O_p(1)$.  Since
$\widehat\Theta_n-\Theta_n^\circ$ is symmetric, Schur's bound gives
\[
  \norm{\widehat\Theta_n-\Theta_n^\circ}_{\mathrm{op}}
  \le
  \max_{1\le t\le n}\sum_{u=1}^n
  \left|\widehat\Theta_n(t,u)-\Theta_n^\circ(t,u)\right|
  =
  o_p\bigl((\log n)^{-2}\bigr).
\]
Hence
\[
\begin{aligned}
  \max_{a,b\in A}
  \left|
    \widehat\Sigma_{ab}(\widehat\Theta_n)
    -
    \widehat\Sigma_{ab}(\Theta_n^\circ)
  \right|
  &=
  \max_{a,b\in A}
  \left|v_a^\top(\widehat\Theta_n-\Theta_n^\circ)v_b\right|
  \\
  &\le
  \left(\max_{a\in A}\norm{v_a}_2\right)^2
  \norm{\widehat\Theta_n-\Theta_n^\circ}_{\mathrm{op}}
  \\
  &=
  O_p(1)o_p\bigl((\log n)^{-2}\bigr)
  =
  o_p\bigl((\log n)^{-2}\bigr).
\end{aligned}
\]
The proof in Section~\ref{subsub:D_bootstrap_covariance} applied with the deterministic benchmark
$\Theta_n^\circ$ gives
\[
  \norm{\widehat\Sigma(\Theta_n^\circ)-\Sigma_\circ}_{\max}
  =
  o_p\bigl((\log n)^{-2}\bigr).
\]
Combining the bounds for
$\norm{\widehat\Sigma(\widehat\Theta_n)-\widehat\Sigma(\Theta_n^\circ)}_{\max}$
and
$\norm{\widehat\Sigma(\Theta_n^\circ)-\Sigma_\circ}_{\max}$ yields
\begin{equation}\label{eq:D_data_bw_covariance}
  \norm{\widehat\Sigma(\widehat\Theta_n)-\Sigma_\circ}_{\max}
  =
  o_p\bigl((\log n)^{-2}\bigr).
\end{equation}
The centered Gaussian approximation requirement
\eqref{eq:roadmap_centered_gaussian} is unchanged because the statistic
defined in \eqref{eq:mathbbT_def} has coordinates
\[
  \mathbb T_a^\circ
  =
  \frac{S_a^\circ}{|I|^{1/2-\gamma}\nu_a^\gamma},
  \qquad
  \Sigma_\circ=\Cov(\mathbb T^\circ),
\]
contains no $\Theta_n$.  Hence the proof of
\eqref{eq:roadmap_centered_gaussian} in
Section~\ref{subsub:D_centered_gaussian} applies without
change.

For Lemma~\ref{lem:random_denominator} with $k_n=b_n^\circ$,
\eqref{eq:RD_Qlower}, \eqref{eq:RD_fourth}, and \eqref{eq:RD_growth}
are the Section~\ref{subsub:D_verify_reduction} bounds with
$b_n=b_n^\circ$, and
\eqref{eq:RD_rows} follows from \eqref{eq:D_hatTheta_row_bounds}.
For \eqref{eq:RD_var}, the diagonal entries of $\Sigma_\circ$ are bounded
above and away from zero, and \eqref{eq:D_data_bw_covariance} gives, with
probability $1-o(1)$ over data,
\[
  c\le
  \min_{a\in A}\widehat\Sigma_{aa}(\widehat\Theta_n)
  \le
  \max_{a\in A}\widehat\Sigma_{aa}(\widehat\Theta_n)
  \le C.
\]
Since $\widehat\Sigma_{aa}(\widehat\Theta_n)=\Var^*(T_{0,a}^*)$,
$Q_a/|I|\asymp1$ with probability $1-o(1)$ over data, and
\[
  \Var^*\!\left(\frac{S_a^*}{\sqrt{Q_a}}\right)
  =
  \left(\frac{|I|}{Q_a}\right)^{1-2\gamma}
  \Var^*(T_{0,a}^*),
\]
\eqref{eq:RD_var} follows.  Finally, \eqref{eq:D_data_bw_covariance}
verifies \eqref{eq:fixed_denominator_covariance_condition} with
$\widehat\Sigma=\widehat\Sigma(\widehat\Theta_n)$.
Therefore the verification in Section~\ref{subsub:D_verify_reduction}
and Lemma~\ref{lem:size_transfer}
apply with $\widehat\Theta_n$ in place of $\Theta_n^\circ$, proving the
same size conclusion as Theorem~\ref{D:thm:main}.
\end{proof}

\subsection{Miscellany}\label{app:miscellany}

This section states the key technical results our proof relies on for completeness.

The next proposition is the independent high-dimensional Gaussian
approximation of
\citet[Theorem~2.1 and Proposition~2.1]{ChernozhukovChetverikovKato2017}.
\begin{proposition}\label{prop:max_gaussian_CLT}
Let $X_1,\dots,X_n$ be independent centered random vectors in $\R^p$, and let
\[
S_n:=\frac{1}{\sqrt n}\sum_{i=1}^n X_i,
\qquad
\Sigma:=\Cov(S_n),
\qquad
G\sim N(0,\Sigma).
\]
Assume that, for some constants $b>0$ and $B_n\ge 1$,
\[
\min_{1\le j\le p}\Sigma_{jj}\ge b,
\qquad
\max_{1\le i\le n,\,1\le j\le p}\|X_{ij}\|_{\psi_1}\le B_n.
\]
Then there exists a constant $C=C(b)<\infty$ such that
\[
\sup_{A\in\mathcal R^p}
\bigl|\Pbb(S_n\in A)-\Pbb(G\in A)\bigr|
\le
C\left(\frac{B_n^2\log^7(pn)}{n}\right)^{1/6},
\]
where $\mathcal R^p$ denotes the class of coordinatewise hyperrectangles in $\R^p$. In particular, if $B_n^2\log^7(pn)=o(n)$, then the left-hand side converges to $0$.
\end{proposition}

The next proposition is the Gaussian anti-concentration bound of
\citet[Theorem~3 and Corollary~1]{CCK15}.
\begin{proposition}\label{prop:gaussian_anti_concentration}
Let $G=(G_1,\dots,G_p)^\top$ be a centered Gaussian vector such that
\[
0<\underline\sigma\le \sqrt{\Var(G_j)}\le \bar\sigma<\infty
\qquad\forall j\in[p].
\]
Then there exists a constant $C<\infty$, depending only on $\underline\sigma$ and $\bar\sigma$, such that for every $\varepsilon>0$,
\[
\sup_{x\in\R}\Pbb\left(|\max_{1\le j\le p}G_j-x|\le \varepsilon\right)
\le
C\varepsilon\left\{1\vee \sqrt{\log(p/\varepsilon)}\right\}.
\]
In particular, if $\varepsilon_n\downarrow 0$ and $\varepsilon_n\sqrt{1\vee\log(p/\varepsilon_n)}\to 0$, then
\[
\sup_{x\in\R}\Pbb\left(|\max_{1\le j\le p}G_j-x|\le \varepsilon_n\right)\to 0.
\]
The same conclusion holds with $\max_{1\le j\le p}|G_j|$ in place of $\max_{1\le j\le p}G_j$ after replacing $p$ by $2p$.
\end{proposition}

The next proposition is the Gaussian comparison inequality from
\citet[Theorem~2]{CCK15}.
\begin{proposition}\label{prop:gaussian_comparison}
Let $X=(X_1,\dots,X_p)^\top$ and $Y=(Y_1,\dots,Y_p)^\top$ be centered Gaussian vectors with covariance matrices $\Sigma^X$ and $\Sigma^Y$. Assume that
\[
0<\underline\sigma
\le \sqrt{\Var(X_j)},\sqrt{\Var(Y_j)}
\le \bar\sigma<\infty
\qquad\forall j\in[p],
\]
and define
\[
\Delta:=\|\Sigma^X-\Sigma^Y\|_{\max}.
\]
Then there exists a constant $C<\infty$, depending only on $\underline\sigma$ and $\bar\sigma$, such that
\[
\sup_{x\in\R}
\left|\Pbb\left(\max_{1\le j\le p}X_j\le x\right)-\Pbb\left(\max_{1\le j\le p}Y_j\le x\right)\right|
\le
C\Delta^{1/3}\bigl\{1\vee \log(p/\Delta)\bigr\}^{2/3},
\]
with the convention that the right-hand side equals $0$ when $\Delta=0$.
In particular, if $\Delta_n=o((\log p)^{-2})$, then the Kolmogorov distance between the two Gaussian maxima converges to $0$.
\end{proposition}

The next proposition gives the Gaussian quadratic-form concentration bound
from \citet[Theorem~6.2.1]{vershynin2018high}.
\begin{proposition}
\label{prop:gaussian_quadratic}
Let $Z\sim N(0,I_m)$ and let $B$ be a deterministic symmetric $m\times m$
matrix.  There is a universal constant $C<\infty$ such that, for every
$u\ge1$,
\[
  \Pbb\left(
    |Z^\top BZ-\operatorname{tr}(B)|
    >C\{\|B\|_F\sqrt u+\|B\|_{\mathrm{op}}u\}
  \right)
  \le2e^{-u}.
\]
\end{proposition}

The next proposition is the dependent high-dimensional CLT from
\citet[Theorem~1]{CCW24}.
\begin{proposition}\label{prop:dependent_CLT}
Let $X_1,\dots,X_n$ be centered random vectors in $\R^p$, and let
\[
S_n:=\frac{1}{\sqrt n}\sum_{t=1}^n X_t,
\qquad
\Xi:=\Cov(S_n),
\qquad
G\sim N(0,\Xi).
\]
Assume that $\{X_t\}$ is geometrically $\alpha$-mixing with
\[
\alpha(\ell)\le K_1e^{-K_2\ell^{\gamma_{\mathrm{mix}}}}
\qquad(\ell\ge 1)
\]
for some constants $K_1>1$, $K_2>0$, and $\gamma_{\mathrm{mix}}>0$, that
$p=\Omega(n^\kappa)$ for some $\kappa>0$, and that
\[
\max_{t\in[n],\,j\in[p]}\|X_{t,j}\|_{\psi_2}\le B_n,
\qquad
\min_{1\le j\le p}\Xi_{jj}\ge b>0.
\]
If, in addition,
\[
(\log p)^{3-\gamma_{\mathrm{mix}}}=o\bigl(n^{\gamma_{\mathrm{mix}}/3}\bigr)
\]
and
\[
B_n^{2/3}\frac{(\log p)^{(1+2\gamma_{\mathrm{mix}})/(3\gamma_{\mathrm{mix}})}}{n^{1/9}}
+
B_n\frac{(\log p)^{7/6}}{n^{1/9}}
\to 0,
\]
then
\[
\sup_{B\in\mathcal R^p}
\bigl|\Pbb(S_n\in B)-\Pbb(G\in B)\bigr|\to 0.
\]
In particular,
\[
\sup_{x\in\R}\left|\Pbb\left(\max_{1\le j\le p}S_{n,j}\le x\right)-\Pbb\left(\max_{1\le j\le p}G_j\le x\right)\right|\to 0.
\]
\end{proposition}

The next proposition gives the kernel covariance consistency bound from
\citet[Theorem~11(i)]{CCW24}, specialized to $\gamma_1=2$.
\begin{proposition}\label{prop:consistency_covariance_matrix}
Let $X_1,\ldots,X_n$ be centered random vectors in $\R^p$.  Suppose
\[
  \max_{t,j}\|X_{t,j}\|_{\psi_{\gamma_1}}\le B_n
\]
for some $\gamma_1\ge1$, and suppose their strong-mixing coefficients obey
\[
  \alpha_X(\ell)\le K_1e^{-K_2\ell^{\gamma_2}},
  \qquad \ell\ge1,
\]
for constants $K_1,K_2,\gamma_2>0$.  Let $K$ satisfy the smoothness and
decay conditions in Assumption~\ref{ass:D_kernel}, and define
\[
\bar X:=\frac1n\sum_{t=1}^n X_t,
\quad
\widehat\Xi_n:=\sum_{j=-n+1}^{n-1}K\!\left(\frac{j}{b_n}\right)\widehat H_j,\]
where
\[
\widehat H_j:=\frac1n\sum_{t=j+1}^n (X_t-\bar X)(X_{t-j}-\bar X)^\top
\quad\text{if } j\ge 0,
\]
\[
\widehat H_j:=\frac1n\sum_{t=-j+1}^n (X_{t+j}-\bar X)(X_t-\bar X)^\top
\quad\text{if } j<0,
\]
and $\Xi:=\Cov\!\bigl(n^{-1/2}\sum_{t=1}^n X_t\bigr)$. Assume
$p>n^\kappa$ for some constant $\kappa>0$ and
$b_n\asymp n^{\rho}$. Then there exist constants $c_1>0$ depending only
on $(\rho,\vartheta)$ and $c_2>0$ depending only on
$(\gamma_1,\gamma_2,\vartheta)$ such that
\[
\|\widehat\Xi_n-\Xi\|_{\max}
=\Op\!\left\{B_n^2 n^{-c_1}(\log p)^{c_2}\right\}+O\!\left(B_n^2 n^{-\rho}\right).
\]
\end{proposition}

The next proposition makes admissible exponents in the covariance
consistency bound of \citet{CCW24} explicit for Regime~D.
\begin{proposition}
\label{prop:chang_chen_wu_explicit_exponents}
Under the assumptions of
Proposition~\ref{prop:consistency_covariance_matrix}, specialize
$\gamma_1=2$ and $\gamma_2=\gamma_{\mathrm{mix}}$.  Then an admissible
choice of the exponents in that proposition is
\[
  c_1(\rho,\vartheta)
  =\frac{\vartheta-1-(3\vartheta-2)\rho}{2\vartheta-1},
  \qquad
  c_2(2,\gamma_{\mathrm{mix}},\vartheta)
  =2+\frac{1}{\gamma_{\mathrm{mix}}}
   +\frac{1}{2\vartheta-1}.
\]
If, in addition, $p\le n^2$, $c_2$ can be improved into
\[
  c_2^\sharp(2,\gamma_{\mathrm{mix}},\vartheta)
  =\frac{\vartheta}{2\vartheta-1},
\]
in the sense that
\[
  \|\widehat\Xi_n-\Xi\|_{\max}
  =\Op\!\left\{
    B_n^2n^{-c_1(\rho,\vartheta)}
    (\log p)^{\vartheta/(2\vartheta-1)}
  \right\}
  +O(B_n^2n^{-\rho}).
\]
For the quadratic spectral kernel, $\vartheta=2$, this becomes
\[
  c_1(\rho,2)=\frac{1-4\rho}{3},
  \qquad
  c_2^\sharp(2,\gamma_{\mathrm{mix}},2)=\frac23,
\]
whereas the generic single-envelope choice is
$c_2(2,\gamma_{\mathrm{mix}},2)=7/3+1/\gamma_{\mathrm{mix}}$.
\end{proposition}

\begin{proof}
\citet[Theorem~11(i)]{CCW24} state only the existence of $c_1$ and
$c_2$, so these exponents are not uniquely defined.  We derive the
admissible choices above from its proof.  Write
$\Delta_{n,r}:=\|\widehat\Xi_n-\Xi\|_{\max}$,
$\ell_p:=\log p$, $x_+:=\max\{x,0\}$, and
\[
  \gamma_*^{-1}
  :=2+\left(\frac{2}{\gamma_1}-1\right)_+
  +\frac{1}{\gamma_2}.
\]
To make \citet[equation~(S.52)]{CCW24} explicit, define the population lag
matrices
\[
  H_j:=
  \begin{cases}
    \displaystyle
    \frac1n\sum_{t=j+1}^n \E(X_tX_{t-j}^{\top}), & j\ge 0,\\[1.2ex]
    \displaystyle
    \frac1n\sum_{t=-j+1}^n \E(X_{t+j}X_t^{\top}), & j<0,
  \end{cases}
  \qquad
  \Xi^\star:=\sum_{j=-n+1}^{n-1}K(j/b_n)H_j.
\]
Since the $X_t$ are centered, grouping the double covariance sum by its lag
gives
\[
  \Xi
  =\frac1n\sum_{s=1}^n\sum_{t=1}^n \E(X_sX_t^\top)
  =\sum_{j=-n+1}^{n-1}H_j.
\]
Consequently, adding and subtracting the deterministic oracle target
$\Xi^\star$ yields the exact decomposition
\[
\begin{aligned}
  \widehat\Xi_n-\Xi
  &=(\widehat\Xi_n-\Xi^\star)+(\Xi^\star-\Xi)\\
  &=\sum_{j=-n+1}^{n-1}K(j/b_n)(\widehat H_j-H_j)
    +\sum_{j=-n+1}^{n-1}\{K(j/b_n)-1\}H_j.
\end{aligned}
\]
Thus \citet[equation~(S.52)]{CCW24} is the max-norm triangle inequality
\[
  \Delta_{n,r}
  \le
  \underbrace{\left\|
    \sum_{j=-n+1}^{n-1}K(j/b_n)(\widehat H_j-H_j)
  \right\|_{\max}}_{\mathcal S_n\ \text{(stochastic estimation error)}}
  +
  \underbrace{\left\|
    \sum_{j=-n+1}^{n-1}\{K(j/b_n)-1\}H_j
  \right\|_{\max}}_{\mathcal B_{n,K}\ \text{(kernel bias)}}.
\]
The bias term is deterministic.  Writing
$H_j=(H_{j,a,b})_{a,b\in[p]}$, symmetry of $K$ and $K(0)=1$ give
\[
  \mathcal B_{n,K}
  \le
  \max_{a,b\in[p]}\sum_{j=1}^{n-1}
  |K(j/b_n)-1|
  \bigl(|H_{j,a,b}|+|H_{-j,a,b}|\bigr).
\]
In the unified notation used in
\citet[proof of Theorem~11(i)--(ii)]{CCW24}, Davydov's
inequality and the mixing condition imply
\[
  |H_{\pm j,a,b}|
  \lesssim B_n^2\exp\{-C L_n^{-r_2}j^{r_2}\},
  \qquad j\ge1,
\]
while the bounded derivative of $K$ implies, by the mean-value theorem,
$|K(j/b_n)-1|\lesssim j/b_n$.  Therefore
\[
\begin{aligned}
  \mathcal B_{n,K}
  &\lesssim
  \frac{B_n^2}{b_n}
  \sum_{j=1}^{n-1}j\exp\{-C(j/L_n)^{r_2}\}\\
  &\lesssim B_n^2b_n^{-1}L_n^2.
\end{aligned}
\]
This is \citet[equation~(S.53)]{CCW24}.  For the $\alpha$-mixing case,
their specialization is $L_n=1$ and $r_2=\gamma_2$; hence
$\mathcal B_{n,K}=O(B_n^2b_n^{-1})=O(B_n^2n^{-\rho})$.

The stochastic part is also explicit.  For $j\ge0$,
\[
\begin{aligned}
  \widehat H_j-H_j
  ={}&\frac1n\sum_{t=j+1}^n
  \{X_tX_{t-j}^\top-\E(X_tX_{t-j}^\top)\}\\
  &-\left(\frac1n\sum_{t=j+1}^nX_t\right)\bar X^\top
  -\bar X\left(\frac1n\sum_{t=j+1}^nX_{t-j}\right)^\top
  +\frac{n-j}{n}\bar X\bar X^\top.
\end{aligned}
\]
After multiplication by $K(j/b_n)$ and summation over $j$,
\citet[equation~(S.54)]{CCW24} labels the max norms of these four
contributions as
$I_1,I_2,I_3,I_4$.  \citet[equation~(S.55)]{CCW24} then splits the leading centered
cross-product term $I_1$ at a lag $M$, with lags $j\le M$ controlled by a
mixing concentration inequality and lags $j>M$ controlled by the polynomial
decay $|K(j/b_n)|\lesssim (j/b_n)^{-\vartheta}$.  These steps give the
aggregate bound
\[
\begin{aligned}
\Delta_{n,r}
={}&O(B_n^2n^{-\rho})
+\Op\!\left(
  B_n^2n^{\vartheta\rho}M^{1-\vartheta}
  \ell_p^{2/\gamma_1}
\right)\\
&+\Op\!\left(
  B_n^2n^{(2\rho-1)/2}M^{1/2}\ell_p^{1/2}
\right)
+\Op\!\left(
  B_n^2n^{\rho-1}M\ell_p^{1/\gamma_*}
\right),
\end{aligned}
\]
where $M$ is the lag-truncation point.  In the proof of
\citet[Theorem~11(i), immediately after equation~(S.55)]{CCW24}, they choose
\[
  M\asymp
  n^{\{1-2\rho+2\vartheta\rho\}/(2\vartheta-1)}
  \ell_p^{(4-\gamma_1)/\{\gamma_1(2\vartheta-1)\}}.
\]
Substitution shows that the first two stochastic terms in the preceding
display have common rate
\[
  B_n^2n^{-c_A}\ell_p^{d_A},
  \qquad
  c_A
  :=\frac{\vartheta-1-(3\vartheta-2)\rho}
          {2\vartheta-1},
  \qquad
  d_A
  :=\frac{\gamma_1(\vartheta-1)+2}
          {\gamma_1(2\vartheta-1)},
\]
whereas the last stochastic term has rate
\[
  B_n^2n^{-c_B}\ell_p^{d_B},
  \qquad
  c_B
  :=\frac{2(\vartheta-1)-(4\vartheta-3)\rho}
          {2\vartheta-1},
  \qquad
  d_B
  :=\gamma_*^{-1}
    +\frac{4-\gamma_1}{\gamma_1(2\vartheta-1)}.
\]
Since
\[
  c_B-c_A
  =\frac{(\vartheta-1)(1-\rho)}{2\vartheta-1}>0,
\]
one admissible single-envelope choice in
Proposition~\ref{prop:consistency_covariance_matrix} is
\[
  c_1=c_A,
  \qquad
  c_2=\max\{d_A,d_B\}.
\]
The bandwidth restriction
$\rho<(\vartheta-1)/(3\vartheta-2)$ is exactly the condition $c_1>0$.

For the Regime~D specialization
$\gamma_1=2$ and $\gamma_2=\gamma_{\mathrm{mix}}$, one has
$d_B>d_A$ and
\[
  d_A=\frac{\vartheta}{2\vartheta-1},
  \qquad
  d_B=2+\frac{1}{\gamma_{\mathrm{mix}}}
      +\frac{1}{2\vartheta-1}.
\]
It follows that the single-envelope choice is
\[
  c_1=c_A
  =\frac{\vartheta-1-(3\vartheta-2)\rho}{2\vartheta-1},
  \qquad
  c_2=d_B
  =2+\frac{1}{\gamma_{\mathrm{mix}}}
      +\frac{1}{2\vartheta-1},
\]
which proves the first assertion.

For the sharper assertion, retain the two stochastic terms separately:
\[
  \Delta_{n,r}
  =\Op(B_n^2n^{-c_A}\ell_p^{d_A})
  +\Op(B_n^2n^{-c_B}\ell_p^{d_B})
  +O(B_n^2n^{-\rho}).
\]
Under the assumptions of
Proposition~\ref{prop:consistency_covariance_matrix},
$p=\Omega(n^\kappa)$; if also $p\le n^2$, then $\ell_p\asymp\log n$.
Since $c_B-c_A>0$,
\[
  \frac{n^{-c_B}\ell_p^{d_B}}
       {n^{-c_A}\ell_p^{d_A}}
  =n^{-(c_B-c_A)}\ell_p^{d_B-d_A}
  \longrightarrow0.
\]
The second stochastic term can therefore be absorbed into the first, so
the sharper leading exponent is
$c_2^\sharp=d_A=\vartheta/(2\vartheta-1)$.  Finally, setting
$\vartheta=2$ gives
\[
  c_1=\frac{1-4\rho}{3},
  \qquad
  c_2^\sharp=\frac23,
  \qquad
  c_2=\frac73+\frac1{\gamma_{\mathrm{mix}}},
\]
as claimed.
\end{proof}

The next proposition gives the bootstrap high-dimensional CLT comparison
from \citet[Theorem~10(i)]{CCW24}.
\begin{proposition}\label{prop:dependent_highd_CLT}
Let $X_1,\dots,X_n$ be centered random vectors in $\R^p$, and define
\[
S_{n,x}:=\frac{1}{\sqrt n}\sum_{t=1}^n X_t,
\qquad
\Xi:=\Cov(S_{n,x}).
\]
Let $G\sim N(0,\Xi)$. Given a random covariance estimator $\widehat\Xi_n$, let $\widehat G$ be a conditionally Gaussian random vector satisfying
\[
\widehat G\mid X_1, \ldots, X_n \sim N(0,\widehat\Xi_n).
\]
For a class $\mathcal A$ of Borel subsets of $\R^p$, define
\[
\rho_n(\mathcal A)
:=
\sup_{A\in\mathcal A}
\bigl|\Pbb(S_{n,x}\in A)-\Pbb(G\in A)\bigr|,
\quad
\widehat\rho_n(\mathcal A)
:=
\sup_{A\in\mathcal A}
\bigl|\Pbb(S_{n,x}\in A)-\Pbb(\widehat G\in A\mid X_1,\ldots,X_n)\bigr|.
\]
Let $\mathcal{A}^{\mathrm{re}}$ denote the class of all hyper-rectangles in $\R^p$, namely all sets of the form
\[
A=\prod_{j=1}^p [a_j,b_j],
\qquad
-\infty\le a_j\le b_j\le \infty.
\]
These sets include one-sided rectangles by allowing infinite endpoints.
For each $j\in[p]$, let
\[
V_{n,j}:=\Var\!\left(\frac{1}{\sqrt n}\sum_{t=1}^n X_{t,j}\right).
\]
Assume that there exists a constant $K_3>0$ such that
$\min_{j\in[p]}V_{n,j}\ge K_3$.
Further, assume that $p\ge n^\kappa$ for some constant $\kappa>0$. Then
\[
\widehat\rho_n(\mathcal{A}^{\mathrm{re}})
\lesssim
\rho_n(\mathcal{A}^{\mathrm{re}})+\norm{\widehat\Xi_n-\Xi}_{\max}^{1/3}(\log p)^{2/3},
\]
where $\lesssim$ hides a universal positive constant.
\end{proposition}

The next proposition is the strong-mixing Bernstein inequality from
\citet[Theorem~1 and Remark~1]{merlevede2011bernstein}.
\begin{proposition}\label{prop:alpha_mixing_concentration}
Let $(X_j)_{j\ge 1}$ be centered real-valued random variables, and define their strong-mixing coefficients by
\[
\alpha_X(\ell)
:=
\sup_{k\ge 1}
\sup_{\substack{A\in\sigma(X_1,\dots,X_k),\\ B\in\sigma(X_{k+\ell},X_{k+\ell+1},\dots)}}
\big|\Pbb(A\cap B)-\Pbb(A)\Pbb(B)\big|,
\qquad \ell\ge 1.
\]
Assume that there exist constants $c>0$, $b>0$, $\gamma_1>0$, and $\gamma_2\in(0,\infty]$ such that
\[
\alpha_X(n) \le \exp\!\bigl(-c n^{\gamma_1}\bigr)
\qquad\text{for all }n\ge 1,
\]
and
\[
\sup_{j\ge 1} \Pbb(|X_j|>t)
\le
\exp\!\Bigl(1-(t/b)^{\gamma_2}\Bigr)
\qquad\text{for all }t>0.
\]
Let
\[
\gamma := \left(\frac{1}{\gamma_1}+\frac{1}{\gamma_2}\right)^{-1},
\qquad
\varphi_M(x):=(x\wedge M)\vee(-M),
\]
and assume $\gamma<1$. Define
\begin{equation}\label{eq:mpr-V}
V
:=
\sup_{M>0}\sup_{i\ge 1}
\left(
\Var\bigl(\varphi_M(X_i)\bigr)
+2\sum_{j>i}\bigl|\Cov\bigl(\varphi_M(X_i),\varphi_M(X_j)\bigr)\bigr|
\right).
\end{equation}
Then there exist constants $C_0,C_1,C_2,\eta_0\in(0,\infty)$, depending only on $(b,c,\gamma_1,\gamma_2)$, such that for every $n\ge 4$ and every $y\ge C_0(\log n)^{\eta_0}$,
\begin{equation}\label{eq:mpr-strong-mixing}
\Pbb\!\left(\max_{1\le j\le n}\left|\sum_{r=1}^j X_r\right|\ge y\right)
\le
(n+1)\exp\!\left(-\frac{y^{\gamma}}{C_1}\right)
+
\exp\!\left(-\frac{y^2}{C_2+C_2 nV}\right).
\end{equation}
\end{proposition}

The next proposition bounds the term $V$ in the strong-mixing Bernstein
inequality \eqref{eq:mpr-strong-mixing} by marginal quantiles and mixing
coefficients.
\begin{proposition}\label{prop:V_bound}
Consider the setting of Proposition \ref{prop:alpha_mixing_concentration}. For $u\in(0,1]$, define the marginal quantiles
\[
q_{|X_i|}(u):=\inf\{t>0:\Pbb(|X_i|>t)\le u\},
\qquad
q(u):=\sup_{i\ge 1} q_{|X_i|}(u),
\]
Then the strong-mixing covariance inequality
\citep[see, e.g., Remark~3]{merlevede2011bernstein} implies
\begin{equation}\label{eq:remark3-strong-mixing}
V
\le
\sup_{i\ge 1}\E[X_i^2]
+
4\sum_{k>0}\int_0^{2\alpha_X(k)}q(u)^2\,du.
\end{equation}
\end{proposition}

\paragraph{Orlicz $\psi_2$ and $\psi_1$ norms.}
Following \citet[Definitions~2.5.6 and 2.7.5]{vershynin2018high}, define for a real-valued random variable $X$
\[
\|X\|_{\psi_2}:=\inf\bigl\{t>0:\ \E\exp(X^2/t^2)\le 2\bigr\},
\qquad
\|X\|_{\psi_1}:=\inf\bigl\{t>0:\ \E\exp(|X|/t)\le 2\bigr\}.
\]
We call $X$ \emph{sub-Gaussian} if $\|X\|_{\psi_2}<\infty$ and \emph{sub-exponential} if $\|X\|_{\psi_1}<\infty$.
By \citet[Proposition~2.5.2]{vershynin2018high}, the $\psi_2$ norm is equivalent to the standard sub-Gaussian tail and moment conditions; by \citet[Proposition~2.7.1]{vershynin2018high}, the $\psi_1$ norm is equivalent to the standard sub-exponential tail and moment conditions.

The next proposition gives maximal bounds for finitely many sub-Gaussian
variables, following
\citet[Proposition~2.5.2 and Exercise~2.5.10]{vershynin2018high}.
\begin{proposition}\label{prop:subgaussian_max}
Let $X_1,\dots,X_m$ be real-valued random variables satisfying
\[
\max_{1\le j\le m}\|X_j\|_{\psi_2}\le K.
\]
Then there exist universal constants $C,c\in(0,\infty)$ such that
\[
\E\Bigl[\max_{1\le j\le m}|X_j|\Bigr]
\le CK\sqrt{1\vee \log m}
\]
and, for every $u\ge 0$,
\[
\Pbb\Bigl(
\max_{1\le j\le m}|X_j|
>
CK\bigl(\sqrt{\log(2m)}+u\bigr)
\Bigr)
\le 2e^{-c u^2}.
\]
In particular,
\[
\max_{1\le j\le m}|X_j|=\Op\bigl(K\sqrt{1\vee\log m}\bigr).
\]
No independence is required.
\end{proposition}

The next proposition gives the standard sub-Gaussian bound for weighted sums
of independent centered variables $X_1,\dots,X_m$, following
\citet[Proposition~2.6.1]{vershynin2018high}.
\begin{proposition}\label{prop:sum_subgaussian}
Let $X_1,\dots,X_m$ be independent centered real-valued random variables, and set
\[
K_j:=\norm{X_j}_{\psi_2},
\qquad
V_2:=\left(\sum_{j=1}^m K_j^2\right)^{1/2}.
\]
Then there exists a universal constant $C<\infty$ such that
\[
\Bigl\|\sum_{j=1}^m X_j\Bigr\|_{\psi_2}
\le C V_2.
\]
More generally, for deterministic weights $a_1,\dots,a_m\in\R$,
\[
\Bigl\|\sum_{j=1}^m a_j X_j\Bigr\|_{\psi_2}
\le
C\left(\sum_{j=1}^m a_j^2 K_j^2\right)^{1/2}.
\]
\end{proposition}

The next proposition gives the standard product rule turning sub-Gaussian
products into sub-exponential variables, following
\citet[Lemma~2.7.7]{vershynin2018high}.
\begin{proposition}\label{prop:product_subgaussian}
  Let $X$ and $Y$ be sub-gaussian random variables. Then $XY$ is sub-exponential. Moreover,
  \[\|XY\|_{\psi_1}\le \|X\|_{\psi_2}\|Y\|_{\psi_2}.\]
\end{proposition}

The next proposition gives maximal bounds for finitely many sub-exponential
variables, following \citet[Proposition~2.7.1]{vershynin2018high}.
\begin{proposition}\label{prop:max_subexponential}
Let $Y_1,\dots,Y_m$ be real-valued random variables satisfying
\[
\max_{1\le j\le m}\|Y_j\|_{\psi_1}\le K.
\]
Then there exist universal constants $C,c\in(0,\infty)$ such that
\[
\E\Bigl[\max_{1\le j\le m}|Y_j|\Bigr]
\le CK(1\vee \log m)
\]
and, for every $u\ge 0$,
\[
\Pbb\Bigl(
\max_{1\le j\le m}|Y_j|
>
CK\bigl(\log(2m)+u\bigr)
\Bigr)
\le 2e^{-c u}.
\]
In particular,
\[
\max_{1\le j\le m}|Y_j|=\Op\bigl(K(1\vee\log m)\bigr).
\]
Again, no independence is required.
\end{proposition}

\clearpage 
\section{Simulation Results for Partially Studentized Scan Tests}
\label{app:simulation_results}

Section~\ref{sec:scantest} of the main draft develops the scan test as a one-sided procedure for
detecting forecast outperformance that may occur only over unknown target-month
intervals and unknown horizons.  The theory justifies the bootstrap scan test
asymptotically, but the empirical application has only $n=215$ target months
and uses short interval collections for expert selection.  This appendix
therefore documents the finite-sample calibration exercises behind the
implementation choices in Section~\ref{sec:scantest}.  The first part calibrates the
partial-studentization exponent $\gamma$ for the horizon-by-horizon scan tests.
The second part repeats the calibration for the horizon-pooled pre/post-COVID
use of the scan test.  The final part benchmarks these choices against the
closely related Giacomini--Rossi fluctuation test of
\citet{GiacominiRossi2010}.

\subsection{Simulation and the choice of $\gamma$}\label{app:sim_gamma}

The scan statistic in Section~\ref{sec:scantest} maximizes a partially studentized statistic
over many intervals and horizons.  This maximization is what lets the test find
localized outperformance without specifying its timing or forecast horizon in
advance, but it also creates the main finite-sample concern: aggressive
studentization can make short, noisy intervals look too extreme.  We therefore
choose $\gamma$ by simulating the null in a way that preserves the empirical
scale, cross-horizon dependence, and date-specific heterogeneity of the observed
loss-differential panels.

The simulation setting is under the null.  It fixes an empirical
loss-differential panel $D_t^{(h)}$ and draws
\[
D_{t,k}^{(h)}=\xi_{t,k}D_t^{(h)},
\qquad
\xi_{t,k}=\operatorname{sign}(Z_{t,k}),
\]
where $Z_{t,k}$ is a stationary Gaussian AR(1) process with autocorrelation
$\rho$.  To choose $\rho$, we fit an AR(1) model to each comparison
(sarima--micro, macro--micro, and sarima--macro) and each horizon $h=1,\ldots,24$.  The box
plots of the estimates are given below.  Thus, we choose
$\rho\in\{-0.25,0,0.25,0.5\}$ in our simulations.

\begin{figure}[H]
\centering
\includegraphics[width=0.70\textwidth]{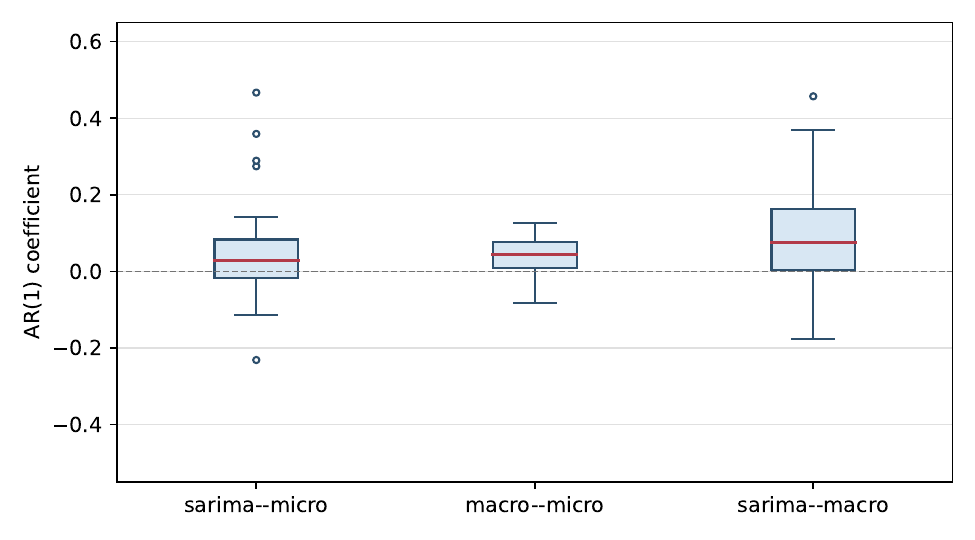}
\caption{AR(1) coefficient estimates across horizons for sarima--micro, macro--micro, and sarima--macro.}
\label{fig:ar1_coefficients_boxplot}
\end{figure}

\subsubsection{Simulation results for occasional outperformance}

Occasional outperformance is the expert-selection use emphasized in Section~\ref{sec:scantest}:
the candidate forecast should enter the pool if it improves on the benchmark
over some unknown short subperiod and horizon.  We use the interval lengths
$\{6,7,\ldots,12\}$ to target half-year to one-year episodes, where a forecast
can be useful during a localized inflation episode even if it is not uniformly
better over the full sample.  This is also the most demanding case for size,
because the scan searches over many short windows.

For each setting in Tables~\ref{tab:sim_sarima_micro}
--\ref{tab:sim_sarima_macro}, we use 1000 Monte Carlo samples and 999
bootstrap draws.  The scan intervals are $\{6,7,\ldots,12\}$ and the target
level is denoted by $\alpha_0$ to avoid conflict with the tuning parameter
$\gamma$.

\begin{table}[H]
\centering
\caption{Simulation for sarima--micro. Entries are rejection probabilities at nominal level $\alpha_0=0.05$.}
\label{tab:sim_sarima_micro}
\begin{subtable}[t]{0.47\textwidth}
\centering
\caption{Regime I}
\label{tab:sim_sarima_micro_regime_i}
\small
\setlength{\tabcolsep}{3pt}
\begin{tabular}{lcccc}
\toprule
$\gamma$ & $\rho=-.25$ & $\rho=0$ & $\rho=.25$ & $\rho=.5$ \\
\midrule
0.00 & 0.000 & 0.000 & 0.003 & 0.005 \\
0.05 & 0.000 & 0.002 & 0.019 & 0.038 \\
0.10 & 0.003 & 0.011 & 0.035 & 0.075 \\
0.15 & 0.005 & 0.015 & 0.038 & 0.089 \\
0.20 & 0.006 & 0.016 & 0.043 & 0.090 \\
0.25 & 0.005 & 0.014 & 0.044 & 0.101 \\
0.30 & 0.049 & 0.064 & 0.098 & 0.123 \\
0.35 & 0.161 & 0.159 & 0.165 & 0.185 \\
0.40 & 0.235 & 0.224 & 0.203 & 0.188 \\
0.45 & 0.213 & 0.217 & 0.180 & 0.178 \\
0.50 & 0.182 & 0.194 & 0.151 & 0.152 \\
\bottomrule
\end{tabular}
\end{subtable}
\hfill
\begin{subtable}[t]{0.47\textwidth}
\centering
\caption{Regime D}
\label{tab:sim_sarima_micro_regime_d}
\small
\setlength{\tabcolsep}{3pt}
\begin{tabular}{lcccc}
\toprule
$\gamma$ & $\rho=-.25$ & $\rho=0$ & $\rho=.25$ & $\rho=.5$ \\
\midrule
0.00 & 0.003 & 0.003 & 0.027 & 0.056 \\
0.05 & 0.004 & 0.002 & 0.026 & 0.058 \\
0.10 & 0.002 & 0.002 & 0.023 & 0.045 \\
0.15 & 0.000 & 0.000 & 0.006 & 0.011 \\
0.20 & 0.000 & 0.000 & 0.000 & 0.000 \\
0.25 & 0.000 & 0.000 & 0.000 & 0.000 \\
0.30 & 0.002 & 0.004 & 0.006 & 0.002 \\
0.35 & 0.053 & 0.043 & 0.044 & 0.038 \\
0.40 & 0.108 & 0.097 & 0.093 & 0.075 \\
0.45 & 0.185 & 0.200 & 0.175 & 0.171 \\
0.50 & 0.193 & 0.214 & 0.185 & 0.189 \\
\bottomrule
\end{tabular}
\end{subtable}
\end{table}

\begin{table}[H]
\centering
\caption{Simulation for macro--micro. Entries are rejection probabilities at nominal level $\alpha_0=0.05$.}
\label{tab:sim_macro_micro}
\begin{subtable}[t]{0.47\textwidth}
\centering
\caption{Regime I}
\label{tab:sim_macro_micro_regime_i}
\small
\setlength{\tabcolsep}{3pt}
\begin{tabular}{lcccc}
\toprule
$\gamma$ & $\rho=-.25$ & $\rho=0$ & $\rho=.25$ & $\rho=.5$ \\
\midrule
0.00 & 0.002 & 0.003 & 0.001 & 0.003 \\
0.05 & 0.006 & 0.010 & 0.010 & 0.010 \\
0.10 & 0.014 & 0.026 & 0.030 & 0.045 \\
0.15 & 0.068 & 0.075 & 0.076 & 0.089 \\
0.20 & 0.205 & 0.196 & 0.210 & 0.211 \\
0.25 & 0.319 & 0.310 & 0.353 & 0.353 \\
0.30 & 0.394 & 0.420 & 0.476 & 0.484 \\
0.35 & 0.400 & 0.409 & 0.452 & 0.447 \\
0.40 & 0.382 & 0.382 & 0.404 & 0.377 \\
0.45 & 0.319 & 0.320 & 0.320 & 0.302 \\
0.50 & 0.227 & 0.232 & 0.209 & 0.221 \\
\bottomrule
\end{tabular}
\end{subtable}
\hfill
\begin{subtable}[t]{0.47\textwidth}
\centering
\caption{Regime D}
\label{tab:sim_macro_micro_regime_d}
\small
\setlength{\tabcolsep}{3pt}
\begin{tabular}{lcccc}
\toprule
$\gamma$ & $\rho=-.25$ & $\rho=0$ & $\rho=.25$ & $\rho=.5$ \\
\midrule
0.00 & 0.003 & 0.003 & 0.000 & 0.001 \\
0.05 & 0.005 & 0.010 & 0.008 & 0.008 \\
0.10 & 0.013 & 0.020 & 0.024 & 0.026 \\
0.15 & 0.074 & 0.067 & 0.071 & 0.068 \\
0.20 & 0.199 & 0.190 & 0.194 & 0.191 \\
0.25 & 0.331 & 0.308 & 0.340 & 0.329 \\
0.30 & 0.418 & 0.422 & 0.474 & 0.493 \\
0.35 & 0.424 & 0.419 & 0.455 & 0.457 \\
0.40 & 0.399 & 0.393 & 0.416 & 0.395 \\
0.45 & 0.346 & 0.337 & 0.328 & 0.312 \\
0.50 & 0.241 & 0.237 & 0.218 & 0.224 \\
\bottomrule
\end{tabular}
\end{subtable}
\end{table}

\begin{table}[H]
\centering
\caption{Simulation for sarima--macro. Entries are rejection probabilities at nominal level $\alpha_0=0.05$.}
\label{tab:sim_sarima_macro}
\begin{subtable}[t]{0.47\textwidth}
\centering
\caption{Regime I}
\label{tab:sim_sarima_macro_regime_i}
\small
\setlength{\tabcolsep}{3pt}
\begin{tabular}{lcccc}
\toprule
$\gamma$ & $\rho=-.25$ & $\rho=0$ & $\rho=.25$ & $\rho=.5$ \\
\midrule
0.00 & 0.000 & 0.000 & 0.000 & 0.000 \\
0.05 & 0.000 & 0.001 & 0.003 & 0.005 \\
0.10 & 0.000 & 0.004 & 0.022 & 0.044 \\
0.15 & 0.004 & 0.010 & 0.035 & 0.075 \\
0.20 & 0.005 & 0.012 & 0.036 & 0.085 \\
0.25 & 0.013 & 0.027 & 0.064 & 0.121 \\
0.30 & 0.049 & 0.085 & 0.128 & 0.183 \\
0.35 & 0.227 & 0.246 & 0.279 & 0.322 \\
0.40 & 0.308 & 0.316 & 0.351 & 0.385 \\
0.45 & 0.309 & 0.293 & 0.320 & 0.322 \\
0.50 & 0.260 & 0.221 & 0.258 & 0.256 \\
\bottomrule
\end{tabular}
\end{subtable}
\hfill
\begin{subtable}[t]{0.47\textwidth}
\centering
\caption{Regime D}
\label{tab:sim_sarima_macro_regime_d}
\small
\setlength{\tabcolsep}{3pt}
\begin{tabular}{lcccc}
\toprule
$\gamma$ & $\rho=-.25$ & $\rho=0$ & $\rho=.25$ & $\rho=.5$ \\
\midrule
0.00 & 0.000 & 0.000 & 0.018 & 0.036 \\
0.05 & 0.003 & 0.000 & 0.016 & 0.039 \\
0.10 & 0.000 & 0.000 & 0.012 & 0.027 \\
0.15 & 0.000 & 0.000 & 0.003 & 0.006 \\
0.20 & 0.000 & 0.000 & 0.000 & 0.000 \\
0.25 & 0.001 & 0.005 & 0.003 & 0.002 \\
0.30 & 0.013 & 0.025 & 0.030 & 0.034 \\
0.35 & 0.086 & 0.089 & 0.104 & 0.085 \\
0.40 & 0.209 & 0.216 & 0.244 & 0.266 \\
0.45 & 0.314 & 0.323 & 0.356 & 0.385 \\
0.50 & 0.300 & 0.291 & 0.331 & 0.343 \\
\bottomrule
\end{tabular}
\end{subtable}
\end{table}

\subsubsection{Simulation results for overall outperformance}

Overall outperformance is the companion expert-selection diagnostic from
Section~\ref{sec:scantest}.  Instead of searching for a localized episode, it fixes the interval
to the full balanced sample and asks whether one forecast improves on another
over the sample as a whole at some horizon.  This removes the interval-search
noise from the occasional scan and can be more stable when predictive gains are
sustained rather than concentrated in a short spell.

For each setting in Tables~\ref{tab:sim_overall_sarima_micro}
--\ref{tab:sim_overall_sarima_macro}, we use 1000 Monte Carlo samples and 999
bootstrap draws.  The scan interval length is $\{215\}$ and the target level
is again denoted by $\alpha_0$.

\begin{table}[H]
\centering
\caption{Simulation for sarima--micro with length $\{215\}$. Entries are rejection probabilities at nominal level $\alpha_0=0.05$.}
\label{tab:sim_overall_sarima_micro}
\begin{subtable}[t]{0.47\textwidth}
\centering
\caption{Regime I}
\label{tab:sim_overall_sarima_micro_regime_i}
\small
\setlength{\tabcolsep}{3pt}
\begin{tabular}{lcccc}
\toprule
$\gamma$ & $\rho=-.25$ & $\rho=0$ & $\rho=.25$ & $\rho=.5$ \\
\midrule
0.00 & 0.008 & 0.012 & 0.026 & 0.042 \\
0.05 & 0.012 & 0.019 & 0.036 & 0.061 \\
0.10 & 0.018 & 0.025 & 0.047 & 0.075 \\
0.15 & 0.025 & 0.038 & 0.054 & 0.094 \\
0.20 & 0.035 & 0.046 & 0.063 & 0.107 \\
0.25 & 0.046 & 0.052 & 0.075 & 0.115 \\
0.30 & 0.054 & 0.055 & 0.071 & 0.105 \\
0.35 & 0.057 & 0.051 & 0.077 & 0.087 \\
0.40 & 0.056 & 0.049 & 0.062 & 0.073 \\
0.45 & 0.055 & 0.049 & 0.061 & 0.063 \\
0.50 & 0.055 & 0.051 & 0.057 & 0.068 \\
\bottomrule
\end{tabular}
\end{subtable}
\hfill
\begin{subtable}[t]{0.47\textwidth}
\centering
\caption{Regime D}
\label{tab:sim_overall_sarima_micro_regime_d}
\small
\setlength{\tabcolsep}{3pt}
\begin{tabular}{lcccc}
\toprule
$\gamma$ & $\rho=-.25$ & $\rho=0$ & $\rho=.25$ & $\rho=.5$ \\
\midrule
0.00 & 0.007 & 0.006 & 0.008 & 0.010 \\
0.05 & 0.011 & 0.010 & 0.011 & 0.019 \\
0.10 & 0.017 & 0.016 & 0.017 & 0.025 \\
0.15 & 0.022 & 0.023 & 0.031 & 0.035 \\
0.20 & 0.037 & 0.031 & 0.037 & 0.047 \\
0.25 & 0.043 & 0.039 & 0.047 & 0.061 \\
0.30 & 0.056 & 0.040 & 0.054 & 0.060 \\
0.35 & 0.058 & 0.047 & 0.057 & 0.058 \\
0.40 & 0.059 & 0.051 & 0.057 & 0.060 \\
0.45 & 0.058 & 0.052 & 0.054 & 0.059 \\
0.50 & 0.056 & 0.053 & 0.055 & 0.059 \\
\bottomrule
\end{tabular}
\end{subtable}
\end{table}

\begin{table}[H]
\centering
\caption{Simulation for macro--micro with length $\{215\}$. Entries are rejection probabilities at nominal level $\alpha_0=0.05$.}
\label{tab:sim_overall_macro_micro}
\begin{subtable}[t]{0.47\textwidth}
\centering
\caption{Regime I}
\label{tab:sim_overall_macro_micro_regime_i}
\small
\setlength{\tabcolsep}{3pt}
\begin{tabular}{lcccc}
\toprule
$\gamma$ & $\rho=-.25$ & $\rho=0$ & $\rho=.25$ & $\rho=.5$ \\
\midrule
0.00 & 0.041 & 0.052 & 0.050 & 0.054 \\
0.05 & 0.040 & 0.057 & 0.055 & 0.059 \\
0.10 & 0.043 & 0.059 & 0.054 & 0.059 \\
0.15 & 0.043 & 0.063 & 0.057 & 0.061 \\
0.20 & 0.048 & 0.070 & 0.059 & 0.068 \\
0.25 & 0.049 & 0.071 & 0.063 & 0.075 \\
0.30 & 0.051 & 0.070 & 0.064 & 0.080 \\
0.35 & 0.050 & 0.074 & 0.066 & 0.079 \\
0.40 & 0.048 & 0.074 & 0.066 & 0.073 \\
0.45 & 0.045 & 0.072 & 0.065 & 0.071 \\
0.50 & 0.044 & 0.072 & 0.064 & 0.072 \\
\bottomrule
\end{tabular}
\end{subtable}
\hfill
\begin{subtable}[t]{0.47\textwidth}
\centering
\caption{Regime D}
\label{tab:sim_overall_macro_micro_regime_d}
\small
\setlength{\tabcolsep}{3pt}
\begin{tabular}{lcccc}
\toprule
$\gamma$ & $\rho=-.25$ & $\rho=0$ & $\rho=.25$ & $\rho=.5$ \\
\midrule
0.00 & 0.039 & 0.052 & 0.046 & 0.053 \\
0.05 & 0.042 & 0.056 & 0.049 & 0.055 \\
0.10 & 0.043 & 0.058 & 0.051 & 0.057 \\
0.15 & 0.048 & 0.063 & 0.054 & 0.060 \\
0.20 & 0.050 & 0.066 & 0.057 & 0.063 \\
0.25 & 0.049 & 0.067 & 0.059 & 0.071 \\
0.30 & 0.052 & 0.067 & 0.062 & 0.074 \\
0.35 & 0.051 & 0.069 & 0.061 & 0.076 \\
0.40 & 0.049 & 0.070 & 0.061 & 0.075 \\
0.45 & 0.047 & 0.070 & 0.060 & 0.076 \\
0.50 & 0.044 & 0.069 & 0.058 & 0.076 \\
\bottomrule
\end{tabular}
\end{subtable}
\end{table}

\begin{table}[H]
\centering
\caption{Simulation for sarima--macro with length $\{215\}$. Entries are rejection probabilities at nominal level $\alpha_0=0.05$.}
\label{tab:sim_overall_sarima_macro}
\begin{subtable}[t]{0.47\textwidth}
\centering
\caption{Regime I}
\label{tab:sim_overall_sarima_macro_regime_i}
\small
\setlength{\tabcolsep}{3pt}
\begin{tabular}{lcccc}
\toprule
$\gamma$ & $\rho=-.25$ & $\rho=0$ & $\rho=.25$ & $\rho=.5$ \\
\midrule
0.00 & 0.003 & 0.012 & 0.016 & 0.046 \\
0.05 & 0.005 & 0.018 & 0.023 & 0.056 \\
0.10 & 0.012 & 0.033 & 0.032 & 0.072 \\
0.15 & 0.019 & 0.044 & 0.042 & 0.087 \\
0.20 & 0.023 & 0.054 & 0.050 & 0.094 \\
0.25 & 0.029 & 0.062 & 0.053 & 0.109 \\
0.30 & 0.038 & 0.057 & 0.067 & 0.109 \\
0.35 & 0.037 & 0.061 & 0.065 & 0.093 \\
0.40 & 0.041 & 0.061 & 0.056 & 0.085 \\
0.45 & 0.038 & 0.062 & 0.056 & 0.080 \\
0.50 & 0.033 & 0.063 & 0.057 & 0.074 \\
\bottomrule
\end{tabular}
\end{subtable}
\hfill
\begin{subtable}[t]{0.47\textwidth}
\centering
\caption{Regime D}
\label{tab:sim_overall_sarima_macro_regime_d}
\small
\setlength{\tabcolsep}{3pt}
\begin{tabular}{lcccc}
\toprule
$\gamma$ & $\rho=-.25$ & $\rho=0$ & $\rho=.25$ & $\rho=.5$ \\
\midrule
0.00 & 0.005 & 0.004 & 0.001 & 0.012 \\
0.05 & 0.006 & 0.007 & 0.004 & 0.017 \\
0.10 & 0.010 & 0.012 & 0.012 & 0.024 \\
0.15 & 0.016 & 0.030 & 0.016 & 0.031 \\
0.20 & 0.027 & 0.037 & 0.025 & 0.045 \\
0.25 & 0.034 & 0.041 & 0.031 & 0.054 \\
0.30 & 0.042 & 0.051 & 0.039 & 0.060 \\
0.35 & 0.041 & 0.056 & 0.046 & 0.064 \\
0.40 & 0.041 & 0.059 & 0.050 & 0.066 \\
0.45 & 0.040 & 0.063 & 0.052 & 0.058 \\
0.50 & 0.041 & 0.063 & 0.052 & 0.058 \\
\bottomrule
\end{tabular}
\end{subtable}
\end{table}

\subsubsection{Choice of $\gamma$}

This subsection converts the two calibration exercises into the tuning rule
used in the empirical scan checks.  Because Section~\ref{sec:scantest} allows either Regime I or
Regime D depending on serial dependence, the decision must control size not only
under each fixed multiplier regime but also after the Ljung--Box pre-test chooses
between them.

The simulation results suggest $\gamma = 0.1$ avoids over-rejections across the
board.  So we choose $\gamma = 0.1$ for the scan test.  A few rejection
probabilities are slightly above the nominal level when $\rho=0.5$, but the
inflation is minor and occurs under the strongest positive serial dependence
considered in the simulation grid.  We will check robustness by also reporting the
empirical results for the adjacent choices $\gamma=0.05$ and $\gamma=0.15$.

To confirm that the pre-test does not inflate size, we run the same simulation
on the scan test with the Ljung--Box pre-test and the above chosen $\gamma$ for
each comparison.  We report the results for both the occasional- and
overall-outperformance interval sets.

\begin{table}[H]
\centering
\caption{Simulation for the Ljung--Box pre-test scan test with $\delta=0.05$ and $\gamma=0.1$ under both multiplier regimes, using occasional-outperformance lengths $\{6,7,\ldots,12\}$. Entries are rejection probabilities at nominal level $\alpha_0=0.05$.}
\label{tab:sim_ljung_box_pretest}
\small
\setlength{\tabcolsep}{6pt}
\begin{tabular}{lcccc}
\toprule
Comparison & $\rho=-.25$ & $\rho=0$ & $\rho=.25$ & $\rho=.5$ \\
\midrule
sarima--micro & 0.003 & 0.012 & 0.022 & 0.063 \\
macro--micro & 0.023 & 0.020 & 0.032 & 0.037 \\
sarima--macro & 0.005 & 0.001 & 0.012 & 0.035 \\
\bottomrule
\end{tabular}
\end{table}

\begin{table}[H]
\centering
\caption{Simulation for the Ljung--Box pre-test scan test with $\delta=0.05$ and $\gamma=0.1$ under both multiplier regimes, using overall-outperformance length $\{215\}$. Entries are rejection probabilities at nominal level $\alpha_0=0.05$.}
\label{tab:sim_ljung_box_pretest_overall}
\small
\setlength{\tabcolsep}{6pt}
\begin{tabular}{lcccc}
\toprule
Comparison & $\rho=-.25$ & $\rho=0$ & $\rho=.25$ & $\rho=.5$ \\
\midrule
sarima--micro & 0.012 & 0.032 & 0.029 & 0.055 \\
macro--micro & 0.045 & 0.051 & 0.056 & 0.061 \\
sarima--macro & 0.016 & 0.027 & 0.035 & 0.066 \\
\bottomrule
\end{tabular}
\end{table}

\subsubsection{Robustness of occasional and overall performance to the choice of $\gamma$}

The simulations above motivate the baseline choice $\gamma=0.1$ by focusing on
size control under null designs calibrated to the observed loss-differential
panels.  The empirical conclusion should not, however, hinge on this single
tuning value.  We therefore report the actual scan-test p-values for the two
expert-selection metrics from Section~\ref{sec:scantest} at $\gamma=0.1$ and at the adjacent
values $\gamma=0.05$ and $\gamma=0.15$.  The occasional-outperformance metric
uses the short interval collection $\{6,7,\ldots,12\}$ and asks whether a
candidate forecast improves on the benchmark during some localized episode and
horizon.  The overall-outperformance metric uses the full balanced sample
length $\{215\}$ and asks whether the improvement is sustained over the sample
as a whole at some horizon.  Reporting both metrics across nearby values of
$\gamma$ checks that the substantive expert-selection evidence is robust to the
degree of partial studentization.

\begin{table}[H]
\centering
\caption{Empirical scan-test p-values with $\gamma=0.05$. The Ljung--Box pre-test uses level $\delta=0.05$. Occasional outperformance uses lengths $\{6,7,\ldots,12\}$; overall outperformance uses length $\{215\}$.}
\label{tab:empirical_outperformance_gamma_05}
\small
\setlength{\tabcolsep}{3pt}
\begin{tabular}{lcccc}
\toprule
Comparison & Regime & LB $p$ & Occ. $p$ & Overall $p$ \\
\midrule
sarima--micro & Regime D & 0.0338 & 0.0352 & 0.1494 \\
macro--micro & Regime I & 0.2260 & 0.2861 & 0.0306 \\
sarima--macro & Regime I & 0.2234 & 0.0572 & 0.0406 \\
\bottomrule
\end{tabular}
\end{table}

\begin{table}[H]
\centering
\caption{Empirical scan-test p-values with $\gamma=0.10$. The Ljung--Box pre-test uses level $\delta=0.05$. Occasional outperformance uses lengths $\{6,7,\ldots,12\}$; overall outperformance uses length $\{215\}$.}
\label{tab:empirical_outperformance_gamma_10}
\small
\setlength{\tabcolsep}{3pt}
\begin{tabular}{lcccc}
\toprule
Comparison & Regime & LB $p$ & Occ. $p$ & Overall $p$ \\
\midrule
sarima--micro & Regime D & 0.0338 & 0.0406 & 0.1524 \\
macro--micro & Regime I & 0.2260 & 0.2575 & 0.0256 \\
sarima--macro & Regime I & 0.2234 & 0.0406 & 0.0324 \\
\bottomrule
\end{tabular}
\end{table}

\begin{table}[H]
\centering
\caption{Empirical scan-test p-values with $\gamma=0.15$. The Ljung--Box pre-test uses level $\delta=0.05$. Occasional outperformance uses lengths $\{6,7,\ldots,12\}$; overall outperformance uses length $\{215\}$.}
\label{tab:empirical_outperformance_gamma_15}
\small
\setlength{\tabcolsep}{3pt}
\begin{tabular}{lcccc}
\toprule
Comparison & Regime & LB $p$ & Occ. $p$ & Overall $p$ \\
\midrule
sarima--micro & Regime D & 0.0338 & 0.0522 & 0.1544 \\
macro--micro & Regime I & 0.2260 & 0.1870 & 0.0210 \\
sarima--macro & Regime I & 0.2234 & 0.0286 & 0.0224 \\
\bottomrule
\end{tabular}
\end{table}

\subsection{Simulations for the horizon-pooled scan test}\label{app:sim_pooled}

Section~\ref{sec:scantest} also describes a complementary use of the scan test for a
pre-specified calendar window, while pooling across horizons.  This question is
different from expert selection: the window, such as the pre-COVID or
post-COVID period, is fixed in advance, and the test asks whether the average
loss differential across horizons is positive within that window.  Because the
statistic is now built from the horizon-pooled series $\bar D_t$ rather than
from the full interval-horizon panel, we run a separate null calibration for
this version of the procedure.

For each empirical panel, we
fit $\bar D_t=c+\rho\bar D_{t-1}+u_t$ to the horizon-pooled series.  The
estimated AR(1) coefficients are reported in Table~\ref{tab:horizon_pooled_ar1}.

\begin{table}[H]
\centering
\caption{AR(1) estimates for the horizon-pooled loss-differential series.}
\label{tab:horizon_pooled_ar1}
\small
\setlength{\tabcolsep}{6pt}
\begin{tabular}{lcc}
\toprule
Comparison & $\hat\rho$ & 95\% confidence interval \\
\midrule
$(\text{sarima}+\text{macro})-\text{three experts}$ & 0.050 & $[-0.086,0.187]$ \\
sarima--micro & 0.185 & $[0.052,0.317]$ \\
sarima--macro & 0.160 & $[0.027,0.293]$ \\
\bottomrule
\end{tabular}
\end{table}

The null DGP is
the same sign-AR(1) construction described above, applied to the horizon-pooled
series.  We run $\gamma\in\{0,0.05,0.1,\ldots,0.5\}$, 1000 Monte Carlo samples,
and 999 bootstrap draws.  Panels (a) and (b) fix Regime I and Regime D,
respectively; panel (c) applies the Ljung--Box pre-test with $\delta=0.05$.
For the horizon-pooled calibration tables below, we simulate only
$\rho\in\{0,0.15,0.3\}$.

\subsubsection{Simulation results for $(\text{sarima}+\text{macro})$ vs three experts}

We show the detailed simulation results for the most important comparison,
$(\text{sarima}+\text{macro})$ vs three experts, in this subsection.
This comparison asks whether adding the micro forecast to the two-expert
combination improves the horizon-pooled pre/post-COVID performance.  We report
the fixed Regime I, fixed Regime D, and pre-test versions here so that the
effect of the regime-selection step is visible in the main horizon-pooled
calibration.  To save space, the next subsection documents the simulation results for
$\gamma\in\{0.05,0.1,0.15\}$ with the Ljung--Box pre-test for the other
comparisons.

\begin{table}[H]
\centering
\caption{Pooled-window simulation for $(\text{sarima}+\text{macro})-\text{three experts}$, pre-COVID window. Entries are rejection probabilities at nominal level $\alpha_0=0.05$.}
\label{tab:window_sim_pre_three_expert}
\begin{subtable}[t]{0.31\textwidth}
\centering
\caption{Regime I}
\label{tab:window_sim_pre_three_expert_regime_i}
\small
\setlength{\tabcolsep}{2pt}
\begin{tabular}{lccc}
\toprule
$\gamma$ & $\rho=0$ & $\rho=.15$ & $\rho=.3$ \\
\midrule
0.00 & 0.053 & 0.060 & 0.046 \\
0.05 & 0.053 & 0.062 & 0.046 \\
0.10 & 0.052 & 0.062 & 0.046 \\
0.15 & 0.051 & 0.060 & 0.045 \\
0.20 & 0.051 & 0.058 & 0.044 \\
0.25 & 0.051 & 0.059 & 0.045 \\
0.30 & 0.050 & 0.056 & 0.044 \\
0.35 & 0.051 & 0.057 & 0.043 \\
0.40 & 0.050 & 0.055 & 0.043 \\
0.45 & 0.050 & 0.055 & 0.042 \\
0.50 & 0.048 & 0.053 & 0.041 \\
\bottomrule
\end{tabular}
\end{subtable}
\hfill
\begin{subtable}[t]{0.31\textwidth}
\centering
\caption{Regime D}
\label{tab:window_sim_pre_three_expert_regime_d}
\small
\setlength{\tabcolsep}{2pt}
\begin{tabular}{lccc}
\toprule
$\gamma$ & $\rho=0$ & $\rho=.15$ & $\rho=.3$ \\
\midrule
0.00 & 0.056 & 0.062 & 0.047 \\
0.05 & 0.057 & 0.061 & 0.046 \\
0.10 & 0.056 & 0.062 & 0.047 \\
0.15 & 0.055 & 0.062 & 0.046 \\
0.20 & 0.055 & 0.061 & 0.048 \\
0.25 & 0.055 & 0.063 & 0.046 \\
0.30 & 0.054 & 0.061 & 0.046 \\
0.35 & 0.054 & 0.061 & 0.047 \\
0.40 & 0.053 & 0.060 & 0.045 \\
0.45 & 0.054 & 0.059 & 0.043 \\
0.50 & 0.052 & 0.058 & 0.043 \\
\bottomrule
\end{tabular}
\end{subtable}
\hfill
\begin{subtable}[t]{0.31\textwidth}
\centering
\caption{Pre-test}
\label{tab:window_sim_pre_three_expert_pretest}
\footnotesize
\setlength{\tabcolsep}{2pt}
\begin{tabular}{lccc}
\toprule
$\gamma$ & $\rho=0$ & $\rho=.15$ & $\rho=.3$ \\
\midrule
0.00 & 0.050 & 0.038 & 0.063 \\
0.05 & 0.051 & 0.037 & 0.060 \\
0.10 & 0.048 & 0.035 & 0.063 \\
0.15 & 0.048 & 0.035 & 0.060 \\
0.20 & 0.048 & 0.035 & 0.060 \\
0.25 & 0.048 & 0.035 & 0.060 \\
0.30 & 0.049 & 0.033 & 0.060 \\
0.35 & 0.048 & 0.032 & 0.058 \\
0.40 & 0.048 & 0.032 & 0.056 \\
0.45 & 0.048 & 0.031 & 0.056 \\
0.50 & 0.048 & 0.031 & 0.056 \\
\bottomrule
\end{tabular}
\end{subtable}
\end{table}

\begin{table}[H]
\centering
\caption{Pooled-window simulation for $(\text{sarima}+\text{macro})-\text{three experts}$, post-COVID window. Entries are rejection probabilities at nominal level $\alpha_0=0.05$.}
\label{tab:window_sim_post_three_expert}
\begin{subtable}[t]{0.31\textwidth}
\centering
\caption{Regime I}
\label{tab:window_sim_post_three_expert_regime_i}
\footnotesize
\setlength{\tabcolsep}{2pt}
\begin{tabular}{lccc}
\toprule
$\gamma$ & $\rho=0$ & $\rho=.15$ & $\rho=.3$ \\
\midrule
0.00 & 0.054 & 0.057 & 0.065 \\
0.05 & 0.054 & 0.058 & 0.065 \\
0.10 & 0.053 & 0.058 & 0.063 \\
0.15 & 0.053 & 0.058 & 0.063 \\
0.20 & 0.053 & 0.059 & 0.065 \\
0.25 & 0.052 & 0.059 & 0.064 \\
0.30 & 0.051 & 0.058 & 0.064 \\
0.35 & 0.051 & 0.058 & 0.064 \\
0.40 & 0.049 & 0.057 & 0.064 \\
0.45 & 0.048 & 0.056 & 0.065 \\
0.50 & 0.049 & 0.056 & 0.065 \\
\bottomrule
\end{tabular}
\end{subtable}
\hfill
\begin{subtable}[t]{0.31\textwidth}
\centering
\caption{Regime D}
\label{tab:window_sim_post_three_expert_regime_d}
\footnotesize
\setlength{\tabcolsep}{2pt}
\begin{tabular}{lccc}
\toprule
$\gamma$ & $\rho=0$ & $\rho=.15$ & $\rho=.3$ \\
\midrule
0.00 & 0.054 & 0.055 & 0.065 \\
0.05 & 0.054 & 0.054 & 0.065 \\
0.10 & 0.054 & 0.057 & 0.064 \\
0.15 & 0.054 & 0.058 & 0.065 \\
0.20 & 0.054 & 0.057 & 0.065 \\
0.25 & 0.053 & 0.055 & 0.065 \\
0.30 & 0.051 & 0.056 & 0.061 \\
0.35 & 0.051 & 0.056 & 0.062 \\
0.40 & 0.052 & 0.054 & 0.062 \\
0.45 & 0.049 & 0.053 & 0.058 \\
0.50 & 0.049 & 0.053 & 0.057 \\
\bottomrule
\end{tabular}
\end{subtable}
\hfill
\begin{subtable}[t]{0.31\textwidth}
\centering
\caption{Pre-test}
\label{tab:window_sim_post_three_expert_pretest}
\footnotesize
\setlength{\tabcolsep}{2pt}
\begin{tabular}{lccc}
\toprule
$\gamma$ & $\rho=0$ & $\rho=.15$ & $\rho=.3$ \\
\midrule
0.00 & 0.068 & 0.055 & 0.070 \\
0.05 & 0.069 & 0.056 & 0.070 \\
0.10 & 0.069 & 0.056 & 0.069 \\
0.15 & 0.070 & 0.054 & 0.066 \\
0.20 & 0.071 & 0.054 & 0.067 \\
0.25 & 0.070 & 0.054 & 0.069 \\
0.30 & 0.068 & 0.051 & 0.067 \\
0.35 & 0.068 & 0.051 & 0.066 \\
0.40 & 0.067 & 0.050 & 0.064 \\
0.45 & 0.066 & 0.050 & 0.062 \\
0.50 & 0.064 & 0.050 & 0.061 \\
\bottomrule
\end{tabular}
\end{subtable}
\end{table}

\subsubsection{Other comparisons}

For the remaining horizon-pooled comparisons in Table~\ref{tab:horizon_pooled_ar1}, we focus on
the pre-test version with $\gamma\in\{0.05,0.1,0.15\}$.  These are the
specifications most relevant for the empirical tables, because the applied
procedure first diagnoses serial dependence and then chooses the corresponding
bootstrap regime.  The neighboring values of $\gamma$ show how sensitive the
size calibration is around the preferred value $\gamma=0.1$.  Each table fixes one
value of $\gamma$ and reports the pre-COVID rows first, followed by the
post-COVID rows, with comparisons ordered as in Table~\ref{tab:horizon_pooled_ar1}.

\begin{table}[H]
\centering
\caption{Horizon-pooled pre-test simulation for the other comparisons with $\gamma=0.05$ and $\delta=0.05$. Entries are rejection probabilities at nominal level $\alpha_0=0.05$.}
\label{tab:window_other_pretest_gamma_05}
\normalsize
\setlength{\tabcolsep}{4pt}
\begin{tabular}{llccc}
\toprule
Window & Comparison & $\rho=0$ & $\rho=.15$ & $\rho=.3$ \\
\midrule
Pre-COVID & sarima--micro & 0.052 & 0.046 & 0.054 \\
Pre-COVID & sarima--macro & 0.055 & 0.059 & 0.047 \\
Post-COVID & sarima--micro & 0.061 & 0.061 & 0.069 \\
Post-COVID & sarima--macro & 0.054 & 0.054 & 0.072 \\
\bottomrule
\end{tabular}
\end{table}

\begin{table}[H]
\centering
\caption{Horizon-pooled pre-test simulation for the other comparisons with $\gamma=0.10$ and $\delta=0.05$. Entries are rejection probabilities at nominal level $\alpha_0=0.05$.}
\label{tab:window_other_pretest_gamma_10}
\normalsize
\setlength{\tabcolsep}{4pt}
\begin{tabular}{llccc}
\toprule
Window & Comparison & $\rho=0$ & $\rho=.15$ & $\rho=.3$ \\
\midrule
Pre-COVID & sarima--micro & 0.053 & 0.047 & 0.052 \\
Pre-COVID & sarima--macro & 0.055 & 0.059 & 0.047 \\
Post-COVID & sarima--micro & 0.061 & 0.061 & 0.070 \\
Post-COVID & sarima--macro & 0.053 & 0.054 & 0.073 \\
\bottomrule
\end{tabular}
\end{table}

\begin{table}[H]
\centering
\caption{Horizon-pooled pre-test simulation for the other comparisons with $\gamma=0.15$ and $\delta=0.05$. Entries are rejection probabilities at nominal level $\alpha_0=0.05$.}
\label{tab:window_other_pretest_gamma_15}
\normalsize
\setlength{\tabcolsep}{4pt}
\begin{tabular}{llccc}
\toprule
Window & Comparison & $\rho=0$ & $\rho=.15$ & $\rho=.3$ \\
\midrule
Pre-COVID & sarima--micro & 0.053 & 0.047 & 0.053 \\
Pre-COVID & sarima--macro & 0.054 & 0.059 & 0.047 \\
Post-COVID & sarima--micro & 0.060 & 0.061 & 0.069 \\
Post-COVID & sarima--macro & 0.055 & 0.054 & 0.072 \\
\bottomrule
\end{tabular}
\end{table}

\subsection{Giacomini--Rossi fluctuation test}\label{app:sim_gr}

Section~\ref{sec:scantest} discusses the \citet{GiacominiRossi2010} fluctuation test as the
closest classical comparison: it also looks for forecast-instability episodes,
but it uses fixed window lengths and fixed horizons rather than a joint scan
over unknown intervals and horizons.  We include it here to make clear what
changes when the localized-outperformance problem is treated with the classical
Brownian-limit calibration instead of the one-sided multiplier-bootstrap scan
used in the main procedure.

For comparison, we implement a one-sided version of this fluctuation test.  For
a fixed horizon $h$ and a fixed window length
$m$, define
\[
G_m^{(h)}=
\max_{1\le s\le n-m+1}
\frac{\sum_{t=s}^{s+m-1}D_t^{(h)}}
{\sqrt{m\widehat\sigma_h^2}},
\]
where $\widehat\sigma_h^2$ is a full-sample Newey--West long-run variance
estimate for the loss differential at horizon $h$.  We use lag $h-1$ in this
estimate.  The alternative is one-sided: large positive values indicate that
the first forecast has larger loss than the second forecast over some
fixed-length window.

Under the local equal-predictive-ability null and a functional central limit
theorem,
\[
\max_s
\frac{\sum_{t=s}^{s+m-1}D_t^{(h)}}{\sqrt{m\widehat\sigma_h^2}}
\Rightarrow
\sup_s \frac{W(s+m/n)-W(s)}{\sqrt{m/n}},
\]
where $W$ is standard Brownian motion.  In the simulations below we approximate
this null distribution on the observed grid by simulating iid standard-normal
increments and taking the maximum rolling sum divided by $\sqrt m$.

We run the same sign-AR(1) null simulation described above with
$m\in\{6,7,\ldots,12\}$, 1000 Monte Carlo samples, and 200000 Brownian-limit
draws for each value of $m$.  The test is horizon-by-horizon, so each table
reports the mean rejection probability across the 24 horizons.  This mirrors
the occasional-outperformance calibration while respecting the fixed-window,
fixed-horizon structure of the Giacomini--Rossi statistic.

\begin{table}[H]
\centering
\caption{Giacomini--Rossi simulation for sarima--micro. Entries are mean rejection probabilities across horizons $h=1,\ldots,24$ at nominal level $\alpha_0=0.05$; brackets give approximate 95\% Monte Carlo intervals.}
\label{tab:gr_section2_sarima_micro}
\footnotesize
\setlength{\tabcolsep}{6pt}
\renewcommand{\arraystretch}{1.24}
\begin{tabular}{lcccc}
\toprule
$m$ & $\rho=-.25$ & $\rho=0$ & $\rho=.25$ & $\rho=.5$ \\
\midrule
6 & 0.41 [0.38, 0.45] & 0.41 [0.38, 0.45] & 0.43 [0.39, 0.46] & 0.42 [0.39, 0.45] \\
7 & 0.39 [0.36, 0.42] & 0.40 [0.36, 0.43] & 0.40 [0.37, 0.43] & 0.40 [0.37, 0.43] \\
8 & 0.37 [0.34, 0.40] & 0.37 [0.34, 0.40] & 0.37 [0.34, 0.40] & 0.37 [0.34, 0.40] \\
9 & 0.34 [0.31, 0.37] & 0.33 [0.30, 0.36] & 0.34 [0.31, 0.36] & 0.33 [0.30, 0.36] \\
10 & 0.31 [0.28, 0.34] & 0.31 [0.28, 0.34] & 0.31 [0.28, 0.34] & 0.31 [0.28, 0.34] \\
11 & 0.29 [0.26, 0.32] & 0.28 [0.25, 0.31] & 0.28 [0.25, 0.31] & 0.28 [0.25, 0.31] \\
12 & 0.27 [0.24, 0.30] & 0.26 [0.23, 0.29] & 0.27 [0.24, 0.29] & 0.26 [0.23, 0.29] \\
\bottomrule
\end{tabular}
\renewcommand{\arraystretch}{1}
\end{table}

\begin{table}[H]
\centering
\caption{Giacomini--Rossi simulation for macro--micro. Entries are mean rejection probabilities across horizons $h=1,\ldots,24$ at nominal level $\alpha_0=0.05$; brackets give approximate 95\% Monte Carlo intervals.}
\label{tab:gr_section2_macro_micro}
\footnotesize
\setlength{\tabcolsep}{6pt}
\renewcommand{\arraystretch}{1.24}
\begin{tabular}{lcccc}
\toprule
$m$ & $\rho=-.25$ & $\rho=0$ & $\rho=.25$ & $\rho=.5$ \\
\midrule
6 & 0.27 [0.24, 0.30] & 0.27 [0.25, 0.30] & 0.27 [0.25, 0.30] & 0.27 [0.25, 0.30] \\
7 & 0.24 [0.21, 0.27] & 0.25 [0.22, 0.27] & 0.25 [0.22, 0.27] & 0.25 [0.22, 0.28] \\
8 & 0.22 [0.20, 0.25] & 0.23 [0.20, 0.25] & 0.23 [0.20, 0.25] & 0.23 [0.20, 0.26] \\
9 & 0.20 [0.18, 0.23] & 0.20 [0.18, 0.23] & 0.21 [0.18, 0.23] & 0.22 [0.19, 0.24] \\
10 & 0.18 [0.16, 0.21] & 0.19 [0.16, 0.21] & 0.19 [0.17, 0.21] & 0.20 [0.17, 0.22] \\
11 & 0.16 [0.14, 0.19] & 0.17 [0.15, 0.19] & 0.18 [0.15, 0.20] & 0.18 [0.16, 0.20] \\
12 & 0.16 [0.14, 0.18] & 0.17 [0.14, 0.19] & 0.17 [0.15, 0.19] & 0.17 [0.15, 0.20] \\
\bottomrule
\end{tabular}
\renewcommand{\arraystretch}{1}
\end{table}

\begin{table}[H]
\centering
\caption{Giacomini--Rossi simulation for sarima--macro. Entries are mean rejection probabilities across horizons $h=1,\ldots,24$ at nominal level $\alpha_0=0.05$; brackets give approximate 95\% Monte Carlo intervals.}
\label{tab:gr_section2_sarima_macro}
\small
\setlength{\tabcolsep}{6pt}
\renewcommand{\arraystretch}{1.24}
\begin{tabular}{lcccc}
\toprule
$m$ & $\rho=-.25$ & $\rho=0$ & $\rho=.25$ & $\rho=.5$ \\
\midrule
6 & 0.51 [0.48, 0.54] & 0.51 [0.48, 0.54] & 0.50 [0.47, 0.53] & 0.50 [0.47, 0.53] \\
7 & 0.48 [0.45, 0.51] & 0.48 [0.45, 0.51] & 0.47 [0.43, 0.50] & 0.47 [0.44, 0.50] \\
8 & 0.44 [0.41, 0.47] & 0.44 [0.41, 0.47] & 0.43 [0.40, 0.47] & 0.45 [0.41, 0.48] \\
9 & 0.41 [0.38, 0.44] & 0.41 [0.38, 0.44] & 0.40 [0.37, 0.43] & 0.41 [0.38, 0.45] \\
10 & 0.38 [0.35, 0.41] & 0.38 [0.35, 0.41] & 0.37 [0.34, 0.40] & 0.38 [0.35, 0.41] \\
11 & 0.35 [0.32, 0.38] & 0.36 [0.33, 0.39] & 0.35 [0.32, 0.38] & 0.35 [0.32, 0.38] \\
12 & 0.33 [0.30, 0.36] & 0.33 [0.30, 0.36] & 0.32 [0.30, 0.35] & 0.33 [0.30, 0.36] \\
\bottomrule
\end{tabular}
\renewcommand{\arraystretch}{1}
\end{table}

We also apply the Giacomini--Rossi statistic to the horizon-pooled sign-AR(1)
simulations used for the pre-COVID and post-COVID horizon-pooled analyses above.  The input series is
$\bar D_t=H^{-1}\sum_{h=1}^{24}D_t^{(h)}$ for
$(\text{sarima}+\text{macro})-\text{three experts}$.  The sign-AR(1)
multiplier is generated on the full balanced sample and then restricted to the
pre-specified calendar window.  Because this series pools horizons up to
$H=24$, the Newey--West long-run variance uses lag 23.

\begin{table}[H]
\centering
\caption{Giacomini--Rossi horizon-pooled simulation for $(\text{sarima}+\text{macro})-\text{three experts}$, pre-COVID window. Entries are rejection probabilities at nominal level $\alpha_0=0.05$; brackets give approximate 95\% Monte Carlo intervals.}
\label{tab:gr_window_pre_covid}
\small
\setlength{\tabcolsep}{6pt}
\renewcommand{\arraystretch}{1.24}
\begin{tabular}{lccc}
\toprule
$m$ & $\rho=-.25$ & $\rho=0$ & $\rho=.25$ \\
\midrule
6 & 0.48 [0.45, 0.51] & 0.51 [0.47, 0.54] & 0.52 [0.49, 0.55] \\
7 & 0.47 [0.44, 0.50] & 0.48 [0.45, 0.51] & 0.50 [0.47, 0.53] \\
8 & 0.46 [0.43, 0.49] & 0.46 [0.43, 0.49] & 0.46 [0.43, 0.49] \\
9 & 0.43 [0.40, 0.46] & 0.45 [0.42, 0.48] & 0.45 [0.42, 0.48] \\
10 & 0.44 [0.41, 0.47] & 0.44 [0.41, 0.47] & 0.44 [0.41, 0.47] \\
11 & 0.39 [0.36, 0.42] & 0.40 [0.37, 0.43] & 0.40 [0.36, 0.43] \\
12 & 0.39 [0.36, 0.42] & 0.42 [0.39, 0.45] & 0.40 [0.36, 0.43] \\
\bottomrule
\end{tabular}
\renewcommand{\arraystretch}{1}
\end{table}

\begin{table}[H]
\centering
\caption{Giacomini--Rossi horizon-pooled simulation for $(\text{sarima}+\text{macro})-\text{three experts}$, post-COVID window. Entries are rejection probabilities at nominal level $\alpha_0=0.05$; brackets give approximate 95\% Monte Carlo intervals.}
\label{tab:gr_window_post_covid}
\small
\setlength{\tabcolsep}{6pt}
\renewcommand{\arraystretch}{1.24}
\begin{tabular}{lccc}
\toprule
$m$ & $\rho=-.25$ & $\rho=0$ & $\rho=.25$ \\
\midrule
6 & 0.42 [0.39, 0.45] & 0.43 [0.40, 0.46] & 0.45 [0.41, 0.48] \\
7 & 0.41 [0.38, 0.44] & 0.42 [0.39, 0.45] & 0.43 [0.40, 0.46] \\
8 & 0.41 [0.38, 0.44] & 0.40 [0.37, 0.43] & 0.41 [0.38, 0.44] \\
9 & 0.36 [0.33, 0.39] & 0.35 [0.32, 0.38] & 0.37 [0.34, 0.40] \\
10 & 0.34 [0.31, 0.37] & 0.35 [0.32, 0.38] & 0.33 [0.30, 0.36] \\
11 & 0.32 [0.29, 0.34] & 0.33 [0.30, 0.36] & 0.33 [0.30, 0.36] \\
12 & 0.29 [0.26, 0.32] & 0.31 [0.28, 0.34] & 0.31 [0.28, 0.33] \\
\bottomrule
\end{tabular}
\renewcommand{\arraystretch}{1}
\end{table}

\clearpage
\section{Appendix Figures}
\label{app:figures}

\vspace*{\fill}
\begin{figure}[H]
  \centering
  \caption{Micro Forecast vs.\ Univariate Forecast: $h=1$}
  \label{fig:xgb_micro_vs_sarima_h1}
  \includegraphics[width=\textwidth]{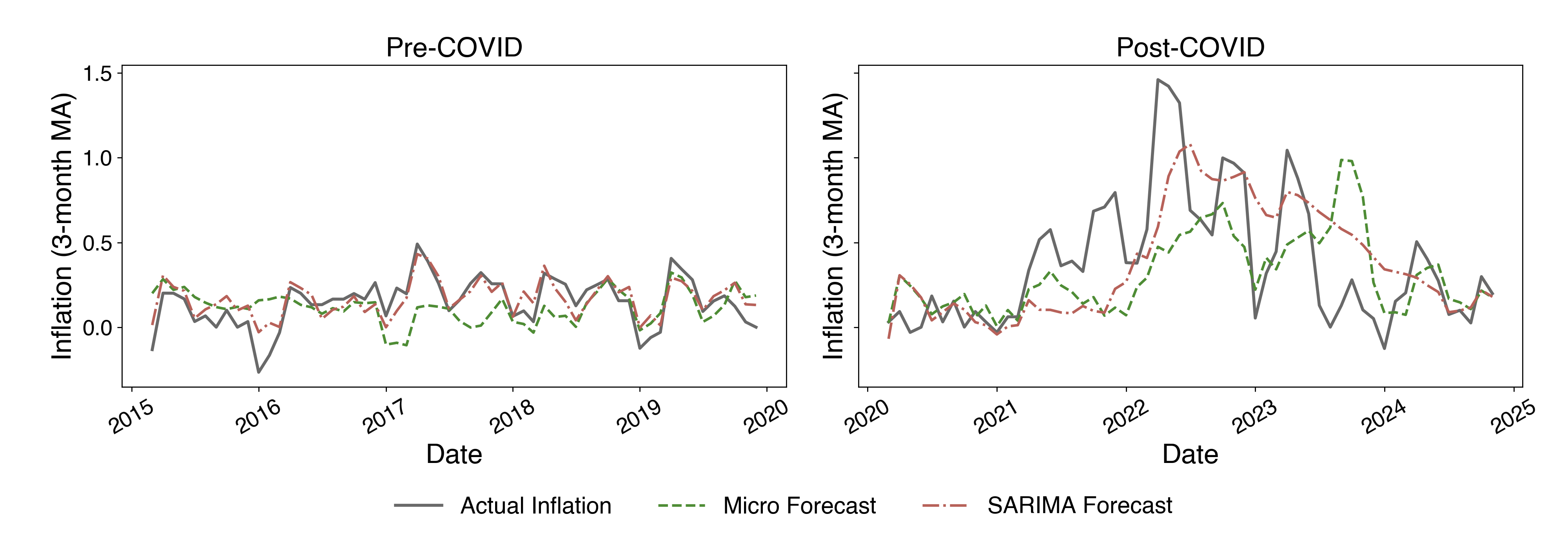}
  \smallskip
  \begin{minipage}{\textwidth}\setstretch{1.0}
    \footnotesize\textit{Notes:} Each panel plots 3-month trailing moving
    averages of realized monthly inflation, the micro 1-month-ahead forecast,
    and the univariate (SARIMA) 1-month-ahead forecast. Dates on the x-axis
    are target months: a forecast dated $t$ was formed at origin $t-1$.
    The panels correspond to the 2015--2019 and 2020--2024 evaluation windows.
  \end{minipage}
\end{figure}
\vspace*{\fill}

\clearpage
\begin{figure}[t!]
  \centering
  \caption{Micro Forecast vs.\ Univariate Forecast: $h=6$}
  \label{fig:xgb_micro_vs_sarima_h6}
  \includegraphics[width=\textwidth]{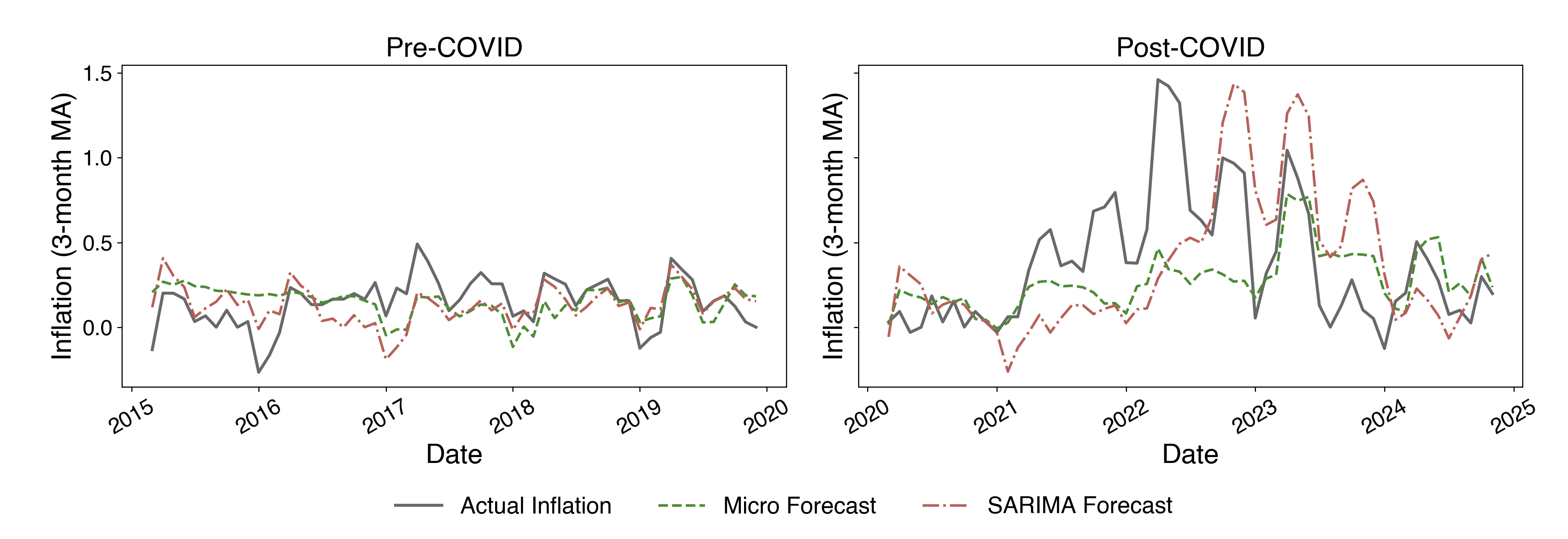}
  \smallskip
  \begin{minipage}{\textwidth}\setstretch{1.0}
    \footnotesize\textit{Notes:} Each panel plots 3-month trailing moving
    averages of realized monthly inflation, the micro 6-month-ahead forecast,
    and the univariate (SARIMA) 6-month-ahead forecast. Dates on the x-axis
    are target months: a forecast dated $t$ was formed at origin $t-6$.
    The panels correspond to the 2015--2019 and 2020--2024 evaluation windows.
  \end{minipage}
\end{figure}

\clearpage
\begin{figure}[t!]
  \centering
  \caption{Micro Forecast vs.\ Univariate Forecast: $h=24$}
  \label{fig:xgb_micro_vs_sarima_h24}
  \includegraphics[width=\textwidth]{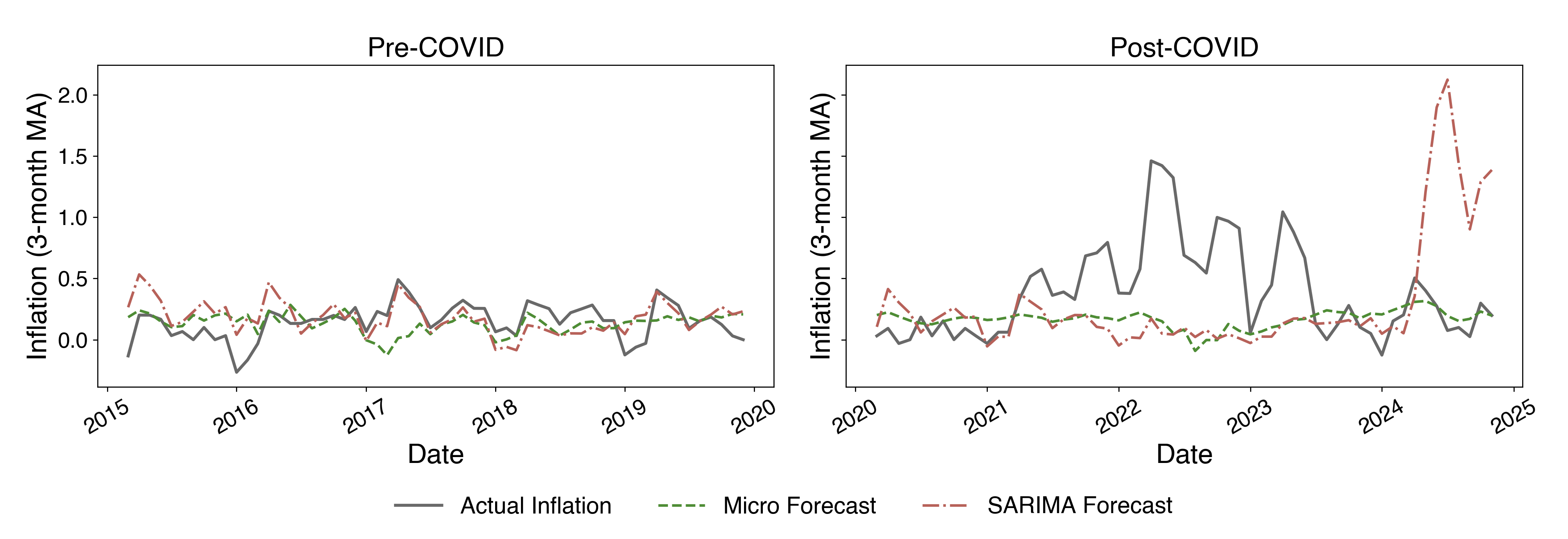}
  \smallskip
  \begin{minipage}{\textwidth}\setstretch{1.0}
    \footnotesize\textit{Notes:} Each panel plots 3-month trailing moving
    averages of realized monthly inflation, the micro 24-month-ahead forecast,
    and the univariate (SARIMA) 24-month-ahead forecast. Dates on the x-axis
    are target months: a forecast dated $t$ was formed at origin $t-24$.
    The panels correspond to the 2015--2019 and 2020--2024 evaluation windows.
  \end{minipage}
\end{figure}

\clearpage
\begin{figure}[t!]
  \centering
  \caption{Combination Weights: $h=1$}
  \label{fig:online_h1}
  \includegraphics[width=\textwidth]{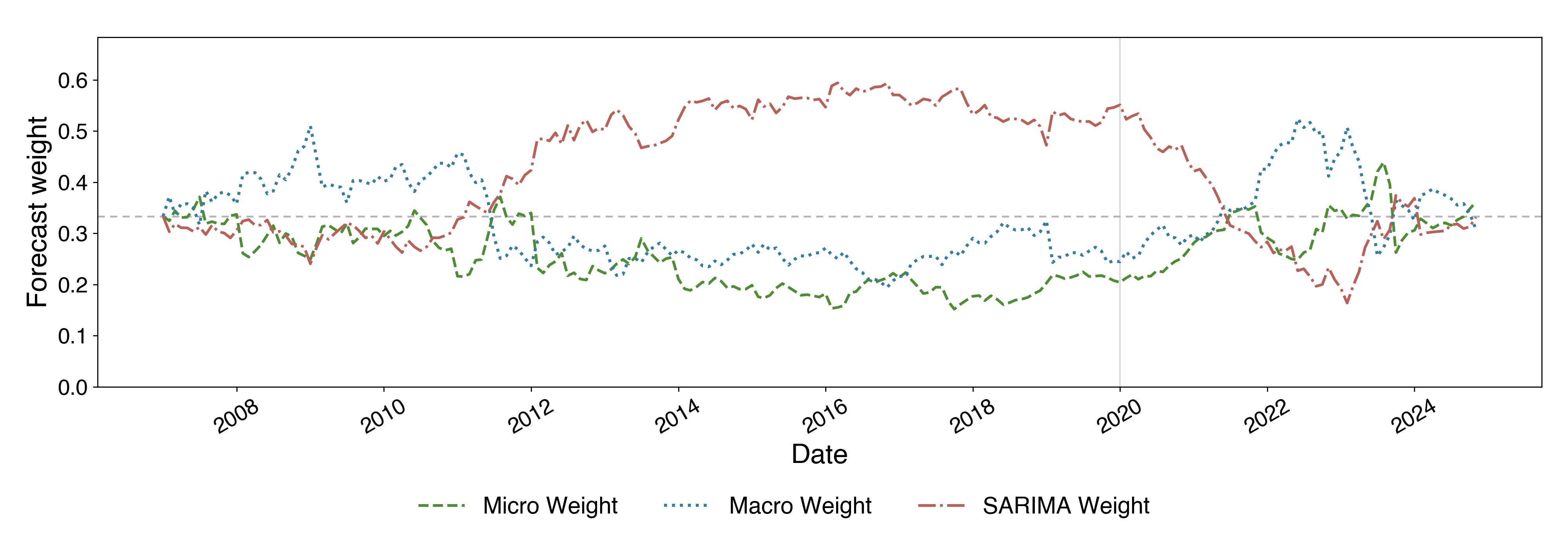}
  \smallskip
  \begin{minipage}{\textwidth}\setstretch{1.0}
    \footnotesize\textit{Notes:} Combination weights of the Fixed Share
    baseline ($\eta=0.5$, $\alpha=0.02$) on the micro, macro, and univariate
    experts for the 1-month-ahead forecast, over the full sample. Dates on
    the x-axis are target months, so date $t$ corresponds to origin $t-1$.
    The dashed horizontal line marks the equal weight $1/3$; the vertical line
    marks 2020.
  \end{minipage}
\end{figure}

\clearpage
\begin{figure}[t!]
  \centering
  \caption{Combination Weights: $h=6$}
  \label{fig:online_h6}
  \includegraphics[width=\textwidth]{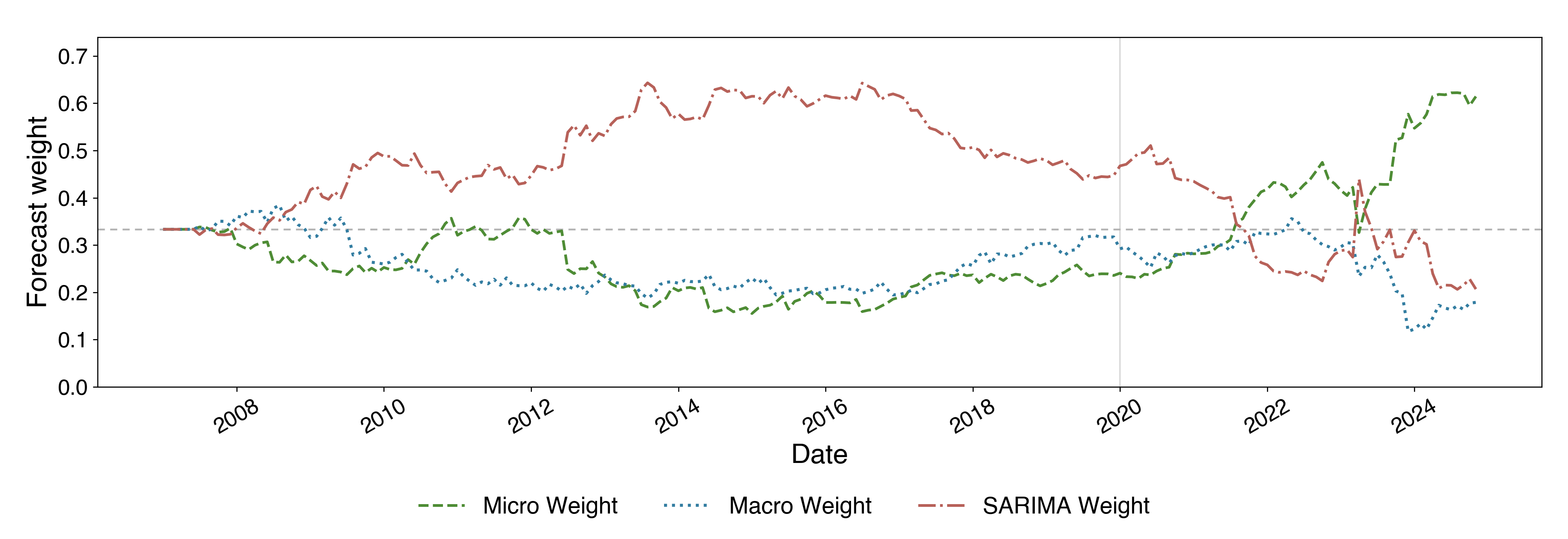}
  \smallskip
  \begin{minipage}{\textwidth}\setstretch{1.0}
    \footnotesize\textit{Notes:} Combination weights of the Fixed Share
    baseline ($\eta=0.5$, $\alpha=0.02$) on the micro, macro, and univariate
    experts for the 6-month-ahead forecast, over the full sample. Dates on
    the x-axis are target months, so date $t$ corresponds to origin $t-6$.
    The dashed horizontal line marks the equal weight $1/3$; the vertical line
    marks 2020.
  \end{minipage}
\end{figure}

\clearpage
\begin{figure}[t!]
  \centering
  \caption{Combination Weights: $h=24$}
  \label{fig:online_h24}
  \includegraphics[width=\textwidth]{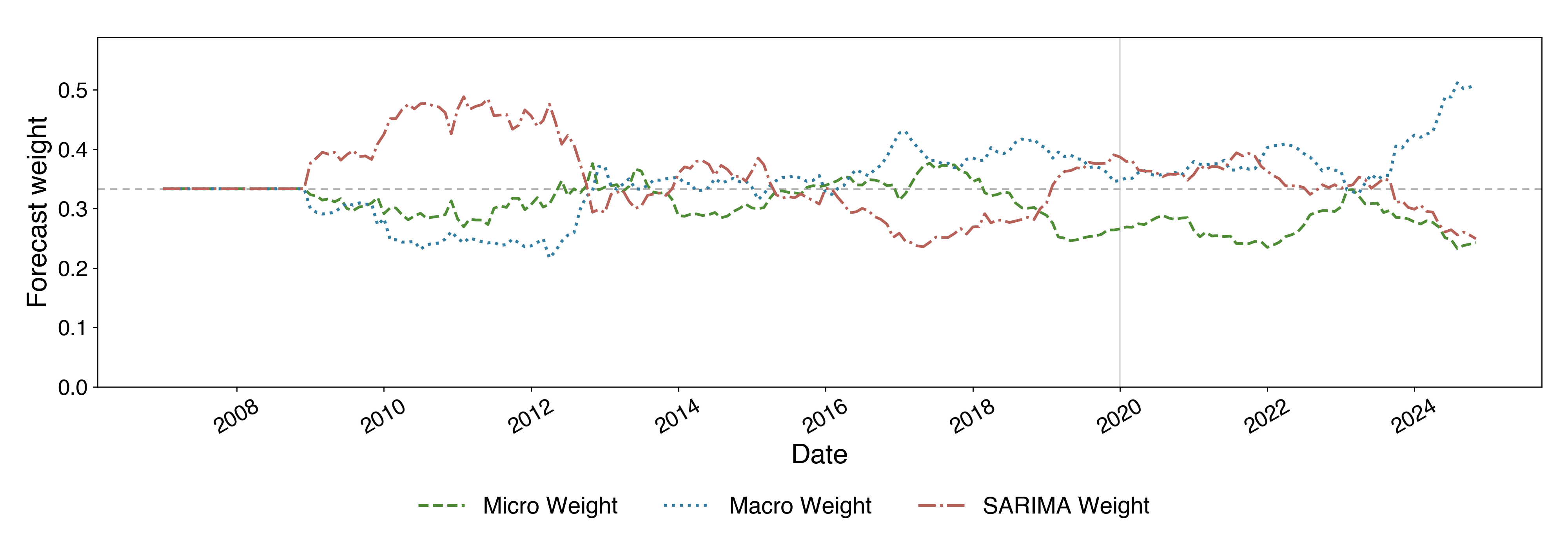}
  \smallskip
  \begin{minipage}{\textwidth}\setstretch{1.0}
    \footnotesize\textit{Notes:} Combination weights of the Fixed Share
    baseline ($\eta=0.5$, $\alpha=0.02$) on the micro, macro, and univariate
    experts for the 24-month-ahead forecast, over the full sample. Dates on
    the x-axis are target months, so date $t$ corresponds to origin $t-24$.
    The dashed horizontal line marks the equal weight $1/3$; the vertical line
    marks 2020.
  \end{minipage}
\end{figure}

\clearpage
\begin{figure}[t!]
  \centering
  \caption{Macro Forecast vs.\ Univariate Forecast: $h=1$}
  \label{fig:xgb_macro_vs_sarima_h1}
  \includegraphics[width=\textwidth]{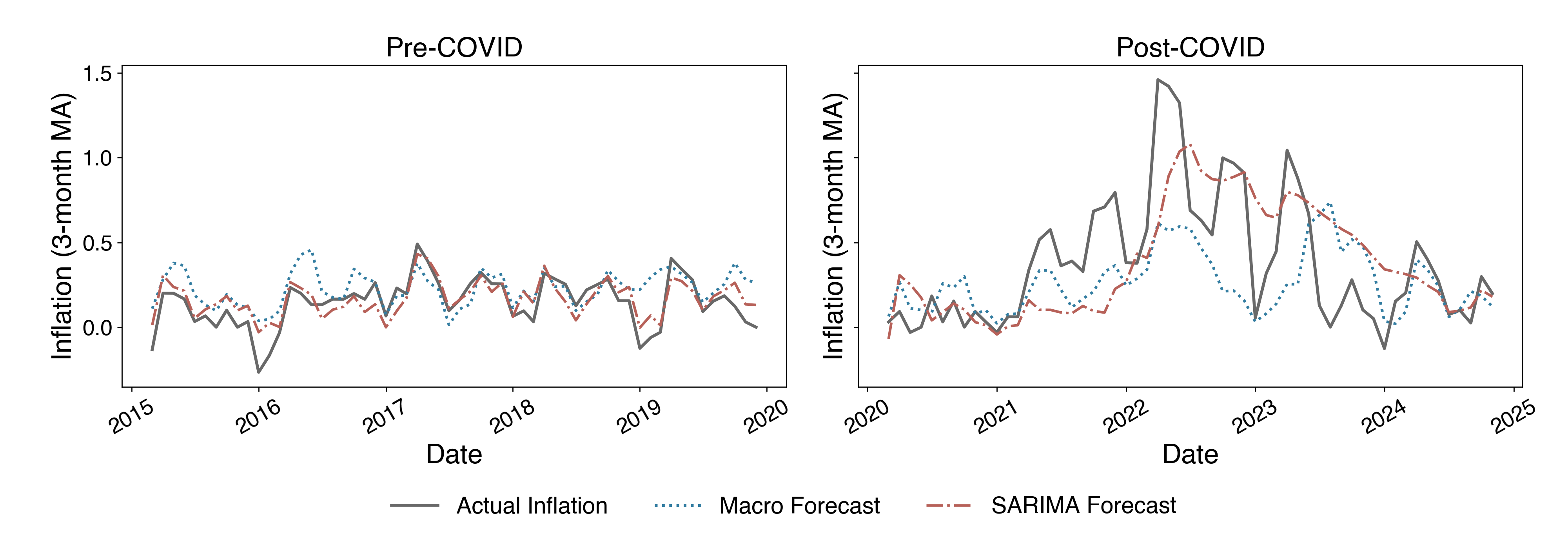}
  \smallskip
  \begin{minipage}{\textwidth}\setstretch{1.0}
    \footnotesize\textit{Notes:} Each panel plots 3-month trailing moving
    averages of realized monthly inflation, the macro 1-month-ahead forecast,
    and the univariate (SARIMA) 1-month-ahead forecast. Dates on the x-axis
    are target months: a forecast dated $t$ was formed at origin $t-1$.
    The panels correspond to the 2015--2019 and 2020--2024 evaluation windows.
  \end{minipage}
\end{figure}

\clearpage
\begin{figure}[t!]
  \centering
  \caption{Macro Forecast vs.\ Univariate Forecast: $h=6$}
  \label{fig:xgb_macro_vs_sarima_h6}
  \includegraphics[width=\textwidth]{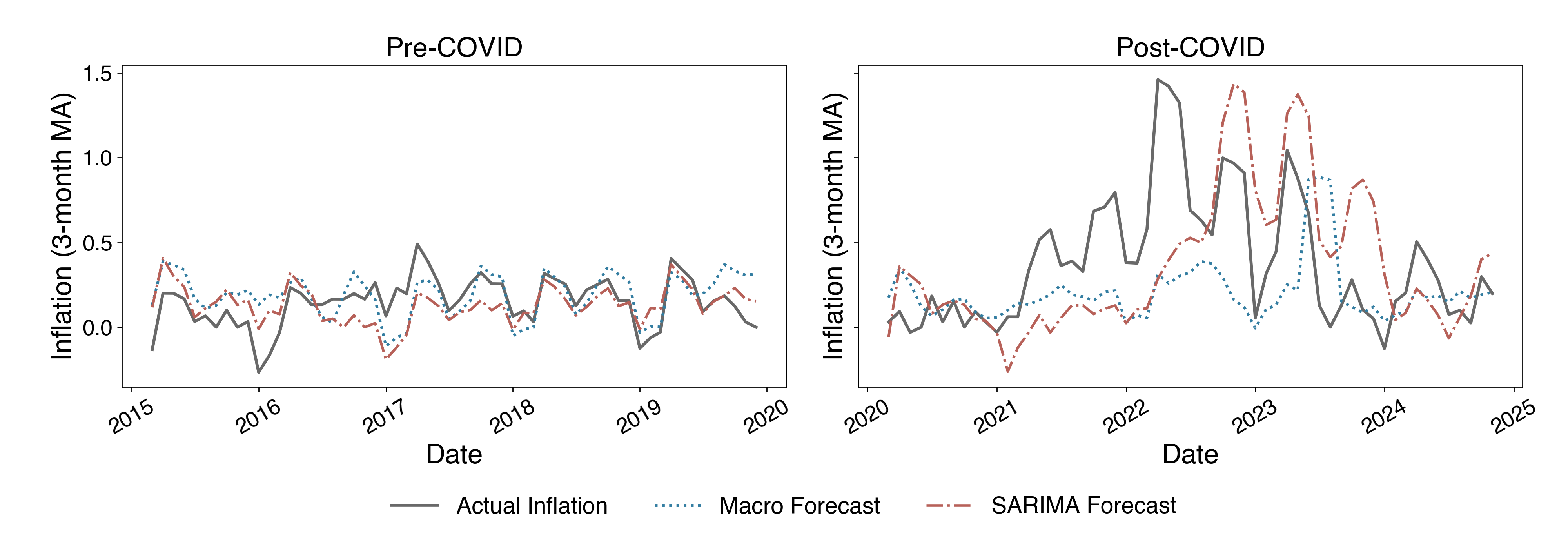}
  \smallskip
  \begin{minipage}{\textwidth}\setstretch{1.0}
    \footnotesize\textit{Notes:} Each panel plots 3-month trailing moving
    averages of realized monthly inflation, the macro 6-month-ahead forecast,
    and the univariate (SARIMA) 6-month-ahead forecast. Dates on the x-axis
    are target months: a forecast dated $t$ was formed at origin $t-6$.
    The panels correspond to the 2015--2019 and 2020--2024 evaluation windows.
  \end{minipage}
\end{figure}

\clearpage
\begin{figure}[t!]
  \centering
  \caption{Macro Forecast vs.\ Univariate Forecast: $h=24$}
  \label{fig:xgb_macro_vs_sarima_h24}
  \includegraphics[width=\textwidth]{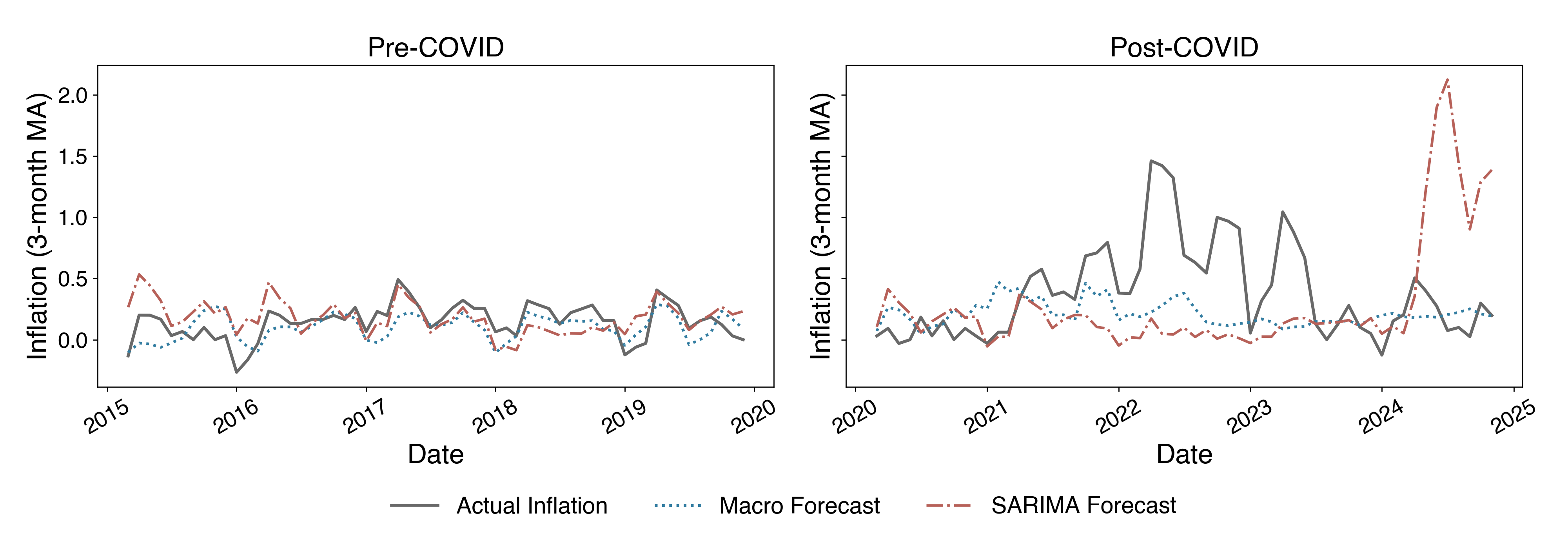}
  \smallskip
  \begin{minipage}{\textwidth}\setstretch{1.0}
    \footnotesize\textit{Notes:} Each panel plots 3-month trailing moving
    averages of realized monthly inflation, the macro 24-month-ahead forecast,
    and the univariate (SARIMA) 24-month-ahead forecast. Dates on the x-axis
    are target months: a forecast dated $t$ was formed at origin $t-24$.
    The panels correspond to the 2015--2019 and 2020--2024 evaluation windows.
  \end{minipage}
\end{figure}

\clearpage
\begin{figure}[t!]
  \centering
  \caption{Two-Expert Combination Weights: $h=1$}
  \label{fig:online_sm_h1}
  \includegraphics[width=\textwidth]{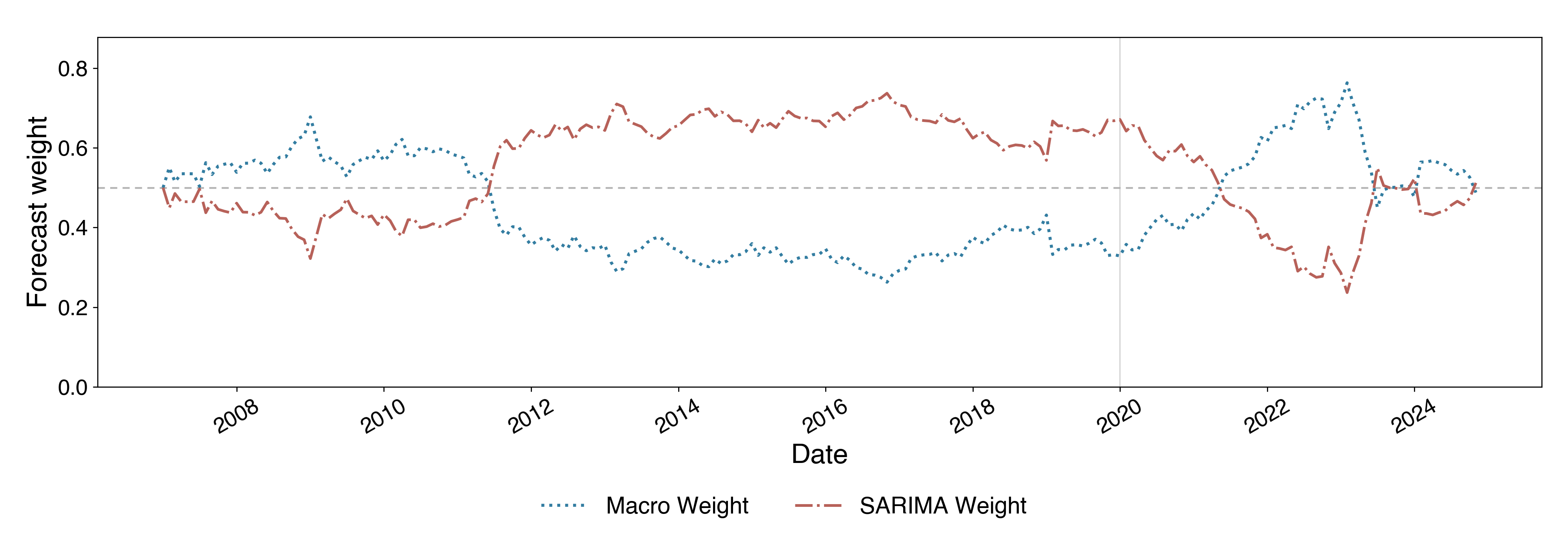}
  \smallskip
  \begin{minipage}{\textwidth}\setstretch{1.0}
    \footnotesize\textit{Notes:} Combination weights of the two-expert Fixed
    Share benchmark ($\eta=0.5$, $\alpha=0.02$) on the macro and univariate
    experts for the 1-month-ahead forecast, over the full sample. Dates on
    the x-axis are target months, so date $t$ corresponds to origin $t-1$.
    The dashed horizontal line marks the equal weight $1/2$; the vertical line
    marks 2020.
  \end{minipage}
\end{figure}

\clearpage
\begin{figure}[t!]
  \centering
  \caption{Two-Expert Combination Weights: $h=6$}
  \label{fig:online_sm_h6}
  \includegraphics[width=\textwidth]{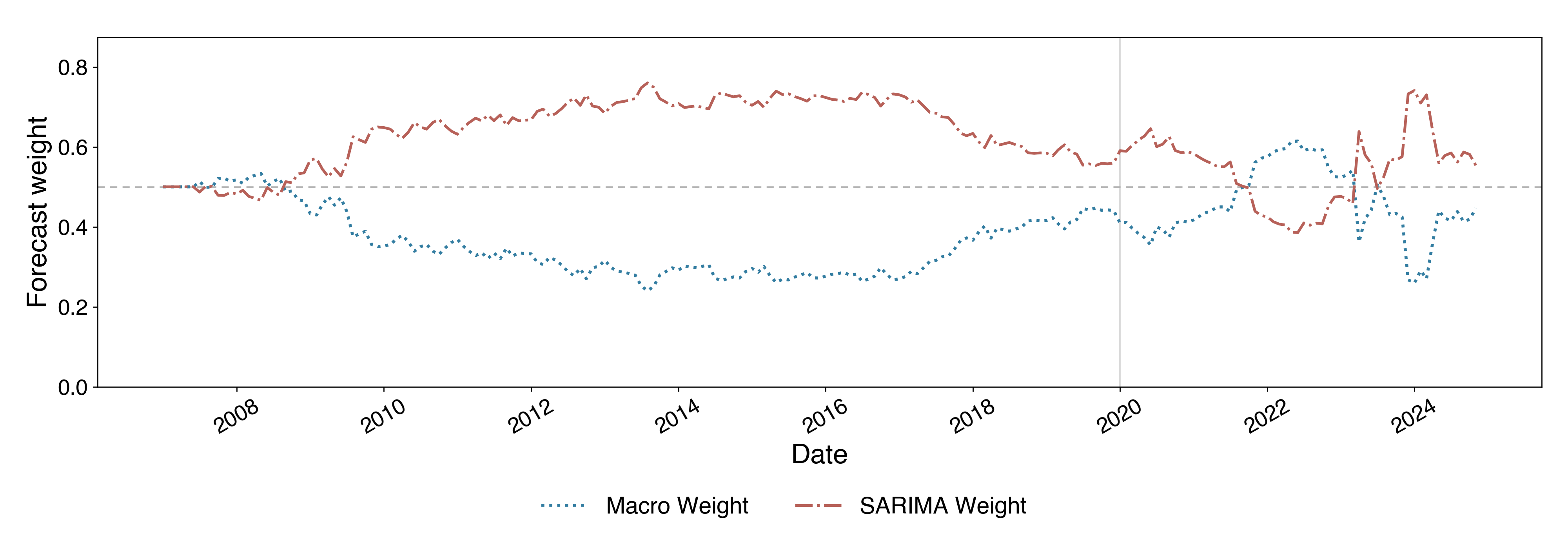}
  \smallskip
  \begin{minipage}{\textwidth}\setstretch{1.0}
    \footnotesize\textit{Notes:} Combination weights of the two-expert Fixed
    Share benchmark ($\eta=0.5$, $\alpha=0.02$) on the macro and univariate
    experts for the 6-month-ahead forecast, over the full sample. Dates on
    the x-axis are target months, so date $t$ corresponds to origin $t-6$.
    The dashed horizontal line marks the equal weight $1/2$; the vertical line
    marks 2020.
  \end{minipage}
\end{figure}

\clearpage
\begin{figure}[t!]
  \centering
  \caption{Two-Expert Combination Weights: $h=24$}
  \label{fig:online_sm_h24}
  \includegraphics[width=\textwidth]{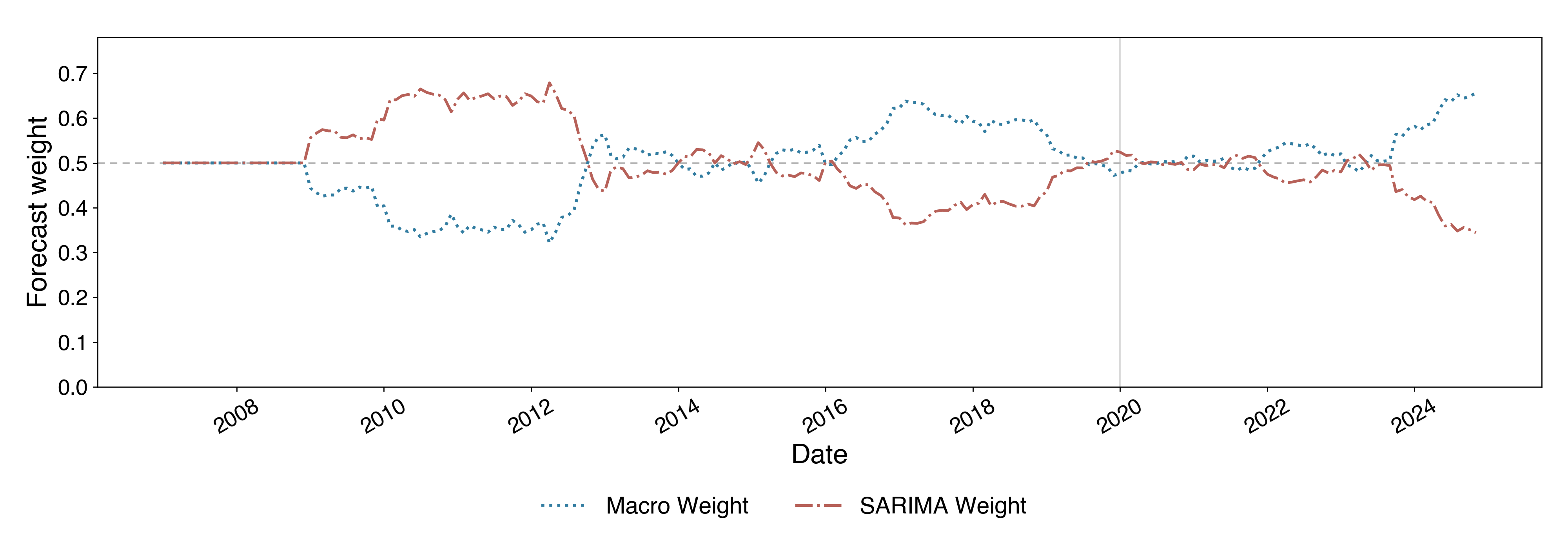}
  \smallskip
  \begin{minipage}{\textwidth}\setstretch{1.0}
    \footnotesize\textit{Notes:} Combination weights of the two-expert Fixed
    Share benchmark ($\eta=0.5$, $\alpha=0.02$) on the macro and univariate
    experts for the 24-month-ahead forecast, over the full sample. Dates on
    the x-axis are target months, so date $t$ corresponds to origin $t-24$.
    The dashed horizontal line marks the equal weight $1/2$; the vertical line
    marks 2020.
  \end{minipage}
\end{figure}

\clearpage
\section{Notation and Pseudocode}
\label{app:pseudocode}

Table~\ref{tab:notation} summarizes the main notation used in the forecasting,
scan-test, online-learning, and Shapley procedures described in
Sections~\ref{sec:methods}, \ref{sec:scantest},
and~\ref{sec:group_shapley}.

\begin{longtable}{p{2.6cm} p{11.6cm}}
  \caption{Main Notation Used in the Methods and Scan-Test Sections}
  \label{tab:notation}                                                                                                                                                                                                                   \\
  \toprule
  Symbol                       & Description                                                                                                                                                                                             \\
  \midrule
  \endfirsthead

  \multicolumn{2}{l}{\textit{Table \ref{tab:notation} continued}}                                                                                                                                                                        \\
  \toprule
  Symbol                       & Description                                                                                                                                                                                             \\
  \midrule
  \endhead

  \midrule
  \multicolumn{2}{r}{\textit{Continued on next page}}                                                                                                                                                                                    \\
  \endfoot

  \bottomrule
  \multicolumn{2}{p{14.2cm}}{\footnotesize\textit{Notes:} This table collects
  the notation used in the machine-learning, SARIMA, Fixed Share, scan-test,
  and Shapley procedures. Forecasts are written in the target-given-origin
  form: $\hat{\pi}_{a\mid b}$ is the forecast of inflation in month $a$ formed
  from information through month $b$. Loss
  differentials are written so that positive values favor the competing method
  $M$ over the benchmark method $B$.} \\
  \endlastfoot

  \multicolumn{2}{l}{\textit{Timing, data, and forecasts}} \\
  \addlinespace
  $t$                          & Calendar month index.                                                                                                                                                                                   \\
  $\tau$                       & Forecast origin: the last month whose information the forecast uses; a forecast formed at origin $\tau$ for horizon $h$ targets month $\tau+h$.                                                        \\
  $h$                          & Forecast horizon in months ahead; $h = 1,\ldots,24$.                                                                                                                                                    \\
  $H$                          & Maximum forecast horizon; $H = 24$.                                                                                                                                                                     \\
  $\pi_t$                      & Realized month-on-month inflation in month $t$.                                                                                                                                                         \\
  $\hat{\pi}_{\tau+h\mid\tau}$ & Forecast of inflation in target month $\tau+h$ formed from information through origin $\tau$; equivalently, $\hat{\pi}_{t\mid t-h}$ is the forecast of target $t$ formed at origin $t-h$.              \\
  $n$                          & Number of evaluated months in the forecast-comparison panel; $n = 215$ at the deployed sample.                                                                                                          \\
  \addlinespace
  \multicolumn{2}{l}{\textit{Features and cached real-time tuning}} \\
  \addlinespace
  $Z_t$                        & Monthly encoding vector observed in month $t$: the micro distribution statistics for the micro expert, or the macro series for the macro expert.                                                        \\
  $m_{t+h}$                    & Target-month indicators, $m_{t+h} \in \{0,1\}^{11}$ (January omitted).                                                                                                                                  \\
  $X_t^h$                      & Origin-$t$ feature vector for horizon $h$: $X_t^h = (Z_{t-11}, \ldots, Z_t, m_{t+h})$; every component is observed at origin $t$.                                                                       \\
  $\Gamma$                     & XGBoost hyperparameter vector (shrinkage $\nu$ and the number of trees; tree depth fixed at 6).                                                                                                        \\
  $\nu$                        & XGBoost learning rate (shrinkage); $\eta$ is reserved for the combination step.                                                                                                                        \\
  $\mathcal{G}$                & Hyperparameter grid; the baseline is the $7 \times 7$ grid $\nu \in \{0.0003,\ldots,0.3\}$, trees $\in \{50,\ldots,1000\}$. (In Section~\ref{sec:group_shapley}, $\mathcal{G}$ instead denotes the set of micro feature groups; see below.) \\
  $L$                          & Rolling training window; $L = 120$ months.                                                                                                                                                              \\
  $W$                          & Trailing validation window: $W = 48$ months for the cached real-time tuning of every expert; $W = 60$ months for the varying-learning-rate rule.                                                        \\
  $\hat{\pi}^{\Gamma}_{t\mid t-h}$ & Cached real-time forecast of target $t$ from candidate $\Gamma$, fitted at origin $t-h$; $e^{\Gamma}_{t\mid t-h}$ is its absolute error, stored once $\pi_t$ is released.                           \\
  $\Gamma^*_{\tau}$            & Hyperparameter selected at origin $\tau$ by the trailing-$W$ MAE rule; reused (frozen) by every Shapley coalition.                                                                                      \\
  \addlinespace
  \multicolumn{2}{l}{\textit{SARIMA benchmark}} \\
  \addlinespace
  $\Upsilon$                   & Candidate SARIMA order $(p,d,q)(P,D,Q,12)$ with $p,d,q,P,D,Q \in \{0,1\}$.                                                                                                                              \\
  $\mathcal{T}$                & SARIMA candidate set: the $2^6$ orders plus the AO-12MA trailing-average candidate.                                                                                                                     \\
  $\Theta^{\Upsilon}_{c}$      & Maximum-likelihood SARIMA coefficients of order $\Upsilon$ fitted on the rolling window ending at origin $c$. (The scan bootstrap's covariance matrix is written $\Theta_n$; see below.)             \\
  $\varepsilon_t$              & SARIMA innovation term.                                                                                                                                                                                 \\
  \addlinespace
  \multicolumn{2}{l}{\textit{Patient Fixed Share combination}} \\
  \addlinespace
  $K$                          & Number of experts in the admitted pool; $K = 3$ in the baseline (micro, macro, univariate).                                                                                                             \\
  $\hat{\pi}^{k}_{t\mid t-h}$  & Expert $k$'s forecast of target $t$ formed at origin $t-h$.                                                                                                                                             \\
  $\ell^{k}_{t\mid t-h}$       & Realized absolute loss of expert $k$ for target $t$: $\lvert\pi_t - \hat{\pi}^{k}_{t\mid t-h}\rvert$.                                                                                                   \\
  $v^{h,k}_t$                  & Internal weight on expert $k$ for horizon $h$, indexed by the target month and updated only with realized $h$-step losses (share step first, then loss step).                                          \\
  $w^{h,k}_t$                  & Forecast weight at origin $t$: $w^{h,k}_t = (1-\alpha)^h v^{h,k}_t + \bigl(1-(1-\alpha)^h\bigr)/K$.                                                                                                     \\
  $\hat{\pi}^{\mathrm{FS}}_{t+h\mid t}$ & Combined (patient Fixed Share) forecast formed at origin $t$ for target $t+h$.                                                                                                                \\
  $\eta$                       & Fixed Share learning rate; baseline $\eta = 0.5$.                                                                                                                                                       \\
  $\alpha$                     & Fixed-share mixing rate; baseline $\alpha = 0.02$, approximately $4/n$.                                                                                                                                 \\
  $t_0$                        & Start of the combination sample: the first combined forecasts target January 2007, each formed $h$ months earlier at its forecast origin.                                                               \\
  \addlinespace
  \multicolumn{2}{l}{\textit{Varying learning rate}} \\
  \addlinespace
  $\mathcal{H}$                & Learning-rate grid $\{0.25, 0.30, \ldots, 1.25\}$ ($G = 21$ values).                                                                                                                                    \\
  $\eta_\tau^\star$            & Learning rate selected at origin $\tau$ by the neighbor-only rolling search over the most recent $W = 60$ realized losses.                                                                              \\
  $\mathcal{I}_\tau$           & Candidate index set at origin $\tau$: the incumbent learning rate and its grid neighbors. (The scan test's interval collection is also written $\mathcal{I}$; see below.)                               \\
  \addlinespace
  \multicolumn{2}{l}{\textit{Scan test and Ljung--Box pre-test}} \\
  \addlinespace
  $r_{t,h}$                    & Absolute-error differential between benchmark method $B$ and competing method $M$ for target month $t$ at horizon $h$, built from the forecasts $\hat{\pi}_{t\mid t-h}$.                                      \\
  $\bar{r}_h$                  & Within-horizon sample mean of the loss-differential series; used in the Ljung--Box pre-test.                                                                                                            \\
  $m$                          & Ljung--Box lag budget; $m = 12$.                                                                                                                                                                        \\
  $Q$                          & Multiplier-bootstrap Ljung--Box statistic, aggregated over horizons.                                                                                                                                    \\
  $D_t^{(h)}$                  & Loss differential for target month $t$ at horizon $h$: the benchmark method $B$ absolute error minus the competing method $M$ absolute error.                                                                 \\
  $I = [s,e]$                  & Candidate contiguous target-month interval with start month $s$ and end month $e$.                                                                                                                    \\
  $\mathcal{I}$                & Collection of candidate intervals considered by the scan test.                                                                                                                                          \\
  $L_{\min}$                   & Minimum admissible interval length in the scan test; $L_{\min} = 6$ for the occasional scan.                                                                                       \\
  $L_{\max}$ & Maximum admissible interval length in the scan test; $L_{\max} = 12$ for the occasional scan, $L_{\min}=L_{\max}=n$ for the overall test.                                        \\
  $p$                          & Total number of (interval, horizon) pairs searched, $p = H\,\lvert\mathcal{I}\rvert$; $34{,}776$ for the deployed occasional scan ($L_{\min}=6$, $L_{\max}=12$).                    \\
  $S_I^{(h)}$                  & Sum of loss differentials over interval $I$ at horizon $h$.                                                                                                                                             \\
  $\bar{D}_I^{(h)}$            & Mean loss differential over interval $I$ at horizon $h$.                                                                                                                                                \\
  $\gamma$                     & Partial-studentization exponent; baseline $\gamma=0.1$.                                                                                                                             \\
  $Q_I^{(h)}$                 & Sum of squared deviations of $D_t^{(h)}$ within interval $I$ at horizon $h$.                                                                                                        \\
  $Q_I^{*(h)}$                & Bootstrap sum of squared deviations of multiplier-weighted local scores within interval $I$ at horizon $h$.                                                                         \\
  $T_I^{(h)}$                 & Partially-studentized scan statistic for interval $I$ and horizon $h$.                                                                                                             \\
  $T^{\max}$                  & Maximum partially-studentized scan statistic across all intervals and horizons.                                                                                                     \\
  $\xi_t^{(b)}$                & Gaussian multiplier in bootstrap draw $b$ for time $t$.                                                                                                                                                 \\
  $\Theta_n$                  & Covariance matrix of the Gaussian multiplier bootstrap: $\Theta_n = I_n$ in Regime~I, kernel-based in Regime~D.                                                                      \\
  $K(\cdot)$                   & Kernel function of the Regime~D covariance matrix.                                                                                                                                                      \\
  $b_n$                        & Bandwidth of the kernel-based covariance estimator.                                                                                                                                                     \\
  $B$                          & Number of bootstrap draws; $B = 4{,}999$.                                                                                                                                          \\
  \addlinespace
  \multicolumn{2}{l}{\textit{Grouped Shapley decomposition}} \\
  \addlinespace
  $\mathcal{G}$                & Set of micro feature groups (the players); $\lvert\mathcal{G}\rvert = G = 5$ in the baseline decomposition.                                                                                             \\
  $S$                          & Coalition of feature groups, $S \subseteq \mathcal{G}$.                                                                                                                                                 \\
  $v_h(S)$                     & Coalition value at horizon $h$: the signed percent MAE reduction of the three-expert forecast built on micro$(S)$ relative to the two-expert benchmark; $v_h(\varnothing) = 0$.                          \\
  $\phi_{g,h}$                 & Shapley value of group $g$ at horizon $h$, in percent of the two-expert benchmark MAE; $\sum_g \phi_{g,h} = v_h(\mathcal{G})$.                                                                          \\
\end{longtable}

This appendix also collects the pseudocode for all algorithms used in
Sections~\ref{sec:methods}, \ref{sec:scantest},
and~\ref{sec:group_shapley}.

\begin{algorithm}[H]
  \caption{\textsc{PrecomputeRollingPredictions}$(\mathcal{G}, h, L)$}
  \label{alg:mlcgi}
  \begin{algorithmic}[1]
    \State \textbf{Input:} Hyperparameter grid $\mathcal{G}$; horizon $h$;
    rolling window $L = 120$; features $\{X_t^h\}$ and inflation $\{\pi_t\}$;
    first usable target $\underline{t}$; origin range
    $[c_{\mathrm{lo}}, c_{\mathrm{hi}}]$.
    \State \textbf{Output:} Cached real-time forecasts and errors
    $\{\hat{\pi}^{\Gamma}_{t\mid t-h},\, e^{\Gamma}_{t\mid t-h}\}$ for all
    $\Gamma \in \mathcal{G}$ and reachable targets $t$.

    \For{$\Gamma \in \mathcal{G}$, $c = c_{\mathrm{lo}}, \ldots, c_{\mathrm{hi}}$}
    \State $s \leftarrow \max\{c - (L-1),\; \underline{t}\}$
    \State Train $\mathrm{ML}(\Gamma)$ on
    $\bigl\{(X_{j-h}^h,\,\pi_j)\bigr\}_{j=s}^{c}$
    \Comment{targets realized by the origin $c$}
    \State $\hat{\pi}^{\Gamma}_{c+h\mid c} \leftarrow
    \mathrm{ML}(\Gamma)\bigl(X_c^h\bigr)$;\quad
    $e^{\Gamma}_{c+h\mid c} \leftarrow
    \bigl|\pi_{c+h} - \hat{\pi}^{\Gamma}_{c+h\mid c}\bigr|$
    \EndFor
  \end{algorithmic}
\end{algorithm}

\begin{algorithm}[H]
  \caption{\textsc{RealTimeTuneAndForecast}$(\tau, h, W)$}
  \label{alg:ml_deploy}
  \begin{algorithmic}[1]
    \State \textbf{Input:} Forecast origin $\tau$; horizon $h$; validation
    window $W = 48$; the cache
    $\{\hat{\pi}^{\Gamma}_{t\mid t-h},\, e^{\Gamma}_{t\mid t-h}\}$ from
    Algorithm~\ref{alg:mlcgi}.
    \State \textbf{Output:} Real-time forecast $\hat{\pi}_{\tau+h\mid\tau}$
    and selected hyperparameter $\Gamma^*$.

    \For{$\Gamma \in \mathcal{G}$}
    \State $\mathrm{MAE}(\Gamma) \leftarrow
    \dfrac{1}{W}\sum_{t'=\tau-W+1}^{\tau} e^{\Gamma}_{t'\mid t'-h}$
    \Comment{trailing errors, all realized by $\tau$}
    \EndFor
    \State $\Gamma^* \leftarrow
    \operatorname*{arg\,min}_{\Gamma\in\mathcal{G}} \mathrm{MAE}(\Gamma)$;\quad
    $\hat{\pi}_{\tau+h\mid\tau} \leftarrow \hat{\pi}^{\Gamma^*}_{\tau+h\mid\tau}$
  \end{algorithmic}
\end{algorithm}

\clearpage

\begin{algorithm}[H]
  \caption{\textsc{PrecomputeRollingSARIMA}$(\mathcal{T}, h, L)$}
  \label{alg:sr1}
  \begin{algorithmic}[1]
    \State \textbf{Input:} Candidate set $\mathcal{T}$; horizon $h$; rolling
    window $L = 120$; inflation $\{\pi_t\}$; first usable month
    $\underline{t}$; origin range $[c_{\mathrm{lo}}, c_{\mathrm{hi}}]$.
    \State \textbf{Output:} Cached real-time forecasts and errors
    $\{\hat{\pi}^{\Upsilon}_{c+h\mid c},\, e^{\Upsilon}_{c+h\mid c}\}$ for all
    $\Upsilon \in \mathcal{T}$ and reachable origins $c$.

    \For{$\Upsilon \in \mathcal{T}$, $c = c_{\mathrm{lo}}, \ldots, c_{\mathrm{hi}}$}
    \State $s \leftarrow \max\{c - (L-1),\; \underline{t}\}$
    \If{$\Upsilon$ is AO-12MA}
    \State $\hat{\pi}^{\Upsilon}_{c+h\mid c} \leftarrow
    \frac{1}{12}\sum_{j=0}^{11}\pi_{c-j}$
    \Comment{degenerate candidate, no fit}
    \Else
    \State Fit $\Theta^{\Upsilon}_{c}$ for $\mathrm{SARIMA}(\Upsilon)$ by
    maximum likelihood on $\{\pi_j\}_{j=s}^{c}$
    \State $\hat{\pi}^{\Upsilon}_{c+h\mid c} \leftarrow$ iterate
    $\mathrm{SARIMA}(\Upsilon)(\Theta^{\Upsilon}_{c})$ $h$ steps forward from $c$
    \EndIf
    \State $e^{\Upsilon}_{c+h\mid c} \leftarrow
    \bigl|\pi_{c+h} - \hat{\pi}^{\Upsilon}_{c+h\mid c}\bigr|$
    \Comment{stored once $\pi_{c+h}$ is released}
    \EndFor
  \end{algorithmic}
\end{algorithm}

\begin{algorithm}[H]
  \caption{\textsc{RealTimeTuneAndForecastSARIMA}$(\tau, h, W)$}
  \label{alg:sr_deploy}
  \begin{algorithmic}[1]
    \State \textbf{Input:} Forecast origin $\tau$; horizon $h$; validation
    window $W = 48$; the cache
    $\{\hat{\pi}^{\Upsilon}_{t\mid t-h},\, e^{\Upsilon}_{t\mid t-h}\}$ from
    Algorithm~\ref{alg:sr1}.
    \State \textbf{Output:} Real-time forecast $\hat{\pi}_{\tau+h\mid\tau}$
    and selected candidate $\Upsilon^*$.

    \For{$\Upsilon \in \mathcal{T}$}
    \State $\mathrm{MAE}(\Upsilon) \leftarrow
    \dfrac{1}{W}\sum_{t'=\tau-W+1}^{\tau} e^{\Upsilon}_{t'\mid t'-h}$
    \Comment{trailing errors, all realized by $\tau$}
    \EndFor
    \State $\Upsilon^* \leftarrow
    \operatorname*{arg\,min}_{\Upsilon\in\mathcal{T}} \mathrm{MAE}(\Upsilon)$;\quad
    $\hat{\pi}_{\tau+h\mid\tau} \leftarrow \hat{\pi}^{\Upsilon^*}_{\tau+h\mid\tau}$
  \end{algorithmic}
\end{algorithm}

\clearpage
\begin{algorithm}[H]
  \caption{Multiplier-Bootstrap Ljung--Box Test (Aggregated over Horizons)}
  \label{alg:mb_ljung_box}
  \begin{algorithmic}[1]

    \State \textbf{Input:} Residual panel $\{r_{t,h}\}$, $t=1,\dots,n$,
    $h=1,\dots,H$,
    where $r_{t,h}=\lvert\pi_{t}-\hat\pi^{B}_{t\mid t-h}\rvert
            -\lvert\pi_{t}-\hat\pi^{M}_{t\mid t-h}\rvert$
    and $t$ indexes target months;
    maximum lag $m$; bootstrap draws $B$.
    \State \textbf{Output:} $p$-value $\hat{p}$.

    \For{$h=1,\dots,H$, $k=1,\dots,m$}
    \State $\hat{\rho}_h(k) \leftarrow
      \dfrac{\sum_{t=k+1}^{n}(r_{t,h}-\bar{r}_h)(r_{t-k,h}-\bar{r}_h)}
      {\sum_{t=1}^{n}(r_{t,h}-\bar{r}_h)^2}$;\quad
    $Q_h \leftarrow n(n+2)\displaystyle\sum_{k=1}^{m}
      \dfrac{\hat{\rho}_h(k)^2}{n-k}$
    \EndFor
    \State $Q \leftarrow \sum_{h=1}^{H} Q_h$

    \For{$b=1,\dots,B$}
    \State Draw $\xi_1^{(b)},\dots,\xi_n^{(b)}\overset{\text{i.i.d.}}{\sim}\mathcal{N}(0,1)$;\quad
    form $r^{*(b)}_{t,h} \leftarrow \xi_t^{(b)}\,r_{t,h}$ for all $t,h$
    \For{$h=1,\dots,H$}
    \State Compute $\hat{\rho}^{*(b)}_h(k)$ from $\{r^{*(b)}_{t,h}\}_{t=1}^n$;\quad
    $Q_h^{*(b)} \leftarrow n(n+2)\displaystyle\sum_{k=1}^{m}
      \dfrac{\hat{\rho}^{*(b)}_h(k)^2}{n-k}$
    \EndFor
    \State $Q^{*(b)} \leftarrow \sum_{h=1}^{H} Q_h^{*(b)}$
    \EndFor
    \State $\hat{p} = \bigl(1 + \sum_{b=1}^{B}\mathbf{1}\{Q^{*(b)}>Q\}\bigr)/(B+1)$
  \end{algorithmic}
\end{algorithm}

\clearpage
\begin{algorithm}[H]
  \caption{Partially-Studentized Scan Test with Gaussian Multiplier Bootstrap (Regime~I: $\Theta_n=I_n$)}
  \label{alg:scan_regime1}
  \begin{algorithmic}[1]

    \State \textbf{Input:} Forecast errors $\{e^{B}_{t\mid t-h},\,e^{M}_{t\mid t-h}\}$, indexed by target month $t$,
    $t=1,\dots,n$, $h=1,\dots,H$; $L_{\min}$; $L_{\max}$; tuning exponent $\gamma$; bootstrap draws $B$.
    \State \textbf{Output:} $p$-value $\hat{p}$.

    \For{$t=1,\dots,n$, $h=1,\dots,H$}
    \State $D_t^{(h)} \leftarrow \lvert e^{B}_{t\mid t-h}\rvert - \lvert e^{M}_{t\mid t-h}\rvert$
    \EndFor
    \State $\mathcal{I} \leftarrow \{[s,e]\subset[n]: L_{\min}\le e-s+1\le L_{\max}\}$;
    \quad $T^{\max} \leftarrow -\infty$
    \For{each $I=[s,e]\in\mathcal{I}$, $h=1,\dots,H$}
    \State $S_I^{(h)} \leftarrow \sum_{t=s}^{e} D_t^{(h)}$;\quad
    $\bar{D}_I^{(h)} \leftarrow S_I^{(h)}/|I|$;\quad
    $Q_I^{(h)} \leftarrow \sum_{t=s}^{e}(D_t^{(h)}-\bar{D}_I^{(h)})^2$
    \State $A_I^{(h)}(\gamma) \leftarrow |I|^{1/2-\gamma}\bigl(Q_I^{(h)}\bigr)^\gamma$
    \State $T_I^{(h)} \leftarrow S_I^{(h)}/A_I^{(h)}(\gamma)$
    if $A_I^{(h)}(\gamma)>10^{-12}$, else $0$;\quad
    $T^{\max} \leftarrow \max(T^{\max}, T_I^{(h)})$
    \EndFor
    \For{$b=1,\dots,B$}
    \State Draw $\xi_1^{(b)},\dots,\xi_n^{(b)}\overset{\text{i.i.d.}}{\sim}\mathcal{N}(0,1)$;\quad
    $T^{*,(b),\max}\leftarrow-\infty$
    \For{each $I=[s,e]\in\mathcal{I}$, $h=1,\dots,H$}
    \State $\bar{D}_{I}^{(h)} \leftarrow |I|^{-1}\sum_{t=s}^e D_t^{(h)}$;\quad
    $e_{t,I}^{(h,b)} \leftarrow \xi_t^{(b)}(D_t^{(h)}-\bar{D}_{I}^{(h)})$
    for $t=s,\dots,e$
    \State $S_I^{*(h)} \leftarrow \sum_{t=s}^{e}e_{t,I}^{(h,b)}$;\quad
    $Q_I^{*(h)} \leftarrow \sum_{t=s}^{e}\bigl(e_{t,I}^{(h,b)}\bigr)^2-\bigl(S_I^{*(h)}\bigr)^2/|I|$
    \State $A_I^{*(h)}(\gamma) \leftarrow |I|^{1/2-\gamma}\bigl(Q_I^{*(h)}\bigr)^\gamma$
    \State $T_I^{*,(h)} \leftarrow S_I^{*(h)}/A_I^{*(h)}(\gamma)$
    if $A_I^{*(h)}(\gamma)>10^{-12}$, else $0$
    \State $T^{*,(b),\max} \leftarrow \max(T^{*,(b),\max}, T_I^{*,(h)})$
    \EndFor
    \EndFor
    \State $\hat{p} = \bigl(1+\sum_{b=1}^{B}\mathbf{1}\{T^{*,(b),\max}>T^{\max}\}\bigr)/(B+1)$
\end{algorithmic}
\end{algorithm}

\clearpage
\begin{algorithm}[H]
  \caption{Partially-Studentized Scan Test with Gaussian Multiplier Bootstrap (Regime~D: Kernel $\Theta_n$)}
  \label{alg:scan_regimeD}
  \begin{algorithmic}[1]

    \State \textbf{Input:} Forecast errors $\{e^{B}_{t\mid t-h},\,e^{M}_{t\mid t-h}\}$, indexed by target month $t$,
    $t=1,\dots,n$, $h=1,\dots,H$; $L_{\min}$; $L_{\max}$; tuning exponent $\gamma$; bootstrap draws $B$;
    kernel $K$.
    \State \textbf{Output:} $p$-value $\hat{p}$.

    \For{$t=1,\dots,n$, $h=1,\dots,H$}
    \State $D_t^{(h)} \leftarrow \lvert e^{B}_{t\mid t-h}\rvert - \lvert e^{M}_{t\mid t-h}\rvert$
    \EndFor
    \For{$h=1,\dots,H$}
    \State Fit $D_t^{(h)} = c_h + \rho_h D_{t-1}^{(h)} + \varepsilon_{h,t}$
    by least squares; obtain $(\hat{\rho}_h, \hat{\sigma}_h^2)$
    \EndFor
    \State $\hat{a} \leftarrow
    \dfrac{\sum_{h=1}^{H} 4\hat{\rho}_h^2\hat{\sigma}_h^4(1-\hat{\rho}_h)^{-8}}
          {\sum_{h=1}^{H} \hat{\sigma}_h^4(1-\hat{\rho}_h)^{-4}}$;\quad
    $b_n \leftarrow 1.3221\,(\hat{a}\,n)^{1/5}$
    \State Compute $\Theta_n \in \mathbb{R}^{n\times n}$ with
    $\Theta_n(t,u) = K\!\left(\dfrac{t-u}{b_n}\right)$ for all $t,u$
    \State $\mathcal{I} \leftarrow \{[s,e]\subset[n]: L_{\min}\le e-s+1\le L_{\max}\}$;
    \quad $T^{\max} \leftarrow -\infty$
    \For{each $I=[s,e]\in\mathcal{I}$, $h=1,\dots,H$}
    \State $S_I^{(h)} \leftarrow \sum_{t=s}^{e} D_t^{(h)}$;\quad
    $\bar{D}_I^{(h)} \leftarrow S_I^{(h)}/|I|$;\quad
    $Q_I^{(h)} \leftarrow \sum_{t=s}^{e}(D_t^{(h)}-\bar{D}_I^{(h)})^2$
    \State $A_I^{(h)}(\gamma) \leftarrow |I|^{1/2-\gamma}\bigl(Q_I^{(h)}\bigr)^\gamma$
    \State $T_I^{(h)} \leftarrow S_I^{(h)}/A_I^{(h)}(\gamma)$
    if $A_I^{(h)}(\gamma)>10^{-12}$, else $0$;\quad
    $T^{\max} \leftarrow \max(T^{\max}, T_I^{(h)})$
    \EndFor
    \For{$b=1,\dots,B$}
    \State Draw $\xi^{(b)} = (\xi_1^{(b)},\dots,\xi_n^{(b)})^\top \sim \mathcal{N}(0,\Theta_n)$;\quad
    $T^{*,(b),\max}\leftarrow-\infty$
    \For{each $I=[s,e]\in\mathcal{I}$, $h=1,\dots,H$}
    \State $\bar{D}_{I}^{(h)} \leftarrow |I|^{-1}\sum_{t=s}^e D_t^{(h)}$;\quad
    $e_{t,I}^{(h,b)} \leftarrow \xi_t^{(b)}(D_t^{(h)}-\bar{D}_{I}^{(h)})$
    for $t=s,\dots,e$
    \State $S_I^{*(h)} \leftarrow \sum_{t=s}^{e}e_{t,I}^{(h,b)}$;\quad
    $Q_I^{*(h)} \leftarrow \sum_{t=s}^{e}\bigl(e_{t,I}^{(h,b)}\bigr)^2-\bigl(S_I^{*(h)}\bigr)^2/|I|$
    \State $A_I^{*(h)}(\gamma) \leftarrow |I|^{1/2-\gamma}\bigl(Q_I^{*(h)}\bigr)^\gamma$
    \State $T_I^{*,(h)} \leftarrow S_I^{*(h)}/A_I^{*(h)}(\gamma)$
    if $A_I^{*(h)}(\gamma)>10^{-12}$, else $0$
    \State $T^{*,(b),\max} \leftarrow \max(T^{*,(b),\max}, T_I^{*,(h)})$
    \EndFor
    \EndFor
    \State $\hat{p} = \bigl(1+\sum_{b=1}^{B}\mathbf{1}\{T^{*,(b),\max}>T^{\max}\}\bigr)/(B+1)$
\end{algorithmic}
\end{algorithm}

\clearpage

\begin{algorithm}[H]
  \caption{\textsc{FixedSharePath}$(\eta, \alpha, h)$ (Patient Fixed Share)}
  \label{alg:fixedshare_path}
  \begin{algorithmic}[1]
    \State \textbf{Input:} $\eta>0$, $\alpha\in[0,1]$, horizon $h$;
    realized $\{\pi_t\}_{t=t_0}^{n}$; expert forecasts
    $\{\hat{\pi}^{k}_{t\mid t-h}\}$ and $\{\hat{\pi}^{k}_{t+h\mid t}\}$,
    $k=1,\dots,K$.
    \State \textbf{Output:} Combined forecasts
    $\{\hat{\pi}^{\mathrm{FS}}_{t+h\mid t}\}$ and forecast weights
    $\{w^{h,k}_t\}$.
    \State Initialize $v^{h,k}_{t_0} \leftarrow 1/K$ for all $k$.
    \Comment{internal weights over experts, indexed by target}
    \For{$t = t_0, \dots, n$}
    \If{$t < t_0 + h$}
    \State $v^{h,k}_{t} \leftarrow 1/K$ for all $k$
    \Comment{no realized $h$-step loss yet}
    \Else
    \State $\ell^{k}_{t\mid t-h} \leftarrow
    \bigl|\pi_t - \hat{\pi}^{k}_{t\mid t-h}\bigr|$ for all $k$
    \Comment{realized loss of target $t$, formed at origin $t-h$}
    \State $v^{h,k}_{t} \leftarrow \Bigl[(1-\alpha)\,v^{h,k}_{t-1}
    + \tfrac{\alpha}{K}\Bigr]\exp\!\bigl(-\eta\,\ell^{k}_{t\mid t-h}\bigr)$,
    then normalize over $k$
    \Comment{share first, then loss}
    \EndIf
    \State $w^{h,k}_{t} \leftarrow (1-\alpha)^h\,v^{h,k}_{t}
    + \dfrac{1-(1-\alpha)^h}{K}$ for all $k$
    \Comment{propagate $h$ loss-free share steps}
    \State $\hat{\pi}^{\mathrm{FS}}_{t+h\mid t} \leftarrow
    \sum_{k=1}^{K} w^{h,k}_{t}\,\hat{\pi}^{k}_{t+h\mid t}$
    \Comment{forecast at origin $t$ for target $t+h$}
    \EndFor
  \end{algorithmic}
\end{algorithm}

\clearpage

\begin{algorithm}[H]
  \caption{Patient Fixed Share with Rolling-Window Neighbor Selection of $\eta$}
  \label{alg:fs_precompute_neighbor_select_eta}
  \begin{algorithmic}[1]
    \State \textbf{Input:} $\alpha\in[0,1]$, horizon $h$; realized
    $\{\pi_t\}$; expert forecasts $\{\hat{\pi}^{k}_{t\mid t-h}\}$,
    $k=1,\dots,K$; grid
    $\mathcal{H}=\{\eta^{(1)}<\cdots<\eta^{(G)}\}$;
    window length $W$; initial value $\eta_0 = \eta^{(g_0)}$.
    \State \textbf{Output:} Combined forecasts
    $\{\hat{\pi}^{\mathrm{FS}}_{\tau+h\mid\tau}\}$.
    \For{each $\eta^{(g)} \in \mathcal{H}$}
    \State Run \textsc{FixedSharePath}$(\eta^{(g)}, \alpha, h)$
    (Algorithm~\ref{alg:fixedshare_path}) to obtain the path
    $\{\hat{\pi}^{\mathrm{FS},(\eta^{(g)})}_{\tau+h\mid\tau}\}$
    \EndFor
    \For{each forecast origin $\tau$, in order}
    \If{fewer than $W$ realized $h$-step forecast losses exist at $\tau$}
    \State $g_\tau \leftarrow g_0$
    \Else
    \State $\mathcal{I}_\tau \leftarrow \{\max(1,g_{\tau-1}-1),\, g_{\tau-1},\,
      \min(G,g_{\tau-1}+1)\}$
    \Comment{incumbent and its grid neighbors}
    \State $g_\tau \leftarrow \arg\min_{g\in\mathcal{I}_\tau}
      \dfrac{1}{W}\sum_{s=\tau-W+1}^{\tau}
      \bigl|\pi_s - \hat{\pi}^{\mathrm{FS},(\eta^{(g)})}_{s\mid s-h}\bigr|$
    \Comment{ties keep the incumbent}
    \EndIf
    \State $\eta_\tau^\star \leftarrow \eta^{(g_\tau)}$;\quad set
    $\hat{\pi}^{\mathrm{FS}}_{\tau+h\mid\tau} \leftarrow
    \hat{\pi}^{\mathrm{FS},(\eta_\tau^\star)}_{\tau+h\mid\tau}$
    \EndFor
  \end{algorithmic}
\end{algorithm}

\clearpage

\begin{algorithm}[H]
  \caption{\textsc{FrozenMicroCoalition}$(S, h)$}
  \label{alg:shapley_micro}
  \begin{algorithmic}[1]
    \State \textbf{Input:} Coalition $S \subseteq \mathcal{G}$ of micro feature
    groups; horizon $h$; rolling window $L = 120$; grouped features
    $\{X_t^h[g]\}$; frozen hyperparameter selections $\{\Gamma^*_{\tau}\}$ from
    Algorithm~\ref{alg:ml_deploy}; inflation $\{\pi_t\}$; first usable month
    $\underline{t}$; origin range $[c_{\mathrm{lo}}, c_{\mathrm{hi}}]$.
    \State \textbf{Output:} Micro coalition forecasts
    $\{\hat{\pi}^{S}_{t+h\mid t}\}$ (none when $S = \varnothing$).
    \If{$S = \varnothing$}
    \State \textbf{return} no forecasts
    \Comment{no micro expert; the combination is macro + univariate only}
    \EndIf
    \For{$c = c_{\mathrm{lo}}, \ldots, c_{\mathrm{hi}}$}
    \State $s \leftarrow \max\{c - (L-1),\; \underline{t}\}$
    \State Train $\mathrm{ML}(\Gamma^*_{c})$ on
    $\bigl\{(X_{j-h}^h[S],\,\pi_j)\bigr\}_{j=s}^{c}$
    \Comment{frozen $\Gamma^*_{c}$, feature subset $S$, no re-tuning}
    \State $\hat{\pi}^{S}_{c+h\mid c} \leftarrow
    \mathrm{ML}(\Gamma^*_{c})\bigl(X_c^h[S]\bigr)$
    \Comment{known at origin $c$}
    \EndFor
  \end{algorithmic}
\end{algorithm}

\begin{algorithm}[H]
  \caption{\textsc{GroupShapley}$(\mathcal{G}, h)$ over the micro feature groups}
  \label{alg:group_shapley}
  \begin{algorithmic}[1]
    \State \textbf{Input:} Groups $\mathcal{G}$ with $|\mathcal{G}| = G$;
    horizon $h$; the macro and univariate expert forecasts; rates
    $(\eta, \alpha)$; evaluation window $\mathcal{E}$ (pre- or post-2020).
    \State \textbf{Output:} Per-horizon Shapley values
    $\{\phi_{g,h}\}_{g \in \mathcal{G}}$.
    \State $\mathrm{MAE}_{\mathrm{baseline},h} \leftarrow
    \dfrac{1}{|\mathcal{E}|}\sum_{t \in \mathcal{E}}
    \bigl|\pi_t - \hat{\pi}^{\mathrm{FS},\varnothing}_{t\mid t-h}\bigr|$,
    where $\hat{\pi}^{\mathrm{FS},\varnothing}$ runs
    Algorithm~\ref{alg:fixedshare_path} on $\{\text{macro}, \text{univariate}\}$
    \Comment{2-expert benchmark}
    \For{each nonempty $S \subseteq \mathcal{G}$}
    \Comment{$2^{G}-1$ coalitions; $S = \mathcal{G}$ is the full 3-expert combination}
    \State $\{\hat{\pi}^{S}_{\cdot\mid\cdot}\} \leftarrow
    \textsc{FrozenMicroCoalition}(S, h)$
    \State $\{\hat{\pi}^{\mathrm{FS},S}_{t+h\mid t}\} \leftarrow$
    Algorithm~\ref{alg:fixedshare_path} on
    $\{\text{micro}(S), \text{macro}, \text{univariate}\}$
    \State $\mathrm{MAE}_{S,h} \leftarrow
    \dfrac{1}{|\mathcal{E}|}\sum_{t \in \mathcal{E}}
    \bigl|\pi_t - \hat{\pi}^{\mathrm{FS},S}_{t\mid t-h}\bigr|$
    \State $v_h(S) \leftarrow 100 \times
    \dfrac{\mathrm{MAE}_{\mathrm{baseline},h} - \mathrm{MAE}_{S,h}}
    {\mathrm{MAE}_{\mathrm{baseline},h}}$
    \Comment{signed \% MAE reduction vs.\ the 2-expert benchmark}
    \EndFor
    \State $v_h(\varnothing) \leftarrow 0$
    \For{each group $g \in \mathcal{G}$}
    \State $\phi_{g,h} \leftarrow
    \displaystyle\sum_{S \subseteq \mathcal{G}\setminus\{g\}}
    \frac{|S|!\,(G-|S|-1)!}{G!}\,
    \bigl[v_h(S \cup \{g\}) - v_h(S)\bigr]$
    \EndFor
    \State $\sum_{g}\phi_{g,h} = v_h(\mathcal{G})$
    \Comment{adding-up: the full 3- vs.\ 2-expert \% gain}
  \end{algorithmic}
\end{algorithm}

\clearpage
\section{Robustness: Alternative Encodings, Aggregations, and Data Filtering}\label{app:robustness_alt}

The baseline micro encoding used in the main text retains 11 features per
COICOP1 category (fraction of zero price changes, nine deciles of
the non-zero price-change distribution, and its mean), aggregates within
each category using CPI weights, and uses the full sales-included price
quotes, for a total of $11\times 11 = 121$ statistics per month. This
appendix documents four robustness checks against this baseline.
Appendix~\ref{app:augmented_encoding} replaces the baseline encoding with
the full augmented encoding (29 features per category, $11\times 29 = 319$
statistics per month) while keeping CPI-weighted aggregation and
sales-included quotes. Appendix~\ref{app:unweighted_encoding} keeps the
11-feature encoding but replaces CPI-weighted aggregation with
equal-weighted (unweighted) aggregation across products within each COICOP1
category. Appendix~\ref{app:nosales_encoding} keeps the 11-feature encoding and
CPI-weighted aggregation but excludes ONS-sale-flagged price quotes from
the underlying micro panel. Appendix~\ref{app:nsa_encoding} likewise keeps the
encoding and aggregation but replaces temporary sale prices with the regular
price inferred by the \citet{NakamuraSteinsson2008} filter (parameterisation
A) before computing the per-category statistics. Each appendix reports the consequences of the
alternative choice at the
standalone micro stage and inside the three-expert online-learning
combination.

\subsection{Augmented Encoding vs Baseline}
\label{app:augmented_encoding}

The full augmented encoding, as described in
Subsection~\ref{sec:encoding_training_tuning}, adds the variance (with Bessel's
correction), the standardised third
and fifth central moments and the excess kurtosis of the non-zero price-change
distribution, the nine deciles of the time-since-last-change distribution, and
the five moments of the same distribution, raising the per-category feature
count from 11 to 29 and the per-month feature count from 121 to 319. This is
the same set of micro variables as the six-group Shapley decomposition
(Appendix~\ref{app:6_group_shapley}), which spans exactly this augmented
feature set; the decomposition's coalition refits reuse the baseline expert's
frozen tuning, whereas the augmented expert here is tuned on its own
cached real-time schedule.

\subsubsection{Standalone Micro Forecast}
\label{app:augmented_standalone}

Table~\ref{tab:augmented_xgb_micro_vs_sarima_two_windows} reports the
percentage MAE improvement of the standalone micro forecast relative to the univariate forecast
under both encodings, for the pre-2020 (2015--2019) and post-2020
(2020--2024) windows and all horizons $h=1,\ldots,24$. Averaging within
four horizon groups ($h=1$ to $6$, $7$ to $12$, $13$ to $18$, and $19$ to
$24$), the augmented encoding's percentage MAE improvements relative to
the univariate forecast are $-35.1\%$, $-25.6\%$, $-20.6\%$, and $-27.1\%$ in the
pre-2020 window and $-0.4\%$, $5.4\%$, $-3.4\%$, and $22.6\%$ in the
post-2020 window. The corresponding gaps relative to the baseline
encoding (augmented minus baseline) are $-5.7$, $-6.8$, $-3.3$, and
$+5.5$ percentage points pre-2020 (averaging a $2.5$ percentage-point
cost) and $-9.0$, $-5.4$, $-11.9$, and $-2.7$ percentage points post-2020
(averaging a $7.3$ percentage-point cost). The augmented encoding
therefore underperforms the baseline at the standalone stage in both
windows, with the cost now concentrated in the post-COVID period, where
the added higher-moment and time-since blocks are noisiest, at short and
medium-long horizons. This is consistent with overfitting risk rising
in the number of predictors when the univariate forecast itself remains competitive, as
discussed in Subsection~\ref{sec:encoding_training_tuning}.

\begin{table}[htbp!]
\centering
\caption{Micro Forecast Percentage MAE Improvement Relative to the Univariate Forecast: Baseline versus Augmented Encoding}
\label{tab:augmented_xgb_micro_vs_sarima_two_windows}
\begin{tabular}{crrrr}
\toprule
 & \multicolumn{2}{c}{Pre-2020} & \multicolumn{2}{c}{Post-2020} \\
\cmidrule(lr){2-3} \cmidrule(lr){4-5}
$h$ & Baseline & Augmented & Baseline & Augmented \\
\midrule
1 & -35.67 & -49.21 & 8.49 & 14.85 \\
2 & -28.92 & -51.74 & 4.46 & -3.26 \\
3 & -29.10 & -37.63 & 2.14 & 1.11 \\
4 & -30.97 & -28.42 & 9.67 & -7.85 \\
5 & -32.52 & -17.42 & 6.73 & -8.23 \\
6 & -19.51 & -26.05 & 20.31 & 0.70 \\
7 & -17.33 & -19.33 & 13.00 & -0.44 \\
8 & -17.05 & -25.52 & 8.63 & 8.03 \\
9 & -27.49 & -39.07 & 14.63 & 8.33 \\
10 & -18.19 & -21.73 & 10.43 & 8.84 \\
11 & -19.84 & -24.42 & 16.04 & 10.14 \\
12 & -12.81 & -23.51 & 1.79 & -2.67 \\
13 & -10.32 & -7.69 & 12.95 & -0.49 \\
14 & -0.13 & -13.61 & 15.28 & -3.30 \\
15 & -22.62 & -28.94 & 14.27 & 0.55 \\
16 & -21.06 & -23.27 & 11.03 & -10.65 \\
17 & -31.47 & -26.04 & -3.33 & -3.28 \\
18 & -18.48 & -23.92 & 0.82 & -3.15 \\
19 & -28.04 & -17.37 & 9.07 & 4.78 \\
20 & -26.65 & -35.14 & 9.34 & 10.56 \\
21 & -32.13 & -32.66 & 35.42 & 31.01 \\
22 & -44.26 & -33.76 & 32.27 & 31.43 \\
23 & -36.87 & -28.26 & 39.33 & 35.81 \\
24 & -27.88 & -15.45 & 26.46 & 21.85 \\
\bottomrule
\end{tabular}
\smallskip
\begin{minipage}{0.92\textwidth}\setstretch{1.0}
\footnotesize\textit{Notes:} Each cell reports the percentage MAE improvement of the standalone micro forecast relative to the univariate forecast at horizon $h$. The baseline uses the 11-feature encoding; the augmented robustness check uses the 29-feature encoding (adding the higher moments of the price-change distribution and the time-since-last-change block), the same feature universe as the six-group Shapley decomposition.
Positive values indicate lower MAE than the univariate forecast. Evaluation
windows are 2015--2019 and 2020--2024.
\end{minipage}
\end{table}

\subsubsection{Three-Expert Online Learning}
\label{app:augmented_3expert}

Table~\ref{tab:augmented_3expert_fs_vs_sarima_two_windows} reports
the percentage MAE improvement of the three-expert patient Fixed Share
combination (fixed $\alpha=0.02$, $\eta=0.5$, the main-text baseline)
relative to the univariate forecast, with the micro model trained on the baseline and
augmented encodings respectively, holding the macro and univariate forecasts fixed.
Averaging within the same four horizon groups, the augmented-encoding
three-expert combination's percentage MAE improvements are $0.8\%$, $-2.5\%$,
$-2.8\%$, and $-1.7\%$ in the pre-2020 window and $8.0\%$, $12.1\%$,
$8.1\%$, and $23.3\%$ in the post-2020 window. The corresponding gaps
relative to the baseline-encoding three-expert combination are $-0.4$, $-1.1$,
$-1.1$, and $+1.3$ percentage points pre-2020 (averaging a $0.3$
percentage-point cost) and $-4.0$, $-1.6$, $-3.8$, and $-1.9$
percentage points post-2020 (averaging a $2.8$ percentage-point cost).
The standalone gap therefore shrinks from $7.3$ to $2.8$ percentage
points once the augmented micro forecast is combined with the macro
and univariate forecasts through the Fixed Share combination. The combination is robust to
encoding choice because the algorithm reallocates weight away from the
data-rich expert when its realized performance deteriorates and keeps a
meaningful share on the univariate benchmark in the pre-COVID window.

\begin{table}[htbp!]
\centering
\caption{Three-Expert Fixed Share: Percent MAE Improvement Relative to the Univariate Forecast under Baseline versus Augmented Encoding of Micro Features ($\alpha=0.02$, $\eta=0.5$)}
\label{tab:augmented_3expert_fs_vs_sarima_two_windows}
\begin{tabular}{crrrr}
\toprule
 & \multicolumn{2}{c}{Pre-2020} & \multicolumn{2}{c}{Post-2020} \\
\cmidrule(lr){2-3} \cmidrule(lr){4-5}
$h$ & Baseline & Augmented & Baseline & Augmented \\
\midrule
1 & -1.58 & -0.84 & 12.03 & 14.19 \\
2 & 1.83 & -0.29 & 7.64 & 5.52 \\
3 & 4.29 & 2.68 & 7.40 & 5.86 \\
4 & 0.66 & 0.85 & 13.17 & 5.78 \\
5 & -0.14 & 2.40 & 13.72 & 7.63 \\
6 & 2.37 & 0.26 & 17.89 & 8.76 \\
7 & -0.08 & -0.38 & 15.63 & 11.38 \\
8 & -1.64 & -3.76 & 14.73 & 14.41 \\
9 & -3.29 & -5.03 & 18.23 & 17.20 \\
10 & -0.93 & -3.97 & 15.22 & 15.78 \\
11 & -1.85 & -0.14 & 15.89 & 13.51 \\
12 & -0.32 & -1.46 & 2.81 & 0.43 \\
13 & 2.33 & 2.56 & 12.26 & 8.38 \\
14 & 4.33 & 0.03 & 15.49 & 6.99 \\
15 & -6.69 & -5.66 & 13.89 & 10.19 \\
16 & -6.57 & -7.24 & 10.52 & 4.60 \\
17 & -5.06 & -5.56 & 8.27 & 7.99 \\
18 & 1.55 & -1.08 & 10.88 & 10.52 \\
19 & 1.15 & 0.16 & 13.64 & 11.18 \\
20 & -5.52 & -8.21 & 15.88 & 13.97 \\
21 & -4.20 & -6.14 & 30.03 & 26.94 \\
22 & -8.20 & -4.51 & 30.22 & 29.06 \\
23 & -4.08 & 0.96 & 35.30 & 34.03 \\
24 & 3.30 & 7.75 & 26.53 & 24.82 \\
\bottomrule
\end{tabular}
\smallskip
\begin{minipage}{0.92\textwidth}\setstretch{1.0}
\footnotesize\textit{Notes:} Each cell reports the percentage MAE improvement of the three-expert patient fixed-share combination (fixed $\alpha=0.02$, $\eta=0.5$) relative to the univariate forecast at horizon $h$, with the micro expert trained on the 11-feature baseline versus the 29-feature augmented encoding, holding the macro and univariate forecasts fixed.
Positive values indicate lower MAE than the univariate forecast. Evaluation
windows are 2015--2019 and 2020--2024.
\end{minipage}
\end{table}

\subsection{Unweighted Aggregation vs Baseline}
\label{app:unweighted_encoding}

The unweighted aggregation keeps the 11-feature baseline encoding but
constructs each per-category statistic as an equal-weighted average
across products rather than weighting by CPI shares. This isolates the
contribution of the CPI-share weighting choice in the baseline.

\subsubsection{Standalone Micro Forecast}
\label{app:unweighted_standalone}

Table~\ref{tab:unweighted_xgb_micro_vs_sarima_two_windows} reports the
percentage MAE improvement of the standalone micro forecast relative to the univariate forecast
under CPI-weighted (baseline) and unweighted aggregation, for all
horizons $h=1,\ldots,24$. Averaging within four horizon groups
($h=1$ to $6$, $7$ to $12$, $13$ to $18$, and $19$ to $24$), the
unweighted aggregation's percentage MAE improvements relative to the univariate forecast
are $-43.1\%$, $-22.4\%$, $-15.2\%$, and $-23.9\%$ in the pre-2020
window and $-1.9\%$, $7.2\%$, $8.4\%$, and $26.4\%$ in the post-2020
window. The corresponding gaps relative to the baseline (unweighted
minus baseline) are $-13.6$, $-3.6$, $+2.1$, and $+8.7$ percentage points
pre-2020 (averaging a $1.6$ percentage-point cost) and $-10.6$, $-3.5$,
$-0.1$, and $+1.1$ percentage points post-2020 (averaging a $3.3$
percentage-point cost). Unweighted aggregation therefore underperforms
CPI-weighted aggregation on average in both windows, with the cost
concentrated at short horizons, where CPI shares concentrate on the most
volatile categories, and a partial offset at the longest horizons.

\begin{table}[htbp!]
\centering
\caption{Micro Forecast Percentage MAE Improvement Relative to the Univariate Forecast: Baseline versus Equal-Weighted Aggregation}
\label{tab:unweighted_xgb_micro_vs_sarima_two_windows}
\begin{tabular}{crrrr}
\toprule
 & \multicolumn{2}{c}{Pre-2020} & \multicolumn{2}{c}{Post-2020} \\
\cmidrule(lr){2-3} \cmidrule(lr){4-5}
$h$ & Baseline & Unweighted & Baseline & Unweighted \\
\midrule
1 & -35.67 & -62.02 & 8.49 & -1.74 \\
2 & -28.92 & -58.46 & 4.46 & -2.48 \\
3 & -29.10 & -39.90 & 2.14 & -0.98 \\
4 & -30.97 & -34.31 & 9.67 & -8.29 \\
5 & -32.52 & -28.48 & 6.73 & -1.42 \\
6 & -19.51 & -35.14 & 20.31 & 3.39 \\
7 & -17.33 & -27.21 & 13.00 & 10.15 \\
8 & -17.05 & -20.85 & 8.63 & 7.05 \\
9 & -27.49 & -27.80 & 14.63 & 14.62 \\
10 & -18.19 & -15.65 & 10.43 & 4.13 \\
11 & -19.84 & -26.44 & 16.04 & 10.33 \\
12 & -12.81 & -16.23 & 1.79 & -2.91 \\
13 & -10.32 & -6.35 & 12.95 & 4.71 \\
14 & -0.13 & -1.56 & 15.28 & 11.37 \\
15 & -22.62 & -12.72 & 14.27 & 8.40 \\
16 & -21.06 & -24.60 & 11.03 & 9.78 \\
17 & -31.47 & -27.43 & -3.33 & 8.97 \\
18 & -18.48 & -18.75 & 0.82 & 7.37 \\
19 & -28.04 & -14.36 & 9.07 & 12.16 \\
20 & -26.65 & -22.72 & 9.34 & 11.21 \\
21 & -32.13 & -38.05 & 35.42 & 35.12 \\
22 & -44.26 & -25.26 & 32.27 & 31.55 \\
23 & -36.87 & -24.48 & 39.33 & 38.22 \\
24 & -27.88 & -18.61 & 26.46 & 30.44 \\
\bottomrule
\end{tabular}
\smallskip
\begin{minipage}{0.92\textwidth}\setstretch{1.0}
\footnotesize\textit{Notes:} Each cell reports the percentage MAE improvement
of the standalone micro forecast relative to the univariate forecast at horizon $h$. The
baseline uses CPI expenditure shares when aggregating within categories; the
equal-weighted robustness check gives each product equal weight within the
category. Positive values indicate lower MAE for the micro forecast than for
the univariate forecast. Evaluation windows are 2015--2019 and 2020--2024.
\end{minipage}
\end{table}

\subsubsection{Three-Expert Online Learning}
\label{app:unweighted_3expert}

Table~\ref{tab:unweighted_3expert_fs_vs_sarima_two_windows}
reports the percentage MAE improvement of the three-expert patient
Fixed Share combination (fixed $\alpha=0.02$, $\eta=0.5$, the main-text
baseline) relative to the univariate forecast, with the micro model trained on
CPI-weighted (baseline) and unweighted aggregations respectively,
holding the macro and univariate forecasts fixed. Averaging within the same four
horizon groups, the unweighted three-expert combination's percentage MAE
improvements are $-0.7\%$, $-1.6\%$, $-1.2\%$, and $-1.4\%$ in the
pre-2020 window and $9.4\%$, $12.9\%$, $11.3\%$, and $25.5\%$ in the
post-2020 window. The corresponding gaps relative to the baseline
three-expert combination are $-2.0$, $-0.3$, $+0.5$, and $+1.6$ percentage
points pre-2020 (averaging to zero) and $-2.6$, $-0.8$, $-0.5$, and $+0.2$
percentage points post-2020 (averaging a $0.9$ percentage-point cost). The
standalone gap shrinks to below one percentage point on average once
the equal-share variant is combined with the macro and univariate forecasts through
the Fixed Share combination, which is therefore robust to
the aggregation-weighting choice.

\begin{table}[htbp!]
\centering
\caption{Three-Expert Fixed Share: Percent MAE Improvement Relative to the Univariate Forecast under Baseline versus Equal-Weighted Aggregation of Micro Features ($\alpha=0.02$, $\eta=0.5$)}
\label{tab:unweighted_3expert_fs_vs_sarima_two_windows}
\begin{tabular}{crrrr}
\toprule
 & \multicolumn{2}{c}{Pre-2020} & \multicolumn{2}{c}{Post-2020} \\
\cmidrule(lr){2-3} \cmidrule(lr){4-5}
$h$ & Baseline & Unweighted & Baseline & Unweighted \\
\midrule
1 & -1.58 & -1.11 & 12.03 & 9.54 \\
2 & 1.83 & -4.37 & 7.64 & 6.25 \\
3 & 4.29 & 2.81 & 7.40 & 7.63 \\
4 & 0.66 & 0.10 & 13.17 & 7.57 \\
5 & -0.14 & -0.49 & 13.72 & 10.32 \\
6 & 2.37 & -1.24 & 17.89 & 14.88 \\
7 & -0.08 & -2.25 & 15.63 & 13.24 \\
8 & -1.64 & -1.83 & 14.73 & 14.76 \\
9 & -3.29 & -2.21 & 18.23 & 19.77 \\
10 & -0.93 & -2.51 & 15.22 & 16.38 \\
11 & -1.85 & -1.59 & 15.89 & 13.60 \\
12 & -0.32 & 0.56 & 2.81 & -0.12 \\
13 & 2.33 & 2.05 & 12.26 & 8.40 \\
14 & 4.33 & 5.75 & 15.49 & 11.44 \\
15 & -6.69 & -3.83 & 13.89 & 12.33 \\
16 & -6.57 & -4.60 & 10.52 & 10.44 \\
17 & -5.06 & -5.20 & 8.27 & 11.79 \\
18 & 1.55 & -1.07 & 10.88 & 13.66 \\
19 & 1.15 & 0.57 & 13.64 & 14.85 \\
20 & -5.52 & -6.73 & 15.88 & 14.65 \\
21 & -4.20 & -6.06 & 30.03 & 30.53 \\
22 & -8.20 & -3.33 & 30.22 & 29.80 \\
23 & -4.08 & 1.13 & 35.30 & 34.91 \\
24 & 3.30 & 6.25 & 26.53 & 27.97 \\
\bottomrule
\end{tabular}
\smallskip
\begin{minipage}{0.92\textwidth}\setstretch{1.0}
\footnotesize\textit{Notes:} Each cell reports the percentage MAE improvement
of the three-expert patient fixed-share combination (fixed $\alpha=0.02$,
$\eta=0.5$) relative to the univariate forecast at horizon $h$, with the micro
expert trained on the baseline (CPI-expenditure-weighted) versus equal-weighted aggregation, holding the macro and univariate forecasts
fixed. Positive values indicate lower MAE than the univariate forecast.
Evaluation windows are 2015--2019 and 2020--2024.
\end{minipage}
\end{table}

\subsection{Sales-Excluded Micro vs Baseline}
\label{app:nosales_encoding}

The sales-excluded variant keeps the 11-feature encoding and CPI-weighted
aggregation but drops ONS-sale-flagged price quotes from the underlying
micro panel before computing the per-category statistics. This isolates
the contribution of temporary-sale price changes to the baseline forecast.

\subsubsection{Standalone Micro Forecast}
\label{app:nosales_standalone}

Table~\ref{tab:nosales_xgb_micro_vs_sarima_two_windows} reports the
percentage MAE improvement of the standalone micro forecast relative to the univariate forecast
under sales-included (baseline) and sales-excluded micro input, for all
horizons $h=1,\ldots,24$. Averaging within four horizon groups ($h=1$ to
$6$, $7$ to $12$, $13$ to $18$, and $19$ to $24$), the sales-excluded
variant's percentage MAE improvements relative to the univariate forecast are $-43.6\%$,
$-18.2\%$, $-16.8\%$, and $-10.4\%$ in the pre-2020 window and $-10.3\%$,
$-1.7\%$, $7.3\%$, and $22.5\%$ in the post-2020 window. The
corresponding gaps relative to the baseline (sales-excluded minus
baseline) are $-14.1$, $+0.6$, $+0.6$, and $+22.2$ percentage points
pre-2020 (averaging a $2.3$ percentage-point gain) and $-19.0$, $-12.5$,
$-1.2$, and $-2.8$ percentage points post-2020 (averaging an $8.9$
percentage-point cost). Dropping sale-flagged quotes therefore helps the
standalone forecast on average in the stable pre-COVID period, driven by
the longest horizons, but hurts it across all horizon groups in the
post-COVID period, most severely at short horizons. Sale-driven price changes
add predictive signal during the inflation surge, while in stable times
they primarily inject noise that the baseline encoding has to absorb.

\begin{table}[htbp!]
\centering
\caption{Micro Forecast Percentage MAE Improvement Relative to the Univariate Forecast: Baseline versus Sales-Excluded Input}
\label{tab:nosales_xgb_micro_vs_sarima_two_windows}
\begin{tabular}{crrrr}
\toprule
 & \multicolumn{2}{c}{Pre-2020} & \multicolumn{2}{c}{Post-2020} \\
\cmidrule(lr){2-3} \cmidrule(lr){4-5}
$h$ & Baseline & Sales-Excluded & Baseline & Sales-Excluded \\
\midrule
1 & -35.67 & -35.11 & 8.49 & -0.90 \\
2 & -28.92 & -54.83 & 4.46 & -4.37 \\
3 & -29.10 & -41.51 & 2.14 & -11.52 \\
4 & -30.97 & -46.75 & 9.67 & -15.50 \\
5 & -32.52 & -41.77 & 6.73 & -17.14 \\
6 & -19.51 & -41.49 & 20.31 & -12.62 \\
7 & -17.33 & -27.53 & 13.00 & 6.65 \\
8 & -17.05 & -23.47 & 8.63 & 3.28 \\
9 & -27.49 & -34.67 & 14.63 & 9.03 \\
10 & -18.19 & -11.75 & 10.43 & -10.30 \\
11 & -19.84 & -13.76 & 16.04 & -5.61 \\
12 & -12.81 & 2.21 & 1.79 & -13.47 \\
13 & -10.32 & -14.37 & 12.95 & 4.25 \\
14 & -0.13 & -6.12 & 15.28 & 10.43 \\
15 & -22.62 & -26.52 & 14.27 & 7.94 \\
16 & -21.06 & -26.98 & 11.03 & 7.79 \\
17 & -31.47 & -22.28 & -3.33 & 3.63 \\
18 & -18.48 & -4.39 & 0.82 & 9.92 \\
19 & -28.04 & -6.07 & 9.07 & 10.54 \\
20 & -26.65 & -17.87 & 9.34 & 7.45 \\
21 & -32.13 & -10.19 & 35.42 & 31.75 \\
22 & -44.26 & -3.72 & 32.27 & 26.02 \\
23 & -36.87 & -10.60 & 39.33 & 35.48 \\
24 & -27.88 & -13.89 & 26.46 & 23.91 \\
\bottomrule
\end{tabular}
\smallskip
\begin{minipage}{0.92\textwidth}\setstretch{1.0}
\footnotesize\textit{Notes:} Each cell reports the percentage MAE improvement
of the standalone micro forecast relative to the univariate forecast at horizon $h$. The
baseline uses sales-included price quotes; the sales-excluded robustness check
drops ONS-sale-flagged quotes before computing the micro features. Positive
values indicate lower MAE for the micro forecast than for the univariate forecast. Evaluation
windows are 2015--2019 and 2020--2024.
\end{minipage}
\end{table}

\subsubsection{Three-Expert Online Learning}
\label{app:nosales_3expert}

Table~\ref{tab:nosales_3expert_fs_vs_sarima_two_windows} reports
the percentage MAE improvement of the three-expert patient Fixed Share
combination (fixed $\alpha=0.02$, $\eta=0.5$, the main-text baseline)
relative to the univariate forecast, with the micro model trained on sales-included
(baseline) and sales-excluded micro input respectively, holding the macro
and univariate forecasts fixed. Averaging within the same four horizon groups,
the sales-excluded three-expert combination's percentage MAE improvements are
$-2.5\%$, $-0.9\%$, $-3.5\%$, and $+1.9\%$ in the pre-2020 window and
$7.8\%$, $10.4\%$, $10.5\%$, and $23.7\%$ in the post-2020 window. The
corresponding gaps relative to the baseline three-expert combination are
$-3.7$, $+0.5$, $-1.8$, and $+4.8$ percentage points pre-2020 (averaging
a $0.1$ percentage-point cost) and $-4.2$, $-3.4$, $-1.4$, and $-1.5$
percentage points post-2020 (averaging a $2.6$ percentage-point cost).
The post-COVID standalone gap therefore shrinks from $8.9$ to $2.6$
percentage points once the sales-excluded micro forecast is combined with
the macro and univariate forecasts through the Fixed Share combination. As with
the other robustness checks, the combination absorbs much of the
standalone difference: when the sales-excluded micro expert weakens, the
algorithm reallocates weight toward the macro and univariate forecasts.

\begin{table}[htbp!]
\centering
\caption{Three-Expert Fixed Share: Percent MAE Improvement Relative to the Univariate Forecast under Baseline versus Sales-Excluded Micro Input ($\alpha=0.02$, $\eta=0.5$)}
\label{tab:nosales_3expert_fs_vs_sarima_two_windows}
\begin{tabular}{crrrr}
\toprule
 & \multicolumn{2}{c}{Pre-2020} & \multicolumn{2}{c}{Post-2020} \\
\cmidrule(lr){2-3} \cmidrule(lr){4-5}
$h$ & Baseline & Sales-Excluded & Baseline & Sales-Excluded \\
\midrule
1 & -1.58 & -0.54 & 12.03 & 8.76 \\
2 & 1.83 & -4.38 & 7.64 & 5.02 \\
3 & 4.29 & -1.16 & 7.40 & 4.76 \\
4 & 0.66 & -4.31 & 13.17 & 7.88 \\
5 & -0.14 & -2.17 & 13.72 & 9.74 \\
6 & 2.37 & -2.38 & 17.89 & 10.52 \\
7 & -0.08 & -1.65 & 15.63 & 12.34 \\
8 & -1.64 & -3.20 & 14.73 & 13.61 \\
9 & -3.29 & -3.42 & 18.23 & 16.54 \\
10 & -0.93 & 0.29 & 15.22 & 11.65 \\
11 & -1.85 & -0.23 & 15.89 & 8.72 \\
12 & -0.32 & 2.96 & 2.81 & -0.73 \\
13 & 2.33 & 0.14 & 12.26 & 9.60 \\
14 & 4.33 & -0.69 & 15.49 & 9.87 \\
15 & -6.69 & -9.64 & 13.89 & 12.23 \\
16 & -6.57 & -7.70 & 10.52 & 7.63 \\
17 & -5.06 & -4.40 & 8.27 & 10.56 \\
18 & 1.55 & 1.30 & 10.88 & 13.24 \\
19 & 1.15 & 0.40 & 13.64 & 14.61 \\
20 & -5.52 & -0.49 & 15.88 & 14.12 \\
21 & -4.20 & 1.38 & 30.03 & 28.54 \\
22 & -8.20 & 2.94 & 30.22 & 27.10 \\
23 & -4.08 & 2.52 & 35.30 & 33.13 \\
24 & 3.30 & 4.73 & 26.53 & 24.97 \\
\bottomrule
\end{tabular}
\smallskip
\begin{minipage}{0.92\textwidth}\setstretch{1.0}
\footnotesize\textit{Notes:} Each cell reports the percentage MAE improvement
of the three-expert patient fixed-share combination (fixed $\alpha=0.02$,
$\eta=0.5$) relative to the univariate forecast at horizon $h$, with the micro
expert trained on the sales-included (baseline) versus sales-excluded micro input, holding the macro and univariate forecasts
fixed. Positive values indicate lower MAE than the univariate forecast.
Evaluation windows are 2015--2019 and 2020--2024.
\end{minipage}
\end{table}

\subsection{Sale-Filtered Micro (Nakamura--Steinsson Filter A) vs Baseline}
\label{app:nsa_encoding}

The sale-filtered variant keeps the 11-feature encoding and CPI-weighted
aggregation but, instead of dropping sale-flagged quotes, replaces temporary
sale prices with the regular price inferred by the
\citet{NakamuraSteinsson2008} filter under its baseline parameterisation
(Filter A of Table~\ref{tab:ns_params}; Appendix~\ref{app:NS_filter} describes
the algorithm). Whereas the sales-excluded variant of
Appendix~\ref{app:nosales_encoding} removes the affected quotes, the filter
keeps them at an imputed regular price, so the comparison isolates the
consequences of imputing rather than deleting sale-driven price changes.
Because the filter's rebound and modal-price rules condition on prices up to
$LL$ months after $t$, the filtered panel is a regular-price measurement
exercise rather than a strictly real-time forecasting input.

\subsubsection{Standalone Micro Forecast}
\label{app:nsa_standalone}

Table~\ref{tab:nsa_xgb_micro_vs_sarima_two_windows} reports the percentage MAE
improvement of the standalone micro forecast relative to the univariate
forecast under sales-included (baseline) and sale-filtered micro input, for
all horizons $h=1,\ldots,24$. Averaging within four horizon groups ($h=1$ to
$6$, $7$ to $12$, $13$ to $18$, and $19$ to $24$), the sale-filtered variant's
percentage MAE improvements relative to the univariate forecast are $-36.2\%$,
$-42.1\%$, $-26.0\%$, and $-29.4\%$ in the pre-2020 window and $-1.8\%$,
$-1.4\%$, $6.7\%$, and $24.3\%$ in the post-2020 window. The corresponding
gaps relative to the baseline (sale-filtered minus baseline) are $-6.8$,
$-23.3$, $-8.6$, and $+3.2$ percentage points pre-2020 (averaging an $8.9$
percentage-point cost) and $-10.4$, $-12.1$, $-1.8$, and $-1.0$ percentage
points post-2020 (averaging a $6.3$ percentage-point cost). Unlike the
sales-excluded variant, which helps the standalone forecast in the stable
pre-COVID period, the filter hurts it in both windows: the imputed regular
prices smooth away sale-driven variation that carries signal after 2020, and
the imputation itself injects noise even in stable times.

\begin{table}[htbp!]
\centering
\caption{Micro Forecast Percentage MAE Improvement Relative to the Univariate Forecast: Baseline versus Nakamura--Steinsson Filter-A Sales Treatment}
\label{tab:nsa_xgb_micro_vs_sarima_two_windows}
\begin{tabular}{crrrr}
\toprule
 & \multicolumn{2}{c}{Pre-2020} & \multicolumn{2}{c}{Post-2020} \\
\cmidrule(lr){2-3} \cmidrule(lr){4-5}
$h$ & Baseline & Sale-Filtered & Baseline & Sale-Filtered \\
\midrule
1 & -35.67 & -34.92 & 8.49 & 6.77 \\
2 & -28.92 & -39.51 & 4.46 & -7.33 \\
3 & -29.10 & -26.05 & 2.14 & 3.57 \\
4 & -30.97 & -38.21 & 9.67 & -8.32 \\
5 & -32.52 & -34.63 & 6.73 & -5.31 \\
6 & -19.51 & -44.06 & 20.31 & 0.06 \\
7 & -17.33 & -50.29 & 13.00 & 4.85 \\
8 & -17.05 & -47.84 & 8.63 & 3.08 \\
9 & -27.49 & -52.95 & 14.63 & 5.89 \\
10 & -18.19 & -40.22 & 10.43 & -11.68 \\
11 & -19.84 & -30.16 & 16.04 & 4.99 \\
12 & -12.81 & -31.29 & 1.79 & -15.45 \\
13 & -10.32 & -25.54 & 12.95 & 8.37 \\
14 & -0.13 & -15.77 & 15.28 & 10.60 \\
15 & -22.62 & -31.96 & 14.27 & 0.97 \\
16 & -21.06 & -37.25 & 11.03 & 3.36 \\
17 & -31.47 & -24.49 & -3.33 & 6.75 \\
18 & -18.48 & -20.91 & 0.82 & 10.27 \\
19 & -28.04 & -24.85 & 9.07 & 5.90 \\
20 & -26.65 & -30.30 & 9.34 & 8.55 \\
21 & -32.13 & -30.94 & 35.42 & 35.91 \\
22 & -44.26 & -30.58 & 32.27 & 27.96 \\
23 & -36.87 & -31.13 & 39.33 & 39.47 \\
24 & -27.88 & -28.39 & 26.46 & 27.92 \\
\bottomrule
\end{tabular}
\smallskip
\begin{minipage}{0.92\textwidth}\setstretch{1.0}
\footnotesize\textit{Notes:} Each cell reports the percentage MAE improvement of the standalone micro forecast relative to the univariate forecast at horizon $h$. The baseline includes sale-flagged quotes; the sale-filtered robustness check replaces temporary sale prices with the regular price inferred by the Nakamura--Steinsson filter, parameterisation A.
Positive values indicate lower MAE than the univariate forecast. Evaluation
windows are 2015--2019 and 2020--2024.
\end{minipage}
\end{table}

\subsubsection{Three-Expert Online Learning}
\label{app:nsa_3expert}

Table~\ref{tab:nsa_3expert_fs_vs_sarima_two_windows} reports the percentage
MAE improvement of the three-expert patient Fixed Share combination (fixed
$\alpha=0.02$, $\eta=0.5$, the main-text baseline) relative to the univariate
forecast, with the micro model trained on sales-included (baseline) and
sale-filtered micro input respectively, holding the macro and univariate
forecasts fixed. Averaging within the same four horizon groups, the
sale-filtered three-expert combination's percentage MAE improvements are
$0.3\%$, $-0.9\%$, $-2.3\%$, and $-0.2\%$ in the pre-2020 window and $8.7\%$,
$10.5\%$, $10.4\%$, and $24.9\%$ in the post-2020 window. The corresponding
gaps relative to the baseline three-expert combination are $-0.9$, $+0.5$,
$-0.6$, and $+2.7$ percentage points pre-2020 (averaging a $0.4$
percentage-point gain) and $-3.3$, $-3.3$, $-1.5$, and $-0.3$ percentage
points post-2020 (averaging a $2.1$ percentage-point cost). The post-COVID
standalone gap therefore shrinks from $6.3$ to $2.1$ percentage points once
the sale-filtered micro forecast is combined with the macro and univariate
forecasts through the Fixed Share combination: when the sale-filtered micro
expert weakens, the algorithm reallocates weight toward the macro and
univariate forecasts.

\begin{table}[htbp!]
\centering
\caption{Three-Expert Fixed Share: Percent MAE Improvement Relative to the Univariate Forecast under Baseline versus Nakamura--Steinsson Filter-A Sales Treatment ($\alpha=0.02$, $\eta=0.5$)}
\label{tab:nsa_3expert_fs_vs_sarima_two_windows}
\begin{tabular}{crrrr}
\toprule
 & \multicolumn{2}{c}{Pre-2020} & \multicolumn{2}{c}{Post-2020} \\
\cmidrule(lr){2-3} \cmidrule(lr){4-5}
$h$ & Baseline & Sale-Filtered & Baseline & Sale-Filtered \\
\midrule
1 & -1.58 & -0.84 & 12.03 & 12.29 \\
2 & 1.83 & -0.61 & 7.64 & 4.46 \\
3 & 4.29 & 4.44 & 7.40 & 5.13 \\
4 & 0.66 & 0.09 & 13.17 & 7.08 \\
5 & -0.14 & 0.41 & 13.72 & 12.16 \\
6 & 2.37 & -1.39 & 17.89 & 11.10 \\
7 & -0.08 & -1.54 & 15.63 & 12.35 \\
8 & -1.64 & -3.40 & 14.73 & 12.96 \\
9 & -3.29 & -3.60 & 18.23 & 15.56 \\
10 & -0.93 & -0.82 & 15.22 & 13.09 \\
11 & -1.85 & 2.03 & 15.89 & 11.65 \\
12 & -0.32 & 2.00 & 2.81 & -2.82 \\
13 & 2.33 & 1.88 & 12.26 & 9.51 \\
14 & 4.33 & -0.51 & 15.49 & 10.33 \\
15 & -6.69 & -5.83 & 13.89 & 9.86 \\
16 & -6.57 & -6.69 & 10.52 & 6.71 \\
17 & -5.06 & -4.74 & 8.27 & 12.13 \\
18 & 1.55 & 1.95 & 10.88 & 13.92 \\
19 & 1.15 & 1.54 & 13.64 & 14.06 \\
20 & -5.52 & -6.62 & 15.88 & 14.78 \\
21 & -4.20 & -1.84 & 30.03 & 29.53 \\
22 & -8.20 & -1.04 & 30.22 & 29.06 \\
23 & -4.08 & 1.91 & 35.30 & 35.47 \\
24 & 3.30 & 4.98 & 26.53 & 26.64 \\
\bottomrule
\end{tabular}
\smallskip
\begin{minipage}{0.92\textwidth}\setstretch{1.0}
\footnotesize\textit{Notes:} Each cell reports the percentage MAE improvement of the three-expert patient fixed-share combination (fixed $\alpha=0.02$, $\eta=0.5$) relative to the univariate forecast at horizon $h$, with the micro expert trained on sales-included (baseline) versus Nakamura--Steinsson Filter-A sale-filtered micro input, holding the macro and univariate forecasts fixed.
Positive values indicate lower MAE than the univariate forecast. Evaluation
windows are 2015--2019 and 2020--2024.
\end{minipage}
\end{table}

\clearpage
\section{Additional Results}\label{app:additional_results}

\subsection{Macro Forecasts versus the Univariate Benchmark}
\label{app:macro_vs_sarima}

While the primary focus of the paper is on microdata, we also document the
performance of the macro forecast, which displays an analogous
regime-dependent pattern.

\paragraph{Visual evidence across forecast horizons.}
Figure~\ref{fig:xgb_macro_vs_sarima_h12} plots both sets of forecasts against
realized inflation at horizon $h=12$. In the pre-2020 panels, the univariate and
macro forecasts track realized inflation similarly, with the univariate forecast generally
remaining slightly closer. The post-2020 panel shows why $h=12$ is one of the
macro forecast's weakest post-2020 horizons: neither forecast anticipates the
2021--22 run-up, and the macro forecast stays flat through the surge, so the
univariate forecast retains a $9.9\%$ MAE advantage at this horizon. The
macro forecast's gains are instead concentrated at long horizons, where the
post-2020 improvements reach $29\%$ to $42\%$ ($h=21$ to $24$;
Table~\ref{tab:fine_macro_rents_vs_sarima_improvement}), as the univariate
forecast overshoots badly in 2023--24 at those horizons while the macro
forecast does not. Figures~\ref{fig:xgb_macro_vs_sarima_h1},
\ref{fig:xgb_macro_vs_sarima_h6}, and~\ref{fig:xgb_macro_vs_sarima_h24} report
the $h=1$, $h=6$, and $h=24$ forecasts.

\begin{figure}[t!]
  \centering
  \caption{Macro Forecast vs.\ Univariate Forecast: $h=12$}
  \label{fig:xgb_macro_vs_sarima_h12}
  \includegraphics[width=\textwidth]{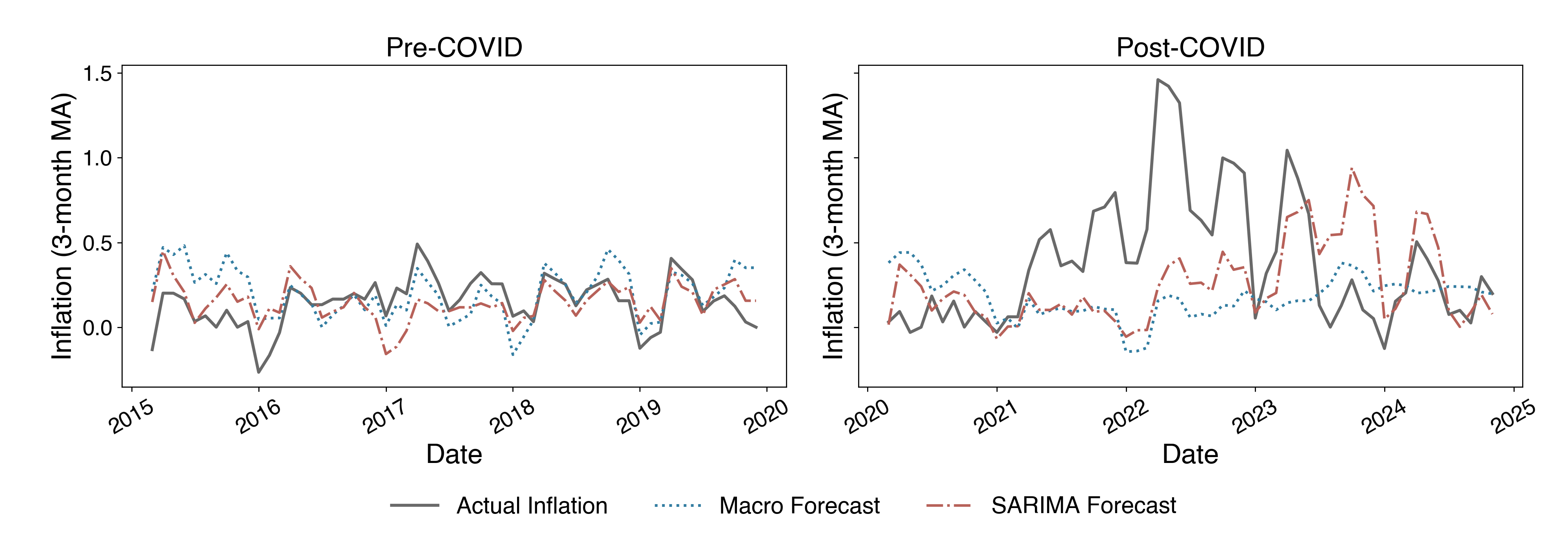}
  \smallskip
  \begin{minipage}{\textwidth}\setstretch{1.0}
    \footnotesize\textit{Notes:} Each panel plots 3-month trailing moving
    averages of realized monthly inflation, the macro 12-month-ahead
    forecast, and the univariate (SARIMA) 12-month-ahead forecast. Dates on
    the x-axis are target months: a forecast dated $t$ was formed at origin
    $t-12$. The panels correspond to the 2015--2019 and 2020--2024 evaluation
    windows.
  \end{minipage}
\end{figure}

\paragraph{Aggregate evidence: pre-2020 versus post-2020 windows.}
Table~\ref{tab:fine_macro_rents_vs_sarima_improvement} reports the
percentage MAE improvement of the macro forecast relative to the univariate forecast
across the two evaluation windows, pre-2020 (2015--2019) and post-2020
(2020--2024), for all horizons $h=1,\ldots,24$;
Figure~\ref{fig:xgb_macro_vs_sarima_improvement_heatmap} visualizes the
same comparison as a heatmap.

Over the pre-2020 window, the macro forecast underperforms the univariate forecast at all
24 horizons; the smallest shortfall is $5.2\%$ at $h=13$. Averaging within
four horizon groups ($h=1$ to $6$, $7$ to $12$,
$13$ to $18$, and $19$ to $24$), the percentage MAE losses are $15.1\%$,
$19.4\%$, $15.9\%$, and $13.5\%$. Macro covariates add no
predictive value beyond the univariate benchmark in a stable-inflation
environment.

After 2020, the ranking changes. The macro forecast improves on the univariate forecast at 19
of the 24 horizons, with the exceptions at $h=2$ to $5$ (between $-1.5\%$
and $-3.8\%$) and at $h=12$ ($-9.9\%$). Averaging within the same four
horizon groups, the post-COVID percentage MAE gains are $-0.4\%$, $9.4\%$,
$6.4\%$, and $26.7\%$, rising with the horizon and peaking at the longest
horizons, where the largest gain reaches $41.9\%$ at $h=23$.
The contrast between the two windows indicates that the value of macro data
is state dependent: limited during normal times but substantially higher
during periods of macroeconomic stress.

\begin{table}[htbp!]
\centering
\caption{Macro percent MAE improvement over SARIMA, pre- and post-COVID.}
\label{tab:fine_macro_rents_vs_sarima_improvement}
\begin{tabular}{crr}
\toprule
$h$ & Pre-COVID \% impr. & Post-COVID \% impr. \\
\midrule
1 & -25.19 & 7.86 \\
2 & -22.25 & -3.14 \\
3 & -14.44 & -3.82 \\
4 & -11.02 & -3.09 \\
5 & -5.59 & -1.51 \\
6 & -11.86 & 1.10 \\
7 & -10.26 & 14.52 \\
8 & -19.80 & 16.59 \\
9 & -25.06 & 19.53 \\
10 & -21.28 & 10.08 \\
11 & -19.13 & 5.34 \\
12 & -20.69 & -9.92 \\
13 & -5.24 & 3.59 \\
14 & -10.14 & 2.51 \\
15 & -23.74 & 11.30 \\
16 & -25.62 & 4.18 \\
17 & -20.12 & 7.66 \\
18 & -10.26 & 9.14 \\
19 & -13.54 & 10.09 \\
20 & -23.46 & 11.46 \\
21 & -20.07 & 33.55 \\
22 & -10.89 & 34.12 \\
23 & -5.73 & 41.87 \\
24 & -7.28 & 29.18 \\
\bottomrule
\end{tabular}
\end{table}

\subsection{Univariate + Macro Online Learning as a Benchmark without Microdata}
\label{app:sarima_macro_online}

To isolate the incremental contribution of microdata to the online-learning
step, we consider a two-expert benchmark that combines only the macro and
univariate forecasts, using the same patient Fixed Share baseline as in
Subsection~\ref{sec:fixed_share} ($\eta=0.5$, $\alpha=0.02$). This exercise asks
whether a flexible combination of macro and
univariate forecasts can already replicate the gains that we attribute to
adding microdata.

\paragraph{Results.}
Table~\ref{tab:fs_mae_two_windows_macsar}
shows that the two-expert benchmark tracks the univariate forecast before
2020 and improves on it broadly afterwards. In the pre-2020 window, the
percentage MAE improvements averaged within four horizon groups ($h=1$ to
$6$, $7$ to $12$, $13$ to $18$, and $19$ to $24$) are $1.1\%$, $-1.5\%$,
$-3.4\%$, and $0.4\%$, with positive gains at 9 of the 24 horizons and an
average gap of $-0.9\%$. In the
post-2020 window, by contrast, the two-expert combination improves on the univariate forecast at
all 24 horizons, with group averages of $8.8\%$, $12.9\%$, $9.6\%$, and
$21.6\%$, averaging $13.2\%$ and strongest at the longest horizons. Fixed
Share combination of macro and univariate
information therefore captures a substantial share of the post-2020
forecasting signal.
Figure~\ref{fig:online_sm_improvement_heatmap} visualizes the
two-expert versus univariate forecast comparison as a heatmap.

\begin{table}[htbp!]
\centering
\caption{Fixed-Share Combination Mean Absolute Error Before and After 2020 (Macro + SARIMA Combination)}
\label{tab:fs_mae_two_windows_macsar}
\begin{tabular}{crrrrrr}
\toprule
& \multicolumn{3}{c}{Pre 2020} & \multicolumn{3}{c}{Post 2020} \\
\cmidrule(lr){2-4} \cmidrule(lr){5-7}
$h$ & FS & Macro & SARIMA & FS & Macro & SARIMA \\
\midrule
1 & 0.1390 & 0.1708 & 0.1364 & 0.3487 & 0.3607 & 0.3915 \\
2 & 0.1352 & 0.1678 & 0.1373 & 0.3596 & 0.3962 & 0.3842 \\
3 & 0.1382 & 0.1629 & 0.1424 & 0.3786 & 0.4137 & 0.3985 \\
4 & 0.1523 & 0.1679 & 0.1512 & 0.3624 & 0.4080 & 0.3957 \\
5 & 0.1440 & 0.1559 & 0.1476 & 0.3553 & 0.4035 & 0.3975 \\
6 & 0.1469 & 0.1681 & 0.1503 & 0.3782 & 0.4221 & 0.4268 \\
7 & 0.1472 & 0.1631 & 0.1479 & 0.4104 & 0.4072 & 0.4764 \\
8 & 0.1546 & 0.1835 & 0.1532 & 0.4082 & 0.4040 & 0.4843 \\
9 & 0.1500 & 0.1843 & 0.1474 & 0.4139 & 0.4079 & 0.5070 \\
10 & 0.1579 & 0.1876 & 0.1547 & 0.3773 & 0.4050 & 0.4504 \\
11 & 0.1513 & 0.1763 & 0.1479 & 0.4283 & 0.4622 & 0.4883 \\
12 & 0.1598 & 0.1878 & 0.1556 & 0.4183 & 0.4656 & 0.4235 \\
13 & 0.1820 & 0.1921 & 0.1825 & 0.4374 & 0.4582 & 0.4753 \\
14 & 0.1848 & 0.2004 & 0.1820 & 0.4539 & 0.4755 & 0.4878 \\
15 & 0.1760 & 0.2042 & 0.1651 & 0.4325 & 0.4340 & 0.4893 \\
16 & 0.1757 & 0.2079 & 0.1655 & 0.4439 & 0.4569 & 0.4768 \\
17 & 0.1683 & 0.1928 & 0.1605 & 0.4158 & 0.4360 & 0.4722 \\
18 & 0.1723 & 0.1868 & 0.1694 & 0.4156 & 0.4303 & 0.4736 \\
19 & 0.1844 & 0.2034 & 0.1792 & 0.4087 & 0.4244 & 0.4720 \\
20 & 0.1814 & 0.2062 & 0.1670 & 0.4083 & 0.4172 & 0.4712 \\
21 & 0.1715 & 0.1988 & 0.1656 & 0.5050 & 0.4339 & 0.6530 \\
22 & 0.1678 & 0.1877 & 0.1693 & 0.4607 & 0.4090 & 0.6208 \\
23 & 0.1600 & 0.1782 & 0.1686 & 0.4850 & 0.4105 & 0.7061 \\
24 & 0.1570 & 0.1907 & 0.1778 & 0.4463 & 0.4113 & 0.5808 \\
\bottomrule
\end{tabular}
\smallskip
\begin{minipage}{0.95\textwidth}\setstretch{1.0}
\footnotesize\textit{Notes:} Mean absolute error of the patient fixed-share
combination (FS, $\alpha=0.02$, $\eta=0.5$, $\rho=0$, grid=fine) against
its experts (Macro, SARIMA), by forecast horizon.
The pre-2020 window is 2015--2019 and the post-2020 window is 2020--2024.
\end{minipage}
\end{table}

\begin{figure}[p!]
\centering
\caption{Macro Forecast MAE Improvement Relative to the Univariate Forecast}
\label{fig:xgb_macro_vs_sarima_improvement_heatmap}
\includegraphics[width=0.9\textwidth,height=0.16\textheight,keepaspectratio]{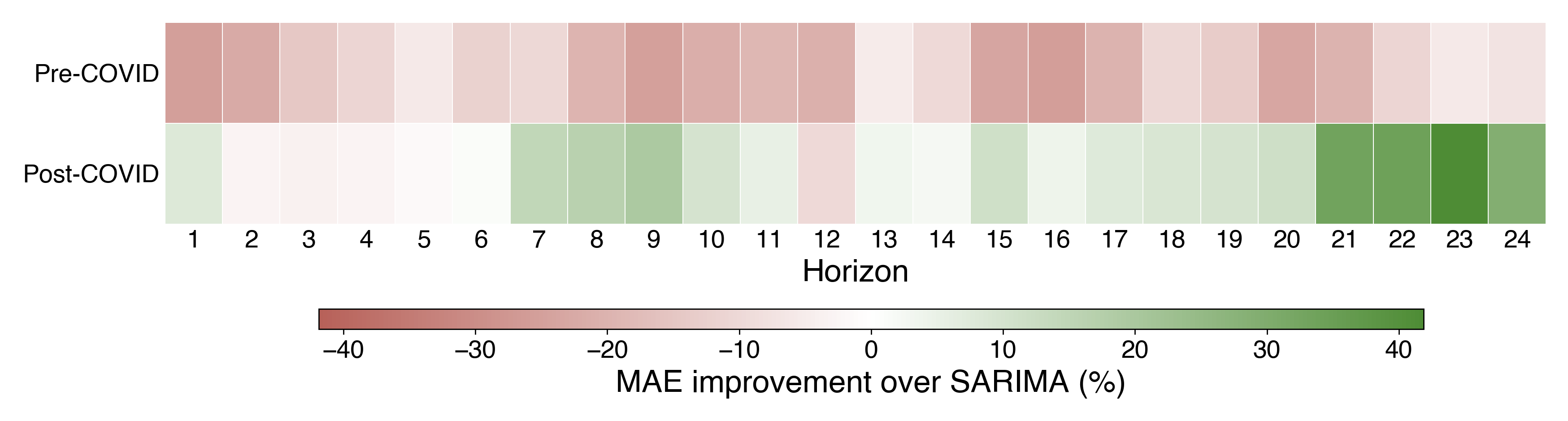}
  \smallskip
  \begin{minipage}{0.96\textwidth}\setstretch{1.0}
    \scriptsize\textit{Notes:} Cells report the percentage MAE improvement
    of the macro forecast relative to the univariate forecast by forecast horizon $h=1,\ldots,24$
    and evaluation window. Positive values indicate lower MAE for the macro
    forecast than for the univariate forecast; negative values indicate the reverse. Green
    shading denotes positive improvements and red shading denotes negative
    improvements. Numerical values are reported in
    Table~\ref{tab:fine_macro_rents_vs_sarima_improvement}.
  \end{minipage}

\vspace{0.6em}

\caption{Univariate + Macro Online-Learning MAE Improvement Relative to the Univariate Forecast}
\label{fig:online_sm_improvement_heatmap}
\includegraphics[width=0.9\textwidth,height=0.16\textheight,keepaspectratio]{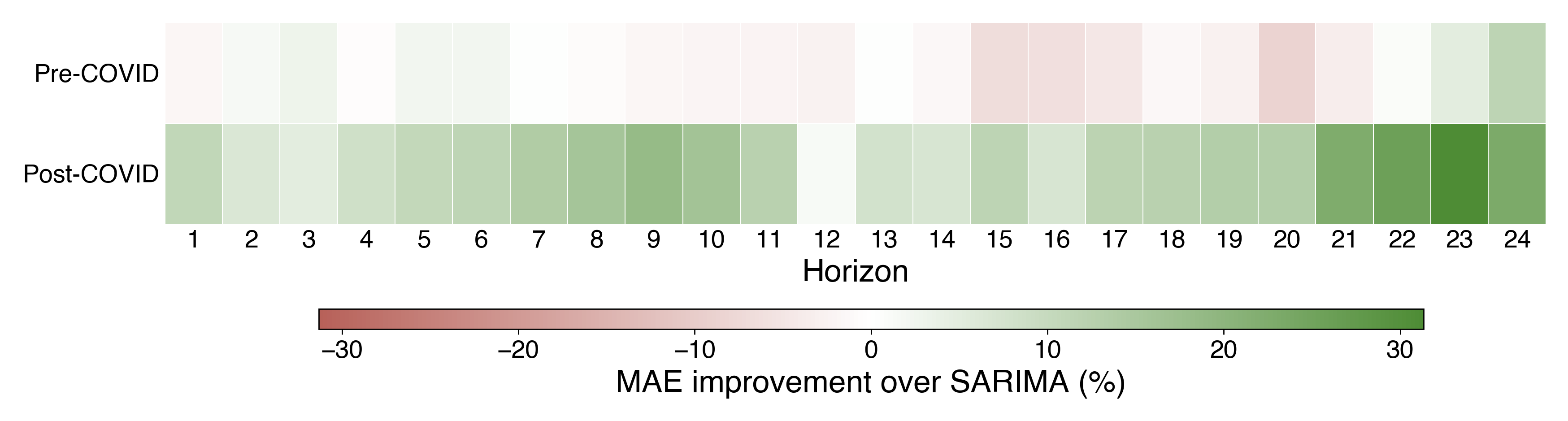}
  \smallskip
  \begin{minipage}{0.96\textwidth}\setstretch{1.0}
    \scriptsize\textit{Notes:} Cells report the percentage MAE improvement
    of the two-expert Fixed Share combination relative to the univariate forecast by
    forecast horizon $h=1,\ldots,24$ and evaluation window. The two experts
    are the macro and univariate forecasts. Positive values indicate lower MAE for the
    combination than for the univariate forecast; negative values indicate the
    reverse. Numerical values are reported in
    Table~\ref{tab:fs_mae_two_windows_macsar}.
  \end{minipage}
\end{figure}

\paragraph{Visual evidence across forecast horizons.}
Figure~\ref{fig:online_sm_h12} displays the aggregate forecast path and expert
weights for the representative horizon $h=12$. In the pre-2020 window, the
two-expert combination tracks realized inflation roughly as closely as the univariate forecast
alone, reflecting modest diversification gains. In the post-2020 window, no
forecast anticipates the 2021--22 run-up at this horizon, and the combined
forecast tracks the univariate forecast through most of the episode,
improving on it only modestly at $h=12$; as with the macro expert itself, the
two-expert combination's gains are concentrated at the longest horizons
(Table~\ref{tab:fs_mae_two_windows_macsar}).
Figures~\ref{fig:online_sm_h1}, \ref{fig:online_sm_h6}, and~\ref{fig:online_sm_h24}
report the corresponding combination weights for $h=1$, $h=6$, and $h=24$.

\begin{figure}[t!]
  \centering
  \caption{Two-Expert Combined Forecast and Weights: $h=12$}
  \label{fig:online_sm_h12}
  \includegraphics[width=\textwidth]{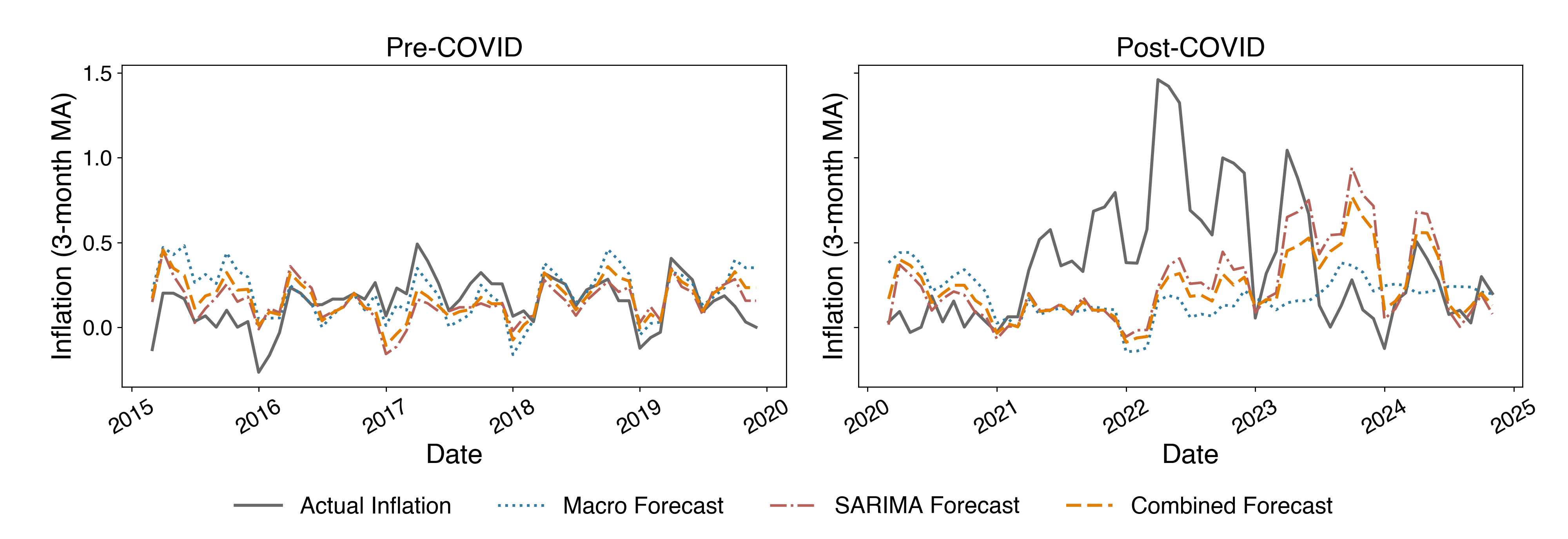}

  \medskip

  \includegraphics[width=\textwidth]{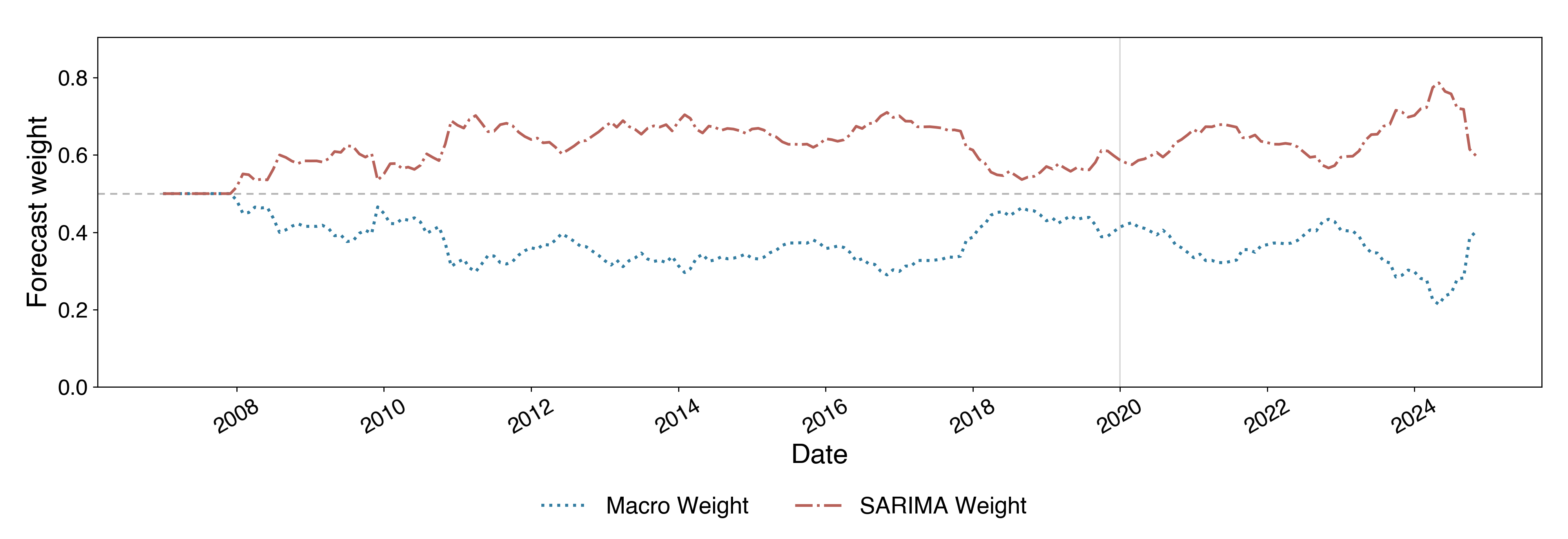}
  \smallskip
  \begin{minipage}{\textwidth}\setstretch{1.0}
    \footnotesize\textit{Notes:} The top panel plots 3-month trailing moving
    averages of realized inflation, the macro and univariate expert
    forecasts, and the two-expert Fixed Share combined forecast for the
    pre-2020 (2015--2019) and post-2020 (2020--2024) windows; dates on the
    x-axis are target months, so a forecast dated $t$ was formed at origin
    $t-12$. The bottom panel plots the combination weights
    on the two experts over the full sample, with the dashed horizontal line
    at the equal weight $1/2$ and the vertical line marking 2020.
  \end{minipage}
\end{figure}

At $h=1$, the combination leans toward the univariate forecast before 2020,
with an average univariate weight of 0.660, but the macro weight rises to
0.539 after 2020 (see Appendix Figure~\ref{fig:online_sm_h1}). At $h=6$, the
average post-2020 weights are 0.457 on the macro forecast and 0.543 on the
univariate forecast (see Appendix Figure~\ref{fig:online_sm_h6}). At $h=12$,
the univariate forecast keeps the majority weight in both windows (0.623
before 2020 and 0.647 after), consistent with $h=12$ being the macro
expert's weakest post-2020 horizon. At $h=24$, the weights stay close to
balanced, with average macro weights of 0.555 before and 0.532 after 2020
(see Appendix Figure~\ref{fig:online_sm_h24}).

\clearpage

\clearpage
\subsection{Three-Expert Combination with a Varying Learning Rate}
\label{app:three_expert_adaptive}

Table~\ref{tab:adaptive_mae_two_windows} reports the MAE of the three-expert
combination when the learning rate is re-selected at each forecast origin by
the neighbor-only rolling search of
Algorithm~\ref{alg:fs_precompute_neighbor_select_eta} ($W=60$ months, fixed
$\alpha=0.02$), alongside the fixed $\eta=0.5$ baseline reported in the main
text. The two specifications are nearly indistinguishable: the adaptive rule
attains a lower MAE at 9 of the 24 horizons before 2020 and at 12 of the 24
horizons after 2020, and the average gap is within $0.3$ percent in both
windows. The transparency of a fixed learning rate therefore costs nothing in
accuracy.

\begin{table}[htbp!]
\centering
\caption{Adaptive Fixed-Share MAE Before and After 2020}
\label{tab:adaptive_mae_two_windows}
\begin{tabular}{crrrr}
\toprule
& \multicolumn{2}{c}{Pre 2020} & \multicolumn{2}{c}{Post 2020} \\
\cmidrule(lr){2-3} \cmidrule(lr){4-5}
$h$ & Adaptive $\eta$ & Fixed $\eta$ & Adaptive $\eta$ & Fixed $\eta$ \\
\midrule
1 & 0.1381 & 0.1385 & 0.3417 & 0.3444 \\
2 & 0.1376 & 0.1348 & 0.3506 & 0.3548 \\
3 & 0.1384 & 0.1363 & 0.3664 & 0.3690 \\
4 & 0.1501 & 0.1502 & 0.3421 & 0.3436 \\
5 & 0.1489 & 0.1478 & 0.3416 & 0.3429 \\
6 & 0.1472 & 0.1467 & 0.3505 & 0.3505 \\
7 & 0.1491 & 0.1481 & 0.4010 & 0.4019 \\
8 & 0.1565 & 0.1557 & 0.4138 & 0.4130 \\
9 & 0.1494 & 0.1523 & 0.4137 & 0.4145 \\
10 & 0.1570 & 0.1562 & 0.3847 & 0.3819 \\
11 & 0.1498 & 0.1507 & 0.4123 & 0.4107 \\
12 & 0.1567 & 0.1561 & 0.4129 & 0.4116 \\
13 & 0.1780 & 0.1783 & 0.4176 & 0.4170 \\
14 & 0.1741 & 0.1741 & 0.4123 & 0.4122 \\
15 & 0.1770 & 0.1761 & 0.4238 & 0.4213 \\
16 & 0.1771 & 0.1763 & 0.4287 & 0.4266 \\
17 & 0.1689 & 0.1686 & 0.4332 & 0.4332 \\
18 & 0.1661 & 0.1668 & 0.4212 & 0.4221 \\
19 & 0.1774 & 0.1771 & 0.4050 & 0.4077 \\
20 & 0.1756 & 0.1762 & 0.3953 & 0.3964 \\
21 & 0.1730 & 0.1725 & 0.4546 & 0.4569 \\
22 & 0.1856 & 0.1831 & 0.4333 & 0.4332 \\
23 & 0.1777 & 0.1754 & 0.4580 & 0.4569 \\
24 & 0.1703 & 0.1719 & 0.4247 & 0.4267 \\
\bottomrule
\end{tabular}
\smallskip
\begin{minipage}{0.95\textwidth}\setstretch{1.0}\footnotesize\textit{Notes:} Mean absolute
error of the varying-learning-rate patient fixed share, which re-selects $\eta$ at each
origin by the neighbor-only rolling search ($W=60$, fixed $\alpha=0.02$,
grid=fine), against the fixed $\eta=0.5$ patient baseline, by forecast
horizon. Pre-2020 is 2015--2019, post-2020 is 2020--2024.\end{minipage}
\end{table}

\clearpage

\subsection{6-Group Shapley Values}
\label{app:6_group_shapley}

\paragraph{Construction.}
We partition the micro variables, computed within each COICOP1 category and
then stacked across categories (except the monthly dummies, which are common
across categories), into six groups:
\begin{enumerate}
  \item the fraction of zero price changes;
  \item the mean price change;
  \item the higher moments of the price-change distribution;
  \item the deciles of the price-change distribution;
  \item the time since last change features;
  \item monthly dummies.
\end{enumerate}

This decomposition is based on the augmented encoding rather than the baseline,
so it adds two feature groups that are absent from the main-text decomposition,
namely the higher moments of the price-change distribution and the hazard-style
statistics based on time since last adjustment. It lets us separate the
location of the price-change distribution (the mean), its coarse shape
(deciles), its higher moments, and the time-since-last-change features.

\begin{table}[htbp!]
\centering
\caption{6-Grouped Shapley Values for the Micro Expert, Averaged over 24 Horizons}
\label{tab:micro_shapley_aggregate_24h_6}
\begin{tabular}{lrr}
\toprule
Feature group & Pre-2020 & Post-2020 \\
\midrule
Fraction of zero price changes & 0.5964 & 0.0743 \\
Mean price change & 0.2853 & 0.7311 \\
Higher moments of price-change distribution & -0.6774 & 0.2298 \\
Deciles of price-change distribution & -0.0886 & 0.7128 \\
Time since last change features & -0.8356 & -2.8658 \\
Monthly dummies & 0.1804 & 0.8718 \\
\bottomrule
\end{tabular}
\begin{minipage}{0.92\textwidth}\vspace{0.4em}\footnotesize\textit{Notes:}
Entries are grouped Shapley contributions of each feature group to the MAE gain
from adding the micro expert to the two-expert online-learning benchmark, which
combines the univariate and macro forecasts. Values are in percent of benchmark
MAE and are averaged over horizons $h=1,\ldots,24$ within each evaluation
window. Positive entries reduce MAE; negative entries increase it.\end{minipage}
\end{table}

Table~\ref{tab:micro_shapley_aggregate_24h_6} decomposes the augmented-encoding
micro expert. In the pre-COVID window, no feature
group makes a large positive contribution: the fraction of zero price changes
is the most positive at $0.60$, followed by the mean price change at $0.29$ and
monthly dummies at $0.18$, while the time-since-last-change
features ($-0.84$), the higher moments ($-0.68$), and the deciles ($-0.09$)
are negative. This again points to limited
incremental forecasting content in microdata during the low and stable
inflation period.

In the post-COVID window, the positive contributions come from monthly
dummies ($0.87$), the mean price change ($0.73$), the deciles of the
price-change distribution ($0.71$), and, more modestly, the higher moments
($0.23$); the fraction of zero price changes is essentially zero ($0.07$).
The
time-since-last-change features, however, contribute $-2.87$, more than
offsetting every positive group. The six contributions therefore sum to the
augmented-encoding three-expert improvement over the two-expert
benchmark, which is negative after 2020 ($-0.25$ percent). The augmented
encoding therefore underperforms the baseline encoding after COVID, and the
shortfall is accounted for by the hazard-style timing block rather than by
the distributional statistics, consistent with the encoding-level robustness
comparison in Subsection~\ref{sec:micro_robustness}.

\begin{figure}[t!]
  \centering
  \caption{6-Horizon Moving Average of Per-Horizon 6-Group Shapley Values}
  \label{fig:micro_shapley_six_horizon_ma_6}
  \begin{subfigure}[t]{0.49\textwidth}
    \centering
    \caption{Pre-COVID}
    \includegraphics[width=\textwidth]{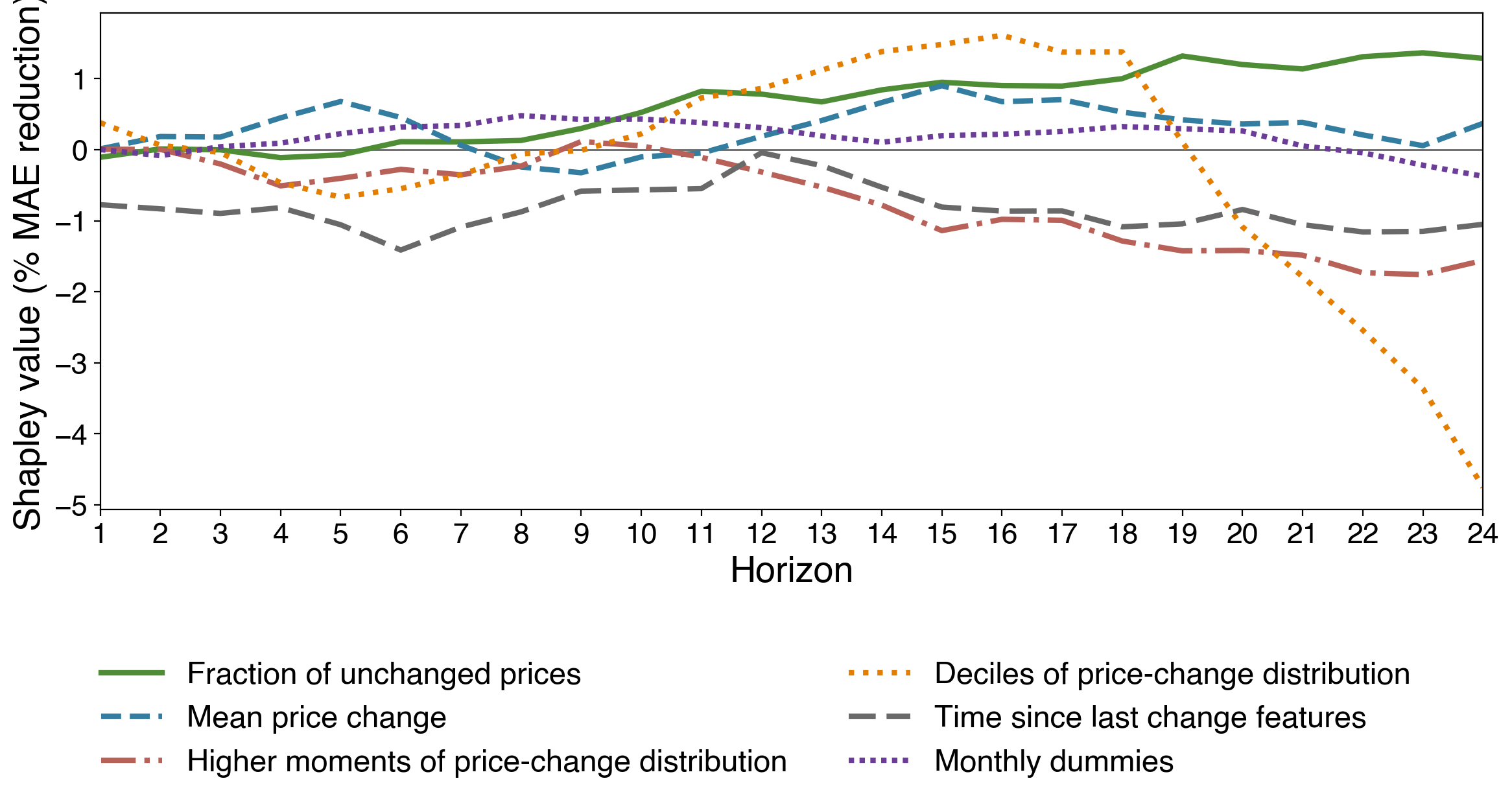}
  \end{subfigure}
  \hfill
  \begin{subfigure}[t]{0.49\textwidth}
    \centering
    \caption{Post-COVID}
    \includegraphics[width=\textwidth]{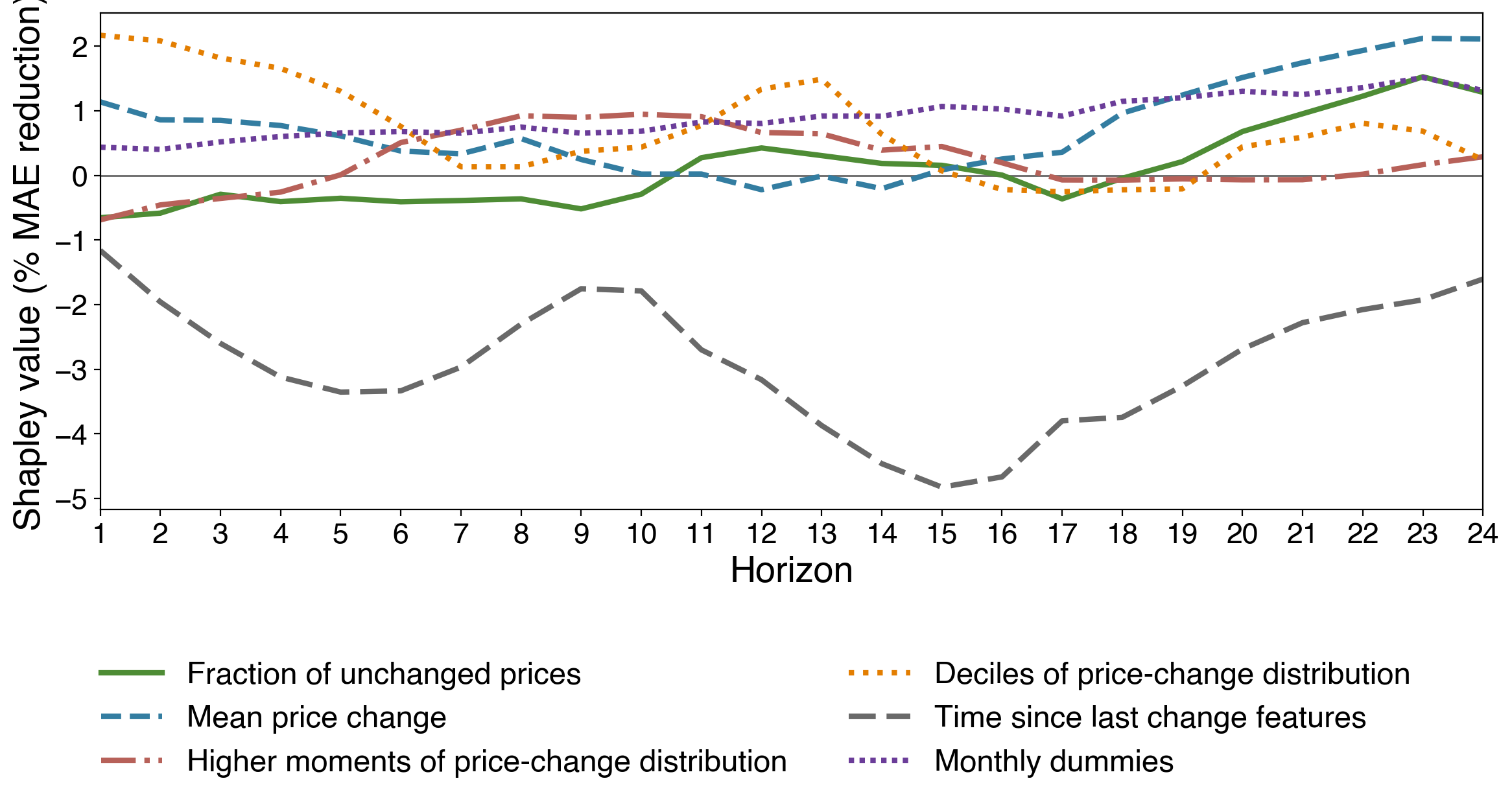}
  \end{subfigure}
  \smallskip
  \begin{minipage}{\textwidth}\setstretch{1.0}
    \footnotesize\textit{Notes:} The figure reports 6-group Shapley
    contributions for the incremental value of the augmented-encoding micro
    forecast in the three-expert combined forecast relative to the two-expert
    benchmark that combines the univariate and macro forecasts. Values are smoothed
    across adjacent horizons using a 6-horizon moving average. Positive
    values indicate reductions in MAE from the corresponding feature group;
    negative values indicate increases in MAE. The horizontal line marks a
    zero contribution.
  \end{minipage}
\end{figure}

Figure~\ref{fig:micro_shapley_six_horizon_ma_6} shows the same pattern at the
horizon level. Before COVID, the per-horizon contributions are small at short
and medium horizons, and the deciles group deteriorates steeply at the
longest horizons, echoing the long-horizon losses of the baseline
decomposition. After COVID, monthly dummies are positive at every horizon,
the mean price change strengthens toward long horizons, and the deciles are
positive at short and medium horizons. The time-since-last-change features
are strongly negative at every horizon in the post-COVID window, reaching
below $-4$ at medium horizons, and are negative at most horizons before COVID
as well. Taken together, the 6-group decomposition suggests that the
post-COVID forecasting value of microdata comes from the location and shape
of the price-change distribution together with seasonal timing, while the
hazard-style duration statistics actively hurt.

\clearpage
\subsection{4-Group Shapley Values}
\label{app:baseline_4_group_shapley}

This subsection reports a coarser 4-group Shapley decomposition. It
partitions the same baseline feature set as the 5-group specification in
Section~\ref{sec:group_shapley}, but instead of separating the center and the
tails of the price-change distribution, it keeps the mean price change and
the deciles as single groups:

\begin{enumerate}
  \item the fraction of zero price changes;
  \item the mean price change;
  \item the deciles of the price-change distribution;
  \item monthly dummies.
\end{enumerate}

\begin{table}[htbp!]
\centering
\caption{4-Group Shapley Values for the Micro Expert, Averaged over 24 Horizons}
\label{tab:micro_shapley_aggregate_24h_baseline4}
\begin{tabular}{lrr}
\toprule
Feature group & Pre-2020 & Post-2020 \\
\midrule
Fraction of zero price changes & 0.1254 & -0.0093 \\
Mean price change & -0.2216 & 0.9672 \\
Deciles of price-change distribution & -0.6009 & 0.6423 \\
Monthly dummies & 0.2993 & 1.3076 \\
\bottomrule
\end{tabular}
\begin{minipage}{0.92\textwidth}\vspace{0.4em}\footnotesize\textit{Notes:}
Entries are grouped Shapley contributions of each feature group to the MAE gain
from adding the micro expert to the two-expert online-learning benchmark, which
combines the univariate and macro forecasts. Values are in percent of benchmark
MAE and are averaged over horizons $h=1,\ldots,24$ within each evaluation
window. Positive entries reduce MAE; negative entries increase it.\end{minipage}
\end{table}

Table~\ref{tab:micro_shapley_aggregate_24h_baseline4} shows that, in the
pre-COVID window, the contributions of the deciles of the price-change
distribution and the mean price change are negative, with values $-0.60$ and
$-0.22$, respectively, while monthly dummies ($0.30$) and the
fraction of zero price changes ($0.13$) are mildly positive. Thus, before
COVID
the grouped micro contributions are small overall, with no group delivering a
large positive contribution. Because the 4-group and 5-group specifications
partition the same baseline feature set, their contributions sum to the same
totals: $-0.40$ percent before 2020 and $2.91$ percent after.

In the post-COVID window, the largest positive contribution comes from monthly
dummies at $1.31$, followed by the mean price change at $0.97$ and the
deciles of the price-change distribution at $0.64$; the fraction of zero price
changes is essentially zero, at $-0.01$. The positive contribution
therefore comes from information in the average and distribution of
price changes, together with seasonal timing, while the extensive-margin
measure contributes essentially nothing on average.

\begin{figure}[t!]
  \centering
  \caption{4-Group Shapley Values by Horizon}
  \label{fig:micro_shapley_six_horizon_ma_baseline4}
  \begin{subfigure}[t]{0.49\textwidth}
    \centering
    \caption{Pre-COVID}
    \includegraphics[width=\textwidth]{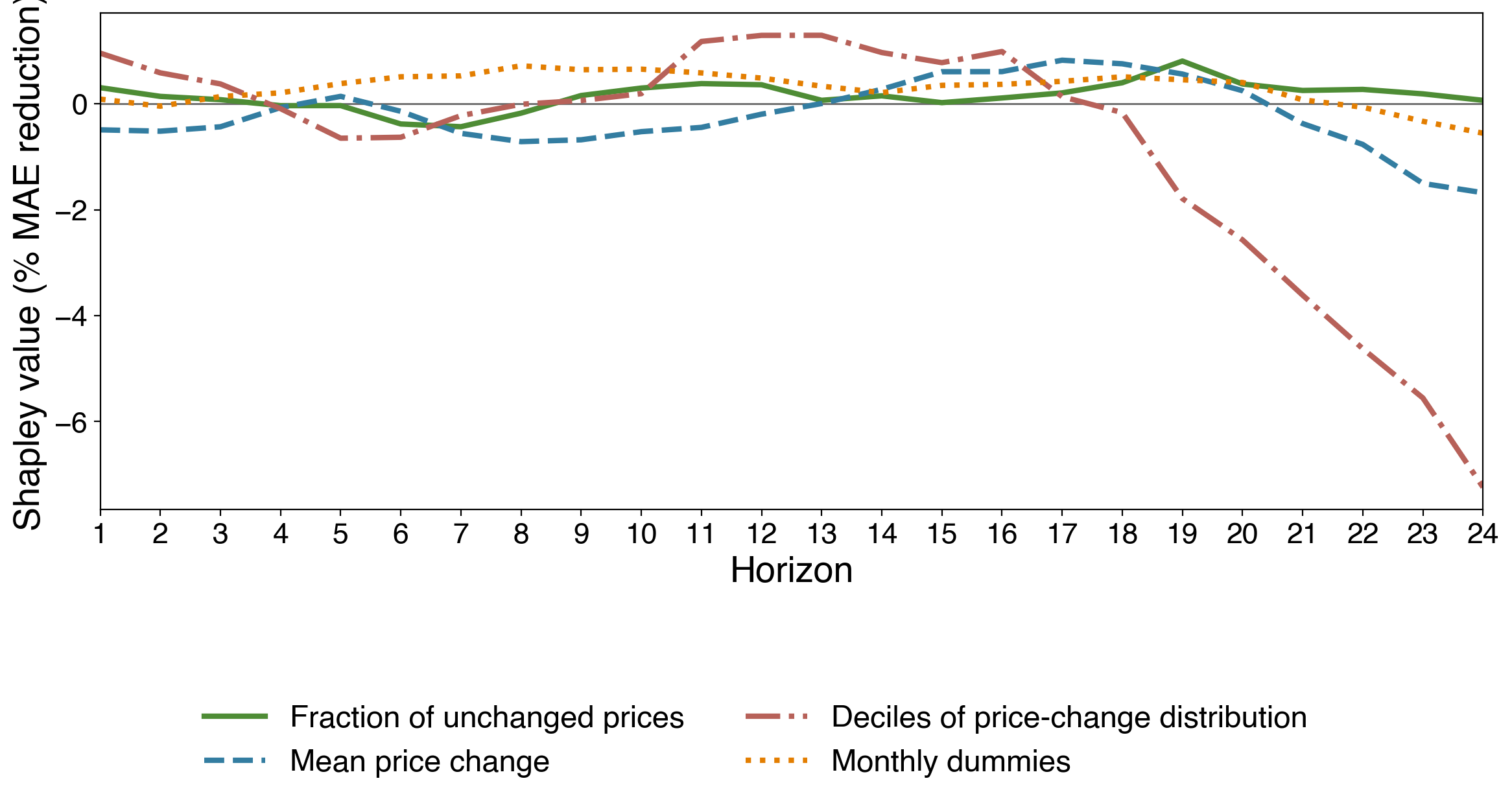}
  \end{subfigure}
  \hfill
  \begin{subfigure}[t]{0.49\textwidth}
    \centering
    \caption{Post-COVID}
    \includegraphics[width=\textwidth]{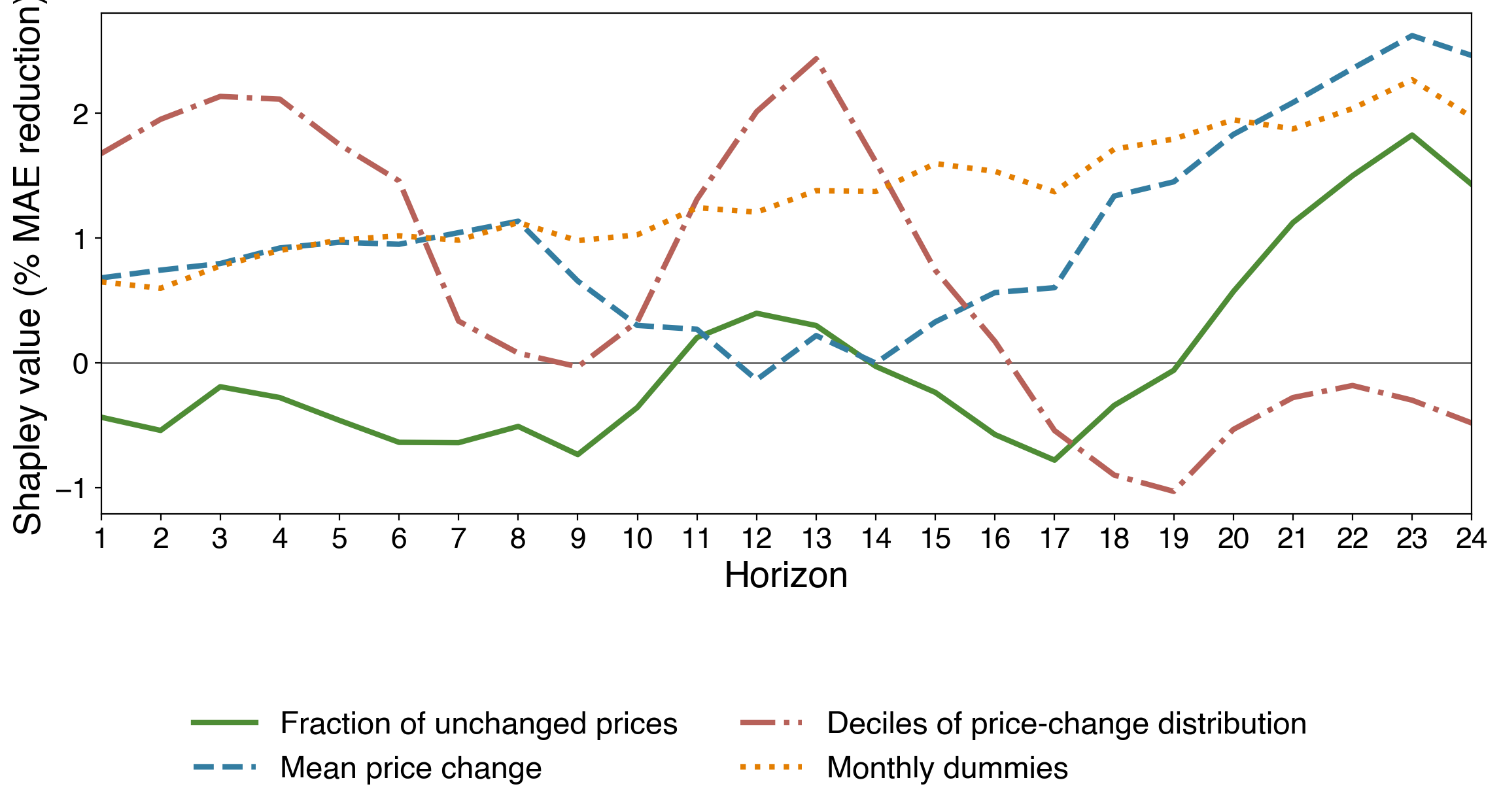}
  \end{subfigure}
  \smallskip
  \begin{minipage}{\textwidth}\setstretch{1.0}
    \footnotesize\textit{Notes:} The figure reports grouped Shapley
    contributions for the incremental value of the micro forecast in the
    three-expert combined forecast relative to the two-expert
    benchmark that combines the univariate and macro forecasts. Values are smoothed
    across adjacent horizons using a 6-horizon moving average. Positive
    values indicate reductions in MAE from the corresponding feature group;
    negative values indicate increases in MAE. The horizontal line marks a
    zero contribution.
  \end{minipage}
\end{figure}

Figure~\ref{fig:micro_shapley_six_horizon_ma_baseline4} confirms the same broad
message at the horizon level: before COVID the contributions are small at
short and medium horizons, with the losses concentrated in the deciles group
at the longest horizons; after
COVID, mean price changes, monthly indicators, and, at short and medium
horizons, the price-change deciles are the
main positive contributors.

\clearpage
\begin{singlespace}
  \putbib
\end{singlespace}
\end{bibunit}

\end{document}